\documentclass[sigconf]{acmart}

\usepackage{comment}
\usepackage{amsthm}
\usepackage{mathtools}
\usepackage{etoolbox}

\usepackage{stmaryrd}
\usepackage[normalem]{ulem}
\usepackage{subcaption}
\usepackage{booktabs}
\usepackage[disable]{todonotes}
\usepackage{graphicx}
\usepackage{listings}
\usepackage{fancyvrb}
\usepackage{caption}
\usepackage{subcaption}
\usepackage{braket}
\usepackage[inline]{enumitem}
\usepackage{xspace}
\usepackage{xstring}
\usepackage{multirow}

\usepackage{url}

\usepackage{breakurl}
\usepackage{cleveref}

\graphicspath{ {figures/} {figs/} }

\usepackage[disable]{todonotes}

\definecolor{white}{rgb}{1,1,1}
\definecolor{black}{rgb}{0,0,0}
\definecolor{grey}{rgb}{0.7,0.7,0.7}
\definecolor{dgrey}{rgb}{0.5,0.5,0.5}
\definecolor{lightgrey}{rgb}{0.88,0.88,0.88}
\definecolor{lgrey}{rgb}{0.9,0.9,0.9}
\definecolor{llgrey}{rgb}{0.93,0.93,0.93}
\definecolor{lllgrey}{rgb}{0.96,0.96,0.96}
\definecolor{tableHeadGray}{rgb}{0.85,0.85,0.85}
\definecolor{oddRowGrey}{rgb}{0.95,0.95,0.95}
\definecolor{evenRowGrey}{rgb}{0.85,0.85,0.85}
\definecolor{yellow}{rgb}{1.0, 1.0, 0.0}
\definecolor{lightyellow}{rgb}{1.0, 1.0, 0.88}
\definecolor{selectiveyellow}{rgb}{1.0, 0.73, 0.0}
\definecolor{shadered}{rgb}{1,0.85,0.85}
\definecolor{red}{rgb}{1,0,0}
\definecolor{shadegreen}{rgb}{0.95,1,0.95}
\definecolor{green}{rgb}{0,1,0}
\definecolor{darkgreen}{rgb}{0,0.5,0}
\definecolor{shadeblue}{rgb}{0.95,0.95,1}
\definecolor{blue}{rgb}{0,0,1}
\definecolor{darkblue}{rgb}{0,0,0.5}
\definecolor{darkpurple}{rgb}{0.5,0,0.5}
\definecolor{darkdarkpurple}{rgb}{0.3,0,0.3}
\newcommand{\BG}[1]{\todo[inline,size=\tiny]{\textbf{Boris says:$\,$} #1}}

\newbool{ShowDetailedProofs}
\newcommand{\detailedproof}[2]{\ifbool{ShowDetailedProofs}{
\begin{proof}
#2
\end{proof}
}{\noindent\textit{Proof Sketch:}#1\qed}\smallskip}
\newrobustcmd{\ifnotdetailedproof}[1]{\ifbool{ShowDetailedProofs}{}{#1}}
\newrobustcmd{\ifnottechreport}[1]{\ifbool{ShowDetailedProofs}{}{#1}}
\newrobustcmd{\iftechreport}[1]{\ifbool{ShowDetailedProofs}{#1}{}}

\newcommand{\mypar}[1]{\smallskip\noindent{\bf #1.}}
\newcommand{\bfcaption}[1]{\caption{#1}}
\newcommand{\trimfigurespacing}{\vspace*{-5mm}}

\DeclareMathAlphabet{\mathbbold}{U}{bbold}{m}{n}

\newtheorem{Theorem}{Theorem}
\newtheorem{Definition}{Definition}
\newtheorem{Lemma}{Lemma}

\newtheorem{Corollary}{Corollary}
\newtheorem{Example}{Example}

\newcommand{\proofpara}[1]{\medskip\noindent\underline{{#1}:}}

\newcommand{\mathtext}[1]{\ensuremath{\,\text{#1}\,}}
\newcommand{\card}[1]{\vert{#1}\vert}
\newcommand{\defas}{\coloneqq}
\newcommand{\mathtab}{\;\;\;}

\DeclareMathAlphabet{\mathbbold}{U}{bbold}{m}{n}

\newcommand{\thead}[1]{\textbf{#1}}
\newcommand{\tabletitle}[1]{\underline{{#1}}}

\newcommand{\schemasymb}{\textsc{Sch}}
\newcommand{\aschema}{S}
\newcommand{\schema}{\schemasymb(D)}
\newcommand{\relschema}{\schemasymb(R)}
\newcommand{\schemaOf}{\textsc{Sch}}
\newcommand{\arity}[1]{arity({#1})}
\newcommand{\db}{D}
\newcommand{\rel}{R}
\newcommand{\query}{Q}

\newcommand{\tup}{t}

\newcommand{\aDom}{\mathbb{D}}

\newcommand{\tuple}[1]{\left<\;{#1}\;\right>}
\newcommand{\comprehension}[2]{\left\{\;#1\;|\;#2\;\right\}}

\newcommand{\selection}{\sigma}
\newcommand{\projection}{\pi}
\newcommand{\join}{\Join}
\newcommand{\crossprod}{\times}
\newcommand{\union}{\cup}

\newcommand{\difference}{-}
\newcommand{\rename}{\rho}
\newcommand{\aggregation}[2]{%
  \IfEqCase{#1}{%
    {}{\gamma_{{#2}}}%
    }[\gamma_{{#1},{#2}}]%
  }

\newcommand{\gbAttrs}{G}
\newcommand{\agroup}{g}
\newcommand{\aout}{\rangeTup_{\asgrp}}
\newcommand{\outof}[1]{\rangeTup_{{#1}}}
\newcommand{\duprem}{\delta}
\newcommand{\combine}{\Psi}
\newcommand{\splitbg}{\bgMarker{\textsc{split}}}
\newcommand{\splitpos}{\ubMarker{\textsc{split}}}
\newcommand{\compressPos}{\textsc{Cpr}}

\newcommand{\raPlus}{\ensuremath{\mathcal{RA}^{+}}}

\newcommand{\raAgg}{\ensuremath{\mathcal{RA}^{agg}}}

\newcommand{\abbrBG}{SG\xspace}

\newcommand{\abbrBGW}{SGW\xspace}
\newcommand{\abbrBGQP}{SGQP\xspace}

\newcommand{\abbrUADB}{UA-DB\xspace}
\newcommand{\abbrUADBs}{UA-DBs\xspace}

\newcommand{\termUAADBs}{attribute-annotated uncertain databases\xspace}
\newcommand{\captialUAADBs}{Attribute-Annotated Uncertain Databases\xspace}
\newcommand{\abbrUAADB}{\textsf{AU-DB}\xspace}
\newcommand{\abbrUAADBs}{\textsf{AU-DBs}\xspace}
\newcommand{\abbrUAARel}{AU-relation\xspace}
\newcommand{\abbrAUDB}{\abbrUAADB}
\newcommand{\abbrAUDBs}{\abbrUAADBs}
\newcommand{\abbrTI}{TI-DB\xspace}
\newcommand{\abbrTIs}{TI-DBs\xspace}
\newcommand{\abbrPTIrel}{probabilistic TI-relation\xspace}
\newcommand{\termTIs}{tuple-independent databases\xspace}

\newcommand{\captialTIs}{tuple-independent database\xspace}

\newcommand{\termBI}{block-independent database\xspace}
\newcommand{\abbrCoddTable}{Codd-table\xspace}

\newcommand{\abbrVtables}{V-tables\xspace}
\newcommand{\abbrCtable}{C-table\xspace}
\newcommand{\abbrCtables}{C-tables\xspace}

\newcommand{\abbrVCtable}{Virtual C-table\xspace}
\newcommand{\abbrVCtables}{Virtual C-tables\xspace}

\newcommand{\abbrXDB}{x-DB\xspace}
\newcommand{\abbrXDBs}{x-DBs\xspace}
\newcommand{\abbrXtable}{x-table\xspace}

\newcommand{\abbrMtables}{m-tables\xspace}

\newcommand{\termBG}{selected-guess\xspace}
\newcommand{\capitalBG}{Selected-Guess\xspace}

\newcommand{\termBGW}{selected-guess world\xspace}
\newcommand{\capitalBGW}{Selected-Guess Worlds\xspace}

\newcommand{\st}{s.t.}

\newcommand{\CT}{CT\xspace}

\newcommand{\ptime}{\texttt{PTIME}\xspace}
\newcommand{\npcomplete}{NP-complete\xspace}
\newcommand{\conpcomplete}{coNP-complete\xspace}
\newcommand{\nphard}{NP-hard\xspace}
\newcommand{\conphard}{coNP-hard\xspace}

\newcommand{\bigO}{\mathcal{O}}

\newcommand{\pdb}{\mathcal{D}}
\newcommand{\prel}{\mathcal{R}}
\newcommand{\bestG}[1]{\textsc{BestGuess}(#1)}
\newcommand{\uadb}{\db_{UA}}

\newcommand{\bgdb}{D_{sg}}
\newcommand{\prob}{P}

\newcommand{\uMarker}[1]{\textcolor{red}{{#1}}}
\newcommand{\cMarker}[1]{#1}
\newcommand{\ubMarker}[1]{\ensuremath{{#1}{}^{\uparrow}}}
\newcommand{\lbMarker}[1]{\ensuremath{{#1}{}^{\downarrow}}}

\newcommand{\bgMarker}[1]{\ensuremath{{#1}{}^{\bgName}}}

\newcommand{\TUL}{\mathcal{L}}
\newcommand{\TTR}{\mathcal{T}}
\newcommand{\TTRTI}{\TTR_{TI}}
\newcommand{\TTRX}{\TTR_{X}}
\newcommand{\TTRC}{\TTR_{C}}

\newcommand{\doTTR}{\textsf{trans}}

\newcommand{\doTTRTI}{\doTTR_{\textsc{\abbrTI}}}
\newcommand{\doTTRX}{\doTTR_{\textsc{\abbrXDB}}}

\newcommand{\doTTRC}{\doTTR_{\textsc{\abbrCtable}}}

\newcommand{\doTTRD}{\doTTR_{\textsc{certain}}}

\newcommand{\makeuncert}[3]{\textsc{MakeUncertain}(#1,#2,#3)}
\newcommand{\makeuncertName}{\ensuremath{\textsc{MakeUncertain}}}

\newcommand{\thenull}{\ensuremath{\mathbf{null}}}
\newcommand{\lcond}{\phi}
\newcommand{\gcond}{\Phi}

\newcommand{\xTup}{\tau}

\newcommand{\xselect}[1]{\textsc{pickMax}(#1)}

\newcommand{\orval}[1]{[{#1}]}

\newcommand{\dataDomain}{\mathbb{D}}
\newcommand{\sexpr}{e}
\newcommand{\vars}{\textsc{vars}}
\newcommand{\seval}[2]{\llbracket{#1}\rrbracket_{#2}}
\newcommand{\sval}{\varphi}
\newcommand{\uval}{{\Phi}}
\newcommand{\rval}{\tilde{\varphi}}
\newcommand{\ifte}[3]{{\bf if}\,{#1}\,{\bf then}\,{#2}\,{\bf else}\,{#3}}

\newcommand{\semN}{\mathbb{N}}
\newcommand{\semB}{\mathbb{B}}
\newcommand{\semNX}{\mathbb{N}[X]}

\newcommand{\semK}{\mathcal{K}}

\newcommand{\onesymbol}{\mathbbold{1}}
\newcommand{\zerosymbol}{\mathbbold{0}}
\newcommand{\multsymb}{\cdot}
\newcommand{\addsymbol}{+}
\newcommand{\monsymbol}{-}

\newcommand{\bfalse}{\bot}
\newcommand{\btrue}{\top}

\newcommand{\kDom}{K}

\newcommand{\addK}{\addsymbol_{\semK}}
\newcommand{\multK}{\multsymb_{\semK}}
\newcommand{\monK}{\monsymbol_{\semK}}
\newcommand{\oneK}{\onesymbol_{\semK}}
\newcommand{\zeroK}{\zerosymbol_{\semK}}
\newcommand{\addOf}[1]{\addsymbol_{#1}}
\newcommand{\multOf}[1]{\multsymb_{#1}}
\newcommand{\monOf}[1]{\monsymbol_{#1}}
\newcommand{\oneOf}[1]{\onesymbol_{#1}}
\newcommand{\zeroOf}[1]{\zerosymbol_{#1}}
\newcommand{\zeroN}{0}
\newcommand{\oneN}{1}
\newcommand{\addN}{+}
\newcommand{\multN}{\cdot}

\newcommand{\ordersymbol}{\preceq}
\newcommand{\geqsymbol}{\succeq}
\newcommand{\ordK}{\ordersymbol_{\semK}}
\newcommand{\geqK}{\geqsymbol_{\semK}}
\newcommand{\ordOf}[1]{\ordersymbol_{#1}}
\newcommand{\ordB}{\ordOf{\semB}}
\newcommand{\ordN}{\ordOf{\semN}}

\newcommand{\lub}{\sqcup}
\newcommand{\glb}{\sqcap}

\newcommand{\glbKof}[1]{\sqcap_{#1}}
\newcommand{\glbK}{\glbKof{\semK}}
\newcommand{\lubKof}[1]{\sqcup_{#1}}
\newcommand{\lubK}{\lubKof{\semK}}

\newcommand{\GlbK}{\sqcap_{\semK}}

\newcommand{\dpK}[2]{{#1}^{#2}}
\newcommand{\dpDom}[2]{{#1}^{#2}}
\newcommand{\doubleK}[1]{\dpK{#1}{2}}
\newcommand{\aDoubleK}{\doubleK{\semK}}
\newcommand{\doubleDom}[1]{\dpDom{#1}{2}}
\newcommand{\doubleN}{\dpK{\semN}{2}}

\newcommand{\semq}[1]{\ensuremath{\dpK{#1}{3}}}
\newcommand{\semkq}{\semq{\semK}}
\newcommand{\semqN}{\semq{\semN}}

\newcommand{\monoid}{M}
\newcommand{\madd}[1]{+_{#1}}
\newcommand{\mzero}[1]{\zeroOf{#1}}
\newcommand{\amadd}{\madd{\monoid}}
\newcommand{\amzero}{\mzero{\monoid}}
\newcommand{\mmax}{\text{\upshape\textsf{MAX}}}
\newcommand{\mmin}{\text{\upshape\textsf{MIN}}}
\newcommand{\msum}{\text{\upshape\textsf{SUM}}}

\newcommand{\rangeMon}{{\monoid}_{I}}
\newcommand{\rangeMof}[1]{{#1}_{I}}

\newcommand{\aggmin}{\ensuremath{\mathbf{min}}\xspace}
\newcommand{\aggmax}{\ensuremath{\mathbf{max}}\xspace}
\newcommand{\aggsum}{\ensuremath{\mathbf{sum}}\xspace}
\newcommand{\aggavg}{\ensuremath{\mathbf{avg}}\xspace}
\newcommand{\aggcount}{\ensuremath{\mathbf{count}}\xspace}

\newcommand{\ggrouping}{\mathcal{G}}
\newcommand{\tgrouping}{\mathbb{g}}
\newcommand{\uncertg}[3]{\textsc{ug}({#1},{#2},{#3})}
\newcommand{\lbagg}[1]{\textsc{lbagg}({#1})}
\newcommand{\ubagg}[1]{\textsc{ubagg}({#1})}

\newcommand{\gstrat}{\mathbb{G}}
\newcommand{\stratGrps}{\mathcal{G}}
\newcommand{\asgrp}{g}
\newcommand{\stratBG}{\psi}
\newcommand{\stratPoss}{\alpha}
\newcommand{\gsdef}{\gstrat_{def}}
\newcommand{\gsdefG}{\stratGrps}

\newcommand{\gsdefP}{\stratPoss}

\newcommand{\smbpair}{\ast}
\newcommand{\smbN}[1]{\smbpair_{{#1}}}
\newcommand{\asmbN}{\smbN{\monoid}}
\newcommand{\smbNAU}[1]{\smbpair_{\uaaK{\semN},{#1}}}
\newcommand{\asmbNAU}{\smbNAU{\monoid}}
\newcommand{\mysmbNAU}[1]{\circledast_{{#1}}}
\newcommand{\amysmbNAU}{\mysmbNAU{\monoid}}

\newcommand{\smpair}{\otimes}
\newcommand{\sm}[2]{{#1} \smpair {#2}}
\newcommand{\asm}{\sm{\semK}{\monoid}}
\newcommand{\kmadd}[2]{\madd{\sm{#1}{#2}}}
\newcommand{\akmadd}{\kmadd{\semK}{\monoid}}

\newcommand{\uaName}{UA}
\newcommand{\uaK}[1]{{#1}_{\uaName}}

\newcommand{\matches}{\simeq}

\newcommand{\bgName}{sg}
\newcommand{\bgOf}[1]{\bgMarker{#1}}

\newcommand{\rangeName}{I}
\newcommand{\rangeDom}{{\dataDomain_{\rangeName}}}
\newcommand{\rangeTup}{\textbf{\sffamily \tup}}
\newcommand{\rangeRel}{\mathbf{\rel}}
\newcommand{\rangeOf}[1]{\mathbf{{#1}}}
\newcommand{\rangeDB}{{\mathbf{\db}}}

\newcommand{\tmatch}{\sqsubseteq}
\newcommand{\ntmatch}{\not\sqsubseteq}
\newcommand{\toverlaps}{\sqcap}
\newcommand{\dbbounds}{\sqsubset}
\newcommand{\rliftK}[1]{\mathcal{M}_{#1}}

\newcommand{\uaaName}{AU}
\newcommand{\uaaK}[1]{{#1}_{\uaaName}}
\newcommand{\uaaDom}[1]{\kDom_{{\uaaName}}}

\newcommand{\uaaAdd}[1]{\addOf{\uaaK{#1}}}
\newcommand{\uaaMult}[1]{\multOf{\uaaK{#1}}}
\newcommand{\uaaZero}[1]{\zeroOf{\uaaK{#1}}}

\newcommand{\uae}[2]{[#2,#1]}
\newcommand{\uv}[3]{\ensuremath{[{#1}/{#2}/{#3}]}}
\newcommand{\ut}[3]{({#1},{#2},{#3})}
\newcommand{\uaaN}{\uaaK{\semN}}
\newcommand{\uaanAdd}{\uaaAdd{\semN}}
\newcommand{\uaanMult}{\uaaMult{\semN}}
\newcommand{\uaaNzero}{\uaaZero{\semN}}

\newcommand{\dbleq}{\preceq_{\rangeName}}
\newcommand{\dble}{\prec_{\rangeName}}

\newcommand{\dbge}{\succ_{\rangeName}}

\newcommand{\certainName}{\textsc{cert}}
\newcommand{\possibleName}{\textsc{poss}}

\newcommand{\pwCertain}{{\certainName}_{\semK}}

\newcommand{\pwPossible}{{\possibleName}_{\semK}}
\newcommand{\pwCertOf}[1]{{\certainName}_{#1}}

\newcommand{\pwPossOf}[1]{{\possibleName}_{#1}}
\newcommand{\pwCertainN}{{\certainName}_{\semN}}

\newcommand{\pwPossibleN}{{\possibleName}_{\semN}}

\newcommand{\TM}{\mathcal{TM}}

\newcommand{\rewrUAA}[1]{\textsc{rewr}({#1})}
\newcommand{\rewrOpt}[1]{\textsc{opt}({#1})}
\newcommand{\qmerge}{\query_{merge}}

\newcommand{\tenc}[1]{\textsc{enc}({#1})}
\newcommand{\tdec}[1]{\textsc{dec}({#1})}
\newcommand{\tdecr}[1]{\textsc{rowdec}({#1})}
\newcommand{\Enc}{\textsc{Enc}}
\newcommand{\Dec}{\textsc{Dec}}

\newcommand{\ubatt}[1]{\ubMarker{{#1}}}
\newcommand{\lbatt}[1]{\lbMarker{{#1}}}
\newcommand{\bgatt}[1]{\bgMarker{{#1}}}
\newcommand{\ratt}{row}
\newcommand{\rattj}{row2}
\newcommand{\rbg}{\bgMarker{\ratt}}
\newcommand{\rub}{\ubMarker{\ratt}}
\newcommand{\rlb}{\lbMarker{\ratt}}
\newcommand{\renameto}{\rightarrow}

\DeclareMathSymbol{\mlq}{\mathord}{operators}{``}
\DeclareMathSymbol{\mrq}{\mathord}{operators}{`'}

\newlength{\eqheight}
\settototalheight{\eqheight}{$=$}

\booltrue{ShowDetailedProofs}

\title{Efficient Uncertainty Tracking for Complex Queries with Attribute-level Bounds}
\subtitle{(Accompanying Technical Report)}

\author{Su Feng}
\affiliation{\institution{Illinois Institute of Technology}}
\email{sfeng14@hawk.iit.edu}

\author{Aaron Huber}
\affiliation{\institution{SUNY Buffalo}}
\email{ahuber@buffalo.edu}

\author{Boris Glavic}
\affiliation{\institution{Illinois Institute of Technology}}
\email{bglavic@iit.edu}

\author{Oliver Kennedy}
\affiliation{\institution{SUNY Buffalo}}
\email{okennedy@buffalo.edu}

\definecolor{lstpurple}{rgb}{0.5,0,0.5}
\definecolor{lstred}{rgb}{1,0,0}
\definecolor{lstreddark}{rgb}{0.7,0,0}
\definecolor{lstredl}{rgb}{0.64,0.08,0.08}
\definecolor{lstmildblue}{rgb}{0.66,0.72,0.78}
\definecolor{lstblue}{rgb}{0,0,1}
\definecolor{lstmildgreen}{rgb}{0.42,0.53,0.39}
\definecolor{lstgreen}{rgb}{0,0.5,0}
\definecolor{lstorangedark}{rgb}{0.6,0.3,0}
\definecolor{lstorange}{rgb}{0.75,0.52,0.005}
\definecolor{lstorangelight}{rgb}{0.89,0.81,0.67}
\definecolor{lstbeige}{rgb}{0.90,0.86,0.45}

\DeclareFontShape{OT1}{cmtt}{bx}{n}{<5><6><7><8><9><10><10.95><12><14.4><17.28><20.74><24.88>cmttb10}{}

\lstdefinestyle{psql}
{
tabsize=2,
basicstyle=\small\upshape\ttfamily,
language=SQL,
morekeywords={PROVENANCE,BASERELATION,INFLUENCE,COPY,ON,TRANSPROV,TRANSSQL,TRANSXML,CONTRIBUTION,COMPLETE,TRANSITIVE,NONTRANSITIVE,EXPLAIN,SQLTEXT,GRAPH,IS,ANNOT,THIS,XSLT,MAPPROV,cxpath,OF,TRANSACTION,SERIALIZABLE,COMMITTED,INSERT,INTO,WITH,SCN,UPDATED},
extendedchars=false,
keywordstyle=\bfseries,
mathescape=true,
escapechar=@,
sensitive=true
}

\lstdefinestyle{psqlcolor}
{
tabsize=2,
basicstyle=\small\upshape\ttfamily,
language=SQL,
morekeywords={PROVENANCE,BASERELATION,INFLUENCE,COPY,ON,TRANSPROV,TRANSSQL,TRANSXML,CONTRIBUTION,COMPLETE,TRANSITIVE,NONTRANSITIVE,EXPLAIN,SQLTEXT,GRAPH,IS,ANNOT,THIS,XSLT,MAPPROV,cxpath,OF,TRANSACTION,SERIALIZABLE,COMMITTED,INSERT,INTO,WITH,SCN,UPDATED},
extendedchars=false,
keywordstyle=\bfseries\color{lstpurple},
deletekeywords={count,min,max,avg,sum},
keywords=[2]{count,min,max,avg,sum},
keywordstyle=[2]\color{lstblue},
stringstyle=\color{lstreddark},
commentstyle=\color{lstgreen},
mathescape=true,
escapechar=@,
sensitive=true
}

\lstdefinestyle{datalog}
{
basicstyle=\footnotesize\upshape\ttfamily,
language=prolog
}

\lstdefinestyle{pseudocode}
{
  tabsize=3,
  basicstyle=\small,
  language=c,
  morekeywords={if,else,foreach,case,return,in,or},
  extendedchars=true,
  mathescape=true,
  literate={:=}{{$\gets$}}1 {<=}{{$\leq$}}1 {!=}{{$\neq$}}1 {append}{{$\listconcat$}}1 {calP}{{$\cal P$}}{2},
  keywordstyle=\color{lstpurple},
  escapechar=&,
  numbers=left,
  numberstyle=\color{lstgreen}\small\bfseries,
  stepnumber=1,
  numbersep=5pt,
}

\lstdefinestyle{xmlstyle}
{
  tabsize=3,
  basicstyle=\small\upshape\ttfamily,
  language=xml,
  extendedchars=true,
  mathescape=true,
  escapechar=£,
  tagstyle=\bfseries,
  usekeywordsintag=true,
  morekeywords={alias,name,id},
  keywordstyle=\color{lstred}
}

\lstdefinestyle{xmlstyle-color}
{
  tabsize=3,
  basicstyle=\small\upshape\ttfamily,
  language=xml,
  extendedchars=true,
  mathescape=true,
  escapechar=£,
  tagstyle=\color{keywordpurple},
  usekeywordsintag=true,
  morekeywords={alias,name,id},
  keywordstyle=\color{lstred}
}

\renewcommand{\shortauthors}{Feng, et al.}

\keywords{uncertainty, incomplete databases, annotations, aggregation}

\newcommand{\reva}[1]{#1}
\newcommand{\revb}[1]{#1}
\newcommand{\revc}[1]{#1}
\newcommand{\revm}[1]{#1}

\begin{document}

\begin{abstract}
  \ifnottechreport{ Incomplete and probabilistic database techniques are
    principled methods for coping with uncertainty in data.  Unfortunately, 
    the class of queries that can be answered efficiently over such databases is
    severely limited, even when advanced approximation techniques are employed.
    We introduce \emph{\termUAADBs} (\abbrUAADBs), an uncertain data model that
    annotates tuples and attribute values with bounds to compactly approximate
    an incomplete database.  \abbrUAADBs are closed under 
    relational
    algebra with aggregation using an efficient evaluation semantics. Using
    optimizations that trade accuracy for performance, our approach scales to
    complex queries and large datasets, and produces accurate
    results.  
  }
  \iftechreport{ Certain answers are a principled method for coping with the
    uncertainty that arises in many practical data management
    tasks. Unfortunately, this method is expensive and may exclude useful (if
    uncertain) answers. Prior work introduced \emph{Uncertainty Annotated
      Databases} (\abbrUADBs), which combine an under- and over-approximation of
    certain answers. \abbrUADBs combine the reliability of certain answers based
    on incomplete K-relations with the performance of classical deterministic
    database systems. However, \abbrUADBs only support a limited class of
    queries and do not support attribute-level uncertainty which can lead to
    inaccurate under-approximations of certain answers.  In this paper, we
    introduce \emph{\termUAADBs} (\abbrUAADBs) which extend the \abbrUADB model
    with attribute-level annotations that record bounds on the values of an
    attribute across all possible worlds. This enables more precise
    approximations of incomplete databases. Furthermore, we extend \abbrUADBs to
    encode an compact over-approximation of possible answers which is necessary
    to support non-monotone queries including aggregation and set difference.
    We prove that query processing over \abbrUAADBs preserves the bounds on
    certain and possible answers and investigate algorithms for compacting
    intermediate results to retain efficiency. Through an compact encoding of
    possible answers, our approach also provides a solid foundation for handling
    missing data.  Using optimizations that trade accuracy for performance, our
    approach scales to complex queries and large datasets, and produces accurate
    results. Furthermore, it significantly outperforms alternative methods for
    uncertain data management.  }
\end{abstract}


\maketitle
\section{Introduction}
\label{sec:introduction}

Uncertainty arises naturally in many application domains due to data entry errors, sensor errors and noise~\cite{jeffery-06-dssdc}, uncertainty in information extraction and parsing~\cite{sarawagi2008information}, ambiguity from data integration~\cite{OP13,AS10,HR06a}, and heuristic data wrangling~\cite{Yang:2015:LOA:2824032.2824055,F08,Beskales:2014:SRC:2581628.2581635}.
Analyzing uncertain data without accounting for its uncertainty can create hard to trace errors with severe real world implications. 
Incomplete database techniques~\cite{DBLP:conf/pods/ConsoleGLT20} have emerged as a principled way to model and manage uncertainty in data\footnote{
\revb{
  Probabilistic databases~\cite{suciu2011probabilistic} generalize incomplete databases with a probability distribution over possible worlds.
  We focus on contrasting with the former for simplicity, but many of the same cost and expressivity limitations also affect probabilistic databases.
}
}.
An incomplete database models uncertainty by encoding a set of possible worlds, each of which is one possible state of the real world.
Under the commonly used \emph{certain answer semantics}~\cite{AK91,DBLP:journals/jacm/ImielinskiL84}, a \revm{query} 
returns the set of answer tuples \emph{guaranteed} to be in the result, regardless of which possible world is correct.
Many computational problems are intractable over incomplete databases. Even approximations (e.g.,~\cite{GP17,DBLP:journals/vldb/FinkHO13,DBLP:conf/pods/KoutrisW18}) are often still not efficient enough, are insufficiently expressive, or exclude useful answers~\cite{FH19,DBLP:conf/pods/ConsoleGLT20}.
%
Thus, typical database users resort to a cruder, but more practical alternative: resolving uncertainty using heuristics and then treating the result as a deterministic database~\cite{Yang:2015:LOA:2824032.2824055}.
In other words, this approach selects one possible world for analysis, ignoring all other possible worlds.
\BG{This is, for example, how uncertainty is addressed in typical ETL processes~\cite{Yang:2015:LOA:2824032.2824055}.}
We refer to this 
approach as \textit{selected-guess} query processing (\emph{\abbrBGQP}).
\abbrBGQP is efficient, since the resulting dataset is deterministic, but discards all information about uncertainty, with the associated potential for severe negative consequences.

\begin{figure*}[t]
  \begin{minipage}{0.7\linewidth}
  \centering
  \begin{subfigure}{0.28\linewidth}
    \centering
    \begin{minipage}{1\linewidth}
      \centering
    \tabletitle{$\pdb$}                                 \\
  {\tiny
  \begin{tabular}{c|c|c}
    \thead{locale} & \thead{rate}  & \thead{size}  \\ \hline
    Los Angeles   & \orval{3\%,4\%}     & metro             \\
    Austin       & 18\%                & \orval{city,metro}             \\
    Houston      & 14\%                & metro             \\
    Berlin       & \orval{1\%,3\%} & \orval{town,city} \\
    Sacramento   & 1\%                 & \thenull       \\
    Springfield  & \thenull            & town             \\
  \end{tabular}
}
\end{minipage}
\begin{minipage}{1\linewidth}
      \centering
    \tabletitle{$\query(\pdb)$}  \\
  {\tiny
  \begin{tabular}{c|cc}
    \thead{size}  & \thead{rate} \\ \hline
    village       & 0\%    \\
    village       & 1\%    \\
    town          & 0\%        \\
    town          & 0.5\%        \\
    \ldots        & \ldots       \\
    metro         & 12\%      \\
  \end{tabular}
}
\end{minipage}
\caption{\label{fig:running-example-xdb}X-DB}
\end{subfigure}
\begin{subfigure}{0.2\linewidth}
\begin{minipage}{1\linewidth}
  \centering
  \tabletitle{$\db_{\abbrBG}$}\\
{\tiny
    \begin{tabular}{c|c|c}
    \thead{locale} & \thead{rate} & \thead{size} \\ \hline
    Los Angeles   & 3\%         & metro           \\
    Austin       & 18\%        & city            \\
    Houston      & 14\%        & metro           \\
      Berlin     & 3\%         & town            \\
      Sacramento & 1\%         & town            \\
    Springfield  & 5\%         & town            \\
  \end{tabular}
}
\end{minipage}
  \begin{minipage}{1\linewidth}
    \centering
  \tabletitle{$\query(\db_{\abbrBG})$}   \\
    {\tiny
  \begin{tabular}{c|cc}
    \thead{size} & \thead{rate} \\ \hline
    metro            & 8.5\%       \\
    city             & 18\%       \\
    town             & 3\%         \\
  \end{tabular}
}
\end{minipage}
\caption{\label{fig:running-example-pos-worlds}A possible world}
\end{subfigure}
%
\begin{subfigure}{0.5\linewidth}
  \centering
  \begin{minipage}{1\linewidth}
    \centering
\tabletitle{$\db_{AU}$}\\
{\tiny
  \begin{tabular}{c|c|cc}
    \thead{locale} & \thead{rate}         & \thead{size}    & \underline{\semqN} \\ \cline{1-3}
    Los Angeles   & \uv{3\%}{3\%}{4\%}     & metro           & \ut{1}{1}{1}       \\
    Austin       & 18\%                   & \uv{city}{city}{metro}            & \ut{1}{1}{1}       \\
    Houston      & 14\%                   & metro           & \ut{1}{1}{1}       \\
    Berlin       & \uv{1\%}{3\%}{3\%}     & \uv{town}{town}{city} & \ut{1}{1}{1}       \\
    Sacramento   & 1\%                    & \uv{village}{town}{metro}            & \ut{1}{1}{1}       \\
    Springfield  & \uv{0\%}{5\%}{100\%}   & town & \ut{1}{1}{1}       \\
  \end{tabular}
}
\end{minipage}
\begin{minipage}{1\linewidth}
  \centering
\tabletitle{$\query(\db_{AU})$}                                   \\
  {\tiny
  \begin{tabular}{c|cc}
    \thead{size}   & \thead{spop}             & \underline{\semqN} \\ \cline{1-2}
    metro          & \uv{6\%}{8.5\%}{12\%}    & \ut{1}{1}{1} \\
    city           & \uv{7.33\%}{18\%}{18\%} & \ut{0}{1}{1} \\
    town           & \uv{0.33\%}{4\%}{100\%}  & \ut{1}{1}{1} \\
    \uv{village}{village}{metro} & 1\%       & \ut{0}{0}{1} \\
  \end{tabular}
  }
\end{minipage}
\caption{\label{fig:running-example-audb}Possible \abbrUAADB Encoding (based on $\db_{\abbrBG}$)}
\end{subfigure}\\[-5mm]
  \caption{\revb{Example incomplete database and query results.}}
  \label{fig:running-example}
\end{minipage}
  \begin{minipage}{0.295\linewidth}
  \centering
  \includegraphics[width=1.0\linewidth]{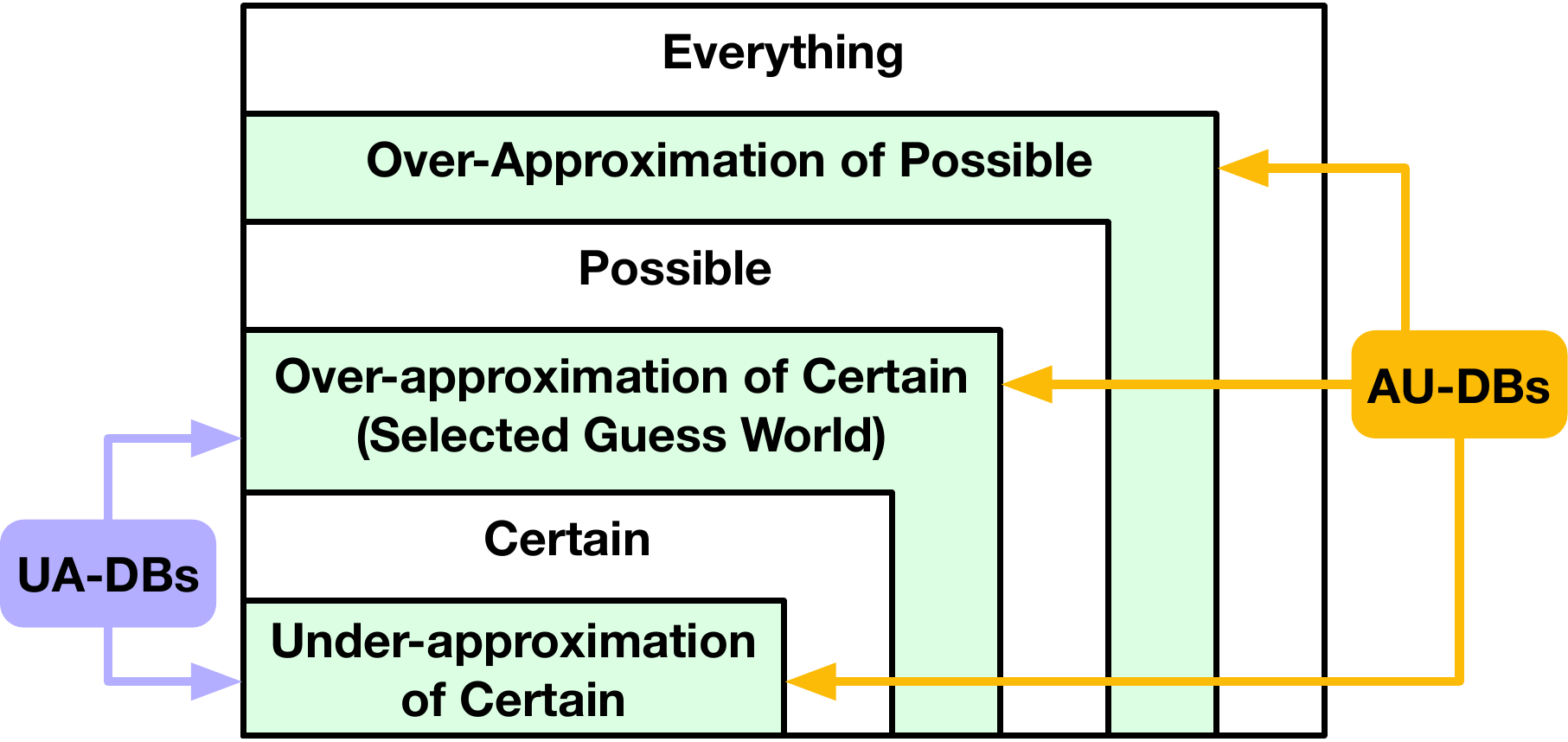}
  \caption{\abbrUAADBs sandwich certain answers between an under-approximation and the \abbrBGW and over-approximate possible answers. \label{fig:uadb-sandwiching}}
\end{minipage}
\end{figure*}


\begin{figure}[ht]
  \centering
  \vspace{-3mm}
  \includegraphics[width=6cm]{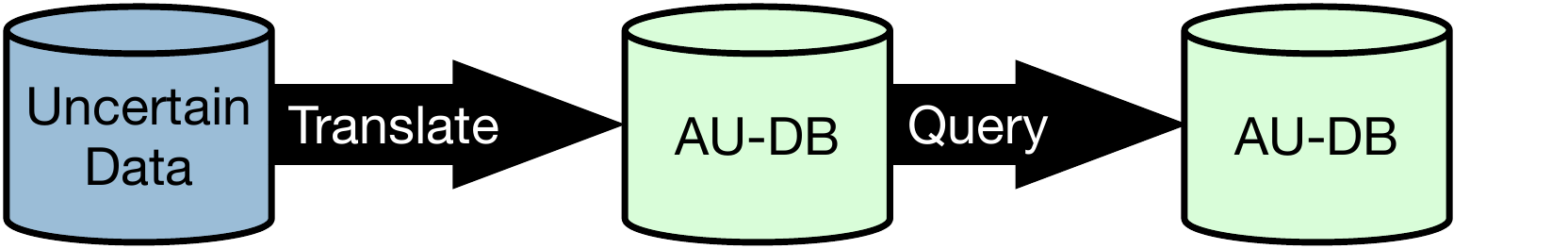}\\[-3mm]
  \caption{\label{fig:audb-ingest-and-querying} \abbrAUDBs are created from uncertain data, possibly represented using 
    incomplete or probabilistic data models. 
  }
\end{figure}

\begin{Example}
\revb{
Alice is tracking the spread of COVID-19 and wants to use data extracted from the web to compare infection rates in population centers of varying size.
\Cref{fig:running-example-xdb} (top) shows example (unreliable) input data.
Parts of this data are trustworthy, while other parts are ambiguous; $[v_1, \ldots, v_n]$ denotes an uncertain value (e.g., conflicting data sources) and $\thenull$ indicates that the value is completely unknown (i.e., any value from the attribute's domain could be correct). 
$\pdb$ encodes a set of possible worlds, each a deterministic database that represents one possible state of the real world.
%
Alice's ETL heuristics select (e.g., based on the relative trustworthiness of each source) one possible world $\db_{\abbrBG}$ (\Cref{fig:running-example-pos-worlds}) by selecting a deterministic value for each ambiguous input (e.g., an infection rate of 3\% for Los Angeles).
Alice next computes the average rate by locale size.
}
{\upshape
\begin{lstlisting}
SELECT size, avg(rate) AS rate
FROM locales GROUP BY size
\end{lstlisting}
}
\revb{
Querying $\db_{\abbrBG}$ 
may produce misleading results, (e.g., an 18\% average infection rate for cities).
Conversely, querying $\pdb$ 
using 
\emph{certain answer semantics} produces no results at all.
Although there must exist a result tuple for metros, the uncertain infection rate of Los Angeles makes it impossible to compute one certain result tuple.
Furthermore, the data lacks a size for Sacramento, which can contribute to any result, rendering all \lstinline!rate! values uncertain, 
even for result tuples with otherwise perfect data. 
An alternative is the 
\emph{possible answer semantics}, which enumerates all possible results.
However, the number of possible results is inordinately large (e.g., \Cref{fig:running-example-xdb}, bottom).
With only integer percentages there are nearly 600 possible result tuples for towns alone.
Worse, enumerating either the (empty) certain or the (large)  possible results is expensive (\conphard/\nphard).
}
\end{Example}

Neither certain answers nor possible answer semantics are meaningful for aggregation over uncertain data (e.g., see \cite{DBLP:conf/pods/ConsoleGLT20} for a deeper discussion), further encouraging the (mis-)use of \abbrBGQP.
One possible solution is to  develop a special query semantics for aggregation, either returning hard bounds on aggregate results (e.g., ~\cite{DBLP:journals/tcs/ArenasBCHRS03,DBLP:conf/pods/AfratiK08,DBLP:journals/tkde/MurthyIW11}), or computing expectations (e.g.,~\cite{DBLP:journals/tkde/MurthyIW11,5447879}) when probabilities are available.
Unfortunately, for such approaches, aggregate queries and non-aggregate queries return incompatible results, and thus the class of queries supported by these approaches is typically quite limited.
For example, most support only a single aggregation as the last operation of a query.
Worse, these approaches are often still computationally intractable.
Another class of solutions represents aggregation results symbolically (e.g., \cite{DBLP:journals/pvldb/FinkHO12,AD11d}).
Evaluating queries over symbolic representations is often tractable (\ptime), but the result may be hard to interpret for a human, and extracting tangible information (e.g., expectations) from 
symbolic instances is again hard.
In summary, prior work on processing complex queries involving aggregation over incomplete (and probabilistic) databases
(i) only supports 
limited query types; (ii) is often expensive; (iii) and/or returns results that are hard to interpret.

We argue that for uncertain data management to be accepted by practitioners it has to be competitive with the \termBG approach in terms of (i) performance and (ii) the class of supported queries (e.g., aggregation).
In this work, we present \abbrAUDBs, an annotated data model that approximates an incomplete database by annotating one of its possible worlds.
As an extension of the recently proposed \abbrUADBs~\cite{FH19}, \abbrAUDBs generalize and subsume current standard practices (i.e., \abbrBGQP).
An \abbrAUDB is  built
on a 
selected world, supplemented with two sets of annotations: lower and upper bounds on both attributes, and tuple  \ifnottechreport{multiplicities.}\iftechreport{ annotations (multiplicities in the case of bag semantics).}
Thus, each tuple in an \abbrAUDB may encode a set of tuples from each possible world, each with attribute values falling within the provided bounds.
In addition to being a strict generalization of \abbrBGQP, an \abbrAUDB relation also includes enough information to bound both the certain and possible answers as illustrated in \Cref{fig:uadb-sandwiching}.

\begin{Example}
\revb{
  \Cref{fig:running-example-audb} shows an \abbrAUDB constructed from one possible world $\db_{\abbrBG}$ of $\pdb$.
  We refer to this world as the \emph{\termBGW} (\abbrBGW).
  Each uncertain attribute is replaced by a 3-tuple, consisting of a lower bound, the value of the attribute in the \abbrBGW, and an upper bound, respectively.
  Additionally, each tuple is annotated with a 3-tuple consisting of a lower bound on its multiplicity across all possible worlds, its multiplicity in the \abbrBGW, and an upper bound on its multiplicity.
  For instance, Los Angeles is known to have an infection rate between 3\% and 4\% with a guess (e.g., based on a typical ETL approach like giving priority to a trusted source) of 3\%.
  The query result is shown in \Cref{fig:running-example-audb}.
  The first row of the result indicates that there is exactly one record for metro areas (i.e., the upper and lower multiplicity bounds are both 1), with an average rate between 6\% and 12\% (with a selected guess of 8.5\%).
  Similarly, the second row of the result indicates that there might (i.e., lower-bound of 0) exist one record for cities with a rate between 7.33\% and 18\%.
  This is a strict generalization of how users presently interact with uncertain data, as ignoring everything but the middle element of each 3-tuple gets us the \abbrBGW.
  However, the \abbrAUDB also captures the data's uncertainty.
}
\end{Example}

As we will demonstrate, \abbrAUDBs have several beneficial properties that make them a good fit for dealing with uncertain data:

\mypar{Efficiency}
Query evaluation over \abbrAUDBs is \ptime, and
by using novel optimizations that compact intermediate results to trade precision for performance,
 our approach scales to large datasets and complex queries.
While still slower than SGQP, AU-DBs are practical, significantly outperforming alternative uncertain data management systems, especially for queries involving aggregation.

\mypar{Query Expressiveness}
The under- and over-approximations encoded by an \abbrAUDB are preserved by queries from the full-relational algebra with \reva{multiple} aggregation\reva{s} ($\raAgg$).
Thus, \abbrAUDBs are closed under $\raAgg$, and are (to our knowledge) the first incomplete database approach to support complex, multi-aggregate queries.

\mypar{Compatibility} Like \abbrUADBs~\cite{FH19}, an \abbrAUDB can be constructed from many existing incomplete and probabilistic data models, including \abbrCtables~\cite{DBLP:journals/jacm/ImielinskiL84} or
\termTIs~\cite{suciu2011probabilistic}, making it possible to re-use existing approaches for exposing uncertainty in data (e.g., \cite{Beskales:2014:SRC:2581628.2581635,DBLP:conf/sigmod/RatnerBER17,Yang:2015:LOA:2824032.2824055,DBLP:conf/pods/KoutrisW18,GP17,DBLP:series/synthesis/2011Bertossi,DBLP:conf/pods/ArenasBC99}).
Moreover, although this paper focuses on bag semantics,
our model is defined for the same class of semiring-annotated databases~\cite{Green:2007:PS:1265530.1265535} as \abbrUADBs~\cite{FH19} which include, e.g.,  set semantics, security-annotations, and provenance. 

\mypar{Compactness}
As observed
elsewhere~\cite{GL17,GL16,L16a}, under\--ap\-prox\-i\-ma\-ting certain answers
for non-monotone queries (like aggregates) requires over-approximating possible answers. 
A single \abbrAUDB tuple can encode a large number of tuples, and can
compactly approximate possible results. 
This over-approximation 
is interesting in its own right to deal with missing data in the spirit
of~\cite{sundarmurthy_et_al:LIPIcs:2017:7061,DBLP:conf/sigmod/LangNRN14,liang-20-frmdcanp}.

\mypar{Simplicity} 
\abbrAUDBs use simple bounds to convey uncertainty, as opposed to the more complex symbolic formulas of
\abbrMtables~\cite{sundarmurthy_et_al:LIPIcs:2017:7061} or tensors~\cite{AD11d}. \revb{Representing uncertainty as ranges has   been shown to lead to better decision-making~\cite{kumari:2016:qdb:communicating}.
 \abbrAUDBs can be integrated into uncertainty-aware user interfaces, e.g., Vizier~\cite{BB19,kumari:2016:qdb:communicating}.}

\section{Related Work}
\label{sec:related-work}

We build on prior research in 
uncertain databases, specifically, 
techniques for approximating certain answers and aggregation.


\iftechreport{
\newcommand{\CmpApproach}[1]{\small #1}
\newcommand{\CmpSupported}{\textcolor{darkgreen}{\checkmark}}
\newcommand{\CmpUnsupported}{\textcolor{red}{\texttimes}}
\newcommand{\CmpSpecialAgg}[1]{\multicolumn{3}{c|}{\small $\vdash$\hspace*{-0.5mm}\textemdash\hspace*{-0.5mm}\textemdash \hspace{1mm} #1  \textemdash\hspace*{-0.5mm}\textemdash\hspace*{-0.5mm}$\dashv$}}
\newcommand{\CmpInputFDs}{\small FD}
\newcommand{\CmpInputDLLite}{\CmpInputFDs}
\newcommand{\CmpInputSTTGDs}{\small TGD}
\newcommand{\CmpInputCTable}{\small C-Tb}
\newcommand{\CmpInputXTable}{\small X-Tb}
\newcommand{\CmpInputTITable}{\small TI}
\newcommand{\CmpInputIA}{\small IA}
\newcommand{\CmpInputVTable}{\small V-Tb}
\newcommand{\CmpInputAny}{\small Any}
\newcommand{\CmpOutputGLBLUB}{\small GLB+LUB \xspace}
\newcommand{\CmpOutputGLB}{\small GLB only \xspace}
\newcommand{\CmpOutputSymbolic}{\small Symbolic \xspace}
\newcommand{\CmpOutputMoments}{\small Moments \xspace}
\newcommand{\CmpOutputSample}{\small Output Sample \xspace}
\newcommand{\CmpOutputSpecial}[1]{\small #1 \xspace}
\newcommand{\CmpNPHard}{\small NP-hard \xspace}
\newcommand{\CmpCoNPHard}{\small coNP-hard \xspace}
\newcommand{\CmpCoNPComplete}{\small coNP-complete \xspace}
\newcommand{\CmpPTime}{\small PTIME \xspace}
\newcommand{\CmpApprox}[1]{#1\small (approx)\xspace}
\newcommand{\CmpTimingSpecial}[1]{\small #1\xspace}

\begin{figure*}[ht]
  \centering
  \bgroup
\def\arraystretch{0.9}%
  \begin{tabular}{l|c|c|c|c|c|c|c|c|c}
    \multirow{2}{*}{\thead{Approach}}
      & \multicolumn{3}{c|}{\thead{Aggregates}}
      & \multicolumn{3}{c|}{\thead{Features}}
      & \multirow{2}{*}{\thead{Input}}
      & \multirow{2}{*}{\thead{Output}}
      & \multirow{2}{*}{\thead{Complexity}} \\

      & \textbf{\footnotesize Sum/Cnt}
      & \textbf{\footnotesize Avg}
      & \textbf{\footnotesize Min/Max}
      & \textbf{\footnotesize Chain}
      & \textbf{\footnotesize Having}
      & \textbf{\footnotesize Group}
      & 
      & 
      \\ \hline

    \CmpApproach{Arenas et. al.~\cite{DBLP:journals/tcs/ArenasBCHRS03}}
      & \CmpSupported   & \CmpUnsupported & \CmpSupported
      & \CmpUnsupported & \CmpUnsupported & \CmpUnsupported
      & \CmpInputFDs
      & \CmpOutputGLBLUB
      & \CmpNPHard
    \\

    \CmpApproach{Fuxmann et. al.~\cite{FF05a}}
      & \CmpSupported   & \CmpUnsupported & \CmpSupported
      & \CmpUnsupported & \CmpUnsupported & \CmpSupported
      & \CmpInputFDs
      & \CmpOutputGLBLUB
      & \CmpCoNPHard / \CmpPTime
    \\

    \CmpApproach{Afrati et. al.~\cite{DBLP:conf/pods/AfratiK08}}
      & \CmpSupported   & \CmpSupported   & \CmpSupported
      & \CmpUnsupported & \CmpUnsupported & \CmpUnsupported
      & \CmpInputSTTGDs
      & \CmpOutputGLBLUB
      & \CmpNPHard / \CmpPTime
    \\

    \CmpApproach{Fink et. al.~\cite{DBLP:journals/pvldb/FinkHO12}}
      & \CmpSupported   & \CmpUnsupported & \CmpSupported
      & \CmpSupported   & \CmpSupported   & \CmpSupported
      & \CmpInputCTable
      & \CmpOutputSymbolic
      & \CmpNPHard
    \\

    \CmpApproach{Murthy et. al.~\cite{DBLP:journals/tkde/MurthyIW11}}
      & \CmpSupported   & \CmpSupported   & \CmpSupported
      & \CmpUnsupported & \CmpUnsupported & \CmpUnsupported
      & \CmpInputXTable
      & \CmpOutputGLBLUB / \CmpOutputMoments
      & \CmpNPHard / \CmpPTime
    \\

    \CmpApproach{Abiteboul et. al. ~\cite{DBLP:conf/icdt/AbiteboulCKNS10}}
      & \CmpSupported   & \CmpSupported   & \CmpSupported
      & \CmpUnsupported & \CmpUnsupported & \CmpUnsupported
      & \CmpInputCTable$^1$
      & \CmpOutputGLBLUB / \CmpOutputMoments
      & \CmpNPHard
    \\

    \CmpApproach{Lechtenborger et. al.~\cite{DBLP:journals/jiis/LechtenborgerSV02}}
      & \CmpSupported   & \CmpUnsupported & \CmpUnsupported
      & \CmpSupported   & \CmpSupported   & \CmpSupported
      & \CmpInputCTable
      & \CmpOutputSymbolic
      & \CmpNPHard
    \\

    \CmpApproach{Re et. al.~\cite{DBLP:journals/vldb/ReS09}}
      & \CmpSpecialAgg{\texttt{HAVING} only}
      & \CmpSupported   & \CmpSupported   & \CmpSupported
      & \CmpInputTITable
      & \CmpOutputMoments
      & \CmpNPHard / \CmpPTime
    \\

    \CmpApproach{Soliman et. al.~\cite{DBLP:journals/tods/SolimanIC08}}
      & \CmpSpecialAgg{\texttt{TOP-K} only}
      & \CmpUnsupported & \CmpUnsupported & \CmpSupported
      & \CmpInputCTable
      & \CmpOutputMoments
      & \CmpNPHard / \CmpPTime
    \\

    \CmpApproach{Chen et. al.~\cite{CC96}}
      & \CmpSupported   & \CmpSupported   & \CmpSupported
      & \CmpUnsupported & \CmpUnsupported & \CmpSupported
      & \CmpInputXTable
      & \CmpOutputGLBLUB
      & \CmpNPHard / \CmpPTime
    \\

    \CmpApproach{Jayram et. al.~\cite{DBLP:conf/soda/JayramKV07}}
      & \CmpSupported   & \CmpSupported   & \CmpSupported
      & \CmpUnsupported & \CmpUnsupported & \CmpUnsupported
      & \CmpInputXTable
      & \CmpOutputMoments
      & \CmpApprox \CmpPTime
    \\

    \CmpApproach{Burdick et. al.~\cite{DBLP:journals/vldb/BurdickDJRV07}}
      & \CmpSupported   & \CmpSupported   & \CmpUnsupported
      & \CmpUnsupported & \CmpUnsupported & \CmpSupported
      & \CmpInputXTable
      & \CmpOutputMoments
      & \CmpApprox \CmpPTime
    \\

    \CmpApproach{Calvanese et. al~\cite{DBLP:conf/cikm/CalvaneseKNT08}}
      & \CmpSupported   & \CmpUnsupported & \CmpUnsupported
      & \CmpUnsupported & \CmpUnsupported & \CmpUnsupported
      & \CmpInputDLLite
      & \CmpOutputGLB
      & \CmpNPHard 
    \\

    \CmpApproach{Kostylev et. al.~\cite{DBLP:conf/aaai/KostylevR13}}
      & \CmpSpecialAgg{\texttt{COUNT} / \texttt{DISTINCT}}
      & \CmpUnsupported & \CmpUnsupported & \CmpUnsupported
      & \CmpInputDLLite
      & \CmpOutputGLB
      & \CmpCoNPComplete 
    \\

    \CmpApproach{Yang et. al.~\cite{DBLP:conf/sigmod/YangWCK11}}
      & \CmpSpecialAgg{Agg Constraint only}
      & \CmpUnsupported & \CmpUnsupported & \CmpUnsupported
      & \CmpInputXTable
      & \CmpOutputSpecial{Sample of Input}
      & \CmpCoNPComplete 
    \\

    \CmpApproach{Jampani et. al.~\cite{jampani2008mcdb}}
      & \CmpSpecialAgg{No restrictions}
      & \CmpSupported & \CmpSupported & \CmpSupported
      & \CmpInputVTable
      & \CmpOutputSample
      & \CmpApprox \CmpPTime 
    \\

    \CmpApproach{Kennedy et. al.~\cite{5447879}}
      & \CmpSupported   & \CmpSupported   & \CmpSupported
      & \CmpUnsupported   & \CmpSupported   & \CmpUnsupported
      & \CmpInputCTable
      & \CmpOutputSample
      & \CmpApprox \CmpPTime 
    \\

    \CmpApproach{Lang et. al.~\cite{DBLP:conf/sigmod/LangNRN14}}
      & \CmpSupported   & \CmpSupported   & \CmpSupported
      & \CmpSupported   & \CmpSupported   & \CmpSupported
      & \CmpInputIA
      & \CmpInputIA
      & \CmpTimingSpecial{Data-Indep.}/\CmpPTime 
    \\

    \CmpApproach{Sismanis et al.~\cite{sismanis-09-rawqanbin}}
      & \CmpSupported   & \CmpSupported   & \CmpSupported
      & \CmpUnsupported   & \CmpUnsupported  & \CmpSupported
      & \CmpInputXTable$^2$
      & \CmpOutputGLBLUB
      & \CmpApprox \CmpPTime 
    \\
    \hline

    \CmpApproach{\textbf{This Paper}}
      & \CmpSupported   & \CmpSupported   & \CmpSupported
      & \CmpSupported   & \CmpSupported   & \CmpSupported
      & \CmpInputAny
      & \CmpOutputGLBLUB
      & \CmpApprox \CmpPTime 
    \\

  \end{tabular}
\egroup
  \\[-2mm]

  \caption{\label{fig:comparison-uncertain-aggregation} Comparison of approaches for  aggregation over uncertain data.  Features include the ability to \emph{chain} aggregates, support \emph{having} predicates, and support \emph{group}ing. 
    FD: Functional Dependency Repair, TGD: Source-Target tgds, C-Tb: C-Table, X-Tb: X-Table (aka Block-Independent), TI: Tuple-Independent, V-Tb: V-Table, IA: Tables (or horizontal partitions) are annotated to indicate whether (i) their attribute values are correct, whether they may not contain all certain tuples, and whether they contain tuples that are not certain, GLB only: under-approximates certain answers (a lower bound on every possible world), GLB+LUB: under-approximates certain answers and over-approximates possible answers (an upper bound on every possible world). $^1$: Probabilistic XML analogous to C-Tables. $^2$: the input is an entity resolution problem that can be represented as using X-Tables.}
\end{figure*}
}


\mypar{Approximations of Certain Answers}
Queries over incomplete databases typically use  certain answer semantics~\cite{DBLP:journals/jacm/ImielinskiL84,AK91,L16a,GL16,GL17} first defined in~\cite{L79a}. 
Computing certain answers is \conpcomplete~\cite{AK91,DBLP:journals/jacm/ImielinskiL84} (data complexity) for relational algebra. 
Several techniques for computing an under-approximation (subset) of certain answers
have been proposed.
Reiter~\cite{R86} proposed a \ptime algorithm 
for positive existential queries. 
Guagliardo and Libkin~\cite{GL17,L16a,GL16} proposed a scheme 
for full relational algebra for Codd- and \abbrVtables, and also studied bag semantics~\cite{CG19,GL17}.
Feng et. al.~\cite{FH19} generalized this approach to new query semantics through Green et. al.'s $\semK$-relations~\cite{Green:2007:PS:1265530.1265535}.
m-tables~\cite{sundarmurthy_et_al:LIPIcs:2017:7061} compactly encode of large amounts of possible tuples, allowing for efficient query evaluation. However, this requires complex symbolic expressions which necessitate schemes for approximating certain answers. 
Consistent query
answering (CQA)~\cite{DBLP:series/synthesis/2011Bertossi,DBLP:conf/pods/ArenasBC99}
computes the certain answers to queries over all possible repairs of a database
that violates a set of constraints.  Variants of this problem
have been studied extensively
(e.g.,~\cite{DBLP:journals/ipl/KolaitisP12,DBLP:conf/pods/CaliLR03,DBLP:conf/pods/KoutrisW18})
and several combinations of classes of constraints and queries 
permit first-order
rewritings~\cite{FM05,GP17,DBLP:journals/tods/Wijsen12,DBLP:conf/pods/Wijsen10}.
Geerts et. al.~\cite{GP17} study first-order under-approximations of certain
answers in the context of CQA.
Notably, \abbrUAADBs build on the approach of~\cite{FH19} (i.e., a selected guess and lower bounds), adding an upper bound on possible answers (e.g., as in~\cite{GL17}) to support aggregations, and bound attribute-level uncertainty with ranges instead of nulls.


\mypar{Aggregation in Incomplete/Probabilistic Databases}
While aggregation of uncertain data has been studied extensively (see\ifnottechreport{~\cite{techreport}}\iftechreport{ \Cref{fig:comparison-uncertain-aggregation}} for a comparison of approaches), general solutions remain an open problem~\cite{DBLP:conf/pods/ConsoleGLT20}.
A key challenge lies in defining a meaningful semantics, as aggregates over uncertain data frequently produce empty certain answers~\cite{DBLP:conf/cikm/CalvaneseKNT08}.
An alternative semantics adopted for CQA and ontologies~\revc{\cite{DBLP:journals/tcs/ArenasBCHRS03,DBLP:conf/cikm/CalvaneseKNT08,FF05a,DBLP:conf/pods/AfratiK08,sismanis-09-rawqanbin}} returns per-attribute bounds over all possible results~\cite{DBLP:journals/tcs/ArenasBCHRS03} instead of a single \emph{certain} answer.
In contrast to prior work, we use bounds as a fundamental building block of our data model. 
Because of the complexity of aggregating uncertain data, most approaches focus on identifying tractable cases and producing statistical moments or other lossy representations~\cite{5447879,DBLP:journals/tkde/MurthyIW11,DBLP:conf/icdt/AbiteboulCKNS10,DBLP:journals/tods/SolimanIC08,CC96,DBLP:conf/soda/JayramKV07,DBLP:journals/vldb/BurdickDJRV07,DBLP:conf/sigmod/YangWCK11}.
Even this simplified approach is expensive (often NP-hard, depending on the query class), and requires approximation.
Statistical moments like expectation may be meaningful as final query answers, but are less useful if the result is to be subsequently queried (e.g., \lstinline!HAVING! queries~\cite{DBLP:journals/vldb/ReS09}).

Efforts to create a lossless symbolic encoding closed under aggregation~\cite{DBLP:journals/pvldb/FinkHO12,DBLP:journals/jiis/LechtenborgerSV02} exist, supporting complex multi-aggregate queries and a wide range of statistics (e.g, bounds, samples, or expectations).
However, even factorizable encodings like aggregate semimodules~\cite{AD11d} usually scale in the size of the aggregate input and not the far smaller aggregate output, making these schemes impractical.
\abbrAUDBs are also closed under aggregation, but replace lossless encodings of aggregate outputs with lossy, but compact \emph{bounds}.

A third approach, exemplified by MCDB~\cite{jampani2008mcdb} queries sampled possible worlds.
In principle, this approach supports arbitrary queries, but is significantly slower than \abbrBGQP~\cite{FH19}, only works when probabilities are available, and only supports statistical measures that can be derived from samples (i.e., moments and epsilon-delta bounds).

A similarly general approach~\cite{DBLP:conf/sigmod/LangNRN14,sundarmurthy_et_al:LIPIcs:2017:7061} determines which parts of a query result over incomplete data are uncertain, and whether the result is an upper or lower bound. 
However, this approach tracks incompleteness coarsely (horizontal table partitions).
\abbrAUDBs are more general, combining both fine-grained uncertainty information (individual rows and attribute values) and coarse-grained information (one row in a \abbrAUDB may encode multiple tuples).


\section{Notation and Background}\label{sec:background}

We now review $\semK$-relations, incomplete $\semK$-relations that generalize classical incomplete databases, and the \abbrUADBs model extended in this work.
A database schema $\schema = \{\schemaOf(\rel_1), \ldots, \schemaOf(\rel_n)\}$ is a set of relation schemas  $\schemaOf(\rel_i) = \tuple{A_1,\; \ldots,\; A_n}$. 
The arity $\arity{\relschema}$ of 
$\relschema$ is the number of attributes in $\relschema$.
\iftechreport{An  instance $\db$ for database schema $\schema$ is a set of relation instances with one relation for each relation schema in $\schema$: $\db=\{\rel_1, \dots, \rel_n\}$.}
Assume a universal domain of attribute values $\aDom$. A tuple with schema $\relschema$ is an element from $\aDom^{\arity{\relschema}}$.
\revm{We assume the existence of a total order over the elements of $\aDom$.}\footnote{
\revm{The order over $\aDom$ may be arbitrary, but range bounds are most useful when the order makes sense for the domain values (e.g., the ordinal scale of an ordinal attribute).}
}


\subsection{K-Relations} \label{sec:data_provenance}
The generalization of incomplete databases 
we use here is based on  \textbf{$\semK$-relations}~\cite{Green:2007:PS:1265530.1265535}.
In this framework, relations are annotated with elements from the domain $\kDom$ of a (commutative) semiring $\semK=\tuple{\kDom,\addK,\multK,\oneK,\zeroK}$, i.e., a mathematical structure with commutative and associative addition ($\addK$) and product ($\multK$) operations where $\addK$ distributes over $\multK$ and $k \multK \zeroK = \zeroK$ for all $k \in K$.
An $n$-nary $\semK$-relation is a function that maps tuples 
to elements from $\kDom$.
Tuples that are not in the relation are annotated with $\zeroK$. 
Only finitely many tuples may be mapped to an element other than $\zeroK$. 
Since $\semK$-relations are functions from tuples to annotations, it is customary to denote the annotation of a tuple $\tup$ in relation $\rel$ as $\rel(\tup)$. 
 \ifnottechreport{
   In this work we are interested in bag semantics, which can be encoded using the semiring of natural numbers with standard addition and multiplication $\tuple{\semN, +, \times, 0, 1}$ to annotate each tuple with its multiplicity. 
   We discuss the applicability of our framework to a larger class of semirings in an accompanying technical report~\cite{techreport}. 
}
\iftechreport{
The specific information encoded by an annotation 
 depends on the choice of semiring.
  For instance, bag and set relations can be encoded as
 semirings: the natural numbers ($\semN$) with addition and multiplication, $\tuple{\semN, +, \times, 0, 1}$, annotates each tuple with its multiplicity; and boolean constants $\semB = \{T,F\}$ with disjunction and conjunction, $\tuple{\semB, \vee, \wedge, F, T}$, annotates each tuple with its set membership.
 Abusing notation, we often use $\semK$ to denote both the domain and the corresponding semiring.
 }

 \ifnottechreport{
Operators of the positive relational algebra ($\raPlus$) over $\semN$-relations 
combine input annotations using 
$\addN$ and $\multN$. 
\begin{align*}
	\textbf{Union: } & (R_1 \union R_2)(t) = R_1(t) \addN R_2(t) \\
	\textbf{Join: } & (R_1 \join R_2)(t) = R_1(t[{\schemaOf(\rel_1}) ]) \multN R_2(t[{\schemaOf(\rel_2}) ]) \\
  \textbf{Projection: } & (\projection_U (R))(t) = \sum_{t=t'[U]}R(t')\\
  \textbf{Selection: } & (\selection_\theta(R))(t) = R(t) \multN \theta(t)
\end{align*}
   }
\iftechreport{

\mypar{Query Semantics}
Operators of the positive relational algebra ($\raPlus$) over $\semK$-relations are defined by combining input annotations using operations $\addK$ and $\multK$. 

\begin{align*}
	\textbf{Union: } & (R_1 \union R_2)(t) = R_1(t) \addK R_2(t) \\
	\textbf{Join: } & (R_1 \join R_2)(t) = R_1(t[\relschema_1 ]) \multK R_2(t[\relschema_2 ]) \\
  \textbf{Projection: } & (\projection_U (R))(t) = \sum_{t=t'[U]}R(t')\\
  \textbf{Selection: } & (\selection_\theta(R))(t) = R(t) \multK \theta(t)
\end{align*}
}
\iftechreport{
For simplicity we assume in the definition above that tuples are of a compatible schema (e.g., $\relschema_1$ for a union $R_1 \union R_2$).
We use $\theta(t)$ to denote a function that returns $\oneK$ iff $\theta$ evaluates to true over tuple $t$ and $\zeroK$ otherwise.
}
\iftechreport{
A \textbf{homomorphism} is a mapping $h: \semK \to \semK'$ from a semiring $\semK$ to a semiring $\semK'$ that maps $\zeroOf{\semK}$ and $\oneOf{\semK}$ to their counterparts in $\semK'$ and distributes over sum and product (e.g., $h(k \addK k') = h(k) \addOf{\semK'} h(k')$).
Any homomorphisms $h$ can be lifted from semirings to $\semK$-relations or $\semK$-databases  by applying $h$ to the annotation of every tuple $\tup$: $h(R)(t) = h(R(t))$. We will use the same symbol for a homomorphism and its lifted variants.
Importantly, queries commute with semiring homomorphisms: $h(\query(\db)) = \query(h(\db))$.
}
We will make use of the so called \textit{natural order} $\ordK$ for a semiring $\semK$
which is the  standard order $\leq$ of natural numbers for $\semN$.
\iftechreport{
Formally, $k \ordK k'$ if it is possible to obtain $k'$ by adding to $k$: $\exists k'': k \addK k'' = k'$. 
Semirings for which the natural order is a partial order are called \textit{naturally ordered}~\cite{Geerts:2010bz}.
\begin{align}
\forall k, k' \in \kDom: \big(k \ordK k'\big) \Leftrightarrow \big(\exists k'' \in \kDom: k \addK k'' = k'\big)
\end{align}
}
\subsection{Incomplete K-Relations}
\label{sec:incompl-prob-k}

\ifnottechreport{
Instead of using the classical set-based definition of incomplete databases and certain answers, we apply the generalization from~\cite{FH19} called incomplete $\semK$-relations, specifically the $\semN$-relations that model bag-incomplete databases.
An incomplete $\semN$-database is a set of $\semN$-databases $\pdb = \{\db_1, \ldots, \db_n\}$  called possible worlds.
Queries over an incomplete $\semN$-database use
possible world semantics: The result of a query $\query$ over
an incomplete $\semN$-database $\pdb$ is the set of possible worlds derived
by evaluating $\query$ over every world in $\pdb$. 
\begin{align*}\label{eq:pqp}
\query(\pdb) \defas \comprehension{\query(\db)}{\db \in \pdb} \tag{possible world semantics}
\end{align*}
}
\iftechreport{
\begin{Definition} 
  Let $\semK$ be a semiring. An \emph{incomplete $\semK$-database} $\pdb$ is a set of $\semK$-databases $\pdb = \{\db_1, \ldots, \db_n\}$   called possible worlds.
\end{Definition}

Queries over an incomplete $\semK$-database use
possible world semantics, i.e., the result of evaluating a query $\query$ over
an incomplete $\semK$-database $\pdb$ is the set of all possible worlds derived
by evaluating $\query$ over every possible world $\db \in \pdb$. 
\begin{align}\label{eq:pqp}
\query(\pdb) \defas \comprehension{\query(\db)}{\db \in \pdb}
\end{align}
}
\ifnottechreport{
Generalizing certain and possible answers,
the certain (possible)
multiplicity 
of a tuple $\tup$ in an incomplete $\semN$-database $\pdb$ is 
the minimum (maximum) multiplicity of the tuple across all possible worlds:\\[-10mm]
\begin{center}
  \begin{align*}
                         \pwCertainN(\pdb,\tup)  &\defas \min(\{\db(\tup) \mid \db \in \pdb \})\\
\pwPossibleN(\pdb,\tup) &\defas \max(\{\db(\tup) \mid \db \in \pdb \})
  \end{align*}
\end{center}
}
\iftechreport{

\subsubsection{Certain and Possible Annotations} 
\label{sec:poss-worlds-cert}

For incomplete $\semK$-relations, we define the certain and possible annotations of tuples as a generalization of certain and possible answers in classical incomplete databases. For these concepts to be well-defined we require that $\semK$ is an l-semiring~\cite{DBLP:conf/icdt/KostylevB12} which means that the natural order forms a lattice. Most commonly considered semirings (e.g., sets, bags, most provenance semirings, \ldots) are l-semirings.
The \emph{certain annotation} of a tuple, is the greatest lower bound (\emph{glb}) of its annotations across all possible world while the \emph{possible annotation} is the least upper bound (\emph{lub}) of these annotations.
We use $\glbK$ (glb) and $\lubK$ (lub)
to denote the $\glb$ and $\lub$
operations for a semiring $\semK$.
The certain (possible)
annotation $\pwCertain(\pdb, \tup)$ ($\pwPossible(\pdb, \tup)$) of a tuple $\tup$ in an incomplete $\semK$-database $\pdb$ is defined as the glb (lub) over the annotations of tuple $\tup$ across all possible worlds of $\pdb$: 
%
\begin{center}
  \begin{align*}
                         \pwCertain(\pdb,\tup)  &\defas \GlbK(\{\db(\tup) \mid \db \in \pdb \})\\
\pwPossible(\pdb,\tup) &\defas \lubK(\{\db(\tup) \mid \db \in \pdb \})
  \end{align*}
\end{center}
Importantly, this coincides with the standard definition of certain and possible answers for set semantics ($\semB$): the natural order of the set semiring $\semB$ is $\bfalse \ordOf{\semB} \btrue$, $k_1 \lubKof{\semB} k_2 = k_1 \vee k_2$,  
and $k_1 \glbKof{\semB} k_2 = k_1 \wedge k_2$. That is, a tuple is certain (has certain annotation $\btrue$) if it exists (is annotated with $\btrue$) in every possible world and possible if it exists in at least one possible world (is annotated with $true$ in one or more worlds). The natural order of $\semN$ is the standard order of natural numbers. We get  $\pwCertOf{\semN} = \min$ and $\pwPossOf{\semN} = \max$.
This coincides with 
the definition of certain and possible multiplicity for bag semantics from~\cite{GL16,CG19,DBLP:conf/pods/ConsoleGLT20}. 
}

\subsection{UA-Databases} \label{sec:UA-model}
\ifnottechreport{
Using $\semK$-relations, Feng et al.~\cite{FH19} introduced \textit{\abbrUADBs} (\textit{uncertainty-annotated databases}) which, in the case of semiring $\semN$, encode a \revb{tuple level} under- and over-\-ap\-prox\-i\-ma\-tion of the certain multiplicity of tuples from an incomplete $\semN$-database $\pdb$. 
In a bag \abbrUADB (semiring $\semN$), every tuple is annotated with a pair $\uae{d}{c} \in \doubleDom{\semN}$ where $d$ is the tuple's multiplicity in a selected possible world $\bgdb \in \pdb$ i.e., $d = \bgdb(t)$ and $c$ is an under-approximation of the tuple's certain multiplicity, i.e., $c \leq \pwCertainN(\pdb,t) \leq d$.
The selected world $\bgdb$ is called the \termBGW (\abbrBGW).
Formally, these pairs $\uae{d}{c}$ are elements from a semiring $\uaK{\semN}$ which is the direct product of semiring $\semN$ with itself ($\semN^2$).
Operations in the product semiring $\doubleN = \tuple{\doubleDom{\semN}, \addOf{\doubleN}, \multOf{\doubleN}, \zeroOf{\doubleN}, \oneOf{\doubleN}}$ are defined pointwise, e.g., $[k_1, {k_1}'] \multOf{\doubleN} [k_2,{k_2}'] = [k_1 \multN k_2, {k_1}' \multN {k_2}']$.
}
\iftechreport{
  Using $\semK$-relations, Feng et al.~\cite{FH19} introduced \textit{\abbrUADBs} (\textit{uncertainty-annotated databases}) which encode an under- and an over-approximation of the certain annotation of tuples from an incomplete $\semK$-database $\pdb$. In the case of semiring $\semN$ this means that every tuple is annotated with  an under- and an over-\-ap\-prox\-i\-ma\-tion of its certain multiplicity. 
That is, in a bag \abbrUADB (semiring $\semN$), every tuple is annotated with a pair $\uae{d}{c} \in \doubleDom{\semN}$ where $d$ is the tuple's multiplicity in a selected possible world $\bgdb \in \pdb$ i.e., $d = \bgdb(t)$ and $c$ is an under-approximation of the tuple's certain multiplicity, i.e., $c \leq \pwCertain(\pdb,t) \leq d$.
The selected world $\bgdb$ is called the \termBGW (\abbrBGW).
Formally, these pairs $\uae{d}{c}$ are elements from a semiring $\uaK{\semK}$ which is the direct product of semiring $\semK$ with itself ($\semK^2$).
Operations in the product semiring $\aDoubleK = \tuple{\doubleDom{\semK}, \addOf{\aDoubleK}, \multOf{\aDoubleK}, \zeroOf{\aDoubleK}, \oneOf{\aDoubleK}}$ are defined pointwise, e.g., $[k_1, {k_1}'] \multOf{\aDoubleK} [k_2,{k_2}'] = [k_1 \multK k_2, {k_1}' \multK  p{k_2}']$.
\begin{Definition}[UA-semiring]
  Let $\semK$ be a semiring. We define the corresponding UA-semiring
  $\uaK{\semK} \defas \semK^2$
\end{Definition}
}
\iftechreport{
\abbrUADBs are created from incomplete or probabilistic data sources by selecting a \abbrBGW $\bgdb$ and generating an under-approximation $\TUL$ of the certain annotation $\pwCertain$ of tuples. In the \abbrUADB, the annotation of each tuple $\tup$ is set to: 
  \begin{align*}
  \forall \tup: \TUL(\tup) \ordK \pwCertain(\pdb, \tup) \ordK \bgdb(\tup)
 \end{align*}
 }
 \ifnottechreport{
   \abbrUADBs are created from an incomplete or probabilistic data source by selecting a \abbrBGW $\bgdb$ and generating an under-approximation $\TUL$ of the certain multiplicty $\pwCertainN$ of tuples. In the \abbrUADB, the annotation of each tuple $\tup$ is set to: 
$
\uadb(\tup) \defas \uae{\bgdb(\tup)}{\TUL(\tup)}
$.
}
\iftechreport{
$$
\uadb(\tup) \defas \uae{\bgdb(\tup)}{\TUL(\tup)}
$$
}
\abbrUADBs constructed in this fashion are said to \textit{bound} $\pdb$ through $\TUL$ and $\bgdb$.
Feng et al.~\cite{FH19} discussed how to create \abbrUADBs that bound \abbrCtables, \abbrVtables, and \abbrXDBs.
\cite[Theorem 1]{FH19} shows that standard $\semN^2$-relational query semantics preserves bounds under $\raPlus$ queries, i.e., if the input bounds an incomplete $\semN$-database $\pdb$, then the result bounds $\query(\pdb)$.
\iftechreport{
Formally, let $\uadb$ be a \abbrUADB created from a pair $(\TUL, \db)$  that approximates an incomplete $\semK$-database $\pdb$. Then for any $\raPlus$ query $\query$, we have that $\query(\uadb)$ approximation $\query(\pdb)$ by encoding $(\query(\db), \query(\TUL))$. \iftechreport{Importantly, this means that \abbrUADBs are closed under $\raPlus$ queries.}
}
\BG{even though it is nice to have the example here, I moved it to the techreport to save space.}
\iftechreport{
\begin{Example}\label{example:incomplete-k-relations}
  Consider the incomplete $\semN$-database  $\pdb$ (bag semantics) with two possible worlds 
  shown below.
  Using semiring $\semN$ each tuple in
  a possible world is annotated with its multiplicity (the number of copies
  of the tuple that exist in the possible world).
%
We also show an $\uaK{\semN}$-database that bounds $\pdb$ by encoding $\db_2$ and the certain multiplicities of tuples ($\TUL$ is exact in this example).
For example, tuple $(IL)$ is annotated with $[2,3]$ since this tuple appears thrice in $\db_2$ and \emph{at least} twice in every possible world, i.e., its certain annotation is
$\pwCertOf{\semN}(\{2,3\}) = min(2,3) = 2$.
Futhermore, consider the incomplete $\semB$-database (set semantics) shown below. Tuples $\tuple{IL}$ and $\tuple{AZ}$ in both possible worlds and, thus are certain (annotated with $[\btrue,\btrue]$). Tuple $\tuple{IN}$ only exists in $\db_2$. Thus, this tuple is not certain, but it is possible (annotated with $[\bfalse,\btrue]$).
\vspace*{-1mm}
\begin{center}
{\upshape
  \begin{minipage}{0.53\linewidth}
    \centering
    \textbf{Incomplete $\semN$-Database}
    \begin{minipage}{0.45\linewidth}
    \centering
    $\db_1$                                     \\
  \begin{tabular}{ c c }
     \textbf{state}     & \underline{$\semN$}   \\
    \cline{1-1}
        IL              & \textbf{2}            \\
        AZ              & \textbf{2}            \\
  \end{tabular}
  \end{minipage}
   \begin{minipage}{0.45\linewidth}
     \centering
     $\db_2$                                    \\
       \begin{tabular}{ c c }
     \textbf{state}     & \underline{$\semN$}   \\
         \cline{1-1}
 IL                     & \textbf{3}            \\
 AZ                     & \textbf{1}            \\
 IN                     & \textbf{5}            \\
  \end{tabular}
\end{minipage}
\end{minipage}
\begin{minipage}{0.4\linewidth}
  \begin{center}
    \textbf{$\uaK{\semB}$-Database}             \\[1mm]
  \begin{tabular}{ c c }
         \textbf{state} & \underline{$\semN^2$} \\
    \cline{1-1}         &                       \\[-3.1mm]
 IL                     & \textbf{[2,3]}        \\
 AZ                     & \textbf{[1,1]}        \\
 IN                     & \textbf{[0,5]}        \\
  \end{tabular}
  \end{center}
\end{minipage}
}
\end{center}
\begin{center}
{\upshape
  \begin{minipage}{0.53\linewidth}
    \centering
    \textbf{Incomplete $\semB$-Database}
    \begin{minipage}{0.45\linewidth}
    \centering
    $\db_1$                                     \\
  \begin{tabular}{ c c }
     \textbf{state}     & \underline{$\semB$}   \\
    \cline{1-1}
        IL              & $\btrue$            \\
        AZ              & $\btrue$            \\
  \end{tabular}
  \end{minipage}
   \begin{minipage}{0.45\linewidth}
     \centering
     $\db_2$                                    \\
       \begin{tabular}{ c c }
     \textbf{state}     & \underline{$\semB$}   \\
         \cline{1-1}
 IL                     & $\btrue$            \\
 AZ                     & $\btrue$            \\
 IN                     & $\btrue$            \\
  \end{tabular}
\end{minipage}
\end{minipage}
\begin{minipage}{0.4\linewidth}
  \begin{center}
    \textbf{$\uaK{\semB}$-Database}             \\[1mm]
  \begin{tabular}{ c c }
         \textbf{state} & \underline{$\dpK{\semB}{2}$} \\
    \cline{1-1}         &                       \\[-3.1mm]
 IL                     & \textbf{[$\btrue$,$\btrue$]}        \\
 AZ                     & \textbf{[$\btrue$,$\btrue$]}        \\
 IN                     & \textbf{[$\bfalse$,$\btrue$]}        \\
  \end{tabular}
\end{center}
\end{minipage}
}
\end{center}
\end{Example}
}



\section{Overview}

\revm{Query evaluation over \abbrUADBs is efficient (\ptime data complexity and experimental performance comparable to \abbrBGQP).}
However, \abbrUADBs may not be as precise \revm{and concise} as possible since uncertainty is only recorded at the tuple-level.
\revb{For example, the encoding of the town tuple in \Cref{fig:running-example-xdb} needs just shy of 600 uncertain tuples, one for each combination of possible values of the uncertain \lstinline{size} and \lstinline{rate} attributes.}
Additionally, \abbrUADB query semantics does not support non-monotone operations like aggregation \revc{and set difference}, as this requires an over-\-ap\-prox\-i\-ma\-tion of possible answers.

We address both shortcomings in \abbrUAADBs through two changes relative to \abbrUADBs:
(i) Tuple annotations include an upper bound on the tuple's possible multiplicity; and
(ii) Attribute values become 3-tuples, with lower- and upper-bounds and a \termBG (\abbrBG) value.
These building blocks, range-annotated scalar expressions and $\uaaN$-relations, are formalized in \Cref{sec:expression,sec:uaa-range}, respectively.

Supporting both attribute-level and tuple-level uncertainty creates ambiguity in how tuples should be represented.
\revb{As noted above, the tuple for towns is certain (i.e., deterministically present) and has uncertain (i.e., multiple-possible values) attributes, but could also be expressed as 600 tuples with certain attribute values whose existence is uncertain.}
This ambiguity makes it challenging to define what it means for an \abbrAUDB to bound an incomplete database, a problem we resolve in \Cref{sec:uaadb-bounds} by defining \emph{tuple matchings} that relate tuples in an \abbrAUDB to those of a possible world.
An \abbrAUDB bounds an incomplete database if such a mapping exists for every possible world.
This ambiguity is also problematic for group-by aggregation, as aggregating a relation with uncertain group-by attribute values may admit multiple, equally viable output {\abbrUAARel}s.
We propose a specific grouping strategy in \Cref{sec:bound-pres-aggr} that mirrors \abbrBGW query evaluation, and show that it behaves as expected.

\reva{Uncertain attributes are defined by ranges, so equi-joins on such attributes degenerate to interval-overlap joins that may produce large results if many intervals overlap.
To mitigate this bottleneck, \Cref{sec:joinOpt} proposes splitting join inputs into large, equi-joinable ``{\abbrBG}'' tables and small, interval-joinable ``possible'' tables.}


\section{Scalar Expressions}
\label{sec:expression}


Recall that $\dataDomain$ denotes a universal domain of values.
We assume that at least boolean values ($\bot$ and $\top$) are included in the domain.
Furthermore, let $\mathbb{V}$ denote a countable set of variables.
 \iftechreport{ %
\begin{Definition}[Expression Syntax]\label{def:expr-syntax}
	For any variable $x \in \mathbb{V}$, $x$ is an expression and for any constant $c \in \dataDomain$, $c$ is an expression.
	If $\sexpr_1$, $\sexpr_2$ and $\sexpr_3$ are expressions, then \ldots
    \begin{align*}
      &\sexpr_1 \wedge \sexpr_2
      &&\sexpr_1 \vee \sexpr_2
      && \neg \sexpr_1
      &&\sexpr_1 = \sexpr_2
      &&\sexpr_1 \neq \sexpr_2
      &&\sexpr_1 \leq \sexpr_2
    \end{align*}
    \begin{align*}
      &&\sexpr_1 + \sexpr_2
      &&\sexpr_1 \cdot \sexpr_2
      &&\frac{1}{\sexpr_1}
      &&\ifte{\sexpr_1}{\sexpr_2}{\sexpr_3}\\
    \end{align*}
	are also expressions. Given an expression $\sexpr$, we denote the variables in $\sexpr$ by $\vars(\sexpr)$.  \end{Definition}
   }
   \ifnottechreport{
	For any variable $x \in \mathbb{V}$, $x$ is an expression and for any constant $c \in \dataDomain$, $c$ is an expression.
	If $\sexpr_1$, $\sexpr_2$ and $\sexpr_3$ are expressions, then \ldots\\[-6mm]
    \begin{align*}
      &\sexpr_1 \wedge \sexpr_2
      &&\sexpr_1 \vee \sexpr_2
      && \neg \sexpr_1
      &&\sexpr_1 = \sexpr_2
      &&\sexpr_1 \neq \sexpr_2
      &&\sexpr_1 \leq \sexpr_2
    \end{align*}\\[-11mm]
    \begin{align*}
      &&\sexpr_1 + \sexpr_2
      &&\sexpr_1 \cdot \sexpr_2
      &&\ifte{\sexpr_1}{\sexpr_2}{\sexpr_3}
    \end{align*}\\[-4mm]
	are also expressions. Given an expression $\sexpr$, we denote the variables of $\sexpr$ by $\vars(\sexpr)$.  
     }
We will also use $\neq$, $\geq$, $<$, $-$, and $>$ since these operators can be defined using the expression syntax above, e.g., $e_1 > e_2 = \neg\, (e_1 \leq e_2)$. \iftechreport{Assuming that $\dataDomain$ contains negative numbers, subtraction can be expressed using addition and multiplication.}
For an expression $\sexpr$, given a valuation $\sval$ that maps variables from $\vars(\sexpr)$ to constants from $\dataDomain$, the expression evaluates to a constant \ifnottechreport{$\seval{e}{\sval}$} from $\dataDomain$. \ifnottechreport{The semantics of these expressions are standard (see \cite{techreport} for explicit definitions).}
\iftechreport{The semantics of expression evaluation is defined below.
%
\begin{Definition}[Expression Semantics]\label{def:expr-semantics}
	Let $\sexpr$ be an expression. Given a valuation  $\sval: \vars(e) \rightarrow \dataDomain$, the result of expression $\sexpr$ over $\sval$ is denoted as $\seval{\sexpr}{\sval}$. Note that $\seval{\frac{1}{\sexpr}}{\sval}$ is undefined if $\seval{\sexpr}{\sval} = 0$. The semantics of expression is defined as shown below:
    \begin{align*}
      \seval{x}{\sval}                                                                 & \defas \sval(x)
                                                                                       & \seval{c}{\sval}                                    & \defas c &
      \seval{\neg \sexpr_1}{\sval}                                                     & \defas \neg \seval{\sexpr_1}{\sval}
    \end{align*}
    \begin{align*}
     \seval{e_1 \wedge e_2}{\sval}                                                     & \defas \seval{e_1}{\sval} \wedge \seval{e_2}{\sval} &
    \seval{e_1 \vee e_2}{\sval}                                                        & \defas \seval{e_1}{\sval} \vee \seval{e_2}{\sval}                                                       \\
                                                          \seval{e_1 + e_2}{\sval}     & \defas \seval{e_1}{\sval} + \seval{e_2}{\sval}                                                          &
                                                                                         \seval{e_1 \cdot e_2}{\sval} & \defas \seval{e_1}{\sval} \cdot \seval{e_2}{\sval}  \\
     \seval{\frac{1}{e_1}}{\sval} &\defas \frac{1}{\seval{e_1}{\sval}}\\
	 \seval{e_1 = e_2}{\sval}                                                          & \defas \seval{e_1}{\sval} = \seval{e_2}{\sval}                                                          &
	 \seval{e_1 \leq e_2}{\sval}                                                       & \defas \seval{e_1}{\sval} \leq \seval{e_2}{\sval}
      \end{align*}
      \begin{align*}
      \seval{\ifte{\sexpr_1}{\sexpr_2}{\sexpr_3}}{\sval}                               & \defas
                                                                                         \begin{cases}
                                                                                           \seval{\sexpr_2}{\sval} &\mathtext{if}\; \seval{\sexpr_1}{\sval}\\
                                                                                           \seval{\sexpr_3}{\sval} &\mathtext{otherwise}
                                                                                         \end{cases}
    \end{align*}
\end{Definition}
}
\subsection{Incomplete Expression Evaluation}\label{sec:expr-incomplete}

We now define evaluation of expressions over incomplete valuations, which are sets of valuations.
Each valuation in such a set, called a possible world, represents one possible input for the expression.
The semantics of expression evaluation are then defined using 
possible worlds semantics: the result of evaluating an expression $\sexpr$ over an incomplete valuation $\uval = \{ \sval_1, \ldots, \sval_n \}$
\ifnottechreport{ denoted as $\seval{\sexpr}{\uval}$
is the set of results obtained by evaluating $\sexpr$ over each $\sval_i$ using deterministic expression evaluation semantics:
$$\seval{\sexpr}{\uval} \defas \{ \seval{\sexpr}{\sval} \mid \sval \in \uval \}$$
Consider an expression $\sexpr \defas x+y$ and an incomplete valuation $\uval = \{ (x=1,y=4), (x=2,y=4), (x=1,y=5) \}$.
We get $\seval{\sexpr}{\uval} = \{ 1+4, 2+4, 1+5 \} = \{5,6\}$. }
\iftechreport{
is the set of results obtained by evaluating $\sexpr$ over each $\sval_i$ using the deterministic expression evaluation semantics defined above.
\begin{Definition}[Incomplete Expression Semantics]\label{def:incomplete-expr-sem}
An incomplete valuation $\uval$ is a set $\{ \sval_1, \ldots, \sval_n \}$ where each $\sval_i$ is a valuation.
The result of evaluating an expression $\sexpr$ over $\uval$ denoted as $\seval{\sexpr}{\uval}$ is:
$$\seval{\sexpr}{\uval} \defas \{ \seval{\sexpr}{\sval} \mid \sval \in \uval \}$$
\end{Definition}
\begin{Example}
	Consider an expression $\sexpr \defas x+y$ and an incomplete valuation with possible bindings $\uval = \{ (x=1,y=4), (x=2,y=4), (x=1,y=5) \}$. Applying deterministic evaluation semantics for each of the three valuations from $\uval$ we get  $1+4=5$ ,$2+5=6$, and $1+5=6$. Thus, the possible outcomes of this expression under this valuation are: $\seval{\sexpr}{\uval} = \{5,6\}$.
\end{Example}
}


\ifnottechreport{
\subsection{Range-Annotated Domains}\label{sec:range-dom}
}
\iftechreport{
\subsection{Range-Annotated Domains}\label{sec:range-dom}
}
We now define \textit{range-annotated values}, which are domain values that are annotated with an interval that bounds the value from above and below. We assume an order $\leq$ for $\dataDomain$ preserved under addition.
\revb{For categorical values where no sensible order can be defined, we impose an arbitrary order.
Note that in the worst-case, we can just annotate a value with the range covering the whole domain to indicate that it is completely uncertain.}
We define an expression semantics for valuations that maps variables to range-annotated values and then prove that if the input bounds an incomplete valuation, then the range-annotated output produced by this semantics bounds the possible outcomes of the incomplete expression.

\begin{Definition}
  \label{def:range-domain}
Let $\dataDomain$ be a domain and let $\leq$ denote a total order over its elements. Then the \emph{range-annotated domain} $\rangeDom$ is defined as:
$$\left\{ \uv{\lbMarker{c}}{\bgMarker{c}}{\ubMarker{c}} \mid \lbMarker{c}, \bgMarker{c}, \ubMarker{c} \in \dataDomain \wedge \lbMarker{c} \leq \bgMarker{c} \leq \ubMarker{c} \right\}$$

\end{Definition}

A value $c = \uv{\lbMarker{c}}{\bgMarker{c}}{\ubMarker{c}}$ from $\rangeDom$ encodes a value $\bgMarker{c} \in \dataDomain$ and two values ($\lbMarker{c}$ and $\ubMarker{c}$) that bound $\bgMarker{c}$ from below and above.
We call a value $c \in \rangeDom$  \emph{certain} if $\lbMarker{c}=\bgMarker{c}=\ubMarker{c}$.
Observe, that the definition requires that for any $c \in \rangeDom$ we have $\lbMarker{c} \leq \bgMarker{c} \leq \ubMarker{c}$.
\iftechreport{
\begin{Example}
  For the boolean domain $\dataDomain=\{\bfalse,\btrue\}$ with order $\bfalse < \btrue$, the corresponding range annotated domain is:
  $$\rangeDom=\{\uv{\btrue}{\btrue}{\btrue},\uv{\bfalse}{\btrue}{\btrue},\uv{\bfalse}{\bfalse}{\btrue},\uv{\bfalse}{\bfalse}{\bfalse}\}$$
\end{Example}
}
We use valuations that map the variables of an expression to elements from $\rangeDom$ 
to 
bound 
incomplete valuations.

\iftechreport{
\begin{Definition}[Range-annotated valuation]\label{def:range-val}
Let $\sexpr$ be an expression. A \emph{range-annotated valuation} $\rval$ for $\sexpr$ is a mapping $\vars(\sexpr) \to \rangeDom$.
\end{Definition}
}
\begin{Definition}
  \label{def:range-expr-bound}
  \ifnottechreport{A \emph{range-annotated valuation} $\rval$ for an expression $\sexpr$ is a mapping $\vars(\sexpr) \to \rangeDom$.}
  Given an incomplete valuation $\uval$ 
  and  a range-annotated valuation $\rval$ for $\sexpr$,
  we say that $\rval$ bounds $\uval$ iff
  \begin{align*}
    \forall x \in \vars(\sexpr): \forall \sval \in \uval: \lbMarker{\rval(x)} \leq \sval(x) \leq \ubMarker{\rval(x)}\\
    \exists \sval \in \uval: \forall x \in \vars(\sexpr): \sval(x) = \bgMarker{\rval(x)}
  \end{align*}
\end{Definition}

\ifnottechreport{
Consider the incomplete valuation $\uval = \{(x=1),(x=2),(x=3)\}$. The range-annotated valuation  $x=\uv{0}{2}{3}$ is a bound for $\uval$, while $x=\uv{0}{2}{2}$ is not a bound.
}
\iftechreport{
\begin{Example}
Consider the incomplete valuation $\uval = \{(x=1),(x=2),(x=3)\}$. The range-annotated valuation  $x=\uv{0}{2}{3}$ is a bound for $\uval$, while $x=\uv{0}{2}{2}$ is not a bound.
  \end{Example}
}
\iftechreport{
\subsection{Range-annotated Expression Evaluation}\label{sec:range-expr-eval}
}
We now define a semantics for evaluating expressions over range-annotated valuations. We then demonstrate that this semantics preserves bounds.

\begin{Definition}\iftechreport{[Range-annotated expression evaluation]}
  \label{def:range-expr-eval}
  Let $\sexpr$ be an expression. Given a range valuation $\rval: \vars(e) \rightarrow \rangeDom$,
we define 
$\bgMarker{\rval}(x) \defas \bgMarker{\rval(x)}$.
  The result of expression $\sexpr$ over $\rval$ denoted as $\seval{\sexpr}{\rval}$ is defined as:
    \begin{align*}
      \seval{x}{\rval} &\defas \uv{\lbMarker{\rval(x)}}{\bgMarker{\rval(x)}}{\ubMarker{\rval(x)}} &
      \seval{c}{\rval} &\defas \uv{c}{c}{c}
    \end{align*}
    \ifnottechreport{
    For any of the following expressions we define $\bgMarker{\seval{e}{\rval}} \defas \seval{e}{\bgMarker{\rval}}$. Let $\seval{e_1}{\rval} = a$, $\seval{e_2}{\rval} = b$, and $\seval{e_3}{\rval} = c$. All expressions omitted below are defined point-wise (e.g., $\lbMarker{\seval{e_1 + e_2}{\rval}} \defas \lbMarker{a} + \lbMarker{b}$).
   \begin{align*}
      \lbMarker{\seval{\neg e_1}{\rval}} &\defas \neg\, \ubMarker{a} &
      \ubMarker{\seval{\neg e_1}{\rval}} &\defas \neg\, \lbMarker{a}
    \end{align*}\\[-10mm]
    \begin{align*}
      \lbMarker{\seval{e_1 \cdot e_2}{\rval}}                                                                                         &\defas \min(\ubMarker{a} \cdot \ubMarker{b},\ubMarker{a} \cdot \lbMarker{b},\lbMarker{a} \cdot \ubMarker{b},\lbMarker{a} \cdot \lbMarker{b})\\
      \ubMarker{\seval{e_1 \cdot e_2}{\rval}}                                                                       &\defas \max(\ubMarker{a} \cdot \ubMarker{b},\ubMarker{a} \cdot \lbMarker{b},\lbMarker{a} \cdot \ubMarker{b},\lbMarker{a} \cdot \lbMarker{b})
    \end{align*}\\[-10mm]
    \begin{align*}
      \lbMarker{\seval{e_1 \leq e_2}{\rval}} &\defas \ubMarker{a} \leq \lbMarker{b} &
 \ubMarker{\seval{e_1 \leq e_2}{\rval}} &\defas \lbMarker{a} \leq \ubMarker{b} \\
\lbMarker{\seval{e_1=e_2}{\rval}} &\defas (\ubMarker{a}=\lbMarker{b} \wedge \ubMarker{b}=\lbMarker{a} ) &
                                                                                                 \ubMarker{\seval{e_1=e_2}{\rval}} &\defas \lbMarker{a} \leq \ubMarker{b} \wedge \lbMarker{b} \leq \ubMarker{a}
    \end{align*}\\[-10mm]
    \begin{align*}
\lbMarker{\seval{\ifte{e_1}{e_2}{e_3}}{\rval}} &\defas
 \begin{cases}
\lbMarker{b} & \text{if } \lbMarker{a}= \ubMarker{a}=\btrue \\
\lbMarker{c} & \text{if } \lbMarker{a}= \ubMarker{a}=\bfalse \\
 \min(\lbMarker{b},\lbMarker{c}) & \text{else}
\end{cases}\\
      \ubMarker{\seval{\ifte{e_1}{e_2}{e_3}}{\rval}} &\defas
 \begin{cases}
 \ubMarker{b} & \text{if } \lbMarker{a}= \ubMarker{a}=\btrue \\
\ubMarker{c} & \text{if } \lbMarker{a}= \ubMarker{a}=\bfalse \\
\max(\ubMarker{b},\ubMarker{c}) & \text{else}
\end{cases}
    \end{align*}
    }
    \iftechreport{
Note that $\seval{\frac{1}{\sexpr}}{\rval}$ is undefined if $\lbMarker{\seval{\sexpr}{\rval}} \leq 0$ and $\ubMarker{\seval{\sexpr}{\rval}} \geq 0$, because then $\rval$ may bound a valuation $\sval$ where $\seval{\sexpr}{\sval} = 0$.
      For any of the following expressions we define $\bgMarker{\seval{e}{\rval}} \defas \seval{e}{\bgMarker{\rval}}$. Let $\seval{e_1}{\rval} = a$, $\seval{e_2}{\rval} = b$, and $\seval{e_3}{\rval} = c$. Then,
      \begin{align*}
	 \lbMarker{\seval{e_1 \wedge e_2}{\rval}} &\defas \lbMarker{a} \wedge \lbMarker{b} &
	 \ubMarker{\seval{e_1 \wedge e_2}{\rval}} &\defas \ubMarker{a} \wedge \ubMarker{b}\\
      \lbMarker{\seval{e_1 \vee e_2}{\rval}} &\defas \lbMarker{a} \vee \lbMarker{b} &
      \ubMarker{\seval{e_1 \vee e_2}{\rval}} &\defas \ubMarker{a} \vee \ubMarker{b} \\
      \lbMarker{\seval{\neg e_1}{\rval}} &\defas \neg\, \ubMarker{a} &
      \ubMarker{\seval{\neg e_1}{\rval}} &\defas \neg\, \lbMarker{a} \\
      \lbMarker{\seval{e_1 + e_2}{\rval}} &\defas \lbMarker{a} + \lbMarker{b} &
      \ubMarker{\seval{e_1 + e_2}{\rval}} &\defas \ubMarker{a} + \ubMarker{b}
    \end{align*}\\[-10mm]
    \begin{align*}
      \lbMarker{\seval{e_1 \cdot e_2}{\rval}}                                                                                         &\defas \min(\ubMarker{a} \cdot \ubMarker{b},\ubMarker{a} \cdot \lbMarker{b},\lbMarker{a} \cdot \ubMarker{b},\lbMarker{a} \cdot \lbMarker{b})\\
      \ubMarker{\seval{e_1 \cdot e_2}{\rval}}                                                                       &\defas \max(\ubMarker{a} \cdot \ubMarker{b},\ubMarker{a} \cdot \lbMarker{b},\lbMarker{a} \cdot \ubMarker{b},\lbMarker{a} \cdot \lbMarker{b})\\
                                                                                                                     \lbMarker{\seval{\frac{1}{e_1}}{\rval}}                                                                                         &\defas \frac{1}{\ubMarker{a}})\\
      \ubMarker{\seval{\frac{1}{e_1}}{\rval}}                                                                       &\defas \frac{1}{\lbMarker{a}}
    \end{align*}
    \begin{align*}
      \lbMarker{\seval{a \leq b}{\rval}} &\defas \ubMarker{a} \leq \lbMarker{b} &
 \ubMarker{\seval{a \leq b}{\rval}} &\defas \lbMarker{a} \leq \ubMarker{b} \\
\lbMarker{\seval{a=b}{\rval}} &\defas (\ubMarker{a}=\lbMarker{b} \wedge \ubMarker{b}=\lbMarker{a} ) &
                                                                                                 \ubMarker{\seval{a=b}{\rval}} &\defas \lbMarker{a} \leq \ubMarker{b} \wedge \lbMarker{b} \leq \ubMarker{a}
    \end{align*}\\[-10mm]
    \begin{align*}
\lbMarker{\seval{\ifte{e_1}{e_2}{e_3}}{\rval}} &\defas
 \begin{cases}
\lbMarker{b} & \text{if } \lbMarker{a}= \ubMarker{a}=\btrue \\
\lbMarker{c} & \text{if } \lbMarker{a}= \ubMarker{a}=\bfalse \\
 \min(\lbMarker{b},\lbMarker{c}) & \text{else}
\end{cases}\\
      \ubMarker{\seval{\ifte{e_1}{e_2}{e_3}}{\rval}} &\defas
 \begin{cases}
 \ubMarker{b} & \text{if } \lbMarker{a}= \ubMarker{a}=\btrue \\
\ubMarker{c} & \text{if } \lbMarker{a}= \ubMarker{a}=\bfalse \\
\max(\ubMarker{b},\ubMarker{c}) & \text{else}
\end{cases}
    \end{align*}
      }
\end{Definition}





\iftechreport{
\subsection{Preservation of Bounds}\label{sec:range-expr-eval-preserves-bounds}
}

Assuming that an input range-annotated valuation bounds an incomplete valuation, we need to prove that the output of range-annotated expression evaluation also  bounds the possible outcomes. 

\begin{Definition}
  \label{def:bounding-vals}
  A value $c \in \rangeDom$ 
  bounds a set of values $S \subseteq \dataDomain$ if: 
  \begin{align*}
    &\forall c_i \in S: \lbMarker{c} \leq c_i \leq \ubMarker{c}
    &&\exists c_i \in S: c_i = \bgMarker{c}
  \end{align*}
\end{Definition}

\begin{Theorem}
  \label{theo:expr-bound}
  Let $\sexpr$ be an expression, $\uval$ an incomplete valuation for $\sexpr$, and $\rval$ a range-annotated valuation 
  that bounds $\uval$, then $\seval{\sexpr}{\rval}$ bounds $\seval{\sexpr}{\uval}$.
\end{Theorem}
\ifnottechreport{
  \begin{proof}[Proof Sketch]
Proven by straightforward induction over the structure of formulas. We present the full proof in \cite{techreport}.
  \end{proof}
}
\iftechreport{
\begin{proof}
  We prove this theorem through induction over the structure of an expression under the assumption that $\rval$ bounds $\uval$.

\proofpara{Base case}
If $e \defas c$ for a constant $c$, then $\lbMarker{e} = \bgMarker{e} = \ubMarker{e} = c$ which is also the result of $e$ in any possible world of $\uval$. If $e \defas x$ for a variable $x$, then since $\rval$ bounds $\uval$, the value of $x$ in any possible world is bounded by $\rval(x)$.

\proofpara{Induction step}
Assume that for expressions $e_1$, $e_2$, and $e_3$, we have that their results under $\uval$ are bounded by their result under $\rval$:

\begin{align*}
\forall i \in \{1,2,3\}: \forall c \in \seval{e_i}{\uval}: \lbMarker{\seval{e_i}{\rval}} \leq c \leq \ubMarker{\seval{e_i}{\rval}}\\
\exists \sval \in \uval: \forall i \in \{1,2,3\}: \bgMarker{\seval{e_i}{\rval}} = \seval{e_i}{\sval}
\end{align*}
Note that the second condition trivially holds since $\bgMarker{\seval{e}{\rval}}$ was defined as applying deterministic expression semantics to $\bgMarker{\rval}$. We, thus, only have to prove  that the lower and upper bounds are preserved for all expressions $e$ that combine these expressions using one of the scalar, conditional, or logical operators.

\proofpara{$e \defas e_1 + e_2$}
Inequalities are preserved under addition. Thus, for any $\sval \in \uval$ we have $\lbMarker{\seval{e_1}{\rval}} + \lbMarker{\seval{e_2}{\rval}} \leq \seval{e_1}{\sval} + \seval{e_2}{\sval} \leq \ubMarker{\seval{e_1}{\rval}} + \ubMarker{\seval{e_2}{\rval}}$.

\proofpara{$e \defas e_1 \cdot e_2$}
We distinguish sixteen cases based on which of $\lbMarker{\seval{e_1}{\rval}}$, $\lbMarker{\seval{e_2}{\rval}}$, $\ubMarker{\seval{e_1}{\rval}}$, and $\lbMarker{\seval{e_2}{\rval}}$ are negative. For instance, if all numbers are positive then clearly $\lbMarker{\seval{e_1}{\rval}} \cdot \lbMarker{\seval{e_2}{\rval}} \leq \seval{e_1}{\sval} \cdot \seval{e_2}{\sval}$. While there are sixteen cases, there are only four possible combinations of lower and upper bounds we have to consider. Thus, if we take the minimal (maximal) value across all these cases, we get a lower (upper) bound on $e$.

\proofpara{$e \defas \frac{1}{e_1}$}
For any pair of numbers $c_1$ and $c_2$ that are either both positive or both negative, we have $c_1 \leq c_2$ implies $\frac{1}{c_1} \geq \frac{1}{c_2}$. Thus, $\frac{1}{\ubMarker{a}}$ is an upper bound on $\frac{1}{c}$ for any $c$ bound by $a$. Analog, $\frac{1}{\lbMarker{a}}$ is an upper bound.

\proofpara{$e \defas e_1 \wedge e_2$ and $e \defas e_1 \vee e_2$}
Both $\vee$ and $\wedge$ are monotone in their arguments wrt. the order $F \ordB T$. Thus, applying these operations to combine lower (upper) bounds preserves these bounds.

\proofpara{$e \defas \neg\,e_1$}
We distinguish three cases: (i) $\seval{e_1}{\sval} = \bfalse$ for all $\sval \in \uval$; (ii)$\seval{e_1}{\sval} = \btrue$ for some $\sval \in \uval$ and $\seval{e_1}{\sval} = \bfalse$ for some $\sval \in \uval$; and (iii) $\seval{e_1}{\sval} = \bfalse$ for all $\sval \in \uval$. In case (i) for $\rval$ to bound the input either $\seval{e_1}{\rval} = \uv{\bfalse}{\bfalse}{\bfalse}$ in which case $\seval{e}{\rval} = \uv{\btrue}{\btrue}{\btrue}$ or $\seval{e_1}{\rval} = \uv{\bfalse}{\bfalse}{\btrue}$ and $\seval{r}{\rval} = \uv{\bfalse}{\btrue}{\btrue}$. We have $\seval{e}{\sval} = \btrue$ for all $\sval \in \uval$ and, thus, in either case $\seval{e}{\rval}$ bounds $\seval{e}{\uval}$. In case (ii), $\lbMarker{\seval{e}{\rval}} = \bfalse$ and $\ubMarker{\seval{e}{\rval}} = \btrue$ which trivially bound $\seval{e}{\uval}$. The last case is symmetric to (i).

\proofpara{$e \defas e_1 \leq e_2$}
Recall that $\bfalse \leq \btrue$.  $e_1 \leq  e_2$ is guaranteed to evaluate to true in every possible world if the upper bound of $e_1$ is lower than or equal to the lower bound of $e_2$. In this case it is safe to set $\lbMarker{\seval{e}{\rval}} = \btrue$. Otherwise, there may exist a possible world where $e_1 \leq e_2$ evaluates to false and we have to set $\lbMarker{\seval{e}{\rval}} = \bfalse$. Similarly, if the lower bound of $e_1$ is larger than the upper bound of $e_2$ then $e_1 \leq e_2$ evaluates to false in every possible world and $\ubMarker{\seval{e}{\rval}} = \bfalse$ is an upper bound. Otherwise, there may exist a world where $e_1 \leq e_2$ holds and we have to set $\ubMarker{\seval{e}{\rval}} = \btrue$.

\proofpara{$\ifte{e_1}{e_2}{e_3}$}
When $e_1$ is certainly true ($\lbMarker{\seval{e_1}{\rval}} = \ubMarker{\seval{e_1}{\rval}} = \btrue$) or certainly false ($\lbMarker{\seval{e_1}{\rval}} = \ubMarker{\seval{e_1}{\rval}} = \bfalse$) then the bounds $e_2$ (certainly true) or $e_3$ (certainly false) are bounds for $e$. Otherwise, $e$ may evaluate to $e_2$ in some worlds and to $e_3$ in others. Taking the minimum (maximum) of the bounds for $e_2$ and $e_3$ is guaranteed to bound $e$ from below (above) in any possible world.

We conclude that the result of range-annotated expression evaluation under $\rval$ which bounds an incomplete valuation $\uval$ bounds the result of incomplete expression evaluation for any expression $e$.
\end{proof}
}


\section{\captialUAADBs}
\label{sec:uaa-range}

We define \termUAADBs (\abbrUAADBs) as a special type of $\semK$-relations over range-annotated domains and demonstrate how to bound an incomplete $\semK$-relation using this model. Afterwards, define a metric for how precise the bounds of an incomplete $\semK$-database encoded by a \abbrUAADB are and proceed to define a query semantics for \abbrUAADBs and prove that this query semantics preserves bounds.
Tuple annotation of \abbrUAADBs are triples of elements from a semiring $\semK$. These triples form a semiring structure $\uaaK{\semK}$. The construction underlying $\uaaK{\semK}$ is well-defined if $\semK$ is an l-semiring, i.e., a semiring where the natural order forms a lattice over the elements of the semiring. Importantly, $\semN$ (bag semantics), $\semB$ (set semantics), and many provenance semirings are l-semirings.

\subsection{\abbrUAADBs}\label{sec:uaadbs}

In addition to allowing for range-annotated values, \abbrUAADBs also differ from \abbrUADBs in that they encode an upper bound of the possible annotation of tuples. Thus, instead of using annotations from $\aDoubleK$, we use $\semkq$ to encode three annotations for each tuple: a lower bound on the certain annotation of the tuple, the annotation of the tuple in the \abbrBGW, and an over-approximation of the tuple's possible annotation.

\begin{Definition}[Tuple-level Annotations]\label{def:uaa-tuple-annot}
  Let $\semK$ be an l-semiring and let $\ordK$ denote its natural order. Then the tuple level range-annotated domain $\uaaDom{\domK}$ is defined as:

  $$\{ \ut{\lbMarker{k}}{k}{\ubMarker{k}} \mid k, \lbMarker{k}, \ubMarker{k} \in \semK \wedge \lbMarker{k} \ordK k \ordK \ubMarker{k} \}$$
We use $\uaaK{\semK}$ to denote semiring $\semkq$ restricted to elements from $\uaaDom{\domK}$.
  \end{Definition}

  Similar to the range-annotated domain, a value $(k_1, k_2, k_3)$ from
  $\uaaK{\semK}$ encodes a semiring element from $\semK$ and two elements ($k_1$
  and $k_3$) that bound the element from below and above.  Given an
  $\uaaK{\semK}$-element $k = \ut{k_1}{k_2}{k_3}$ we define
  $\lbMarker{k} = k_1$, $\bgMarker{k} = k_2$, and $\ubMarker{k} = k_3$.  Note
  that $\uaaK{\semK}$ is a semiring since when combining two elements of
  $\uaaK{\semK}$ with $\addOf{\semkq}$ and $\multOf{\semkq}$, the result
  $(k_1, k_2, k_3)$ fulfills the requirement $k_1 \ordK k_2 \ordK k_3$. This is
  the case because semiring addition and multiplication preserves
  the natural order of $\semK$ and these operations in $\semkq$ are defined as pointwise
  application of $\addK$ and $\multK$, e.g.,
  $\ut{\lbMarker{k}}{k}{\ubMarker{k}} \addOf{\semkq} \ut{\lbMarker{l}}{l}{\ubMarker{l}} = \ut{\lbMarker{l} \addK \lbMarker{l}}{k \addK l}{\ubMarker{k}}
  \addK \ubMarker{l}))$ and
  $k_1 \ordK k_2 \wedge k_3 \ordK k_4 \Rightarrow k_1 \addK k_3 \ordK k_2 \ordK
  k_4$ for any $k_1, k_2, k_3, k_4 \in \semK$.

\begin{Definition}[$\uaaK{\semK}$-relations]\label{def:uaa-rels}
Given a range-annotated data domain $\rangeDom$ and l-semiring $\semK$, an $\uaaK{\semK}$-relation of arity $n$ is a function $\rel: \rangeDom^n \rightarrow \uaaK{\semK}$.
\end{Definition}

As a notational convenience we show certain values, i.e., values $c \in \rangeDom$ where $\lbMarker{c} = \bgMarker{c} = \ubMarker{c} = c'$, as the deterministic value $c'$ they encode.

\subsection{Extracting \capitalBGW}\label{sec:uaab-get-bgw}

Note that the same tuple $\tup$ may appear more than once in a $\uaaK{\semK}$-relation albeit with different value annotations.
We can extract the \termBGW encoded by a $\uaaK{\semK}$-relation by grouping tuples by the \abbrBG of their attribute values and then summing up their tuple-level \abbrBG annotation.

\begin{Definition}\label{def:bool-to-K}
We lift function $\bgName$ from values to tuples: $\bgName: \rangeDom^n \to \dataDomain^n$, i.e., given an \abbrUAADB tuple $\rangeTup = (v_1, \ldots, v_n)$, 
$\bgOf{\rangeTup} \defas (\bgMarker{v_1}, \ldots, \bgMarker{v_n})$. For a  $\uaaK{\semK}$-relation $\rangeRel$, $\bgOf{\rangeRel}$, the \abbrBGW encoded by $\rangeRel$, is then defined as:
$$\bgOf{\rangeRel}(\tup) \defas \sum_{\bgOf{\rangeTup} = \tup} \bgOf{\rangeRel(\rangeTup)}$$
\end{Definition}

\begin{Example}\label{ex:audb-instance}
Figure~\ref{table:UAAR_inst} shows an instance of a $\uaaK{\semN}$-relation $\rel$ where each attribute is a triple showing the lower bound, \termBG and upper bound of the value. Each tuple is annotated by a triple showing the lower bound, \termBG and upper bound of the annotation value. Since this is a $\uaaK{\semN}$ relation, the annotations encode multiplicities of tuples. For example, the first tuple represents a tuple $(1,1)$ that appears at least twice in every possible world (its lower bound annotation is $2$), appears twice in the \abbrBGW, and may appear in any possible world at most thrice.
	Figure~\ref{table:UAAR_bg} shows the \abbrBGW encoded by the \abbrUAADB produced by summing up the annotations of tuples with identical \abbrBG values. For instance, the first two tuples both represent tuple $(1,1)$ and their annotations sum up to $5$, i.e., the tuple $(1,1)$ appears five times in the chosen \abbrBGW.
\end{Example}

\begin{figure}[t]
	\centering

    \begin{minipage}{0.485\linewidth}
	\begin{subtable}{\linewidth}
	\centering
	\begin{tabular}{ c|cc}
      \textbf{A}  & \textbf{B}  & \underline{\semqN} \\
      				\cline{1-2}
		$\uv{1}{1}{1}$ & $\uv{1}{1}{1}$ & \ut{2}{2}{3}\\
		$\uv{1}{1}{1}$ & $\uv{1}{1}{3}$ & \ut{2}{3}{3}\\
		$\uv{1}{2}{2}$ & $\uv{3}{3}{3}$ & \ut{1}{1}{1}\\
	\end{tabular}
	\caption{Example \abbrUAADB instance}
	\label{table:UAAR_inst}
	\end{subtable}
  \end{minipage}
  \begin{minipage}{0.485\linewidth}
	\begin{subtable}{\linewidth}
	\centering
	\begin{tabular}{ c|cc}
		\textbf{A}  & \textbf{B}  & \underline{$\semN$}  \\
		\cline{1-2}
		$1$ & $1$ & 5\\
		$2$ & $3$ & 1\\
	\end{tabular}
	\caption{\termBGW}
	\label{table:UAAR_bg}
	\end{subtable}
  \end{minipage}
	\caption{Example \abbrUAADB relation and the \abbrBGW it encodes}
\end{figure}
\subsection{Encoding Bounds}\label{sec:uaadb-bounds}

We now formally define what it means for an \abbrUAADB to bound a an incomplete $\semK$-relation from above and below. For that we first define bounding of deterministic tuples by range-annotated tuples.

\begin{Definition}[Tuple Bounding]\label{def:bounding-tuples}
  Let  $\rangeTup$ be a  range-annotated tuple with schema $(a_1, \ldots, a_n)$ and   $\tup$ be a tuple  with same schema as $\rangeTup$. We say that $\rangeTup$ bounds $\tup$ written as $\tup \tmatch \rangeTup$ iff
  $$\forall{i \in \{1, \ldots ,n\}}:
  \lbMarker{\rangeTup.a_i} \leq \tup.a_i \leq \ubMarker{\rangeTup.a_i}$$
\end{Definition}


Obviously, one \abbrUAADB tuple can bound multiple different conventional tuples and vice versa. We introduce \textit{tuple matchings} as a way to match the annotations of  tuples of a $\uaaK{\semK}$-database (or relation) with that of one possible world of an incomplete $\semK$-database (or relation).
Based on tuple matchings we then define how to  bound possible worlds.

\begin{Definition}[Tuple matching]\label{def:tuple-matching}
  Let $n$-ary \abbrUAARel  $\rangeRel$ and an $n$-ary database $R$. A tuple matching $\TM$ for $\rangeRel$ and $\rel$ is a function $(\rangeDom)^{n} \times \dataDomain^n \to \semK$.
  \st
$$\forall \rangeTup \in \rangeDom^n: \forall \tup \ntmatch \rangeTup: \TM(\rangeTup,\tup) = \zeroK$$
and
$$\forall \tup \in \dataDomain^n: \sum_{\rangeTup \in \rangeDom^n} \TM(\rangeTup,\tup)=\rel(\tup)$$
\end{Definition}

Intuitively, a tuple matching distributes the annotation of a tuple from $\rel$ over one or more matching tuples from $\rangeRel$. That is, multiple tuples from a \abbrUADB may encode the same tuple from an incomplete database. This is possible when the multidimensional rectangles of their attribute-level range annotations overlap. For instance, range-annotated tuples $(\uv{1}{2}{3})$ and $(\uv{2}{3}{5})$ both match the tuple $(2)$.

\begin{Definition}[Bounding Possible Worlds]\label{def:bounding-worlds}
  Given an n-ary \abbrUAADB relation $\rangeRel$ and a n-ary deterministic relation $\rel$ (a possible world of an incomplete $\semK$-relation),  relation $\rangeRel$ is a lower bound for $\rel$ iff there exists a tuple matching $\TM$ for $\rangeRel$ and $\rel$ \st
  \begin{align}\label{eq:lower-bound-db}
    \forall \rangeTup \in \rangeDom^n:\sum_{\tup \in \dataDomain^n} \TM(\rangeTup,\tup) \geqK \lbMarker{\rangeRel(\rangeTup)}
  \end{align}
  and is upper bounded by $\rangeRel$ iff there exists a tuple matching $\TM$ for $\rangeRel$ and $\rel$ \st
  \begin{align}\label{eq:upper-bound-db}
    \forall \rangeTup \in \rangeDom^n: \sum_{\tup \in \dataDomain^n} \TM(\rangeTup,\tup) \ordK \ubMarker{\rangeRel(\rangeTup)}
  \end{align}
  A \abbrUAARel $\rangeRel$ bounds a relation $\rel$ written as $\rel \dbbounds \rangeRel$ iff there exists a tuple matching $\TM$ for $\rangeRel$ and $\rel$ that fulfills both \Cref{eq:lower-bound-db,eq:upper-bound-db}.
  \end{Definition}

Having defined when a possible world is bound by a $\uaaK{\semK}$-relation, we are ready to define bounding of incomplete $\semK$-relations.

\begin{Definition}[Bounding Incomplete Relations]\label{def:bounding-incomplet-dbs}
	Given an incomplete $\semK$-relation $\prel$ and a \abbrUAARel $\rangeRel$, we say that $\rangeRel$ bounds $\rel$, written as $\prel \dbbounds \rangeRel$
  iff
    \begin{align}
    \forall \rel \in \prel: \rel \dbbounds \rangeRel\\
    \exists \rel \in \prel: \bgOf{\rangeRel} = \rel
    \end{align}
\end{Definition}

Note that all bounds we define for relations are extended to databases in the obvious way.

\begin{Example}\label{def:bounding-incomplete-dbs}
  Consider the \abbrUAADB from \Cref{ex:audb-instance} and the two possible world shown below.

  \begin{minipage}{0.49\linewidth}
    \centering
    \underline{$\db_1$}                                                                                    \\
    	\begin{tabular}{lc|cc}
                                           & \textbf{A}                 & \textbf{B} & \underline{$\semN$} \\
		\cline{2-3}
$\tup_1$                                   & $1$                        & $1$        & 5                   \\
$\tup_2$                                   & $2$                        & $3$        & 1                   \\
	\end{tabular}
  \end{minipage}
  \begin{minipage}{0.49\linewidth}
    \centering
    \underline{$\db_2$}                                                                                    \\
    \begin{tabular}{lc|cc}
                                           & \textbf{A}                 & \textbf{B} & \underline{$\semN$} \\
		\cline{2-3}
 $\tup_3$                                  & $1$                        & $1$        & 2                   \\
 $\tup_4$                                  & 1                          & 3          & 2                   \\
$\tup_5$                                   & $2$                        & $4$        & 1                   \\
	\end{tabular}
  \end{minipage}
  This \abbrUAADB bounds these worlds, since there exist tuple matchings that provides both a lower and an upper bound for the annotations of the tuples of these worlds. For instance, denoting the tuples from this example as
  \begin{align*}
    \rangeTup_1                            & \defas (\uv{1}{1}{1}, \uv{1}{1}{1})                           \\
                               \rangeTup_2 & \defas (\uv{1}{1}{1}, \uv{1}{1}{3})                           \\
                               \rangeTup_3 & \defas (\uv{1}{2}{2}, \uv{3}{3}{3})
  \end{align*} %
  \iftechreport{tuple matchings $\TM_1$ and $\TM_2$ shown below to bound $\db_1$ and $\db_2$.}
  \ifnottechreport{tuple matching $\TM_1$ shown below bounds $\db_1$.}
  \begin{align*}
    \TM_1(\rangeTup_1, \tup_1)             & = 2
                                           & \TM_1(\rangeTup_2, \tup_1) & = 3
                                           & \TM_1(\rangeTup_3, \tup_1) & = 0                                 \\
    \TM_1(\rangeTup_1, \tup_2)             & = 0
                                           & \TM_1(\rangeTup_2, \tup_2) & = 0
                                           & \TM_1(\rangeTup_3, \tup_2) & = 1
    \iftechreport{\\[2mm]
    \TM_2(\rangeTup_1, \tup_3)             & = 2
                                           & \TM_2(\rangeTup_2, \tup_3) & = 0
                                           & \TM_2(\rangeTup_3, \tup_3) & = 0                                 \\
    \TM_2(\rangeTup_1, \tup_4)             & = 0
                                           & \TM_2(\rangeTup_2, \tup_4) & = 2
                                           & \TM_2(\rangeTup_3, \tup_4) & = 0                                 \\
    \TM_2(\rangeTup_1, \tup_5)             & = 0
                                           & \TM_2(\rangeTup_2, \tup_5) & = 0
                                           & \TM_2(\rangeTup_3, \tup_5) & = 1
}
  \end{align*}
    \end{Example}


\subsection{Tightness of Bounds}
\label{sec:approximation-factor}

\Cref{def:bounding-incomplet-dbs} defines what it means for an \abbrUAADB to bound an incomplete databases. However, given an incomplete database, there may be many possible \abbrUAADBs that bound it that differ in how tight the bounds are. For instance, both $\tup_1 \defas (\uv{1}{15}{100})$ and $\tup_2 \defas (\uv{13}{14}{15})$ bound tuple $(15)$, but intuitively the bounds provided by the second tuple are tighter. In this section we develop a metric for the tightness of the approximation provided by an \abbrUAADB and prove that finding a \abbrUAADB that maximizes tightness 
is intractable. Intuitively, given two \abbrUAADBs $\rangeDB$ and $\rangeDB'$ that both bound an incomplete $\semK$-database $\db$, $\rangeDB$ is a tighter bound than $\rangeDB'$ if the set of deterministic databases bound by $\rangeDB$ is a subset of the set of deterministic databases bound by $\rangeDB'$. As a sanity check, consider $\rangeDB_1 \defas \{t_1\}$ and $\rangeDB_2 \defas \{t_2\}$ using $t_1$ and $t_2$ from above and assume that $\dataDomain = \semN \cup \semB$. Then $\rangeDB_2$ is a tighter bound than $\rangeDB_1$ since the three deterministic databases it bounds $\{(13)\}$, $\{(14)\}$ and $\{(15)\}$ are also bound by $\rangeDB_1$, but $\rangeDB_1$ bounds additional databases, e.g., $\{(2)\}$ that are not bound by $\rangeDB_2$.

\begin{Definition}[Bound Tightness]\label{def:bound-tightness}
  Consider two $\uaaK{\semK}$-databases $\rangeDB$ and $\rangeDB'$ over the same schema
$\aschema$. We say that $\rangeDB$ is at least as tight as $\rangeDB'$, written as $\rangeDB \dbleq \rangeDB'$, if for all $\semK$-databases $\db$ with schema $\aschema$ we have:
  \begin{align*}
 \db \dbbounds \rangeDB \rightarrow \db \dbbounds \rangeDB'
  \end{align*}
  We say that $\rangeDB$ is a strictly tighter than $\rangeDB$, written as $\rangeDB \dble \rangeDB'$ if $\rangeDB \dbleq \rangeDB'$ and there exists $\db \dbbounds \rangeDB'$ with $\db \not\dbbounds \rangeDB$. Furthermore, we call $\rangeDB$  a maximally tight bound for an incomplete $\semK$-database $\pdb$ if:
  \begin{align*}
    &\pdb \dbbounds \rangeDB
   &&\not \exists \rangeDB': \rangeDB' \dble \rangeDB
  \end{align*}
\end{Definition}

Note that the notion of tightness is well-defined even if the data domain $\dataDomain$ is infinite. For instance, if we use the reals $\mathbb{R}$ instead of natural numbers as the domain in the example above, then still $\rangeDB_1 \dbge \rangeDB_2$.
In general \abbrUAADBs that are tighter bounds are preferable. However, computing a maximally tight bound is intractable.

\begin{Theorem}[Finding Maximally Tight Bounds]\label{theo:intractability-of-tight-bounds}
Let $\pdb$ be an incomplete $\semN$-database encoded as a \abbrCtable~\cite{DBLP:journals/jacm/ImielinskiL84}.  
Computing a maximally tight bound $\rangeDB$ for $\pdb$ is \nphard.
\end{Theorem}
\begin{proof}
  Note that obviously, \abbrCtables which apply set semantics cannot encode every possible incomplete $\semN$-database. However, the class of all $\semN$-databases where no tuples appear more than once can be encoded using \abbrCtables. To prove the hardness of computing maximally tight bounds it suffices to prove the hardness of finding bounds for this subset of all $\semN$-databases. We prove the claim through a reduction from the \npcomplete 3-colorability decision problem. A graph $G = (V,E)$ is 3-colorable if  each node $n$ can be assigned a color $C(n) \in \{r,g,b\}$ (red, green, and blue) such that for every edge $e = (v_1,v_2)$ we have $C(v_1) \neq C(v_2)$. Given such a graph, we will construct a \abbrCtable $\prel$ encoding an incomplete $\semB$-relation (\abbrCtables use set semantics) with a single tuple and show that the tight  upper bound on the annotation of the tuple is $\btrue$ iff the graph $G$ is 3-colorable. We now briefly review \abbrCtables for readers not familiar with this model. Consider a set of variables $\Sigma$. A \abbrCtable~\cite{DBLP:journals/jacm/ImielinskiL84} $\prel = (\rel, \lcond, \gcond)$ is a relation $\rel$ paired with (i) a global condition $\gcond$ which is also a logical condition over $\Sigma$ and (ii)  a function $\lcond$ that assigns to each tuple $\tup \in \rel$ a logical condition over $\Sigma$. Given a valuation $\mu$ that assigns to each variable from $\Sigma$ a value, the global condition and all local conditions evaluate to either $\btrue$ or $\bfalse$. The incomplete database represented by a \abbrCtable $\prel$ is the set of all relations $\rel$ such that there exists a valuation $\mu$ for which $\mu(\gcond)$ is true and $\rel = \{ \tup \mid \mu(\lcond(\tup)) \}$, i.e., $\rel$ contains all tuples for which the local condition evaluates to true. Given an input graph $G$, we associate a variable $x_v$ with each vertex $v \in V$. Each possible world of the \abbrCtable we construct encodes one possible assignment of colors to the nodes of the graph. This will be ensured through the global condition which is a conjunction of conditions of the form $(x_v = r \lor x_v = g \lor x_v = b)$ for each node $v \in V$. The \abbrCtable contains a single tuple $\tup_{one} = (1)$ whose local condition tests whether the assignment of nodes to colors is a valid 3-coloring of the input graph. That is, the local condition is a conjunction of conditions of the form $x_{v_1} \neq x_{v_2}$  for every edge $e = (v_1, v_2)$.
Thus, the \abbrCtable $(\rel, \lcond, \gcond)$ we construct for $G$ is:
  \begin{align*}
    \rel &= \{ \tup_{one} \}\;\text{for}\;\tup_{one} = (1)\\
    \gcond &= \bigwedge_{v \in V} (x_v = r \lor x_v = g \lor x_v = b)\\
    \lcond(\tup_{one}) &= \bigwedge_{(v_1, v_2) \in E} x_{v_1} \neq x_{v_2}
  \end{align*}
Note that in any possible world $\rel'$ represented by $\prel$, each $x_v$ is assigned one of the valid colors, because otherwise the global condition would not hold. For each such coloring, the tuple $\tup_{one} = (1)$ exists $\rel'(\tup_{one}) = \btrue$  if no adjacent vertices have the same color, i.e., the graph is 3-colorable. Thus, if $G$ is not 3-colorable, then $\rel'(\tup_{one}) = \bfalse$ in every possible world and if $G$ is 3-colorable, then $\rel'(\tup_{one}) = \btrue$ in at least one possible world. Thus, the tight upper bound on $\tup_{one}$'s annotation is $\btrue$ iff $G$ is 3-colorable.
\end{proof}

In the light of this result, any efficient methods for translating incomplete and probabilistic databases into \abbrUAADBs can not guarantee tight bounds. Nonetheless, comparing the tightness of \abbrUAADBs is useful for evaluating how tight bounds are in practice as we will do in \Cref{sec:experiments}.
Furthermore, note that even if we were able to compute tight bounds for an input incomplete database, preserving the bounds under  queries is computationally hard. This follows from hardness results for computing tight bounds for the results of an aggregation query over incomplete databases (e.g., see~\cite{DBLP:journals/tcs/ArenasBCHRS03}).
\section{\abbrUAADB Query Semantics}
\label{sec:uaadb-query-semantics}

In this section we first introduce a semantics for $\raPlus$ queries over \abbrUAADBs that preserves bounds, i.e., if the input of a query $\query$ bounds an incomplete $\semK$-database $\pdb$, then the output bounds $\query(\pdb)$. Conveniently, it turns out that the standard query semantics for  $\semK$-relations with a slight extension to deal with uncertain boolean values in conditions is sufficient for this purpose. Recall from \Cref{sec:expression} that conditions (or more generally scalar expressions) over range-annotated values evaluate to triples of boolean values, e.g., $\uv{F}{F}{T}$ would mean that the condition is false in some worlds, is false in the \abbrBGW, and may be true in some worlds. Recall the standard semantics for evaluating selection conditions over $\semK$-relations. For a selection $\selection_{\theta}(R)$ the annotation of a tuple $t$ in annotation of $\tup$ in the result of the selection is computed by multiplying $\rel(\tup)$ with $\theta(\tup)$ which is defined as a function $\semB \to \{ \zeroK, \oneK \}$ that returns $\oneK$ if $\theta$ evaluates to true on $\tup$ and $\zeroK$ otherwise. In $\uaaK{\semK}$-relations tuple $\tup$ is a tuple of range-annotated values and, thus,  $\theta$ evaluates to an range-annotated Boolean value as described above.
Using the range-annotated semantics for expressions from \Cref{sec:expression}, a selection condition evaluates to a triple of boolean values $\semB^3$. We need to map such a triple to a corresponding $\uaaK{\semK}$-element to define a semantics for selection that is compatible with $\semK$-relational query semantics.

\begin{Definition}[Boolean to Semiring Mapping]\label{def:lift-bool-to-K}
  Let $\semK$ be a semiring. We define function $\rliftK{\semK}:\semB^3 \rightarrow \semK^3$ as:
  \begin{align*}
    \rliftK{\semK}(b_1,b_2,b_3) &\defas (k_1,k_2,k_3) \,\,\,\text{where}\\
    \forall i \in \{1,2,3\}: k_i &\defas
    \begin{cases} \oneK & \text{if } b_i = true\\
      \zeroK & \text{otherwise}
    \end{cases}
    \end{align*}
\end{Definition}

We use the mapping of range-annotated Boolean values to $\uaaK{\semN}$ elements to define evaluation of selection conditions.

\begin{Definition}[Conditions over Range-annotated Tuples]\label{def:range-selection}
  Let  $\rangeTup$ be a range-annotated tuple and $\theta$ be a Boolean condition over variables representing  attributes from $\rangeTup$. Furthermore, let $\rval_{\rangeTup}$ denote the range-annotated valuation that maps each variable to the corresponding value from $\rangeTup$. We define $\theta(\rangeTup)$, the result of the condition $\theta$ applied to $\rangeTup$ as:
  \begin{align*}
   \theta(\rangeTup) \defas \rliftK{\semN}(\seval{\theta}{\rval_{\rangeTup}})
  \end{align*}
  \end{Definition}

  \begin{Example}\label{ex:selections}
    Consider the example $\uaaK{\semN}$-relation $\rel$ shown below. The single tuple $\rangeTup$ of this relation exists at least once in every possible world, twice in the \abbrBGW, and no possible world contains  more than $3$ tuples bound by this tuple.

    \begin{center}
      {\upshape
    \begin{tabular}{c|cc}
      \thead{A}    & \thead{B} & \underline{$\uaaK{\semN}$} \\ \cline{1-2}
      \uv{1}{2}{3} & $2$         & $\ut{1}{2}{3}$
    \end{tabular}
    }
  \end{center}

    To evaluate query $\query \defas \selection_{A = 2}(\rel)$ over this relations, we first evaluate the expression $A = 2$ using range-annotated expression evaluation semantics. We get $\uv{1}{2}{3} = \uv{2}{2}{2}$ which evaluates to $\uv{F}{T}{T}$. Using $\rliftK{\semN}$, this value is mapped to $\ut{0}{1}{1}$. To calculate the annotation of the tuple in the result of the selection we then multiply these values with the tuple's annotation in $\rel$ and get:
    $$\rel(\rangeTup) \multOf{\uaaK{\semN}} \theta(\rangeTup) = \ut{1}{2}{3} \cdot \ut{0}{1}{1} = \ut{0}{2}{3}$$
    Thus, the tuple may not exist in every possible world of the query result, appears twice in the \abbrBGW query result, and occurs at most three times in any possible world.
  \end{Example}

\subsection{Preservation of Bounds}

For this query semantics to be useful, we need to prove that it preserves bounds. Intuitively, this is true because expressions are evaluated using our range-annotated expression semantics which preserves bounds on values and queries are evaluated in a direct-product semiring $\uaaK{\semN}$ for which semiring operations are defined point-wise. Furthermore, we utilize a result we have proven in \cite[Lemma 2]{FH18}: the operations of l-semirings preserve the natural order, e.g., if $k_1 \ordK k_2$ and $k_3 \ordK k_4$ then $k_1 \addK k_3 \ordK k_2 \addK k_4$.

\begin{Theorem}[\raPlus Queries Preserve Bounds]\label{lem:ra-plus-preserves-bounds}
Let $\pdb$ be an incomplete $\semN$-database, $\query$ be a $\raPlus$ query, and $\rangeDB$ be an $\uaaK{\semN}$-database that bounds $\db$. Then $\query(\rangeDB)$ bounds $\query(\pdb)$.
\end{Theorem}
\begin{proof}
	We prove this lemma using induction over the structure of a relational algebra expression under the assumption that $\rangeDB$ bounds the input $\pdb$.

	\proofpara{Base case}
	The query $\query$ consists of a single relation access $\rel$. The result is bounded following from $\rangeDB \dbbounds \pdb$.

	\proofpara{Induction step}
Let $\rangeRel$ and $\rangeOf{S}$ bound $n$-ary relation $\rel$ and $m$-ary relation $S$. Consider $\db \in \pdb$ and let $\TM_\rel$ and $\TM_{S}$ be two tuple matchings based on which these bounds can be established for $\db$. We will demonstrate how to construct a tuple matching $\TM_Q$ based on which $\query(\rangeDB)$ bounds $\query(\db)$. From this then immediately follows that $\query(\pdb) \dbbounds \query(\rangeDB)$. Note that by definition of $\uaaK{\semK}$ as the 3-way direct product of $\semK$ with itself, semiring operations are point-wise, e.g., $\lbMarker{k_1 \addOf{\uaaK{\semK}} k_2} = \lbMarker{k_1} \addK \lbMarker{k_2}$. Practically, this means that queries are evaluated over each dimension individually. We will make use of this fact in the following. We only prove that $\TM_Q$ is a lower bound since the proof for $\TM_Q$ being an upper bound is symmetric.

\proofpara{$\projection_U(\rangeRel)$}
Recall that for $\TM$ to be a tuple matching, two conditions have to hold: (i) $\TM(\rangeTup, \tup) = \zeroK$ if $\tup \not\tmatch \rangeTup$ and (ii) $\sum_{\rangeTup \in \rangeDom^n} \TM(\rangeTup,\tup) = \rel(\tup)$.
Consider an $U$-tuple $\tup$. Applying the definition of projection for $\semK$-relations we have:
\[
\projection_U(\rel)(\tup) = \sum_{\tup = \tup'[U]} \rel(\tup)
\]

Since $\TM_R$ is a tuple matching based on which $\rangeRel$ bounds $\rel$, we know that by the definition of tuple matching the sum of annotations assigned to a tuple $\tup$ by the tuple matching is equal to the annotation of the tuple in $\rel$):

\begin{equation}\label{eq:project-proof-annotation-equal-tm-sum}
\sum_{\tup = \tup'[U]} \rel(\tup) = \sum_{\tup = \tup'[U]} \sum_{\rangeTup \in \rangeDom^n} \TM_R(\rangeTup, \tup')
\end{equation}

By definition for any tuple matching $\TM$ we have $\TM(\rangeTup,\tup) = \zeroK$ if $\tup \not\tmatch \rangeTup$. Thus, \Cref{eq:project-proof-annotation-equal-tm-sum} can be rewritten as:

\begin{equation}
  \label{eq:2}
= \sum_{\tup = \tup'[U]} \sum_{\tup' \tmatch \rangeTup} \TM_R(\rangeTup, \tup')
\end{equation}

Observe that for any n-ary range-annotated $\rangeTup$ and n-ary tuple $\tup$ it is the case that $\tup \tmatch \rangeTup$ implies $\tup[U] \tmatch \rangeTup[U]$ (if $\rangeTup$ matches $\tup$ on all attributes, then clearly it matches $\tup$ on a subset of attributes). For pair  $\tup$ and $\rangeTup$ such that $\tup[U] \tmatch \rangeTup[U]$, but $\tup \not\tmatch \rangeTup$ we know that $TM_R(\rangeTup, \tup) = \zeroK$. Thus,

\begin{equation}
  \label{eq:3}
= \sum_{\tup = \tup'[U]} \sum_{\rangeTup = \rangeTup'[U] \wedge \tup \tmatch \rangeTup} \TM_R(\rangeTup', \tup')
\end{equation}

So far we have established that:

\begin{equation}
  \label{eq:4}
  \projection_U(\rel)(\tup) = \sum_{\tup = \tup'[U]} \sum_{\rangeTup = \rangeTup'[U] \wedge \tup \tmatch \rangeTup} \TM_R(\rangeTup', \tup')
\end{equation}

We now define $\TM_Q$ as shown below:
\begin{equation}
  \label{eq:5}
  \TM_Q(\rangeTup,\tup) \defas \sum_{\forall \rangeTup', \tup: \rangeTup'[U]=\rangeTup \wedge t'[U]=t}\TM_{R}(\rangeTup', \tup')
\end{equation}

$\TM_Q$ is a tuple matching since \Cref{eq:4} ensures that $\forall \tup \in \dataDomain^n: \sum_{\rangeTup \in \rangeDom^n} \TM(\rangeTup,\tup)=\rel(\tup)$ (second condition in the definition) and we defined $\TM_Q$ such that $\TM_Q(\rangeTup,\tup) = \zeroK$ if $\tup \not\tmatch \rangeTup$. What remains to be shown is that $\projection_U(\rangeRel)$  bounds $\projection_U(\rel)$  based on $\TM_Q$. Let $\card{U} = m$, we have to show that

\[
  \forall {\rangeTup \in \rangeDom^m}: \lbMarker{\projection_U(\rangeRel)(\rangeTup)} \ordK \sum_{\tup \in \dataDomain^m}\TM_Q(\rangeTup, \tup)
  \]

  Since addition in $\uaaK{\semK}$ is pointwise application of $\addK$, using the definition of projection over $\semK$-relations we have

  \[
    \lbMarker{\projection_U(\rangeRel)(\rangeTup)} = \sum_{\rangeTup'[U] = \rangeTup} R(\rangeTup')
  \]

Furthermore, since $\TM_R$ is a tuple matching based on which $\rangeRel$ bounds $\rel$,
\[
= \sum_{\rangeTup'[U] = \rangeTup \wedge t \in \dataDomain^n} \TM_R(\rangeTup, \tup)
  \]

  Using again the fact that $\tup \not \tmatch \rangeTup$ implies $\TM_R(\rangeTup, \tup) = \zeroK$,

  \begin{align*}
    = &\sum_{\forall \tup', \rangeTup': \rangeTup'[U] = \rangeTup \wedge \tup' \tmatch \rangeTup'} \TM_R(\rangeTup', \tup')\\
    = &\sum_{\tup \tmatch \rangeTup} \sum_{\forall \tup', \rangeTup': \rangeTup'[U] = \rangeTup \wedge \tup'[U] = \tup} \TM_R(\rangeTup', \tup')\\
    = &\sum_{\tup \tmatch \rangeTup} \TM_Q(\rangeTup,\tup)     = \sum_{\tup \in \dataDomain^m} \TM_Q(\rangeTup,\tup)
  \end{align*}

Since we have established that $ \lbMarker{\projection_U(\rangeRel)(\rangeTup)} \ordK \sum_{\tup \in \dataDomain^m} \TM_Q(\rangeTup,\tup)$, $\projection_U(\rangeRel)$ lower bounds $\projection_U(\rel)$ via $\TM_Q$.

\proofpara{$\selection_\theta(\rangeRel)$}
By definition of selection and based on (i) and (ii) as in the proof of projection we have
$$\query(\rangeRel)\lbMarker{(\rangeTup)} = \lbMarker{\rangeRel(\rangeTup)} \multK \lbMarker{ \theta(\rangeTup)} \ordK \sum_{t \in \dataDomain^n} \TM_{\rel}(\rangeTup,t) \multK \lbMarker{ \theta(\rangeTup)}$$
Assume that for a tuple $\tup$ we have $t \not\tmatch \rangeTup$, then by \Cref{def:tuple-matching} it follows that $\TM_{rel}(\rangeTup, \tup) = \zeroK$. In this case we get $\lbMarker{\query(\rangeRel)(\rangeTup)} = \rangeRel(\rangeTup) \multK \zeroK = \zeroK$. Since $\zeroK \addK k = k$ for any $k \in \semK$, we get
$$\sum_{t \in \dataDomain^n} \TM_{\rel}(\rangeTup,t) \multK \lbMarker{ \theta(\rangeTup)} = \sum_{t \tmatch \rangeTup} \TM_{\rel}(\rangeTup,t) \multK \lbMarker{ \theta(\rangeTup)}$$
Note that based on \Cref{theo:expr-bound}, we have  $\lbMarker{\theta(\rangeTup)} \ordK \theta(\tup)$ since $\tup \tmatch \rangeTup$ from which follows that:
$\sum_{t \tmatch \rangeTup} \TM_{\rel}(\rangeTup,t) \multK \lbMarker{ \theta(\rangeTup)} \ordK \sum_{t \tmatch \rangeTup} \TM_{\rel}(\rangeTup,t) \multK \lbMarker{ \theta(\tup)}$.
It follows that $Q(\rangeRel)$ lower bounds $Q(\rel)$ through $\TM(\rangeTup,t) \defas \TM_{\rel}(\rangeTup,t) \multK \lbMarker{ \theta(t)}$.

\proofpara{$\rangeRel \times \rangeOf{S}$}
Based on the definition of cross product for $\semK$-relations, (i) from above, and that semiring multiplication preserves natural order we get
$\lbMarker{(\rangeRel \times \rangeOf{S})(\rangeTup)} = \lbMarker{\rangeRel(\rangeTup[\rel])}  \multK \lbMarker{\rangeOf{S}(\rangeTup[S])} \ordK \sum_{t \in \dataDomain^n}\TM_R(\rangeTup[R], \tup) \multK \sum_{t' \in \dataDomain^m}\TM_{S}(\rangeTup[S], \tup')$.
Thus, $\rangeRel \times \rangeOf{S}$ lower bounds $R \times S$ via $\TM_Q(\rangeTup, \tup) \defas \sum_{t \in \dataDomain^n}\TM_R(t,\rangeTup[R]) \multK \sum_{t \in \dataDomain^m}\TM_{S}(t,\rangeTup[S])$.
\proofpara{$\rangeRel \union \rangeOf{S}$}
Assume that $R$ and $S$ are n-ary relations.
Substituting the  definition of union and by (i) and (ii) from above we get:
$\lbMarker{(\rangeRel \union
  \rangeOf{S})(\rangeTup)}=\lbMarker{\rangeRel(\rangeTup)} \addK \lbMarker{\rangeOf{S}(\rangeTup)} \ordK \sum_{t \in \dataDomain^n}\TM_R(\rangeTup, \tup) \addK \sum_{\tup \in \dataDomain^n} \TM_{S}(\rangeTup, \tup)$. Thus, $\rangeRel \union \rangeOf{S}$ lower bounds $R \union S$ via $\TM_Q(\rangeTup, \tup) \defas \sum_{\tup \in \dataDomain^n}\TM_R(\rangeTup, \tup) \addK \sum_{\tup \in \dataDomain^n}\TM_{S}(\rangeTup, \tup)$.
\end{proof}


\section{Set Difference}
\label{sec:set-difference}

In this section, we discuss the evaluation of queries with set difference over \abbrUAADBs.

\subsection{\capitalBG Combiner}\label{sec:BG-combiner}
\BG{Should this be a section by itself?}
In this section we introduce an auxiliary operator for defining set difference  over $\uaaK{\semK}$-relations that merges tuples that have the same values in the \abbrBGW.
The purpose of this operator is to ensure that a tuple in the \abbrBGW is encoded as a single tuple in the \abbrUAADB.

\BG{Explain briefly what its purpose is}

The main purpose for using the merge operator is to prevent tuples from over-reducing or over counting. And make sure we can still extract \abbrBGW from the non-monotone query result.

\begin{Definition}[\abbrBG-Combiner]\label{def:t-combiner}
  Given a \abbrUAADB relation $\rangeRel$, the combine operator $\combine$ yields a \abbrUAADB relation by grouping tuples with the same $\abbrBGW$ attribute values:
  \begin{align*}
\combine(\rangeRel)(\rangeTup) &\defas
\begin{cases}
\sum_{\rangeTup': \bgMarker{\rangeTup}=\bgMarker{\rangeTup'}}\rangeRel(\rangeTup') & \text{if } \rangeTup=Comb(\rangeRel,\bgMarker{\rangeTup}) \\
0_k & \text{else}
\end{cases}
\end{align*}
where $Comb(\rangeRel,t)$ defined below computes the mimimum bounding box for the ranges of all tuples from $\rangeRel$ that have the same \abbrBGW values as $t$ and are not annotated with $\zeroOf{\uaaK{\semK}}$. Let $a$ be an attribute from the schema of $\rangeRel$, then
\begin{align*}
Comb(\rangeRel,t).\lbMarker{a} &=\min_{\bgMarker{\rangeTup}=t \wedge \rangeRel(\rangeTup) \neq \zeroOf{\uaaK{\semK}}} \rangeTup.\lbMarker{a}\\
\bgMarker{Comb(\rangeRel,t).a} &=\bgMarker{\rangeTup.a}\\
Comb(\rangeRel,t).\ubMarker{a} &=\max_{\bgMarker{\rangeTup}=t \wedge \rangeRel(\rangeTup) \neq \zeroOf{\uaaK{\semK}}}\rangeTup.\ubMarker{a}
\end{align*}
\end{Definition}

The \abbrBG-combiner merges all tuples with the same \abbrBG attribute values are combined by merging their attribute ranges and summing up their annotations. For instance, consider a relation $\rangeRel$ with two tuples $(\uv{1}{2}{2}, \uv{1}{3}{5})$ and $(\uv{2}{2}{4}, \uv{3}{3}{4})$ which are annotated with $\ut{1}{2}{2}$ and $\ut{3}{3}{4}$, respectively. Applying \abbrBG-combiner to this relation the two tuples are combined (they have the same \abbrBGW values) into a tuple $(\uv{1}{2}{4}, \uv{1}{3}{5})$ annotated with $\ut{1+3}{2+3}{2+4} = \ut{4}{5}{6}$. Before moving on and discussing semantics for set difference and aggrgeation we first establish that the \abbrBG-combiner preserves bounds.

\begin{Lemma}\label{def:combine preserves bounds}
Let $\rangeRel$ by a $\uaaK{\semK}$-relation that bounds an n-nary $\semK$-relation $\rel$. Then $\combine(\rangeRel)$ bounds $\rel$.
\end{Lemma}
\begin{proof}
Consider a tuple $\tup \in \dataDomain^n$ and let $supp(\rangeRel, \tup)$ denote the set $\{ \rangeTup \mid \rangeRel(\rangeTup) \neq \zeroOf{\uaaK{\semK}} \wedge \bgMarker{\rangeTup} = \tup \}$. Observe that $Comb(\rangeRel, \tup)$ merges the range annotations $supp(\rangeRel, \tup)$. Let $\rangeTup$ be the result of $Comb(\rangeRel, \tup)$. Then
$$\lbMarker{\combine(\rangeRel)(\rangeTup)}=\sum_{\bgMarker{\rangeTup}=\bgMarker{\rangeTup'}}\rangeRel(\rangeTup') \ordK \sum_{\bgMarker{\rangeTup}=\bgMarker{\rangeTup'}} \sum_{\tup \in \dataDomain^n}\TM(\rangeTup,\tup)$$
Thus, $\combine(\rangeRel)$
bounds $\rel$ through $\TM'(\rangeTup, \tup)=\sum_{\bgMarker{\rangeTup}=\bgMarker{\rangeTup'}}\sum_{\tup \in \dataDomain^n}\TM(\rangeTup'), \tup$.
\end{proof}

\subsection{Set Difference}
\label{sec:set-difference-2}

Geerts~\cite{Geerts:2010bz} did extend $\semK$-relations to support set difference through m-semirings which are semirings equipped with a monus operation that is used to define difference. The monus operation $\monK$ is defined based on the natural order of semirings as $k_1 \monK k_2 = k_3$ where $k_3$ is the smallest element from $\semK$ \st $k_2 \addK k_3 \geqK k_1$. For instance, the monus of semiring $\semN$ is truncating subtraction: $k_1 \monOf{\semN} k_2 = max(0, k_1 - k_2)$. The
monus construction for a semiring $\semK$ can be lifted through point-wise application to  $\semkq$ since $\semkq$-semirings are direct products. We get
\[
  \ut{k_1}{k_2}{k_3} \monOf{\uaaK{\semK}} \ut{l_1}{l_2}{l_3} = \ut{k_1 \monK l_1}{k_2 \monK l_2}{k_3 \monK l_3}
\]
However, the result of $k \monOf{\semkq} k'$ for $k, k' \in \uaaK{\semK}$ is not necessarily in $\uaaK{\semK}$, i.e., this semantics for set difference does not preserve bounds even if we disallow range-annotated values.
For instance, consider an incomplete $\semN$-relations with two possible worlds: $\db_1 = \{ R(1) \mapsto 2, S(2) \mapsto 1\}$ and $\db_2 = \{ R(1) \mapsto 1, R(2) \mapsto 1,  S(1) \mapsto 3\}$. Here we use $\tup \mapsto k$ to denote that tuple $\tup$ is annotated with $k$. Without using range-annotations, i.e., we can bound these worlds using $\uaaK{\semN}$-database $\rangeDB_1$:
$\rangeDB \defas \{ \rangeRel(1) \mapsto (1,2,2), \rangeRel(2) \mapsto (0,0,1), \rangeOf{S}(1) \mapsto (0,0,3), \rangeOf{S}(2) \mapsto (0,1,1) \}$. Consider the query $\rel \difference S$. Applying the definition of set difference from Geerts~\cite{Geerts:2010bz} which is $(\rel - S)(\tup) \defas \rel(\tup) \monK S(\tup)$, for tuple $\rangeTup \defas (1)$  we get the annotation  $\rangeRel(\rangeTup) \monOf{\uaaK{\semN}} \rangeOf{S}(\rangeTup) = (1,2,2) \monOf{\uaaK{\semN}} (0,0,3) = (\max(1 - 0,0), \max(2 - 0,0), \max(2 - 3,0)  = (1,2,0)$. However,  $\lbMarker{(1,2,0)} = 1$ is not a lower bound on the certain annotation of $\rangeTup$ , since $\rangeTup$ is not in the result of the query in $db_2$ ($max(1-3,0) = 0$). This failure of the point-wise semantics to preserve bounds is not all uprising if we consider the following  observation from~\cite{GL17}: because of the negation in set difference, a lower bound on certain answers can turns into  an upper bound. To calculate an lower (upper) bound for the result one has to combine a lower bound for the LHS input of the set difference with an upper bound of the RHS. Thus, we can define
\[
(\rangeRel \difference \rangeOf{S})(\rangeTup) \defas
\ut
{\lbMarker{\rangeRel(\rangeTup)} \monK \ubMarker{\rangeOf{S}({\rangeTup})}}
{\bgMarker{\rangeRel(\rangeTup)} \monK \bgMarker{\rangeOf{S}({\rangeTup})}}
{\ubMarker{\rangeRel(\rangeTup)} \monK \lbMarker{\rangeOf{S}({\rangeTup})}
}
\]
to get a result that preserves bounds. For instance, for $\rangeTup \defas (1)$ we get $(\max(1 - 3,0), \max(2 - 0,0), \max(2 - 0,0)) = (0,2,2)$.

This semantics is however still not sufficient if we consider range-annotated values. For instance, consider the following $\uaaK{\semN}$-database $\rangeDB_2$ that also bounds our example incomplete  $\semN$-database :
$\{ \rangeRel(1) \mapsto (1,1,1), \rangeRel(\uv{1}{1}{2}) \mapsto (1,1,1), \rangeOf{S}(\uv{1}{1}{2}) \mapsto (1,1,3) \}$.
Observe that tuple $(1))$ from the \abbrBGW ($\db_1$) is encoded as two tuples in $\rangeDB_2$. To calculate the annotation of this tuple in the \abbrBGW we need to sum up the annotations of all such tuples in the LHS and RHS. To calculate lower bound annotations, we need to also use the sum of annotations of all tuples representing the tuple and then compute the monus of this sum with the sum of all annotations of tuples from the RHS that could be equal to this tuple in some world. Two range-annotated tuples may represent the same tuple in some world if all of their attribute values overlap.
Conversely, to calculate an upper bound it is sufficient to use annotations of RHS tuples if both tuples are certain (they are the same in every possible world).
We use the \abbrBG-combiner operator define above to merge tuples with the same \abbrBG values and then apply the monus using the appropriate set of tuples from the RHS.

 \begin{Definition}[Set Difference]\label{def:set-diff-semantics}
Let $\rangeTup$ and $\rangeTup'$ be n-ary range-annotated tuples with schema $(a_1, \ldots, a_n)$.
   We define a predicate $\rangeTup \equiv \rangeTup'$ that evaluates to true iff $\rangeTup=\rangeTup'$ and both $\rangeTup$ and $\rangeTup'$ are certain and a predicate $\rangeTup \matches \rangeTup'$ that evaluates to true iff $\forall i \in \{1, \ldots, n\}: \lbMarker{\rangeTup.a_i} \ordK \lbMarker{\rangeTup'.a_i} \ordK \ubMarker{\rangeTup.a_i} \vee \lbMarker{\rangeTup.a_1} \ordK \ubMarker{\rangeTup'.a_i} \ordK \ubMarker{\rangeTup.a_i}$. Using these predicates we define set difference as shown below.
   \begin{align*}
     \lbMarker{(\rangeRel_1 - \rangeRel_2)(\rangeTup)}  & \defas \lbMarker{\combine(\rangeRel_1)(\rangeTup)} \monK \sum_{\rangeTup \matches \rangeTup'} \ubMarker{\rangeRel_2(\rangeTup')} \\
     \bgMarker{(\rangeRel_1 - \rangeRel_2)(\rangeTup)} & \defas \bgMarker{\combine(\rangeRel_1)(\rangeTup)} \monK \sum_{\bgOf{\rangeTup} = \bgOf{\rangeTup'}} \bgMarker{\rangeRel_2(\rangeTup')} \\
     \ubMarker{(\rangeRel_1 - \rangeRel_2)(\rangeTup)}      & \defas \ubMarker{\combine(\rangeRel_1)(\rangeTup)} \monK \sum_{\rangeTup \equiv \rangeTup'} \lbMarker{\rangeRel_2(\rangeTup')}
   \end{align*}
 \end{Definition}

\subsection{Bound Preservation}
\label{sec:bound-preservation-diff}

We now demonstrate that the semantics we have defined for set difference preserves bounds.

\begin{Theorem}[Set Difference Preserves Bounds]\label{theo:set-difference-prese}
Let $\query \defas R - S$, $\rel$ and $S$ be incomplete $\semK$-relations, and $\rangeRel$ and $\rangeOf{S}$ be $\uaaK{\semK}$-relations that bound $\rel$ and $S$. Then $\query(\rangeRel,\rangeOf{S})$ bounds $\query(\rel,S)$.
\end{Theorem}
\begin{proof}
\BG{Write up proof}
Given all input tuples are bounded, we first prove that the lower bound of the query semantics reserves the bound. We assume relation $\rangeRel$ is pre-combined s.t. $\rangeRel=\combine(\rangeRel)$ and $R$ preserves the bound. \\
For lower bounds $\lbMarker{(\rangeRel-\rangeOf{S})(\rangeTup)}$, on the L.H.S. of $\monK$ we have $$\lbMarker{\rangeRel(\rangeTup)} \ordK \sum_{\tup \in \dataDomain^n}\TM(\rangeTup,\tup)$$
On the R.H.S. we have
$$\sum_{\rangeTup \matches \rangeTup'} \ubMarker{\rangeOf{S}(\rangeTup')} \geqK \sum_{\rangeTup \matches \rangeTup'}\sum_{\tup' \in \dataDomain^n}\TM(\rangeTup',\tup')$$
thus
\begin{align*}
	\lbMarker{(\rangeRel - \rangeOf{S})(\rangeTup)} \defas & \lbMarker{\rangeRel)(\rangeTup)} \monK \sum_{\rangeTup \matches \rangeTup'}\ubMarker{\rangeOf{S}(\rangeTup')} \\
	\ordK & \sum_{\tup \in \dataDomain^n} \TM(\rangeTup,\tup) \monK \sum_{\rangeTup \matches \rangeTup'}\sum_{\tup' \in \dataDomain^n}\TM(\rangeTup',\tup') \\
	= & \sum_{\tup \in \dataDomain^n}(\TM(\rangeTup,\tup) \monK \sum_{\tup' \in \dataDomain^n}\TM(\rangeTup',\tup'))
\end{align*}
So the lower bounds is bounded by tuple-matching $\forall_{\rangeTup \in \rangeRel}:\TM(\rangeTup,\tup)=\TM(\rangeTup,\tup) \monK \sum_{\tup' \in \dataDomain^n}\TM(\rangeTup',\tup')$.

\end{proof}


\section{Aggregation}\label{sec:aggregation}
\BG{We previously claimed that this works for bags and sets. Does it work for a more general class of semirings? Maybe all semirings where the symbolic semimodule elements can always be reduced to a concrete value?}

We now introduce a semantics for aggregation over \abbrUAADBs that preserves
bounds. We leave a generalization to other semirings to future work. See~\cite{techreport} for a discussion of the challenges involved with that. Importantly, our semantics has \ptime data complexity. One major challenge in defining aggregation over $\semK$-relations
which also applies to our problem setting is that one has to take the
annotations of tuples into account when calculating aggregation function
results. For instance, under bag semantics (semiring $\semN$) the multiplicity
of a tuple affects the result of \lstinline!SUM! aggregation.
We based our semantics for aggregation on earlier results from~\cite{AD11d}.
  For \abbrUAADBs we have to overcome two major new challenges: (i)
since the values of group-by attributes may be uncertain, a tuple's group
membership may be uncertain too and (ii) we are aggregating over range-bounded
values.
\BG{Do we need these solution overview here:?}
To address (ii) we  utilitze our expression semantics for
range-bounded values from \Cref{sec:expression}. However, additional
complications arise when taking the $\uaaK{\semN}$-annotations of tuples into
account. For (i) we will reason about all possible group memberships of
range-annotated tuples to calculate bounds on group-by values, aggregation
function results, and number of result groups.


\subsection{Aggregation Monoids}
Amsterdamer et al.~\cite{AD11d} introduced a semantics for aggregation queries
over $\semK$-relations that commutes with homomorphisms and under which
aggregation results can be encoded with polynomial space.
Contrast this with the
aggregation semantics for c-tables
from~\cite{DBLP:journals/jiis/LechtenborgerSV02} where aggregation results may
be of size exponential in the input size.
\cite{AD11d} deals with aggregation functions that
are commutative monoids $\tuple{\monoid, \amadd, \amzero}$, i.e., where
the values from
$\monoid$ that are the input to aggregation are combined through an operation
$\amadd$ which has a neutral element $\amzero$. Abusing notation, we will use $\monoid$ to both denote the monoid and its domain. A monoid is a
  mathematical structure $\tuple{\monoid, \amadd, \amzero}$ where $\amadd$ is a
  commutative and associative binary operation over $\monoid$, and $\amzero$ is
  the neutral element of $\monoid$. For instance,
  $\msum \defas \tuple{\mathbb{R}, +, 0}$, i.e., addition over the reals can be used
  for sum aggregation. Most standard aggregation
functions (\aggsum, \aggmin, \aggmax, and \aggcount) can be expressed as monoids
or, in the case of \aggavg, can be derived from multiple monoids (count and
sum).
As an example, consider the monoids for $\aggsum$ and $\aggmin$:
$\msum \defas \tuple{\mathbb{R}, +, 0}$ and $\mmin \defas \tuple{\mathbb{R}, \min,
\infty}$.
For $\monoid \in \{\msum, \mmin, \mmax\}$ (\aggcount uses \msum), we define a corresponding monoid  $\rangeMon$ using
range-annotated expression semantics (\Cref{sec:expression}). Note that this gives us aggregation functions which can be applied to range-annotated values and are bound preserving, i.e., the result of the aggregation function bounds all possible results for any set of values bound by the inputs.
For example, $\aggmin$ is expressed as $\min(v,w) \defas \ifte{v \leq w}{v}{w}$.

\begin{Lemma}\label{lem:range-annotated-monoids}
$\rangeMof{\msum}$, $\rangeMof{\mmin}$, $\rangeMof{\mmax}$ are monoids.
\end{Lemma}
\begin{proof}
  Addition in $\rangeDom$ is applied point-wise. Thus, addition in $\rangeDom$
  is commutative and associative and has neutral element $\uv{0}{0}{0}$. Thus,
  $\rangeMof{\msum}$ is a monoid. For $\rangeMof{\mmin}$ if we substitute the
  definition of $\ifte{v \leq w}{v}{w}$ and $v \leq w$ and simplify the
  resulting expression we get
  \begin{align*}
    &\min(\uv{a_1}{a_2}{a_3}, \uv{b_1}{b_2}{b_3})\\
    = &\uv{\min(a_1,b_1)}{\min(a_2,b_2)}{\min(a_3,b_3)}
  \end{align*}
   That is, the operation is
  again applied pointwise and commutativity, associativity, and identity of the
  neutral element ($\uv{\infty}{\infty}{\infty}$) follow from the fact that
  $\mmin$ is a monoid. The proof for max is symmetric.
\end{proof}

Based on \Cref{lem:range-annotated-monoids}, aggregation functions
over range-annotated values preserve bounds.

\begin{Corollary}[Aggregation Functions Preserve Bounds]\label{cor:aggregation-functions-preserve}
  Let $\mathcal{S} = \{c_1, \ldots, c_n\} \subseteq \rangeDom$ be a set of
  range-annotated values and $S = \{d_1, \ldots d_n\} \subseteq \dataDomain$ be
  a set of values such that $c_i$ bounds $d_i$, and
  $\rangeMon \in \{\rangeMof{\msum}, \rangeMof{\mmax}, \rangeMof{\mmin}\}$, then
  using the addition operation of $\rangeMon$ ($\monoid$) we have that
  $\sum \mathcal{S} $ bounds $\sum S $.
\end{Corollary}
\begin{proof}
  The corollary follows immediately from \Cref{theo:expr-bound} and \Cref{lem:range-annotated-monoids}.
\end{proof}

\begin{figure}[t]
  \centering
  \begin{align}
    k \smbpair_{\monoid} (m_1 \amadd m_2) & = k \smbpair_{\semK} m_1 \amadd k \smbpair_{\semK} m_2 \label{eq:semimodule-distributes-over-m-plus}\\
    (k_1 \addK k_2) \smbpair_{\monoid} m  & = k_1 \smbpair_{\semK} m \amadd k_2 \smbpair_{\semK} m \label{eq:semimodule-law-distributes-over-addition}\\
    \oneK \smbpair_{\monoid} m            & = m                                                    \label{eq:semimodule-with-k-one-yields-m} \\
    \zeroK \smbpair_{\monoid} m           & = \amzero                                             \label{eq:semimodule-with-k-zero-yields-m-zero} \\
    k \smbpair_{\monoid} \amzero          & = \amzero                                            \label{eq:semimodule-with-m-zero-yields-m-zero}  \\
    (k_1 \multK k_2) \smbpair_{\monoid} m & = k_1 \smbpair_{\monoid} (k_2 \smbpair_{\monoid} m) \label{eq:semimodule-law-multiplication-is-semimodule}
  \end{align}
  \caption{Semimodule laws (semiring $\semK$ paired with monoid $\monoid$)}
  \label{fig:semimodule-laws}
\end{figure}

\mypar{Semimodules}
One challenge of supporting aggregation over
$\semK$-relations is that the annotations of tuples have to be factored into the
aggregation computation. For instance, consider an $\semN$-relation $R(A)$ with
two tuples $\tuple{30} \mapsto 2$ and $\tuple{40} \mapsto 3$, i.e., there are two duplicates
of tuple $\tuple{30}$ and $3$ duplicates of tuple $\tuple{40}$. Computing the sum over $A$
we expect to get $30 \cdot 2 + 40 \cdot 3 = 180$. More generally speaking, we
need an operation $\smbpair_{\monoid}: \semN \times \monoid \rightarrow \monoid$
that combines semiring elements with values from the aggregation function's
domain. As observed in~\cite{AD11d}
this operation has to be a semimodule, i.e., it has to fulfill a set of equational laws, two of which are shown for $\semN$ in \Cref{fig:semimodule-laws}. Note that in the example above we made use of the fact that
$\smbpair_{\semN,\msum}$ is $\cdot$ to get
$30 \smbpair_{\semN} 2 = 30 \cdot 2 = 60$.
Operation  $\smbpair$ is not well-defined for
all semirings, but it is defined for $\semN$ and all of the monoids we consider.
We show the definition for $\smbpair_{\semN,\monoid}$ for all considered monoids below:
\begin{align*}
  k \smbpair_{\semN,\msum} m                        & = k \cdot m             \\
  k \smbpair_{\semN,\mmin} m = k \smbpair_{\mmax} m & =
  \begin{cases}
    m                                         & \text{if}\,\,\,k \neq 0 \\
    0                                         & \text{else}             \\
  \end{cases}
\end{align*}

\mypar{The Tensor Construction}
Amsterdamer et al. demonstrated that
there is no meaningful way to define semimodules for all combinations of
semirings and standard aggregation function monoids. To be more
  precise, if aggregation function results are concrete values from the
  aggregation monoid, then it is not possible to retain the important property
  that queries commute with homomorphisms. Intuitively, that is the case because
  applying a homomorphism to the input may change the aggregation function
  result. Hence, it is necessary to delay the computation of concrete
  aggregation results by keeping the computation symbolic. For instance,
  consider an $\semNX$-relation $R(A)$ (provenance polynomials) with a single
  tuple $(30) \mapsto x_1$. If we compute the sum over $A$, then under a
  homomorphism $h_1: x_1 \to 1$ we get a result of $30 \cdot 1 = 30$ while under
  a homomorphism $h_2: x_1 \to 2$ we get $30 \cdot 2 = 60$. The solution
presented in~\cite{AD11d} uses monoids whose elements are symbolic expressions
that pair semiring values with monoid elements. Such monoids are compatible with
a larger class of semirings including, e.g., the provenance polynomial
semiring. Given a semiring $\semK$ and aggregation monoid $\monoid$, the
symbolic commutative monoid has as domain bags of elements from
$\semK \times \monoid$ with bag union as addition (denoted as $\akmadd$) and the
emptyset as neutral element. This structure is then extended to a
$\semK$-semimodule $\asm$ by defining
$k \smbpair_{\asm} \sum k_i \smpair m_i \defas \sum (k \multK k_i) \smpair m_i$
and taking the quotient (the structure whose elements are equivalent classes)
wrt. the semimodule laws. For some semirings, e.g., $\semN$ and $\semB$, the
symbolic expressions from $\asm$ correspond to concrete aggregation result
values from $\monoid$.\footnote{That is the case when $\monoid$ and $\asm$ are
  isomorphic.} However, this is not the case for every semiring and aggregation
monoid. For instance, for most provenance semirings these expressions cannot be
reduced to concrete values. Only by applying homomorphisms to semirings for
which this construction is isomorphic to the aggregation monoid is it possible
to map such symbolic expressions back to concrete values. For instance,
computing the sum over the $\semNX$-relation $R(A)$ with tuples
$(30) \mapsto x_1$ and $(20) \mapsto x_2$ yields the symbolic expression
$30 \smpair x_1 \kmadd{\semNX}{\msum} 20 \smpair x_2$. If the input tuple
annotated with $x_1$ occurs with multiplicity $2$ and the input tuple annotated
with $x_2$ occurs with multiplicity $4$ then this can be expressed by applying a
homomorphism $h: \semNX \to \semN$ defined as $h(x_1) = 2$ and $h(x_2) =
4$. Applying this homomorphism to the symbolic aggregation expression
$30 \smpair x_1 \kmadd{\semNX}{\msum} 20 \smpair x_2$, we get the expected
result $30 \cdot 2 + 20 \cdot 4 = 140$. If we want to support aggregation for
$\uaaK{\semK}$-relations with this level of generality then we would have to
generalize range-annotated values to be symbolic expressions from $\asm$ and
would have to investigate how to define an order over such values to be able to
use them as bounds in range-annotated values. For instance, intuitively we may
bound $(x_1 + x_2) \smpair 3 \kmadd{\semNX}{\msum} x_3 \smpair 2$ from below
using $x_1 \smpair 3 \kmadd{\semNX}{\msum} x_3 \smpair 1$ since
$x_1 \ordOf{\semNX} x_1 + x_2$ and $2 < 3$.  Then we would have to show that
aggregation computations preserve such bounds to show that queries with
aggregation with this semantics preserve bounds. We trade generality for
simplicity by limiting the discussion to semirings where $\asm$ is isomorphic to
$\monoid$. This still covers the important cases of bag semantics and set
semantics ($\semN$ and $\semB$), but has the advantage that we are not burdening
the user with interpreting bounds that are complex symbolic expressions. For
instance, consider an aggregation without group-by over a relation with millions
of rows. The resulting bound expressions for the aggregation result value may
contain millions of terms which would render them completely useless for human
consumption. Additionally, while query evaluation is still \ptime when using
$\asm$, certain operations like joins on aggregation results are
inefficient.\footnote{Comparing symbolic expressions requires an extension of
  annotations to treat these comparisons symbolically. The reason is that since
  an aggregation result cannot be mapped to a concrete value, it is also not
  possible to determine whether such a values are equal. The net result is that
  joins on such values may degenerate to cross products.}

\subsection{Applying Semimodules to $\uaaK{\semN}$-Relations}
\label{sec:semim-uaaks-relat}

As we will demonstrate in the following, even though it may be possible to
define $\uaaK{\semK}$-semimodules, such semimodules cannot be bound preserving
and, thus, would be useless for our purpose. We then demonstrate that it is
possible to define bound preserving operations that combine $\uaaK{\semN}$
elements with $\rangeDom$ elements and that this is sufficient for defining a
bound preserving semantics for aggregation.

\begin{Lemma}[Bound preserving $\uaaK{\semN}$-semimodules are impossible]\label{lem:bound-preserving-impossible}
The semimodule for $\uaaK{\semN}$ and $\msum$, if it exists, cannot be bound preserving.
\end{Lemma}
\begin{proof}
  For sake of contadiction assume that this semimodule exists and is bound preserving.
  Consider $k = \ut{1}{1}{2}$ and $m = \uv{0}{0}{0}$. Then by semimodule law \ref{eq:semimodule-with-m-zero-yields-m-zero} we have $k \mysmbNAU{\msum} m = m = \uv{0}{0}{0}$. Now observe that for $m_1 = \uv{-1}{-1}{-1}$ and $m_2 = \uv{1}{1}{1}$ we have $m = m_1 + m_2$. Let $m_1' = k \mysmbNAU{\msum} m_1$. We know that $m_1' = \uv{l_1}{-1}{u_1}$ for some $l_1$ and $u_1$. Since the semimodule is assumed to be bound preserving we know that $l_1 \leq -2$ and $u_1 \geq -1$ ($-1 \cdot 2=-2$ and $-1 \cdot 1 = -1$). Analog, let $m_2' = \uv{l_2}{1}{u_2} = k \mysmbNAU{\msum} m_2$. By the same argument we get $l_2 \leq 1$ and $u_2 \geq 2$. Applying semimodule law \ref{eq:semimodule-distributes-over-m-plus} we get $k \mysmbNAU{\msum} m = k \mysmbNAU{\msum} (m_1 + m_2)  = m_1' + m_2'  = \uv{l_1 + l_2}{0}{u_1 + u_2}$. Let $l' = l_1 + l_2$, $u' = u_1 + u_2$ and $m'' = \uv{l'}{0}{u'}$. Based on the inequalities constraining $l_i$ and $u_i$ we know that $l_1 + l_2 \leq -1$ and $u_1 + u_2 \geq 1$. Thus, we have the contradiction $k \mysmbNAU{\msum} m = \uv{0}{0}{0} \neq \uv{l'}{0}{u'} = k \mysmbNAU{\msum} m$.
\end{proof}

In spite of this negative result, not everything is lost. Observe that it not necessary for the operation that combines semiring elements (tuple annotations) with elements of the aggregation monoid to follow semimodule laws. After all, what we care about is that the operation is bound-preserving.
Below we define operations $\mysmbNAU{\monoid}$ that are not semimodules, but are bound-preserving.
To achieve bound-preservation we can rely on the bound-preserving expression semantics we have defined in~\Cref{sec:expression}.
For example, 
since $\smbpair_{\semN,\msum}$ is multiplication, we can define
$\smbpair_{\uaaK{\semN}, \msum}$ using our definition of multiplication for range-annotated expression evaluation.
It turns out that this approach of computing the bounds as the  minimum and maximum over all pair-wise combinations of value and tuple-annotation bounds also works for $\mmin$ and $\mmax$:

\begin{Definition}\label{def:range-semimodules}
  Consider an aggregation monoid $\monoid$ such that $\smbpair_{\semN,\monoid}$ is well defined.
  Let $(\lbMarker{m},m,\ubMarker{m})$ be a range-annotated value from $\rangeDom$ and $(\lbMarker{k},k,\ubMarker{k}) \in \uaaK{\semN}$.
  We define $\amysmbNAU$ as shown below. 
  \begin{flalign*}
      \ut{\lbMarker{k}}{k}{\ubMarker{k}} \amysmbNAU \uv{\lbMarker{m}}{m}{\ubMarker{m}} =
  \end{flalign*}
    \begin{align*}
		(& min(\lbMarker{k} \smbpair_{\semN,\monoid} \lbMarker{m},\lbMarker{k} \smbpair_{\semN,\monoid} \ubMarker{m},\ubMarker{k} \smbpair_{\semN,\monoid} \lbMarker{m},\ubMarker{k} \smbpair_{\semN,\monoid} \ubMarker{m}), \\
		& {k} \smbpair_{\semN,\monoid} {m}, \\
		& max(\lbMarker{k} \smbpair_{\semN,\monoid} \lbMarker{m},\lbMarker{k} \smbpair_{\semN,\monoid} \ubMarker{m},\ubMarker{k} \smbpair_{\semN,\monoid} \lbMarker{m},\ubMarker{k} \smbpair_{\semN,\monoid} \ubMarker{m}))
	\end{align*}

\end{Definition}

As the following theorem demonstrates $\asmbNAU$ is in fact bound preserving.

\begin{Theorem}\label{theo:semimodulish-bound-preserving}
  Let $\monoid \in \{\mmin, \mmax\}$ and $\semK \in \{\semB, \semN\}$ or $\monoid = \msum$ and $\semK = \semN$. Then $\mysmbNAU{\monoid}$ preserves bounds.
\end{Theorem}
\begin{proof}
We first prove the theorem for $\semN$.
  We have to show for all $\monoid \in \{\mmin, \mmax, \msum\}$ that for any $\vec{k} = \ut{\lbMarker{k}}{\bgMarker{k}}{\ubMarker{k}} \in \uaaK{\semN}$ and $\vec{m} = \uv{\lbMarker{m}}{\bgMarker{m}}{\ubMarker{m}} \in \rangeDom$ we have that $\vec{k} \mysmbNAU{\monoid} \vec{m}$ bounds $k \smbpair_{\semN,\monoid} m$ for any $k$ bound by $\vec{k}$ and $m$ bound by $\vec{m}$. We prove the theorem for each  $\monoid \in \{\mmin, \mmax, \msum\}$.

\proofpara{$\monoid = \msum$}
We have $k \smbpair_{\semN,\msum} m \defas k \cdot m$. We distinguish four cases:

\proofpara{$\lbMarker{m} < 0$, $\ubMarker{m} < 0$}
We have that $\ubMarker{k} \cdot \lbMarker{m} = min(\lbMarker{k} \cdot \lbMarker{m},\lbMarker{k} \cdot \ubMarker{m},\ubMarker{k} \cdot \lbMarker{m},\ubMarker{k} \cdot \ubMarker{m}$) and $\lbMarker{k} \cdot \ubMarker{m} = max(\lbMarker{k} \cdot \lbMarker{m},\lbMarker{k} \cdot \ubMarker{m},\ubMarker{k} \cdot \lbMarker{m},\ubMarker{k} \cdot \ubMarker{m})$.
Thus,
\begin{align*}
\vec{k} \mysmbNAU{\monoid} \vec{m} &= \uv{\ubMarker{k} \cdot \lbMarker{m}}{\bgMarker{k} \cdot\bgMarker{m}}{\lbMarker{k} \cdot \ubMarker{m}}
\end{align*}
Now for any $k$ bound by $\vec{k}$ and $m$ bound by $\vec{m}$ we have:
$\ubMarker{k} \cdot \lbMarker{m} \leq k \cdot m$ because $m$ is a negative number and $k \leq \ubMarker{k}$. Analog, $k \cdot m \leq \lbMarker{k} \cdot \ubMarker{m}$, because $m$ is negative and $\ubMarker{m} \geq m$ and $\lbMarker{k} \leq k$. Thus, $\vec{k} \mysmbNAU{\monoid} \vec{m}$ bounds $k \smbpair_{\semN,\monoid} m$.

\proofpara{$\lbMarker{m} \geq 0$, $\ubMarker{m} \geq 0$}
We have that $\lbMarker{k} \cdot \lbMarker{m} = min(\lbMarker{k} \cdot \lbMarker{m},\lbMarker{k} \cdot \ubMarker{m},\ubMarker{k} \cdot \lbMarker{m},\ubMarker{k} \cdot \ubMarker{m}$) and $\ubMarker{k} \cdot \ubMarker{m} = max(\lbMarker{k} \cdot \lbMarker{m},\lbMarker{k} \cdot \ubMarker{m},\ubMarker{k} \cdot \lbMarker{m},\ubMarker{k} \cdot \ubMarker{m})$.
Thus,
\begin{align*}
\vec{k} \mysmbNAU{\monoid} \vec{m} &= \uv{\lbMarker{k} \cdot \lbMarker{m}}{\bgMarker{k} \cdot\bgMarker{m}}{\ubMarker{k} \cdot \ubMarker{m}}
\end{align*}
Now for any $k$ bound by $\vec{k}$ and $m$ bound by $\vec{m}$ we have:
$\ubMarker{k} \cdot \lbMarker{m} \leq k \cdot m$ because $m$ is a positive number and $k \geq \lbMarker{k}$. Analog, $k \cdot m \leq \ubMarker{k} \cdot \ubMarker{m}$, because $m$ is negative and $\ubMarker{m} \geq m$ and $\ubMarker{k} \geq k$. Thus, $\vec{k} \mysmbNAU{\monoid} \vec{m}$ bounds $k \smbpair_{\semN,\monoid} m$.

\proofpara{$\lbMarker{m} < 0$, $\ubMarker{m} \geq 0$}
We have that $\ubMarker{k} \cdot \lbMarker{m} = min(\lbMarker{k} \cdot \lbMarker{m},\lbMarker{k} \cdot \ubMarker{m},\ubMarker{k} \cdot \lbMarker{m},\ubMarker{k} \cdot \ubMarker{m}$) and $\ubMarker{k} \cdot \ubMarker{m} = max(\lbMarker{k} \cdot \lbMarker{m},\lbMarker{k} \cdot \ubMarker{m},\ubMarker{k} \cdot \lbMarker{m},\ubMarker{k} \cdot \ubMarker{m})$.
Thus,
\begin{align*}
\vec{k} \mysmbNAU{\monoid} \vec{m} &= \uv{\ubMarker{k} \cdot \lbMarker{m}}{\bgMarker{k} \cdot\bgMarker{m}}{\ubMarker{k} \cdot \ubMarker{m}}
\end{align*}

Now consider some $k$ bound by $\vec{k}$ and $m$ bound by $\vec{m}$. If $m$ is positive, then trivially $\lbMarker{k} \cdot \lbMarker{m}$ bounds $k \cdot m$ from below since $\lbMarker{m}$ is negative. Otherwise, the lower bound holds using the argument for the case of $\lbMarker{m} < 0$, $\ubMarker{m} < 0$. If $m$ is negative, then trivally $\ubMarker{k} \cdot \ubMarker{m}$ bounds $k \cdot m$ from above since $\ubMarker{m}$ is positive. Otherwise, the upper bound holds using the argument for the case of $\lbMarker{m} \geq 0$, $\ubMarker{m} \geq 0$. Thus, $\vec{k} \mysmbNAU{\monoid} \vec{m}$ bounds $k \smbpair_{\semN,\monoid} m$.

\proofpara{$\monoid = \mmin$}
We have
\[
  k \smbpair_{\semN,\mmin} m \defas
\begin{cases}
  0 &\mathtext{if} k = 0\\
  m &\mathtext{otherwise}\\
\end{cases}
\]

We distinguish three cases.

\proofpara{$\lbMarker{k} \geq 0$}
If $\lbMarker{k} \geq 0$, then $k \smbpair_{\semN,\mmin} m$ returns $m$ and $\vec{k} \mysmbNAU{\mmin} \vec{m}$ returns $\vec{m}$. Since $\vec{m}$ bounds $m$ also $\vec{k} \mysmbNAU{\mmin} \vec{m}$ bounds  $k \smbpair_{\semN,\mmin} m$.

\proofpara{$\lbMarker{k} = 0$, $\ubMarker{k} \geq 0$}
Now consider the remaining case: $\lbMarker{k} = 0$. Then the result of $\vec{k} \mysmbNAU{\mmin} \vec{m}$ simplifies to $\uv{min(0,\lbMarker{m})}{\bgMarker{k} \cdot \bgMarker{m}}{max(0,\ubMarker{m})}$. Now consider some $k$ bound by $\vec{k}$ and $m$ bound by $\vec{m}$. If $k \neq 0$ then $k \smbpair_{\semN,\mmin} m = m$ and the claim to be proven holds. Otherwise, $k \smbpair_{\semN,\mmin} m = 0$ which is bound by $\uv{min(0,\lbMarker{m})}{\bgMarker{k} \cdot \bgMarker{m}}{max(0,\ubMarker{m})}$.

\proofpara{$\lbMarker{k} = 0$, $\ubMarker{k} = 0$}
In this case $\vec{k} \mysmbNAU{\mmin} \vec{m} = \uv{0}{0}{0}$ and since $k = 0$ because $\vec{k}$ bounds $k$ we have $k \smbpair_{\semN,\mmin} m = 0$ which is trivally bound by $\uv{0}{0}{0}$.

\proofpara{$\monoid = \mmax$}
We have
\[
  k \smbpair_{\semN,\mmax} m \defas
\begin{cases}
  0 &\mathtext{if} k = 0\\
  m &\mathtext{otherwise}\\
\end{cases}
\]
The proof for $\mmax$ is analog to the proof for $\mmin$.

For semiring $\semB$ and $\monoid = \mmin$ or $\monoid = \mmax$ observe that $k \smbpair_{\semB,\monoid} m$ is the identity for $m$ if $k \neq \zeroOf{\semB}$ and $k \smbpair_{\semB,\monoid} m = 0$ otherwise. Thus, the proof is analog to the proof for semiring $\semN$.
\end{proof}

\subsection{Bound-Preserving Aggregation}
\label{sec:bound-pres-aggr}
We now define a bound preserving aggregation semantics based on the
$\amysmbNAU$ operations. As mentioned above, the main
challenge we have to overcome is to deal with the uncertainty of group
memberships plus the resulting uncertainty in the number of groups and of which
inputs contribute to a group's aggregation function result values.  In general,
the number of possible groups encoded by an input \abbrUAADB-relation may be
very large. Thus, enumerating all concrete groups is not a viable option.  While
\abbrUAADBs can be used to encode an arbitrary number of groups as a single
tuple, we need to decide how to trade conciseness of the representation for
accuracy. Furthermore, we need to ensure that the aggregation result in the
\abbrBGW is encoded by the result. There are many possible strategies for how to
group possible aggregation results. We, thus, formalize grouping strategies and define a semantics for aggregation that preserves bounds for any such grouping semantics. Additionally, we present a reaonsable default strategy.
We define our aggregation semantics in three steps: (i) we introduce grouping strategies and our default grouping strategy that matches \abbrBG and possible input groups to output tuples (each output tuple will represent exactly one group in the  \abbrBGW and one or more possible groups); (ii) we calculate group-by attribute ranges for output tuples based on the assignment of input tuples to output tuples; (iii) we calculate the multiplicities (annotations) and bounds for aggregation function results for each output tuple. 

\subsection{Grouping Strategies}
\label{sec:grouping-strategies}

A \textit{grouping strategy} $\gstrat$ is a function that takes as input a n-ary $\uaaK{\semK}$-relation $\rangeRel$ and list of group-by attributes $\gbAttrs$ and returns a triple $(\stratGrps, \stratBG, \stratPoss)$ where $\stratGrps$ is a set of output groups, $\stratBG$
is a  function associating each input tuple $\rangeTup$ from $\rangeRel$ where $\bgMarker{\rangeRel(\rangeTup)} \neq \zeroK$ with an output from $\stratGrps$,
and $\stratPoss$ is a  function associating each input tuple $\rangeTup$ from $\rangeRel$ where $\rangeRel(\rangeTup) \neq \uaaZero{\semK }$ with an output from $\stratGrps$. Note that the elements of  $\stratGrps$ are just unique identifiers for output tuples.  The actual range-annotated output tuples returned by an aggregation operator are not returned by the grouping strategy directly but are constructed by our aggregation semantics based on the information returned by a grouping strategy.
Intuitively, $\stratBG$ takes care of the association of groups in the \abbrBGW with an output while $\stratPoss$ does the same for all possible groups. For $\gstrat$ to be a grouping strategy we require that for any input relation $\rangeRel$ and list of group-by attributes $\gbAttrs$ we have:

\begin{align*}
  \forall \rangeTup, \rangeTup' \in \rangeDom^n:\,\,\, &\rangeRel(\rangeTup) \neq \zeroOf{\uaaK{\semK}}
  \wedge \rangeRel(\rangeTup) \neq \zeroOf{\uaaK{\semK}}
  \wedge \bgMarker{\rangeTup.\gbAttrs} = \bgMarker{\rangeTup'.\gbAttrs}\\
  &\rightarrow
  \stratBG(\rangeTup) = \stratBG(\rangeTup')
\end{align*}

This condition ensures that for every tuple that exists in the \abbrBGW, all inputs that exists in the \abbrBGW and belong this group are associated with a single output. Our aggregation semantics relies on this property to produce the correct result in the \abbrBGW. We use function $\stratPoss$ to ensure that every possible group is accounted for by the $\uaaK{\semK}$-relation returned as the result of aggregation. Inuitively, every range-annotated input tuple may correspond to several possible groups based on the range annotations of its group-by attribute values.
Our aggregation semantics takes ensures that an output tuple's group-by ranges bound the group-by attribute ranges of every input associated to it by $\stratPoss$.

\subsection{Default Grouping Strategy}
Our \textit{default grouping strategy}
takes as input a n-ary $\uaaK{\semN}$-relation $\rangeRel$ and list of group-by attributes $\gbAttrs$ and returns a pair $(\stratGrps, \stratPoss)$ where $\stratGrps$ is a set of output tuples --- one for every \abbrBG group, i.e., an input tuple's group-by values in the \abbrBGW.
$\stratPoss$ assigns each input tuple to one output tuple based on its \abbrBG group-by values.
Note that even if the \abbrBG annotation of an input tuple is $0$, we still use its \abbrBG values
  to assign it to an output tuple.
  Only tuples that are not possible (annotated with $\uaaZero{\semN} = \ut{0}{0}{0}$) are not considered.
  Since output tuples are identified by their \abbrBG group-by values, we will use these values to  identify elements from $\stratGrps$.

\begin{Definition}[Default Grouping Strategy]\label{def:default-grouping-strategy}
Consider a query $\query \defas \aggregation{\gbAttrs}{f(A)}(\rangeRel)$.
 Let $\rangeTup \in \rangeDom^n$ such that $\rangeRel(\rangeTup) \neq \zeroOf{\uaaK{\semN}}$ and $\rangeTup' \in \rangeDom^n$ such that $\bgMarker{\rangeRel(\rangeTup')} \neq \zeroN$. The default grouping strategy $\gsdef \defas (\gsdefG,  \gsdefP)$ is defined as shown below.
\begin{align*}
 \gsdefG                      & \defas  \{ t.\gbAttrs \mid \exists \rangeTup: \bgMarker{\rangeTup} = t \wedge \rangeRel(\rangeTup) \neq \uaaZero{\semN} \} &
  \gsdefP(\rangeTup)          & \defas \bgMarker{\rangeTup.\gbAttrs}
\end{align*}
\end{Definition}

For instance, consider
three tuples $\rangeTup_1 \defas \tuple{\uv{1}{2}{2}}$ and
$\rangeTup_2 \defas \tuple{\uv{2}{2}{4}}$ and $\rangeTup_3 \defas \tuple{\uv{2}{3}{4}}$ over schema $\rel(A)$.
Furthermore, assume that $\rangeRel(\rangeTup_1) = \ut{1}{1}{1}$, $\rangeTup(\rangeTup_2) = \ut{0}{0}{1}$, and
$\rangeTup(\rangeTup_3) = \ut{0}{0}{3}$.
Grouping on $A$,
the default strategy will generate  two output groups $\agroup_1$ for \abbrBG group $(2)$ and
$\agroup_2$ for \abbrBG group $(3)$.
Based on their \abbrBG group-by values, the possible grouping function $\gsdefP$ assigns $\rangeTup_1$ and $\rangeTup_2$ to $\agroup_1$ and $\rangeTup_3$ to $\agroup_2$.

\subsection{Aggregation Semantics}
We now introduce an aggregation semantics based on this grouping strategy.  For
simplicity we define aggregation without group-by as a special case of
aggregation with group-by (the only difference is how tuple annotations are
handled). We first define how to construct a result tuple $\aout$ for each
output group $\asgrp$ returned by the grouping strategy and then present how to
calculate tuple annotations. The construction of an output tuple is divided into
two steps: (i) determine range annotations for the group-by attributes and (ii)
determine range annotations for the aggregation function result attributes.

\mypar{Group-by Bounds}
To ensure that all possible groups an input tuple $\rangeTup$ with
$\gsdefP(\rangeTup) = \aout$ belongs to are contained in $\aout.\gbAttrs$ we
have to merge the group-by attribute bounds of all of these tuples. Furthermore,
we set $\bgMarker{\aout.\gbAttrs} = \aout$, i.e., we use the unique \abbrBG
group-by values of all input tuples assigned to $\aout$ (i.e., $\bgMarker{\aout.\gbAttrs} = \asgrp$) as the output's \abbrBG group-by value.

\begin{Definition}[Range-bounded Groups]\label{def:range-bounded-groups}
Consider a result group $\asgrp \in \gsdefG(\gbAttrs, \rangeRel)$ for an aggregation with group-by attributes $\gbAttrs$ over a $\uaaK{\semN}$-relation $\rangeRel$. The bounds for the group-by attributes values of $\aout$ are defined as shown below. Let $\agroup$ be the unique element from the set $\{ \bgMarker{\rangeTup.\gbAttrs} \mid \bgMarker{\rangeRel(\rangeTup)} \neq 0 \wedge \stratPoss(\rangeTup) = \aout \}$. For all $a \in \gbAttrs$ we define:
\begin{align*}
\lbMarker{\aout.a} &= \min_{\rangeTup: \stratPoss(\rangeTup) = \aout} \lbMarker{\rangeTup.a} &
\bgMarker{\aout.a} &= \agroup.a &
\ubMarker{\aout.a} &= \max_{\rangeTup: \stratPoss(\rangeTup) = \aout} \ubMarker{\rangeTup.a}
\end{align*}
\end{Definition}

Note that in the definition above, $\min$ and $\max$ are the minimum and maximum
wrt. to the order over the data domain $\dataDomain$ which we used to define
range-annotated values.  Reconsider the three example tuples and two result groups from above.
The group-by range annotation for output tuple $\outof{\agroup_1}$ is
$\uv{\min(1,2)}{2}{\max{2,4}} = \uv{1}{2}{4}$. Observe that $\uv{1}{2}{4}$ bounds
any group $\rangeTup_1$ and $\rangeTup_2$ may belong to in some possible world.

\mypar{Aggregation Function Bounds}
To calculate bounds on the result of an aggregation function for one group, we have to reason about the minimum and maximum possible aggregation function result based on the bounds of aggregation function input values, their row annotations, and their possible and guaranteed group memberships (even when a value of the aggregation function input attribute is certain the group membership of the tuple it belongs too may be uncertain). To calculate a conservative lower bound of the aggregation function result for an output tuple $\aout$, we use $\amysmbNAU$ to pair the aggregation function value of each tuple $\rangeTup$ with $\stratPoss(\rangeTup) = \asgrp)$ with the tuple's annotation and then extract the lower bound from the resulting range-annotated value.
For some tuples their  group membership is uncertain because either their group-by values are uncertain or they may not exist in all possible worlds (their certain multiplicity is $0$).
We take this into account by taking the minimum of the neutral element of the aggregation monoid and the result of $\amysmbNAU$ for such tuples. Towards this goal we introduce a predicate $\uncertg{\gbAttrs}{\rangeRel}{\rangeTup}$ that is defined as shown below.
\begin{align*}
  \uncertg{\gbAttrs}{\rangeRel}{\rangeTup} \defas (\exists a \in \gbAttrs: \lbMarker{\rangeTup.a} \neq \ubMarker{\rangeTup.a}) \vee \lbMarker{\rangeRel(\rangeTup)} = 0
\end{align*}
We then sum up the resulting values in the aggregation monoid. Note that here summation is assumed to  use addition in $\monoid$. 
The upper bound calculation is analog (using the upper bound and maximum instead). The \abbrBG result is a calculated using standard $\semK$-relational semantics. In the definition we will use $\rangeTup \toverlaps \rangeTup'$ to denote that the range annotations of tuples $\rangeTup$ and $\rangeTup'$ with the same schema $(A_1, \ldots, A_n)$ overlap on each attribute $A_i$, i.e.,
\begin{align*}
\rangeTup \toverlaps \rangeTup' \defas \bigwedge_{i \in \{1, \ldots, n\}} [ \lbMarker{\rangeTup.A_i}, \ubMarker{\rangeTup.A_i} ] \cap  [ \lbMarker{\rangeTup'.A_i}, \ubMarker{\rangeTup'.A_i} ] \neq \emptyset
\end{align*}

\begin{Definition}[Aggregation Function Result Bounds]\label{def:agg-function-bounds}
  Consider an output $\asgrp \in \gsdefG$, input $\rangeRel$, group-by
  attributes $\gbAttrs$, and aggregation function $f(A)$ with monoid $\monoid$.
  We use $\tgrouping(\asgrp)$ to denote the set of input tuples whose group-by
  attribute bounds overlap with $\aout.\gbAttrs$, i.e., they may be belong to a
  group represented by $\aout$:
  \begin{align*}
    \tgrouping(\asgrp) \defas \{ \rangeTup \mid \rangeRel(\rangeTup) \neq \uaaZero{\semN} \wedge \rangeTup.\gbAttrs \toverlaps \aout.\gbAttrs\}
  \end{align*}
  The bounds on the aggregation function result for tuple $\aout$ are defined as:
  \begin{align*}
  \lbMarker{\aout.f(A)} &= \sum_{\rangeTup \in \tgrouping(\asgrp)} \lbagg{\rangeTup}\\
  \lbagg{\rangeTup}     &= \begin{cases}
                            min(\amzero,\lbMarker{(\rangeRel(\rangeTup) \amysmbNAU \rangeTup.A)}) &\mathtext{if }\uncertg{\gbAttrs}{\rangeRel}{\rangeTup} \\
                            \lbMarker{(\rangeRel(\rangeTup) \amysmbNAU \rangeTup.A)} &\mathtext{otherwise}\\
                          \end{cases}\\
   \bgMarker{\aout.f(A)} &= \sum_{\rangeTup \in \tgrouping(\asgrp)} \bgMarker{(\rangeRel(\rangeTup) \amysmbNAU \rangeTup.A)} \\
  \ubMarker{\aout.f(A)} &= \sum_{\rangeTup \in \tgrouping(\asgrp)} \ubagg{\rangeTup}\\
  \ubagg{\rangeTup}     &= \begin{cases}
                            max(\amzero,\ubMarker{(\rangeRel(\rangeTup) \amysmbNAU \rangeTup.A)}) &\mathtext{if }\uncertg{\gbAttrs}{\rangeRel}{\rangeTup} \\
                            \ubMarker{(\rangeRel(\rangeTup) \amysmbNAU \rangeTup.A)}&\mathtext{otherwise}\\
                          \end{cases}
\end{align*}
\end{Definition}

\begin{Example}
  For instance, consider calculating the sum of $A$ grouping on $B$ for a relation $R(A,B)$  which consists of two tuples $\rangeTup_3 \defas \tuple{\uv{3}{5}{10},\uv{3}{3}{3}}$ and $\rangeTup_4 \defas \tuple{\uv{-4}{-3}{-3},\uv{2}{3}{4}}$ which are both annotated with $(1,2,2)$ (appear certainly once and may appear twice). Consider calculating the aggregation function result bounds for the result tuple $\aout$ for the output group $\asgrp$ which  corresponds to \abbrBG group $\agroup \defas \tuple{3}$. The lower bound on {\upshape \lstinline!sum(A)!} is calculated as shown below:
  \begin{align*}
    &\sum_{\rangeTup \in \tgrouping(\asgrp)} \lbagg{\rangeTup}\\
    = &\lbMarker{\left(\ut{1}{2}{2} \cdot \uv{3}{5}{10}  \right)} + min(0, \lbMarker{\left(\ut{1}{2}{2} \cdot \uv{-4}{-3}{-3}\right)})\\
    = &\lbMarker{\uv{3}{10}{20}} + min(0, \lbMarker{\uv{-8}{-6}{-3}})\\
    =&3 + min(0,-8) = -5
  \end{align*}
The aggregation result is guaranteed to be greater than or equal to $-5$ since $\rangeTup_3$ certainly belongs to $\agroup$ (no minimum operation), because its group-by attribute value $\uv{3}{3}{3}$ is certain and the tuple certainly exists ($\lbMarker{\ut{1}{2}{1}} > 0$). This tuple contributes $3$ to the sum and $\rangeTup_4$ contributions at least $-8$. While it is possible that $\rangeTup_4$ does not belong to $\agroup$ this can only increase the final result ($3 + 0 > 3 + -8$).
\end{Example}

\mypar{Aggregation Without Group-by}
Having defined how each output tuple of aggregation is constructed we still need to calculate the row annotation for each result tuple. For aggregation without group-by there will be exactly one result tuple independent of what the input is.
In this case there exists a single possible \abbrBG output group (the empty tuple $\tuple{}$ ) and all input tuples are assigned to it through $\gsdefP$. Let $\rangeTup_{\tuple{}}$ denote this single output tuple.
Recalling that all remaining tuples have multiplicity $\zeroN$, we define:   

\begin{Definition}[Aggregation Without Group-By]\label{def:aggr-op-semantics-wo-gb}
  Consider a query $\query \defas \aggregation{}{f(A)}(\rangeRel)$. Let
$\gsdef(\emptyset,\rangeRel) = (\gsdefG, \gsdefP)$,
  and $\rangeTup$  be a range-annotated tuple with the same schema as $\query$ and $\asgrp$ denote the single output group in $\gsdefG$. Then
  \begin{align*}
    \lbMarker{\aggregation{}{f(A)}(\rangeRel)(\rangeTup)}
    =     \bgMarker{\aggregation{}{f(A)}(\rangeRel)(\rangeTup)}
    =     \ubMarker{\aggregation{}{f(A)}(\rangeRel)(\rangeTup)}
    \defas \begin{cases}
      \oneN &\text{if}\,\,\, \rangeTup = \rangeTup_{\tuple{}}\\
      \zeroN & \text{otherwise}
    \end{cases}
  \end{align*}
\end{Definition}

\mypar{Aggregation With Group-by}
For aggregation with group-by
in order to calculate the upper bound on the possible multiplicity for a result tuple of a group-by aggregation,
we have to determine the maximum number of distinct groups each output tuple could correspond to. We compute the bound for an output $\aout$ based on $\ggrouping$ making the worst-case assumption that (i) each input tuple $\rangeTup$ from $\ggrouping(\asgrp)$ occurs with the maximal multiplicity possible ($\ubMarker{\rangeRel(\rangeTup)}$) and that each tuple $\tup$ encoded by $\rangeTup$ belongs to a separate group and (ii) that the sets of groups produced from two inputs $\rangeTup$ and $\rangeTup'$ do not overlap. We can improve this bound by partitioning the input into two sets: tuples with uncertain group-by attribute values and tuple's whose group membership is certain. For the latter we can compute the maximum number of groups for an output $\aout$ by simplying counting the number of groups using \abbrBG values for each input tuple that overlaps with the group-by bounds of $\aout$. For the first set we still apply the worst-case assumption.
To determine the lower bound on the certain annotation of a tuple we have to reason about which input tuples certainly belong to a group. These are inputs whose group-by attributes are certain. For such tuples we sum up their tuple annotation lower bounds.
We then need to derive the annotation of a result tuple from the annotations of the relevant input tuples. For this purpose, \cite{AD11d} did extend semirings with a duplicate elimination operator $\duprem_\semN$
defined as $\duprem_\semN(k) = 0$ if $k = 0$ and $\duprem_{\semN}(k) = 1$ otherwise.

\begin{Definition}[Aggregation With Group-By]\label{def:aggr-op-semantics-gb}
  Consider a query $\query \defas \aggregation{\gbAttrs}{f(A)}(\rangeRel)$. Let $\gsdef(\rangeRel,\gbAttrs) = (\gsdefG, \gsdefP)$.
Consider a tuple $\rangeTup$ such that $\exists \asgrp \in \gsdefG$ with $\rangeTup = \aout$. Then,

\begin{align*}
	\lbMarker{\aggregation{\gbAttrs}{f(A)}(\rangeRel)(\rangeTup)} &\defas  \duprem_{\semN}\left(\sum_{\rangeTup': \gsdefP(\rangeTup') = \asgrp \wedge \neg\, \uncertg{\gbAttrs}{\rangeRel}{\rangeTup'}}\lbMarker{\rangeRel(\rangeTup')}\right) \\
  \bgMarker{\aggregation{\gbAttrs}{f(A)}(\rangeRel)(\rangeTup)} &\defas       \duprem_{\semN}\left(\sum_{\rangeTup': \gsdefP(\rangeTup') = \asgrp} \bgMarker{\rangeRel(\rangeTup')}\right)\\
     \ubMarker{\aggregation{\gbAttrs}{f(A)}(\rangeRel)(\rangeTup)} &\defas
\sum_{\rangeTup': \gsdefP(\rangeTup') = \asgrp} \ubMarker{\rangeRel(\rangeTup')}
\end{align*}
For any tuple $\rangeTup$ such that $\neg \exists \asgrp \in \gsdefG$ with $\rangeTup = \aout$, we define
\[
  \lbMarker{\aggregation{\gbAttrs}{f(A)}(\rangeRel)(\rangeTup)} =
  \bgMarker{\aggregation{\gbAttrs}{f(A)}(\rangeRel)(\rangeTup)} =
  \ubMarker{\aggregation{\gbAttrs}{f(A)}(\rangeRel)(\rangeTup)}
  = 0
\]
\end{Definition}

The following example illustrates the application of the aggregation semantics we have defined in this section.

\begin{Example}[Aggregation]\label{def:aggregation}
  Consider the relation shown in \Cref{fig:AU-aggregation} which records addresses (street, street number, number of inhabitants). For the street attribute, instead of showing range annotations we mark values in red to indicate that their bound encompass the whole domain of the street attribute. Street values $v$ in black are certain, i.e.,  $\lbMarker{v} = \bgMarker{v} = \ubMarker{v}$. In this example, we are uncertain about particular street numbers and the number of inhabitants at certain addresses. Furthermore, several tuples may represent more than one address. Finally, we are uncertain about the street for the address represented by the second tuple. Consider the aggregation query without group-by shown in \Cref{fig:AU-agg-example-no-GB}. We are calculating the number of inhabitants. In the \abbrBGW there are 7 inhabitants ($1 \cdot 1 + 2 \cdot 1 + 2 \cdot 2$). As another example consider, the query shown in \Cref{fig:AU-agg-example-GB}.
  Consider the second result tuple (group \textit{State}). This tuple certainly exists since the 3rd tuple in the input appears twice in every possible world and its group-by value is certain. Thus, the count for group \textit{State} is at least two. Possibly, the second input tuple could also belong to this group and, thus, the count could be $3$ (the upper bound on the aggregation result).
\end{Example}

\begin{figure}[t]
	\centering

	\begin{subtable}{\linewidth}
	\centering
	\begin{tabular}{ c|c|cc}
		\textbf{street}  & \textbf{number}          & \thead{\#inhab} & \underline{$\uaaK{\semN}$} \\
		\cline{1-3}
      \cMarker{Canal}    & $\uv{165}{165}{165}$     & $\uv{1}{1}{1}$  & (1,1,2)                    \\
		\uMarker{Canal}  & $\uv{154}{153}{156}$     & $\uv{1}{2}{2}$  & (1,1,1)                    \\
		\cMarker{State}  & $\uv{623}{623}{629}$     & $\uv{2}{2}{2}$  & (2,2,3)                    \\
		\cMarker{Monroe} & $\uv{3574}{3550}{3585}$  & $\uv{2}{3}{4}$  & (0,0,1)                    \\
	\end{tabular}
	\caption{Input Relation address}
        \label{fig:AU-agg-example-input}
  \end{subtable}

  \begin{subtable}{\linewidth}
	\centering
\begin{lstlisting}
SELECT sum(#inhab) AS pop FROM address;
\end{lstlisting}
	\begin{tabular}{ cc}
 \textbf{pop}   & \underline{$\uaaK{\semN}$} \\
\cline{1-1}
 $\uv{6}{7}{14}$ & (1,1,1)                    \\
	\end{tabular}
	\caption{Aggregation without Group-by}
    \label{fig:AU-agg-example-no-GB}
  \end{subtable}

  \begin{subtable}{\linewidth}
	\centering
\begin{lstlisting}
SELECT street, count(*) AS cnt
FROM address GROUP BY street;
\end{lstlisting}
	\begin{tabular}{ c|cc}
		\textbf{street}  & \textbf{cnt}   & \underline{$\uaaK{\semN}$} \\
		\cline{1-2}
		\uMarker{Canal}  & $\uv{1}{2}{3}$ & (1,1,2)                    \\
		\cMarker{State}  & $\uv{2}{2}{4}$ & (1,1,1)                    \\
		\cMarker{Monroe} & $\uv{1}{1}{2}$ & (0,0,1)                    \\
	\end{tabular}
    \caption{Aggregation with Group-by}
    \label{fig:AU-agg-example-GB}
  \end{subtable}

\caption{Aggregation over \abbrUAADBs}
\label{fig:AU-aggregation}
\end{figure}

\subsection{Preservation of bounds}

We now demonstrate that our aggregation semantics for \abbrUAADBs is bound-preserving.
In the proof of this fact, we will make use of two auxiliary lemmas.

\begin{Lemma}\label{lem:K-addition-distributes-over-smb}
  For $\monoid \in \{\msum,\mmin,\mmax\}$ we have for all $k_1, k_2 \in \uaaN$ and $m \in \monoid^3$:
  $(k_1 \uaanAdd k_2) \amysmbNAU m = k_1 \amysmbNAU m \madd{\rangeMof{\monoid}}  k_2 \amysmbNAU m$
\end{Lemma}
\begin{proof}
  Consider $k = k_1 + k_2$ and $m \in \rangeDom$.
  Recall the definition of $\amysmbNAU$:

  \begin{align*}
\lbMarker{k \amysmbNAU m}     & = \min(\lbMarker{k} \asmbN \lbMarker{m},\lbMarker{k} \asmbN \ubMarker{m},        \\
                            & \hspace{1.2cm}\ubMarker{k} \asmbN \lbMarker{m},\ubMarker{k} \asmbN \ubMarker{m}) \\
    \ubMarker{k \amysmbNAU m} & = \max(\lbMarker{k} \asmbN \lbMarker{m},\lbMarker{k} \asmbN \ubMarker{m},        \\
                            & \hspace{1.2cm}\ubMarker{k} \asmbN \lbMarker{m},\ubMarker{k} \asmbN \ubMarker{m})
  \end{align*}

\proofpara{$\mmin$}
Consider $\lbMarker{k \mysmbNAU{\mmin} m}$.
$\smbN{\mmin}$ is the identify on $\dataDomain$ except for $k = 0$. Furthermore, $\mzero{\mmin} = \infty$ and $\madd{\mmin} = \min$. We distinguish three cases: $\lbMarker{k} = \ubMarker{k} = 0$, $\lbMarker{k} = 0 \land \ubMarker{k} > 0$ and $\lbMarker{k} > 0$.

If $\lbMarker{k} = \ubMarker{k} = 0$, then $\lbMarker{(k \mysmbNAU{\mmin} m)} =
\lbMarker{(k_1 \amysmbNAU m)} \madd{\mmin} \lbMarker{(k_2 \amysmbNAU m)} = \ubMarker{(k \mysmbNAU{\mmin} m)} =
\ubMarker{(k_1 \amysmbNAU m)} \madd{\mmin} \ubMarker{(k_2 \amysmbNAU m)} = \mzero{\mmin}$.

If $\lbMarker{k} = 0 \land \ubMarker{k} > 0$, then
\begin{align*}
\lbMarker{(k \mysmbNAU{\mmin} m)} & = \min(\lbMarker{k} \smbN{\mmin} \lbMarker{m},\lbMarker{k} \smbN{\mmin} \ubMarker{m},                \\
                                & \hspace{1.2cm}\ubMarker{k} \smbN{\mmin} \lbMarker{m},\ubMarker{k} \smbN{\mmin} \ubMarker{m})         \\
                                & = \ubMarker{k} \smbN{\mmin} \lbMarker{m}                                                             \\
                                & = \ubMarker{(k_1 + k_2)} \smbN{\mmin} \lbMarker{m}                                                   \\
                                & = \lbMarker{m}
\end{align*}
Since $\ubMarker{k} > 0$, at least one of $\ubMarker{k_1}$ and $\ubMarker{k_2}$
is larger than $0$. WLOG $\ubMarker{k_1} > 0$, then
$\lbMarker{(k_1 \smbN{\mmin} m)} = \lbMarker{m}$. $\lbMarker{(k_2 \smbN{\mmin} m)}$ is either $\lbMarker{m}$ or $\mzero{\mmin}$. Since $\min$ is idempotent, in either case we get:
\begin{align*}
                                & = \ubMarker{(k_1 \amysmbNAU m)} \madd{\mmin} \ubMarker{(k_2 \amysmbNAU m)}                               \\
\ubMarker{(k \mysmbNAU{\mmin} m)} & = \max(\lbMarker{k} \smbN{\mmin} \lbMarker{m},\lbMarker{k} \smbN{\mmin} \ubMarker{m},                \\
                                & \hspace{1.2cm}\ubMarker{k} \smbN{\mmin} \lbMarker{m},\ubMarker{k} \smbN{\mmin} \ubMarker{m})         \\
                                & = \lbMarker{k} \smbN{\mmin} \ubMarker{m}                                                             \\
                                & = \lbMarker{(k_1 + k_2)} \smbN{\mmin} \ubMarker{m}                                                   \\
                                & = (\lbMarker{k_1} \smbN{\mmin} \ubMarker{m}) \madd{\mmin} (\lbMarker{k_2} \smbN{\mmin} \ubMarker{m}) \\
                                & = \ubMarker{(k_1 \amysmbNAU m)} \madd{\mmin} \ubMarker{(k_2 \amysmbNAU m)}
\end{align*}

If $\lbMarker{k} > 0$, then
\begin{align*}
\lbMarker{(k \mysmbNAU{\mmin} m)}   & = \min(\lbMarker{k} \smbN{\mmin} \lbMarker{m},\lbMarker{k} \smbN{\mmin} \ubMarker{m},                \\
                                  & \hspace{1.2cm}\ubMarker{k} \smbN{\mmin} \lbMarker{m},\ubMarker{k} \smbN{\mmin} \ubMarker{m})         \\
                                  & = \lbMarker{k} \smbN{\mmin} \lbMarker{m}                                                             \\
                                  & = \lbMarker{(k_1 + k_2)} \smbN{\mmin} \lbMarker{m}                                                   \\
                                  & = \lbMarker{m}
  \end{align*}
  Since $\lbMarker{k} > 0$, at least one of $\lbMarker{k_1}$ and $\lbMarker{k_2}$
is larger than $0$. WLOG $\lbMarker{k_1} > 0$, then
$\lbMarker{(k_1 \smbN{\mmin} m)} = \lbMarker{m}$. Applying the same argument as above, we get:
  \begin{align*}
                                  & = \lbMarker{(k_1 \amysmbNAU m)} \madd{\mmin} \lbMarker{(k_2 \amysmbNAU m)}                               \\
  \ubMarker{(k \mysmbNAU{\mmin} m)} & = \max(\lbMarker{k} \smbN{\mmin} \lbMarker{m},\lbMarker{k} \smbN{\mmin} \ubMarker{m},                \\
                                  & \hspace{1.2cm}\ubMarker{k} \smbN{\mmin} \lbMarker{m},\ubMarker{k} \smbN{\mmin} \ubMarker{m})         \\
                                  & = \ubMarker{k} \smbN{\mmin} \ubMarker{m}                                                             \\
                                  & = \ubMarker{(k_1 + k_2)} \smbN{\mmin} \ubMarker{m}                                                   \\
  \end{align*}
  Since $\ubMarker{k} > 0$, at least one of $\ubMarker{k_1}$ and $\ubMarker{k_2}$
is larger than $0$. WLOG $\ubMarker{k_1} > 0$, then
$\ubMarker{(k_1 \smbN{\mmin} m)} = \ubMarker{m}$. Applying the same argument as above, we get:
  \begin{align*}
                                  & = \ubMarker{(k_1 \amysmbNAU m)} \madd{\mmin} \ubMarker{(k_2 \amysmbNAU m)}                               \\
\end{align*}

\proofpara{$\mmax$}
The proof is analog to the proof for $\mmin$.

\proofpara{$\msum$}
Consider $\lbMarker{k \mysmbNAU{\msum} m}$. We first address that case $\lbMarker{m} < 0$.

\begin{align*}
\lbMarker{(k \mysmbNAU{\msum} m)} & =
                                \min(\lbMarker{k} \cdot \lbMarker{m},
                                \lbMarker{k} \cdot \ubMarker{m},
                                \ubMarker{k} \cdot \lbMarker{m},
                                \ubMarker{k} \cdot \ubMarker{m})                                            \\
\end{align*}
Since $\lbMarker{m} < 0$:
\begin{align*}
  =                             & \ubMarker{k} \cdot \lbMarker{m}                                           \\
  =                             & \ubMarker{(k_1 + k_2)} \cdot \lbMarker{m}                                 \\
  =                             & (\ubMarker{k_1} + \ubMarker{k_2}) \cdot \lbMarker{m}                      \\
  =                             & (\ubMarker{k_1} \cdot \lbMarker{m}) + (\ubMarker{k_2} \cdot \lbMarker{m}) \\
  =                             & \lbMarker{(k_1 \mysmbNAU{\msum} m)} \madd{\msum} \lbMarker{(k_2 \mysmbNAU{\msum} m)}
\end{align*}
Now consider the case $\lbMarker{m} > 0$.
\begin{align*}
                                & \lbMarker{(k \mysmbNAU{\msum} m)}                                           \\
  =                             & \lbMarker{k} \cdot \lbMarker{m}                                           \\
  =                             & \lbMarker{(k_1 + k_2)} \cdot \lbMarker{m}                                 \\
  =                             & (\lbMarker{k_1} \cdot \lbMarker{m}) + (\lbMarker{k_2} \cdot \lbMarker{m}) \\
  =                             & \lbMarker{(k_1 \mysmbNAU{\msum} m)} \madd{\msum} \lbMarker{(k_2 \mysmbNAU{\msum} m)}
\end{align*}

We now prove that  $\ubMarker{k \mysmbNAU{\msum} m} = \ubMarker{(k_1 \mysmbNAU{\msum} m)} \madd{\msum} \ubMarker{(k_2 \mysmbNAU{\msum} m)}$.
\begin{align*}
\ubMarker{(k \mysmbNAU{\msum} m)} & =
                                \max(\lbMarker{k} \cdot \lbMarker{m},
                                \lbMarker{k} \cdot \ubMarker{m},
                                \ubMarker{k} \cdot \lbMarker{m},
                                \ubMarker{k} \cdot \ubMarker{m})
\end{align*}

First consider $\ubMarker{m} < 0$.
\begin{align*}
    & \ubMarker{(k \mysmbNAU{\msum} m)}            \\
  = & \lbMarker{k} \cdot \ubMarker{m}            \\
  = & \lbMarker{(k_1 + k_2)} \cdot \ubMarker{m}  \\
  = & (\lbMarker{k_1} \cdot \ubMarker{m}) + (\lbMarker{k_2} \cdot \ubMarker{m}) \\
  = & \ubMarker{(k_1 \mysmbNAU{\msum} m)} \madd{\msum} \ubMarker{(k_2 \mysmbNAU{\msum} m)}
\end{align*}
Now consider $\ubMarker{m} \geq 0$.
\begin{align*}
    & \ubMarker{(k \mysmbNAU{\msum} m)}            \\
  = & \ubMarker{k} \cdot \ubMarker{m}            \\
  = & \ubMarker{(k_1 + k_2)} \cdot \ubMarker{m}  \\
  = & (\ubMarker{k_1} \cdot \ubMarker{m}) + (\ubMarker{k_2} \cdot \ubMarker{m}) \\
  = & \ubMarker{(k_1 \mysmbNAU{\msum} m)} \madd{\msum} \ubMarker{(k_2 \mysmbNAU{\msum} m)}
\end{align*}
\end{proof}
In addition we will prove that for $\monoid \in \{ \msum, \mmin, \mmax \}$ and for all $m_1,m_2,m_3,m_4 \in \monoid$ such that $m_1 \leq m_2$ and $m_3 \leq m_4$ (here $<$ is the order of $\dataDomain$), we have $m_1 \amadd m_2 \leq m_3 \amadd m_4$. This is implies as a special case $m \amadd m' \leq m \amadd \amzero$ for  $m, m' \in \monoid$ and $m' \leq \amzero$.
\begin{Lemma}\label{lem:add-0-to-m-is-larger}
  Let $\monoid \in \{ \msum, \mmin, \mmax \}$. $\forall m_1, m_2,m_3,m_4 \in \monoid:$
  \[
 m_1 \leq m_2 \land m_3 \leq m_4 \Rightarrow m_1 \amadd m_2 \leq m_3 \amadd m_4
    \]
\end{Lemma}
\begin{proof}
  \proofpara{\mmin}
WLOG assume that $m_1 \leq m_2$ and $m_3 \leq m_4$ (the other cases are analog).
  \[
  \min(m_1,m_2) = m_1 \leq m_3 = \min(m_3,m_4)
\]

\proofpara{\msum}
Since addition preserves inequalities, we get
\[
  m_1 \madd{\msum} m_2 = m_1 + m_2 \leq m_3 + m_4 = m_3 \madd{\msum} m_4
\]

\proofpara{\mmax}
WLOG assume that $m_1 \leq m_2$ and $m_3 \leq m_4$ (the other cases are analog).
  \[
  \max(m_1,m_2) = m_2 \leq m_4 = \max(m_3,m_4)
  \]
\end{proof}

Having proven this lemma, we are ready to proof that aggregation preserves bounds.

\begin{Theorem}\label{theo:aggregation-preserves-bounds}
Let $\query \defas \aggregation{\gbAttrs}{f(A)}(\rel)$ or $\query \defas \aggregation{}{f(A)}(\rel)$ and $\prel$ be an incomplete $\semK$-relation that is bound by an $\uaaK{\semK}$-relation $\rangeRel$. Then $\query(\rangeRel)$ bounds $\query(\prel)$.
\end{Theorem}
\begin{proof}
  We first consider the case of aggregation with group-by, i.e., $\query \defas \aggregation{\gbAttrs}{f(A)}(\rel)$.
Let $\rangeTup_{\asgrp}$ be the output tuple corresponding to $\asgrp \in \gsdefG$.
  Abusing notation, we will understand $\gsdefP(\rangeTup) = \rangeTup_{\asgrp}$ to mean $\gsdefP(\rangeTup) = \asgrp$.
  Consider one possible world $\rel \in \prel$ and let $\TM_{\rel}$ be a tuple matching based on which $\rangeRel$ bounds $\rel$. We will prove the existence of a tuple matching $\TM_{\query}$ between $\query(\rel)$ and $\query(\rangeRel)$ and demonstrate that $\query(\rangeRel)$ bounds $\query(\rel)$ based on this tuple matching. For that we first prove that for each result tuple $\tup \in \query(\db)$ the set $\mathbf{S}_\tup = \{ \rangeTup \mid \tup \tmatch \rangeTup \wedge \query(\rangeRel)(\rangeTup) \neq \uaaZero{\semN} \}$ is non-empty. Intuitively, the set $\mathbf{S}_\tup$ contains potential candidates for which we can set $\TM_{\query}(\rangeTup, \tup)$ to a non-zero value, because only tuples that bound $\tup$ can be associated with $\tup$ in a tuple matching. Because, for $\TM_\query$ to be a  tuple matching we have to assign that annotation $\query(\rel)(\tup)$ to a set of tuples such that $\sum_{\rangeTup} \TM_\query(\rangeTup, \tup) = \query(\rel)(\tup)$. Note that since each aggregation result in $\query(\rel)$ is annotated with $1$, this boils down to assigning $\tup$ to exactly one  $\rangeTup \in \mathbf{S}_\tup$.

  Afterwards, we show that for each $\rangeTup \in \query(\rangeRel)$ it is possible to set $\TM_{\query}(\rangeTup, \tup)$ for all $\tup \in \query(\db)$ for which $\tup \tmatch \rangeTup$ such that (1) $\lbMarker{\query(\rangeRel)(\rangeTup)} \ordN \sum_{\tup: \tup \tmatch \rangeTup} \query(\db)(\tup) \ordN \ubMarker{\query(\rangeRel)(\rangeTup)}$ and (2) for all $\tup \in \query(\db)$ we have $\sum_{\rangeTup: \tup \tmatch \rangeTup} \TM_{\query}(\rangeTup, \tup) = \query(\db)(\tup)$. The consequence of these two steps and \Cref{def:bounding-incomplete-dbs} is that $\query(\rangeRel)$ bounds $\query(\rel)$ based on $\TM_\query$.

  We will make use of the following notation. Let $\mathcal{G} = \{ \tup.G \mid \rel(\tup) \neq 0 \}$, i.e., the set of groups in the possible world $\rel$. For a group $g \in \mathcal{G}$, we define

  \begin{align*}
    T_g                    & = \{ \tup \mid \tup.G = g \wedge \rel(\tup) \neq 0 \}\\
    \mathbf{S}_g           & = \{ \rangeTup \mid \exists \tup: \rel(\tup) \neq 0 \land \tup.G = g \land \TM_\rel(\rangeTup, \tup) \neq 0 \}                   \\
    \mathbf{O}_g           & = \{\rangeTup \mid \rangeTup \in \query(\rangeRel) \land \exists \rangeTup' \in \mathbf{S}_g: \gsdefP(\rangeTup') = \rangeTup \}
  \end{align*}

  Furthermore, for any $\rangeOf{o} \in \mathbf{O}_{g}$, we define
  \begin{align*}
        \mathbf{N}_\rangeOf{o} & = \{ \rangeTup \mid \gsdefP(\rangeTup) = \rangeOf{o} \land \rangeTup \not\in \mathbf{S}_g \}
  \end{align*}

  Consider a group $g \in \mathcal{G}$ and let $\tup_g$ denote the result tuple in $\query(\rel)$ corresponding to $g$. There is at least on $\tup \in T_g$, otherwise $g$ would not be in $\mathcal{G}$.
  Consider an arbitrary $\rangeOf{o} \in \mathbf{O}_g$. At least one such $\rangeOf{o}$ exists since $\TM(\rangeTup, \tup) \neq 0$ for one or more $\rangeTup$ with $\rangeRel(\rangeTup) \neq \uaaZero{\semN}$ and $\rangeTup$ has to be associated with at least one output $\rangeOf{o}$ by $\gsdefP$.
Let $\mathbf{S}_{\rangeOf{o}} = \{ \rangeTup \mid \rangeTup \in \mathbf{S}_g \wedge \gsdefP(\rangeTup) = \rangeOf{o}\}$.
We will show that $\tup_g \tmatch \rangeOf{o}$.

\proofpara{$\tup_g.G \tmatch \rangeOf{o}.G$}
For all $\rangeTup \in \mathbf{S}_{\rangeOf{o}}$ we know that $g \tmatch \rangeTup.G$ because for $\rangeTup$ to be in $\mathbf{S}_g$ it has to be the case there exists $\tup$ with $\tup.G = g$ such that $\TM_\rel(\rangeTup, \tup) \neq 0$. This  implies $\tup \tmatch \rangeTup$ which in turn implies $g = \tup.G  \tmatch \rangeTup$. Since $\gsdefP(\rangeTup) = \rangeOf{o}$ and since by~\Cref{def:range-bounded-groups} the range annotations of $\rangeOf{o}.G$ are defined as the union of the range annotations of all $\rangeTup \in \mathbf{S}_{\rangeOf{o}}$ (and any other $\rangeTup$ with $\gsdefP(\rangeTup)$). Thus, $g \tmatch \rangeOf{o}.G$.

\proofpara{$\tup_g.f(A) \tmatch \rangeOf{o}.f(A)$}
Based on the definition of aggregation over $\semN$-relation, we have:

\begin{align}
  \tup_g.f(A) = \sum_{\tup \in T_g} \rel(\tup) \asmbN \tup.A \label{eq:possible-world-group-result}
\end{align}
Let $f_g = \tup_g.f(A)$.
Note that based on \Cref{def:agg-function-bounds}, $\rangeOf{o}.f(A)$ is calculated over all tuples from $\mathbf{S}_g$ and $\mathbf{N}_\rangeOf{o}$.
Observe that if $\rangeTup \in \mathbf{N}_\rangeOf{o}$ then either $\uncertg{G}{\rangeRel}{\rangeTup}$ or $\lbMarker{\rangeRel(\rangeTup)} = 0$. To see why this has to be the case consider that if $\rangeTup.G$ is certain and $\rangeTup$ exists in every possible world ($\lbMarker{\rangeRel(\rangeTup)} > 0$ then for $\TM_{\rel}$ to be a tuple matching based on which $\rangeRel$ bounds $\rel$ there has to exist some $\tup \in T_g$ for which $\TM_{\rel}(\rangeTup, \tup) \neq 0$ which would lead to the contradiction $\rangeTup \in \mathbf{S}_g$.
Define
\begin{align*}
  \mathbf{S}_g^{uncertain} &= \{ \rangeTup \mid \rangeTup \in \mathbf{S}_g \wedge \uncertg{G}{\rangeRel}{\rangeTup} \}\\
  \mathbf{S}_g^{certain} &= \mathbf{S}_g -   \mathbf{S}_g^{uncertain}
\end{align*}

We have to show that $\lbMarker{\rangeOf{o}.f(A)} \leq f_g \leq \ubMarker{\rangeOf{o}.f(A)}$.

\proofpara{$\lbMarker{\rangeOf{o}.f(A)} \leq f_g$}
Substituting \Cref{def:agg-function-bounds} we get for $\lbMarker{\rangeOf{o}.f(A)}$:

\begin{align}
  \lbMarker{\rangeOf{o}.f(A)} &=
  \sum_{\rangeTup \in {\mathbf{S}_g^{certain}}} \lbMarker{(\rangeRel(\rangeTup) \amysmbNAU \rangeTup.A)} \notag{}\\
  &\amadd \sum_{\rangeTup \in {\mathbf{S}_g^{uncertain}}} \min(\lbMarker{(\rangeRel(\rangeTup) \amysmbNAU \rangeTup.A)}, \amzero) \notag{}\\
  &\amadd \sum_{\rangeTup \in \mathbf{N}_g} \min(\lbMarker{(\rangeRel(\rangeTup) \amysmbNAU \rangeTup.A)}, \amzero) \label{eq:o-agg-lb}
\end{align}

Using \Cref{lem:add-0-to-m-is-larger}, we know that

\[
  \sum_{\rangeTup \in \mathbf{N}_g} \min(\lbMarker{(\rangeRel(\rangeTup) \amysmbNAU \rangeTup.A)}, \amzero) \leq \sum_{\rangeTup \in \mathbf{N}_g} \amzero = \amzero
  \]

Thus, we can bound \Cref{eq:o-agg-lb} from above:

\begin{align}
 \leq &\sum_{\rangeTup \in {\mathbf{S}_g^{certain}}} \lbMarker{(\rangeRel(\rangeTup) \amysmbNAU \rangeTup.A)} \notag{}\\ &\amadd \sum_{\rangeTup \in {\mathbf{S}_g^{uncertain}}} \min(\lbMarker{(\rangeRel(\rangeTup) \amysmbNAU \rangeTup.A)}, \amzero) \label{eq:o-agg-no-N}
\end{align}

We next will relate \Cref{eq:o-agg-no-N} to $f_g$ through $\TM_\rel$. Towards this goal for any $\rangeTup \in \mathbf{S}_g$ we define $T_\rangeTup = \{ \tup \mid \TM_{\rel}(\rangeTup, \tup) > 0\}$.
We know that for any $\rangeTup \in \mathbf{S}_g$, we have $\lbMarker{\rangeRel(\rangeTup)} \leq \sum_{\tup \in T_\rangeTup} \TM_\rel(\rangeTup, \tup) \ubMarker{\rangeRel(\rangeTup)}$ because $\TM_\rel$ is a tuple matching based on which $\rangeRel$ bounds $\rel$.
Consider the sum in \Cref{eq:o-agg-no-N} which ranges over $\mathbf{S}_g^{certain}$ first.
Since $\lbMarker{\rangeRel(\rangeTup)} \leq \sum_{\tup \in T_\rangeTup} \TM_\rel(\rangeTup, \tup) \leq \ubMarker{\rangeRel(\rangeTup)}$, based on \Cref{theo:semimodulish-bound-preserving} we have that $(\sum_{\tup \in T_\rangeTup} \TM_\rel(\rangeTup, \tup)) \asmbN \lbMarker{\rangeTup.A}$ is bound from below by $\lbMarker{(\rangeRel(\rangeTup) \amysmbNAU \rangeTup.A)}$. Thus,

\begin{align}
  &\sum_{\rangeTup \in {\mathbf{S}_g^{certain}}} \lbMarker{(\rangeRel(\rangeTup) \amysmbNAU \rangeTup.A)} \notag{}\\
  \leq &\sum_{\rangeTup \in {\mathbf{S}_g^{certain}}} \left(\sum_{\tup: \tup \in T_{\rangeTup}} \TM_R(\rangeTup, \tup)\right) \asmbN \lbMarker{\rangeTup.A} \label{eq:agg-proof-range-bound-by-tm-sum}
\end{align}

Note that for any $\rangeTup \in \mathbf{S}_g^{certain}$, $\tup \not\in T_g \Rightarrow \tup \not\tmatch \rangeTup$ since $\rangeTup.G$ is certain.

\begin{align*}
  &= \sum_{\rangeTup \in {\mathbf{S}_g^{certain}}}  \left(\sum_{\tup \in T_g} \TM_R(\rangeTup, \tup) \right) \asmbN \lbMarker{\rangeTup.A}
\end{align*}

For any semimodule and thus also every $\semN$-semimodule the law $(k_1 + k_2) \amysmbNAU m = k_1 \asmbN m \amadd k_2 \asmbN m$ holds. Applying this law we can factor out the inner sum:

\begin{align*}
  &=   \sum_{\rangeTup \in {\mathbf{S}_g^{certain}}} \sum_{\tup \in T_g} \TM_R(\rangeTup, \tup) \asmbN \lbMarker{\rangeTup.A}
\end{align*}

Using commutativity and associativity of $\amadd$, we commute the two sums:

\begin{align*}
  &=  \sum_{\tup \in T_g} \sum_{\rangeTup \in {\mathbf{S}_g^{certain}}} \TM_R(\rangeTup, \tup) \asmbN \lbMarker{\rangeTup.A}
\end{align*}

Since, $\tup \tmatch \rangeTup$ for any $\tup \in T_\rangeTup$, we have $\rangeTup.A < \tup.A$ from which follows that:

\begin{align*}
  &<  \sum_{\tup \in T_g} \sum_{\rangeTup \in {\mathbf{S}_g^{certain}}} \TM_R(\rangeTup, \tup) \asmbN \tup.A
\end{align*}

To recap so far we have shown that

\begin{align}
    &\sum_{\rangeTup \in {\mathbf{S}_g^{certain}}} \lbMarker{(\rangeRel(\rangeTup) \amysmbNAU \rangeTup.A)} \notag{}\\
  < &\sum_{\tup \in T_g} \sum_{\rangeTup \in {\mathbf{S}_g^{certain}}} \TM_R(\rangeTup, \tup) \asmbN \tup.A \label{eq:first-sum-of-agg-bound}
\end{align}

Next we will bound the second sum from \Cref{eq:o-agg-no-N} which ranges over $\mathbf{S}_g^{uncertain}$ in a similar fashion and then show that $f_g$ is lower bound by the sum of these bounds.
For $\rangeTup \in \mathbf{S}_g^{uncertain}$ since $\rangeTup.G$ is uncertain, some $\tup \in T_\rangeTup$ may belong to a group $g' \neq g$. We will have to treat this case differently in the following. For that we define $T_\rangeTup^{+} = \{ \tup \mid \tup \in T_\rangeTup \land \tup.G = g\}$ and $T_\rangeTup^{-} = \{ \tup \mid \tup \in T_\rangeTup \land \tup.G \neq g\}$.
Let $k_{\rangeTup}^+ = \sum_{\tup \in T_\rangeTup^+} \TM_R(\rangeTup, \tup)$ and $k_{\rangeTup}^- =  \sum_{\tup \in T_\rangeTup^-} \TM_R(\rangeTup, \tup)$. Using the same argument as for \Cref{eq:agg-proof-range-bound-by-tm-sum}, we get:

\begin{align*}
  &\sum_{\rangeTup \in {\mathbf{S}_g^{uncertain}}} \min(\lbMarker{(\rangeRel(\rangeTup) \amysmbNAU \rangeTup.A)}, \amzero) \\
  \leq &\sum_{\rangeTup \in {\mathbf{S}_g^{uncertain}}} \min\left( \left( k_{\rangeTup}^+ + k_{\rangeTup}^-\right)  \amysmbNAU \lbMarker{\rangeTup.A}, \amzero \right)
\end{align*}

We consider two cases: (i) $\lbMarker{\rangeRel(\rangeTup)} \leq k_\rangeTup^+ \leq \ubMarker{\rangeRel(\rangeTup)}$ and (ii)
$k_\rangeTup^+ < \lbMarker{\rangeRel(\rangeTup)} \leq \ubMarker{\rangeRel(\rangeTup)}$. For case $(i)$ first consider that based on \Cref{lem:add-0-to-m-is-larger}, we have

\begin{align*}
\leq &\sum_{\rangeTup \in {\mathbf{S}_g^{uncertain}}} \lbMarker{\rangeRel(\rangeTup) \amysmbNAU \rangeTup.A}
\end{align*}

From (i) follows that $\sum_{\tup \in T_\rangeTup} \TM_\rel(\rangeTup, \tup) \asmbN \lbMarker{\rangeTup}$ is bound from below by $\lbMarker{\rangeRel(\rangeTup) \amysmbNAU \rangeTup.A}$ for any $\rangeTup \in \mathbf{S}_g^{uncertain}$. Applying the same steps as for $\mathbf{S}_g^{certain}$, we get:

\begin{align}
\leq &\sum_{\tup \in T_g} \sum_{\rangeTup \in {\mathbf{S}_g^{uncertain}}} \TM_R(\rangeTup, \tup) \asmbN \tup.A \label{eq:second-sum-of-agg-bound-case-one}
\end{align}

Now we have to prove the same for (ii), i.e., when $k_\rangeTup^+ < \lbMarker{\rangeRel(\rangeTup)}$. Let $min_A = \min(\{ \tup.A \mid \tup \in T_\rangeTup^+\})$. We now prove for each $\monoid \in \{ \msum, \mmin, \mmax \}$ that under assumption (ii) for any $\rangeTup \in \mathbf{S}_g^{uncertain}$ the following holds:

\begin{align}
       & \min((k_{\rangeTup}^+ + k_{\rangeTup}^-) \asmbN \lbMarker{\rangeTup.A}, \amzero ) \notag \\
  \leq & k_{\rangeTup}^+   \asmbN min_a \notag{}                                           \\
  =    & \left( \sum_{\tup \in T_\rangeTup^+} \TM_R(\rangeTup, \tup) \right) \asmbN min_a \label{eq:sum-of-agg-bound-case-two-each-t}
\end{align}

\proofpara{$\msum$}
Recall that $\smbN{\msum}$ is multiplication  and $\mzero{\msum} = 0$.
We distinguish two cases: $min_a \leq 0$ and $\min_a > 0$. If $min_a \leq 0$, since $\lbMarker{\rangeTup.A} \leq min_a$, it follows that $\min((k_{\rangeTup}^+ + k_{\rangeTup}^-) \smbN{\msum} \lbMarker{\rangeTup.A}, \amzero)  = \min((k_{\rangeTup}^+ + k_{\rangeTup}^-) \cdot \lbMarker{\rangeTup.A}, 0) = (k_{\rangeTup}^+ + k_{\rangeTup}^-) \cdot \lbMarker{\rangeTup.A} < k_{\rangeTup}^+ \cdot \lbMarker{\rangeTup.A} < k_{\rangeTup}^+ \cdot min_a$. If $min_a > 0$, then $k_{\rangeTup}^+ \cdot min_a \geq 0$ and since $\min(m,0) \leq 0$ for any $m$, we get $\min((k_{\rangeTup}^+ + k_{\rangeTup}^-) \smbN{\msum} \lbMarker{\rangeTup.A}, 0) \leq k_{\rangeTup}^+ \cdot min_a$.

\proofpara{$\mmin$}
Since $\smbN{\mmin}$ is the identity on $\mmin$ except for when $k = 0$ and because  $\mzero{\mmax} = \infty$, we get
$\min((k_{\rangeTup}^+ + k_{\rangeTup}^-) \smbN{\mmin} \lbMarker{\rangeTup.A}, \amzero) = (k_{\rangeTup}^+ + k_{\rangeTup}^-) \smbN{\mmin} \lbMarker{\rangeTup.A}$. Distinguish two cases. If $k_{\rangeTup}^+ = 0$, then
$k_{\rangeTup}^+ \smbN{\mmin} \lbMarker{\rangeTup.A} = \infty > m$ for any $m$ including $(k_{\rangeTup}^+ + k_{\rangeTup}^-) \smbN{\mmin} \lbMarker{\rangeTup.A}$. Otherwise, since $\smbN{\mmin}$ is the identify on $\mmin$ if $k \neq 0$, we have
$(k_{\rangeTup}^+ + k_{\rangeTup}^-) \smbN{\mmin} \lbMarker{\rangeTup.A} = k_{\rangeTup}^+ \smbN{\mmin} \lbMarker{\rangeTup.A}$.

\proofpara{$\mmax$}
Since $\mzero{\mmax} = - \infty$, we get $\min((k_{\rangeTup}^+ + k_{\rangeTup}^-) \smbN{\mmax} \lbMarker{\rangeTup.A}, \amzero)  = - \infty \leq k_{\rangeTup}^+   \smbN{\mmax} min_a$.

Using \Cref{eq:sum-of-agg-bound-case-two-each-t} proven above, we can apply the same steps as in the proof of \Cref{eq:first-sum-of-agg-bound} to deduce that:

\begin{align}
       & \sum_{\rangeTup \in {\mathbf{S}_g^{uncertain}}} \min(\lbMarker{\rangeRel(\rangeTup) \amysmbNAU \rangeTup.A}, \amzero)  \notag{}    \\
  \leq & \sum_{\rangeTup \in {\mathbf{S}_g^{uncertain}}} \left(\sum_{\tup \in T_\rangeTup^+} \TM_R(\rangeTup, \tup)\right) \amysmbNAU min_a \notag{} \\
=      & \sum_{\rangeTup \in {\mathbf{S}_g^{uncertain}}} \left(\sum_{\tup \in T_g} \TM_R(\rangeTup, \tup)\right) \amysmbNAU min_a \notag{} \\
  =    & \sum_{\rangeTup \in {\mathbf{S}_g^{uncertain}}} \sum_{\tup \in T_g} \TM_R(\rangeTup, \tup) \amysmbNAU min_a                       \notag{}  \\
  \leq & \sum_{\rangeTup \in {\mathbf{S}_g^{uncertain}}} \sum_{\tup \in T_g} \TM_R(\rangeTup, \tup) \amysmbNAU \tup.A                    \notag{}    \\
  =    & \sum_{\tup \in T_g} \sum_{\rangeTup \in {\mathbf{S}_g^{uncertain}}} \TM_R(\rangeTup, \tup) \amysmbNAU \tup.A
         \label{eq:second-sum-of-agg-bound-case-two}
\end{align}


Combining \Cref{eq:first-sum-of-agg-bound} with \Cref{eq:second-sum-of-agg-bound-case-one} and \Cref{eq:second-sum-of-agg-bound-case-two} and using \Cref{lem:add-0-to-m-is-larger} we get

\begin{align*}
  \lbMarker{\rangeOf{o}.f(A)} & \leq \sum_{\tup \in T_g} \sum_{\rangeTup \in {\mathbf{S}_g^{certain}}} \TM_R(\rangeTup, \tup) \asmbN \tup.A \\
                              & \amadd \sum_{\tup \in T_g} \sum_{\rangeTup \in {\mathbf{S}_g^{uncertain}}} \TM_R(\rangeTup, \tup) \asmbN \tup.A\\
                              & = \sum_{\tup \in T_g} \sum_{\rangeTup \in {\mathbf{S}_g}} \TM_R(\rangeTup, \tup) \asmbN \tup.A              \\
                              & = f_g
\end{align*}

\proofpara{$\ubMarker{\rangeOf{o}.f(A)} \geq f_g$}
We still need to prove that $\ubMarker{\rangeOf{o}.f(A)} \geq f_g$.

\begin{align*}
  \ubMarker{\rangeOf{o}.f(A)} &=
  \sum_{\rangeTup \in {\mathbf{S}_g^{certain}}} \ubMarker{(\rangeRel(\rangeOf{o}) \amysmbNAU \rangeTup.A)}\\
  &\amadd \sum_{\rangeTup \in {\mathbf{S}_g^{uncertain}}} \max(\ubMarker{(\rangeRel(\rangeOf{o}) \amysmbNAU \rangeTup.A)}, \amzero)\\
  &\amadd \sum_{\rangeTup \in \mathbf{N}_g} \max(\ubMarker{(\rangeRel(\rangeOf{o}) \amysmbNAU \rangeTup.A)}, \amzero)\\
\end{align*}

This prove is analog to the prove for $\lbMarker{\rangeOf{o}.f(A)} \leq f_g$ except that it is always the case that $k_\rangeTup^+ \leq \ubMarker{\rangeRel(\rangeTup)}$ which simplifies the case for $\mathbf{S}_g^{uncertain}$.

\proofpara{$\lbMarker{\query(\rangeRel)(\rangeOf{o})} \leq \sum_{\tup \tmatch \rangeOf{o}} \TM_\query(\rangeOf{o},\tup) \leq \ubMarker{\query(\rangeRel)(\rangeOf{o})}$}
So far we have established that for any $\rangeOf{o} \in \mathbf{O}_g$ we have $\tup_g \tmatch \rangeOf{o}$. For that follows that when constructing a tuple matching $\TM_\query$ based on which $\query(\rangeTup)$ bounds $\query(\rel)$ we can associate $\tup_g$ with any subset of $\mathbf{O}_g$. It remains to be shown that we can find such a tuple matching such that $\lbMarker{\query(\rangeRel)(\rangeOf{o})} \leq \sum_{\tup \tmatch \rangeOf{o}} \TM_\query(\rangeOf{o},\tup) \leq \ubMarker{\query(\rangeRel)(\rangeOf{o})}$. Since each aggregation result in $\query(\rel)$ appears exactly once, this boils down to proving that $\lbMarker{\query(\rangeRel)(\rangeOf{o})} \leq \card{\{ \tup \mid \TM_\query(\rangeOf{o},\tup) \neq 0 \}} \leq \ubMarker{\query(\rangeRel)(\rangeOf{o})}$. We will make use of the following notation:
\begin{align*}
  \mathbf{Tup}_\rel &= \{  \tup \mid \rel(\tup) \neq 0 \}\\
  \mathbf{Tup}_\rangeRel &= \{ \rangeTup \mid \rangeRel(\rangeTup) \neq 0 \}\\
  \mathbf{Tup}_{output} &= \{ \rangeOf{o} \mid \query(\rangeRel)(\rangeOf{o}) \neq 0 \}
\end{align*}

Recall that $\mathcal{G}$ denotes the set of groups in $\rel$.
For the construction of $\TM_\query$ we will use a mapping $gcover: \mathcal{G} \to \mathbf{Tup}_\rel \times \mathbf{Tup}_{\rangeRel} \times \mathbf{Tup}_{output}$ such that for any $g \in \mathcal{G}$ for which $gcover(g) = (\tup, \rangeTup, \rangeOf{o})$ the following conditions hold:
\begin{align*}
  \tup &\in T_g
  &\TM_R(\rangeTup, \tup) &\neq 0
  &\gsdefP(\rangeTup) &= \rangeOf{o}
\end{align*}
We will refer to such a mapping as a group cover.
The purpose of a group cover $gcover$ is to assign each group $g$ in the possible world to an output $\rangeOf{o}$ which represents this group and to justify this assignment through an input $\rangeTup$ that is associated by $TM_\rel$ with at least one tuple from group $g$ and is assigned by the grouping strategy to the \abbrUAADB output tuple $\rangeOf{o}$. We will first prove that at least one group cover exists and then prove that a group cover induces a tuple matching $\TM_\query$ for which the condition we want to prove ($\lbMarker{\query(\rangeRel)(\rangeOf{o})} \leq \sum_{\tup \tmatch \rangeOf{o}} \TM_\query(\rangeOf{o},\tup) \leq \ubMarker{\query(\rangeRel)(\rangeOf{o})}$) holds for all $\rangeOf{o} \in \mathbf{Tup}_{output}$.

\proofpara{Group cover exists}
To prove the existence of a group cover, we will show how to construct such a group cover for any aggregation query $\query$, input $\rel$ that is bound by a $\uaaN$-relation $\rangeRel$. Consider a group $g \in \mathbf{G}$ and pick a arbitrary tuple $\tup \in T_g$ and $\rangeTup \in \mathbf{S}_g$ such that $\TM_{\rel}(\rangeTup, \tup) \neq$. At least one such $\tup$ has to exist for $g$ to be a group in the result of $\query(\rel)$. Furthermore, since $\TM_\rel$ is a tuple matching based on which $\rangeRel$ bounds $\rel$, there has to exist at least one such $\rangeTup$. Now recall that $\gsdefP$ associates each tuple $\rangeTup$ for which $\rangeRel(\rangeTup) \neq \uaaZero{\semN}$ with one output in $\query(\rangeRel)$. WLOG let $o = \gsdefP(\rangeRel)$. We set $gcover(g) = (\tup, \rangeTup, \rangeOf{o})$. By construction $gcover$ is a group cover.

\proofpara{$\lbMarker{\query(\rangeRel)(\rangeOf{o})} \leq \sum_{\tup \tmatch \rangeOf{o}} \TM_\query(\rangeOf{o},\tup) \leq \ubMarker{\query(\rangeRel)(\rangeOf{o})}$}
It remains to be shown that it is possible to construct a tuple matching $\TM_\query$ such that for any $\rangeOf{o}$ in $\query(\rangeRel)$ we have
\[
  \lbMarker{\query(\rangeRel)(\rangeOf{o})} \leq \sum_{\tup \tmatch \rangeOf{o}} \TM_\query(\rangeOf{o},\tup) \leq \ubMarker{\query(\rangeRel)(\rangeOf{o})}
\]
which implies that $\query(\rangeRel)$ bounds $\query(\rel)$ based on $\TM_\query$. Since aggregation returns a single result tuple $t_g$ for each group $g$, we know that that $\query(\rel)(\tup_g) = 1$.
Using $gcover$, we construct $\TM_\query$ as shown below:
\begin{align*}
\TM_\query(\tup_g,\rangeOf{o}) =
  \begin{cases}
    1 &\mathtext{if} \exists \tup,\rangeTup: gcover(g) = (\tup,\rangeTup,\rangeOf{o})\\
    0 &\mathtext{otherwise}\\
  \end{cases}
\end{align*}
Obviously, $\sum_{\rangeTup} \TM_\query(\tup_g, \rangeTup) = 1 = \query(\rel(\tup_g))$. Thus, $\TM_{\query}$ is a tuple matching. It remains to be shown that for each $\rangeOf{o}$ in $\query(\rangeRel)$ we have $\lbMarker{\query(\rangeRel)(\rangeOf{o})} \leq \sum_{\tup} \TM_\query(\tup, \rangeRel) \leq \ubMarker{\query(\rangeRel)(\rangeOf{o})}$. Observe that based on how we have constructed $gcover$, the following holds for any such $\rangeOf{o}$:
\begin{align*}
  \sum_{\tup_g} \TM_{\query}(\rangeOf{o},\tup_g) = \sum_{g \in \mathbf{G}: \exists \tup, \rangeTup: gcover(g) = (\tup, \rangeTup, \rangeOf{o})} 1
\end{align*}
For any group cover if $gcover(g) = (\tup, \rangeTup, \rangeOf{o})$ then $\TM_{\rel}(\rangeTup, \tup) \neq 0$. Then,
\begin{align*}
\leq &\sum_{g \in \mathbf{G}: \exists \tup, \rangeTup: gcover(g) = (\tup, \rangeTup, \rangeOf{o})} \ubMarker{\rangeRel(\rangeTup)}
\end{align*}
Since $\gsdefP$ may assign to $\rangeOf{o}$ additional tuples which do not co-occur with $\rangeOf{o}$ in  $gcover$, i.e., where $\neg \exists g, \rangeTup, tup: gcover(g) = (\tup, \rangeTup, \rangeOf{o})$, we get:
\begin{align*}
  \leq &\sum_{\rangeTup: \gsdefP(\rangeTup) = \rangeOf{o}} \ubMarker{\rangeRel(\rangeTup)}
     = &\ubMarker{\query(\rangeRel)(\rangeOf{o})}
\end{align*}
It remains to be shown that $\lbMarker{\query(\rangeRel)(\rangeOf{o})} \leq \sum_{\tup \tmatch \rangeOf{o}} \TM_\query(\rangeOf{o},\tup)$. From the construction of $\TM_\query$ follows that:
\begin{align*}
  &\sum_{\tup \tmatch \rangeOf{o}} \TM_\query(\rangeOf{o},\tup)\\
  = &\sum_{g \in \mathbf{G}: \exists \tup, \rangeTup: gcover(g) = (\tup, \rangeTup, \rangeOf{o})} 1
\end{align*}
Based on \Cref{def:aggr-op-semantics-gb},
\begin{align}
\lbMarker{\query(\rangeRel)(\rangeOf{o})} = \duprem_{\semN}\left(\sum_{\rangeTup': \gsdefP(\rangeTup') = \asgrp}\lbMarker{\rangeRel(\rangeTup')}\right) \label{eq:agg-proof-tuple-annot-lb}
\end{align}
First consider the case where the set $\{ \rangeTup \mid \gsdefP(\rangeTup)  = \rangeOf{o} \}$ is empty. It follows that $\lbMarker{\query(\rangeRel)(\rangeOf{o})} = 0$ and the claim trivially holds.

If the set is non empty, then $\sum_{\tup \tmatch \rangeOf{o}} \TM_\query(\rangeOf{o},\tup) \geq 1$. Also
\[
  \duprem_{\semN}\left(\sum_{\rangeTup': \gsdefP(\rangeTup') = \rangeOf{o}}\lbMarker{\rangeRel(\rangeTup')}\right) \leq 1
  \], because
for any $k \in \semN$, $\duprem_{\semN}(k) \leq 1$ if $k \neq 0$. Thus,
\begin{align*}
  \leq &1  &\leq &\sum_{\tup \tmatch \rangeOf{o}} \TM_\query(\rangeOf{o},\tup)
\end{align*}
Thus, we have shown that $\lbMarker{\query(\rangeRel)(\rangeOf{o})} \leq \sum_{\tup \tmatch \rangeOf{o}}$ which together with $\sum_{\tup \tmatch \rangeOf{o}} \leq \ubMarker{\query(\rangeRel)(\rangeOf{o})}$ and the fact that $TM_\query$ only assigns non-zero annotations to $\tup$ and $\rangeOf{o}$ if $\tup \tmatch \rangeOf{o}$ that we have proven above implies that $\query(\rangeRel)$ bounds $\query(\rel)$ based on $\TM_\query$.

\proofpara{Aggregation without group-by}
The proof for aggregation without group-by is analog except for that both $\query(\rel)$ and $\query(\rangeRel)$ contain a single result tuple annotated with $1$ and $\ut{1}{1}{1}$ respectively. Let $\tup$ and $\rangeOf{o}$ denote this single result tuple. Then we trivially define $\TM_\query(\rangeTup, \tup) = 1$ and $\TM_\query(\tup',\rangeTup') = 0$ if either $\tup' \neq \tup$ or $\rangeTup \neq \rangeTup')$. Then, $\lbMarker{\query(\rangeRel)(\rangeOf{o})} \leq \sum_{\tup \tmatch \rangeOf{o}} \TM_\query(\rangeOf{o},\tup) \leq \ubMarker{\query(\rangeRel)(\rangeOf{o})}$. The proof of $\tup \tmatch \rangeOf{o}$ is analog to the proof for group-by aggregation.
\end{proof}

From \Cref{theo:aggregation-preserves-bounds}\iftechreport{, \Cref{theo:set-difference-prese},}\ifnottechreport{,\cite[Theorem X]{techreport} (bound preservation for set difference)} and \Cref{lem:ra-plus-preserves-bounds} follows our main technical result: Our query semantics for $\raAgg$ queries preserves bounds.

\begin{Corollary}[Preservation of bounds for $\raAgg$]\label{cor:preservation-of-boun}
  Let $\query$ be an $\raAgg$ query and $\pdb$ an incomplete $\semN$-database that is bound by an $\uaaK{\semN}$-database $\rangeDB$. Then $\query(\rangeDB)$ bounds $\query(\pdb)$.
  $$ \pdb \dbbounds \rangeDB \Rightarrow \query(\pdb) \dbbounds \query(\rangeDB)$$
\end{Corollary}

Note that our semantics for $\raAgg$ queries over $\uaaN$-relations has \ptime data complexity.

\begin{Theorem}[Data Complexity of $\raAgg$ Queries]\label{theo:data-complexity-of-raAgg}
Evaluation of $\raAgg$ queries over $\uaaN$-relations has \ptime data complexity.
\end{Theorem}
\begin{proof}
Query evaluation for $\raPlus$ over $\semK$-relations is known to be in \ptime. For $\raPlus$, our semantics only differs in the evaluation of expressions which adds an overhead that is independent on the size of the input database. For set difference, the semantics according to \Cref{def:set-diff-semantics} is in worst-case accessing the annotation of every tuple in the right-hand side input to calculate the annotation of a result tuple. Since each result tuple belongs to the left-hand side input, the complexity is certainly in $O(n^2)$ which is \ptime. Finally, for aggregation, the number of result tuples is at most linear in the input size and even a naive implementation just has to test for each input whether it contributes to a particular output. Thus, aggregation is certainly in $O(n^2)$ and we get an overall \ptime data complexity for evaluation of $\raAgg$ queries over $\uaaN$-relations.
\end{proof}


\section{Implementation}
\label{sec:implementation}
In this section we discuss about the implementation of our \abbrAUDB model as a middleware running on top of conventional database systems. For that we define an encoding of $\uaaK{\semN}$-relations as classical bag semantics relations implemented as a function $\Enc$ which maps a $\uaaK{\semN}$-database to a bag semantics database. We use $\Dec$ to denote the inverse of $\Enc$.
Using the encoding we apply query rewriting to propagate annotations and implement $\uaaK{\semN}$-relational query semantics over the encoding.  
Our frontend rewriting engine receives a query  $\query$ over an $\uaaK{\semN}$-annotated database $\rangeDB$ and rewrites this into a query $\rewrUAA{\query}$ that evaluated over $\Enc(\rangeDB)$ returns the encoding of $\query(\rangeDB)$. That is, we will show that:
\begin{align*}
  \query(\rangeDB) = \Dec(\qmerge(\Enc(\rangeDB)))
\end{align*}

\subsection{Relational encoding of \abbrAUDBs}

We now define $\Enc$ for a single $\uaaN$-relation $\rangeRel$.  $\Enc(\rangeDB)$ is then defined as the database generated by applying $\Enc$ to each relation $\rangeRel \in \rangeDB$.
We use $\bar{A}$ to denote a set of attributes.
We use $\schemaOf(\rel)$ to denote the schema of input relation $\rel$.
The schema of $\Enc(\rangeRel)$ for an $\uaaN$-relation $\rangeRel$
with schema $\schemaOf{\rangeRel} = (a_1, \ldots, a_n)$ is

\[
  \schemasymb(\Enc(\rangeRel)) = (\bar{A},\ubMarker{\bar{A}},\lbMarker{\bar{A}}, \rlb, \ratt, \rub )
\]

\begin{align*}
	\mathtext{where } \bar{A}=\{\bgMarker{A_1}, \ldots, & \bgMarker{A_n}\}, \\
		& \ubMarker{\bar{A}}=\{\ubMarker{A_1}, \ldots, \ubMarker{A_n}\}, \lbMarker{\bar{A}}=\{\lbMarker{A_1}, \ldots, \lbMarker{A_n}\}
\end{align*}

\begin{Example}
	The schema of $\Enc(\rangeRel)$ for \abbrAUDB relation $\rangeRel(A,B)$ is $(A,B,\lbatt{A},\lbatt{B},\ubatt{A},\ubatt{B},\rlb,\ratt,\rub)$.
\end{Example}

For each tuple $\rangeTup$ with $\rangeRel(\rangeTup) \neq \uaaZero{\semN}$, there will be one tuple $\tup = \tenc{\rangeTup,\rangeRel(\rangeTup)}$ in $\Enc(\rangeRel)$ where $\tenc{}$ is a function that maps tuples from $\rangeRel$ and their annotations to the corresponding tuple from $\Enc(\rangeRel)$. Attributes $\rlb$, $\ratt$, and $\rub$ are used to store $\rangeRel(\rangeTup)$:
\begin{align*}
  \tenc{\rangeTup, k}.\rlb &= \lbMarker{k} \\
 \tenc{\rangeTup, k}.\ratt &= \bgMarker{k} \\
 \tenc{\rangeTup, k}.\rub  &= \ubMarker{k}
\end{align*}

For each attribute $A_i$, the three attributes $\lbMarker{A_i}$, $\bgMarker{A_i}$, and $\ubMarker{A_i}$ are used to store the range-annotated value $\rangeTup.A_i$:
\begin{align*}
  \tenc{\rangeTup,k}.\lbMarker{A_i} & = \lbMarker{\rangeTup.A_i}\\
  \tenc{\rangeTup,k}.\bgMarker{A_i} & = \bgMarker{\rangeTup.A_i}\\
  \tenc{\rangeTup,k}.\ubMarker{A_i} & = \ubMarker{\rangeTup.A_i}
\end{align*}

In addition we define a function $\tdec{}$ which takes a tuple in the encoding and returns the corresponding range-annotated tuple $\rangeTup$. Given a tuple $\tup$ with schema $\schemasymb(\Enc(\rangeRel)) = (\bar{A},\ubMarker{\bar{A}},\lbMarker{\bar{A}}, \rlb, \ratt, \rub )$ where $\rangeRel$ is a $\uaaN$-relation $\rangeRel$, $\tdec{}$ returns a tuple with schema $\schemasymb(\rangeRel) = (\bar{A})$ such that for all $A_i \in \schemaOf(\rangeRel)$:

\begin{align*}
  \tdec{\tup}.A_i = \uv{\tup.\lbatt{A_i}}{\tup.\bgatt{A_i}}{\tup.\ubatt{A_i}}
\end{align*}

Furthermore, we define a function $\tdecr{\tup}$ which extract the row annotation encoded by a tuple $\tup$ in the encoding:
\[
  \tdecr{\tup} = \ut{\tup.\rlb}{\tup.\rbg}{\tup.\rub}
\]
Having defined the schema and tuple-level translation, we define $\Enc$ and its inverse $\Dec$ below.

\begin{Definition}[Relational Encoding]\label{def:relational-encoding}
  Let $\rangeRel$ be a $\uaaN$-relation with schema $(A_1, \ldots, A_n)$ and let $\rel = \Enc(\rangeRel)$. Furthermore, let $\rangeTup$ be a tuple with schema $\schemaOf(\rangeRel)$ and $\tup$ be a tuple with schema $\schemaOf(\rel)$.
  \begin{align*}
    \Enc(\rangeRel)(\tup) &\defas
                            \begin{cases}
                              1 &\mathtext{if}\,\,\, \exists \rangeTup: \tup = \tenc{\rangeTup} \land \rangeRel(\rangeTup) > \uaaNzero\\
                              0 &\mathtext{otherwise}\\
                            \end{cases}
  \end{align*}
  \begin{align*}
    \Dec(\rel)(\rangeTup) &\defas \sum_{\tup: \tdec{\tup} = \rangeTup} \tdecr{\tup} \cdot \ut{\rel(\tup)}{\rel(\tup)}{\rel(\tup)}
  \end{align*}
\end{Definition}
\subsection{Rewriting}\label{sec:rewrite-rules}

We now define the rewriting $\rewrUAA{\cdot}$. We assume that for each input relation $\rangeRel$ of a query, $\Enc(\rangeRel)$ has been materialized as a relation $\rel_{\Enc}$. We will discuss how to create \abbrUAADBs in \Cref{sec:creating-abbruaadbs}.  These techniques enable \abbrUAADBs to be generated as part of the rewritten query in addition to supporting reading from a materialized input table. We call two tuples in the the relational encoding value equivalent if they are equal after projecting away the row annotation attributes ($\rlb$, $\rbg$, and $rub$. Note that $\Enc$ does never produce an output where two tuples are value-equivalent. To ensure that intermediate results that are valid encodings, we have to sum up the row annotations of value-equivalent tuples which requires aggregation for operators like projection and union. Observe that we only need to ensure that the final result of a rewritten query is a valid encoding. Thus, we can allow for valid-equivalent tuples as long as we ensure that they are merged in the final result.
In the following let $\query_1$ be a query with result schema $\bar{A} = (A_1, \ldots, A_n)$ and $\query_2$ be a query with schema $\bar{B} = (B_1, \ldots, B_m)$. We use $e \renameto A$ in generalized projections to denote that the projection onto expression $e$ renaming the result to $A$, e.g., $\projection_{A + B \renameto C, D \renameto E}$ has schema $(C,E)$. We will use $\Enc(\rangeTup)$ to refer to the deterministic tuple in $\Enc(\rangeRel)$ that encodes the \abbrUAADB tuple $\rangeTup$ and its annotation $\rangeRel(\rangeTup)$.


\mypar{Merge Annotations}
After rewriting a $\query$ using the rewriting scheme $\rewrUAA{\cdot}$ shown below, we merge the annotation of value-equivalent tuples to generate the final encoding. Given $\rewrUAA{\query}$, we return a rewritten query $\qmerge$ to realize this:

\begin{align*}
  \qmerge &\defas \gamma_{\bar{A},\lbatt{\bar{A}},\ubatt{\bar{A}},e_{c},e_{sg},e_{p}}(\rewrUAA{\query})\\
  e_{c} &\defas \aggsum(\rlb) \renameto \rlb\\
  e_{sg} &\defas  \aggsum(\rbg) \renameto \rbg\\
  e_{p} &\defas  \aggsum(\rub) \renameto \rub
\end{align*}
%

\mypar{Table Access}
Each access to a table $\rangeRel$ is rewritten into an access to $\Enc(\rangeRel)$ which is materialized as $\rel_{\Enc}$.

\[ \rewrUAA{\rangeRel} \defas \rel_{\Enc} \]

\mypar{Selection}
For a selection we only filter out tuples $\Enc(\rangeTup)$ which are guaranteed to not fulfill the selection condition $\theta$, i.e., where $\ubMarker{\theta(\rangeTup)} = \bfalse$. Given an expression $e$, we use $\ubMarker{e}$ ($\lbMarker{e}$, and $\bgMarker{e}$) to denote an expression that if applied to $\Enc(\rangeTup)$ for range-annotated tuple $\rangeTup$ returns $\ubMarker{\seval{e}{\rval_{\rangeTup}}}$ ($\lbMarker{\seval{e}{\rval_{\rangeTup}}}$, $\bgMarker{\seval{e}{\rval_{\rangeTup}}}$), i.e., the upper (lower, \abbrBG) result of evaluating $e$ over the $\rangeTup$ using range-annotated expression semantics (\Cref{def:range-expr-eval}). Recall that $\rval_{\rangeTup}$ denotes the range-annotated valuation that assigns tuple $\rangeTup$'s attribute values to the variables of expression $e$. We will use  $\rval_{\Enc{\rangeTup}}$ to denote the valuation that contains three variables $\lbMarker{A}, \bgMarker{A}, \ubMarker{A}$ for each variable $A$ in $\rval$ and assigns these variables to values from $\Enc(\rangeTup)$ as follows:

\begin{align*}
  \rval_{\Enc(\rangeTup)}(\lbMarker{A}) &= \Enc(\rangeTup).\lbMarker{A}  = \lbMarker{\rval_\rangeTup(A)}\\
  \rval_{\Enc(\rangeTup)}(\bgMarker{A}) &= \Enc(\rangeTup).\bgMarker{A}  = \bgMarker{\rval_\rangeTup(A)}\\
  \rval_{\Enc(\rangeTup)}(\ubMarker{A}) &= \Enc(\rangeTup).\ubMarker{A}  = \ubMarker{\rval_\rangeTup(A)}
\end{align*}

Note that the expression semantics of \Cref{def:range-expr-eval} only uses deterministic expression evaluation and it is always possible to generate deterministic expressions $\lbMarker{e}$, $\bgMarker{e}$, and $\ubMarker{e}$. For instance, for a condition $e \defas A \leq B$, we would generate $\ubMarker{e} \defas \ubMarker{A} \leq \lbMarker{B}$. For instance, consider a  tuple $\rangeTup = (\uv{1}{1}{1},  \uv{0}{1}{2})$ with $\rangeRel(\rangeTup) = \ut{1}{1}{1}$. This tuple would be encoded in $\Enc(\rangeRel)$ as $(\bgMarker{A}: 1, \bgMarker{B}: 1, \lbMarker{A}: 1, \lbMarker{B}: 0, \ubMarker{A}: 1, \ubMarker{B}: 1,\rlb: 1,\rbg: 1,\rub: 1)$. We get $\seval{\ubMarker{e}}{\rval_{\Enc(\rangeTup)}} = 1 \leq 0 = \bfalse$.
In the result of selection the annotation of tuples is determined based on their annotation in the input and whether they certainly or in the \abbrBG world fulfill the selection conditions (see \Cref{def:range-selection}).
\begin{align*}
  \rewrUAA{\selection_{\theta}(Q_1)} &\defas \projection_{\bar{A},\ubatt{\bar{A}}, \lbatt{\bar{A}}, e_{c},e_{sg},\rub}(\selection_{\ubMarker{\theta}}(\rewrUAA{Q_1}))\\
 e_{c} &\defas (\ifte{\lbMarker{\theta}}{1}{0})*\rlb \renameto \rlb \\
 e_{sg} &\defas (\ifte{\bgMarker{\theta}}{1}{0})*\rbg \renameto \rbg
\end{align*}

\mypar{Projection}
For a generalized projection $\projection_U(\rangeRel)$ with $U \defas e_1 \renameto A_1, \ldots, e_k \renameto A_k$, we rewrite each projection expression $e_i$ into three expressions $\lbMarker{e_i}$, $\bgMarker{e_i}$, and $\ubMarker{e_i}$ as explained for selection above. Then let $\lbatt{U} = \lbMarker{e_1} \renameto \lbatt{A_1}, \ldots, \lbMarker{e_k} \renameto \lbatt{A_k}$ and let $\bgMarker{U}$ and $\ubMarker{U}$ be defined analog.

\[ \rewrUAA{\projection_{U}(Q_1)}=\projection_{\bgatt{U},\ubatt{U},\lbatt{U},\rlb,\rbg,\rub}(\rewrUAA{Q_1}) \]

\mypar{Cross Product}
Recall that $\uaanMult$ is defined as pointwise multiplication. Thus, for crossproduct we have to multiply the bounds of row annotations of input tuples.

\begin{align*}
  \rewrUAA{Q_1 \times Q_2} & \defas \projection_{\bar{A},\bar{B},\lbatt{\bar{A}},\lbatt{\bar{B}}\ubatt{\bar{A}},\ubatt{\bar{B}},e_{c},e_{sg},e_{p}}(\query_{prod}) \\
 \query_{prod}             & \defas \rewrUAA{Q_1} \times \rewrUAA{Q_2}\\
 e_{c}                     & \defas Q_1.\rlb \cdot Q_2.\rlb \renameto \rlb                                                                                                           \\
 e_{sg}                    & \defas Q_1.\rbg \cdot Q_2.\rbg \renameto \rbg                                                                                                           \\
 e_{p}                     & \defas Q_1.\rub \cdot Q_2.\rub \renameto \rub
\end{align*}


\mypar{Union}
A union is rewritten as the union of its rewritten inputs.
\begin{align*}
  \rewrUAA{\query_1 \union \query_2} & \defas
                                       \rewrUAA{\query_1} \union \rewrUAA{\query_2} \\
\end{align*}

\mypar{Set Difference}
For set difference we need to develop a rewrite that implements the combiner operator $\combine$ which merges all tuples with the same values in the \abbrBGW. This rewrite is shown below. To calculate the lower bound of the range-bounded annotation for a tuple in the result of a set difference operator we have to determine for each tuple $\rangeTup$ from the left input the set of all tuples $\rangeTup'$ from the right input whose values overlap with $\rangeTup$, i.e., where $\rangeTup \matches \rangeTup'$.
These are tuples that may be equal to $\rangeTup$ in some possible world. To calculate the lower bound we have to assume that all these tuples are equal to $\rangeTup$ and appear with the maximum possible multiplicity. That is, we have to subtract from the lower annotation bound of $\rangeTup$ the sum of the upper annotation bounds of these tuples.
For that we join the inputs on a condition $\theta_{join}$ shown below that checks whether $\rangeTup \matches \rangeTup'$ holds by checking that $[ \lbMarker{\rangeTup.A}, \ubMarker{\rangeTup.A}]$ overlaps with $[ \lbMarker{\rangeTup'.B}, \ubMarker{\rangeTup'.B}]$ for each attribute $A$ of the left input and the corresponding attribute $B$ of the right input. The \abbrUAADB-annotation of a tuple is then computed by grouping on the \abbrBG values of the LHS, summing up the upper bounds of tuples from the RHS and then subtract them from the lower bound of the LHS tuple's annotation. Below $\query_{SumRight}$ implements this step. To calculate the upper bound of a tuple's annotation we only subtract the lower bound annotations of tuples from the RHS if the tuples are guaranteed to be equal to the LHS in all possible worlds. That is the case if both the LHS and RHS tuple's attribute values are all certain (the lower bound is equal to the upper bound) and the tuples are equal. This is checking using condition $\theta_c$ shown below. Using a conditional expression $e_{pv}$ we only sum up the lower bounds of the annotation of RHS tuples fulfilling $\theta_c$. Finally, to calculate the multiplicity of a tuple in the \abbrBGW, we sum of the \abbrBG-annotations of RHS which are equal to the LHS tuple wrt. the tuples' \abbrBG-values. Finally, since this can result in negative multiplicities.

\begin{align*}
  \rewrUAA{Q_1 \difference Q_2} & \defas \selection_{\rub>0}(\projection_{\bar{A},\lbatt{\bar{A}},\ubatt{\bar{A}},e_{c},e_{sg},e_{p}}(\query_{sumright}))         \\
  e_{c}                         & \defas \aggmax(\rlb - \lbatt{rrow}, 0) \renameto \rlb                                                                              \\
  e_{sg}                        & \defas \aggmax(\rbg-\bgatt{rrow},0)  \renameto \rbg                                                                                \\
  e_{p}                         & \defas \aggmax(\rub - \ubatt{rrow},0) \renameto \rub                                                                               \\
  \query_{sumright}             & \defas \gamma_{\bar{A},\ubatt{\bar{A}}, \lbatt{\bar{A}},\rlb,\rbg, \rub, e_{sc}, e_{ssg}, e_{sp}}(\query_{preagg})                                                      \\
  e_{sc}                        & \defas \aggsum(\lbatt{rrow}) \renameto \lbatt{rrow}                                                                                     \\
  e_{ssg}                       & \defas \aggsum(\bgatt{rrow}) \renameto \bgatt{rrow}                                                                                      \\
  e_{sp}                        & \defas \aggsum(\ubatt{rrow}) \renameto \ubatt{rrow}                                                                                       \\
  \query_{preagg} &\defas \projection_{\bar{A},\ubatt{\bar{A}}, \lbatt{\bar{A}}, Q_1.\rlb, Q_1.\rbg, Q_1.\rub, e_{cv}, e_{sg}, e_{pv}}(\query_{join})\\
 e_{cv}                          & \defas Q_2.\rub \renameto \lbatt{rrow}\\
  e_{sgv}                       & \defas \ifte{\theta_{sg}}{Q_2.\rbg}{0}  \renameto \bgatt{rrow}                                                                                      \\
  \theta_{sg}                   & \defas \bigwedge_{A \in \bar{A}, B \in \bar{B}} \bgatt{A} = \bgatt{B}                                                           \\
  e_{pv}                        & \defas \ifte{\theta_c}{Q_2.\rlb}{0} \renameto \ubatt{rrow}                                                                                            \\
  \theta_c                      & \defas \bigwedge_{A \in \bar{A}, B \in \bar{B}} \lbatt{A} = \ubatt{A} \wedge \ubatt{A} = \lbatt{B} \wedge \lbatt{B} = \ubatt{B} \\
  \query_{join}                 & \defas \rewrUAA{\combine(Q_1)} \join_{\theta_{join}} \rewrUAA{Q_2}                                                              \\
  \theta_{join}                 & \defas \bigwedge_{i \in \{1, \ldots, n\}}  \ubatt{A_i} \geq \lbatt{B_i} \wedge \ubatt{B_i} \geq \lbatt{A_i}                     \\
\end{align*}

\mypar{$\combine(\query)$}
The \abbrBG combiner merges all tuples with the same values in the \abbrBGW by summing up their annotations and by merging their range-annotated values. We can implement this in relational algebra using aggregation.
\begin{align*}
   \rewrUAA{\combine(\query)}
  & \defas \gamma_{\bar{A},U_c,U_p,e_c,e_{sg},e_{p}}(\rewrUAA{\query})       \\
U_c                             & \defas \min(\lbatt{A_1}) \to \lbatt{A_1}, \ldots, \min(\lbatt{A_n}) \to \lbatt{A_n} \\
  U_p                           & \defas \max(\ubatt{A_1}) \to \ubatt{A_1}, \ldots, \max(\ubatt{A_n}) \to \ubatt{A_n} \\
  e_{c}                         & \defas sum(\rlb) \renameto \rlb                                                    \\
 e_{sg}                         & \defas sum(\rbg) \renameto \rbg                                                    \\
  e_{p}                         & \defas sum(\rub) \renameto \rub
\end{align*}

\mypar{Aggregation}
Our rewrite for aggregation support $\aggmax$, $\aggmin$, $\aggsum$, and $\aggcount$ directly. For $\aggavg$, we calculate $\aggsum$ and $\aggcount$ and then calculate $\aggavg(A) = \ifte{\aggcount(*) = 0}{0.0}{\frac{\aggsum(A)}{\aggcount(*)}}$ using projection (not shown here).
In the rewrite, we first determine output groups and ranged-bounded values for the group-by attributes of each of this output. This is achieved by grouping the input tuples based on the group-by \abbrBG values and calculating the minimum/maximum bounds of group-by values (query $\query_{gbounds}$). Each such output is then joined with the aggregation's input to match all inputs with an output that could contribute to the groups represented by this output. For that we have to check output's group-by bounds overlap with the input's group-by bound (query $\query_{join}$). Afterwards, we determine the bounds on the number of groups represented by each output and prepare expressions that calculate bounds for aggregation function results. These expression ($lba$, $sba$, and $uba$) are specific to the aggregation function $f$ and are explained below. Finally, we use aggregation to calculate aggregation function result bounds and row annotations. Recall that $\bar{A}$ denotes the attributes from relation $\rangeRel$.

\begin{align*}
  \rewrUAA{\aggregation{G}{f(A)}(R)} & \defas \aggregation{\bgatt{G},\ubatt{G},\lbatt{G}}{e_{aggbounds}}(Q_{proj})
\end{align*}\\[-12mm]
\begin{align*}
	e_{aggbounds}                            \defas
                                     & f(\bgatt{A}),f(\ubatt{A}), f(\lbatt{A})      \\
                                     & \aggmax(\rlb) \renameto \rlb,                                                                           \\
                                     & \aggmax(\rbg) \renameto \rbg,                                                                           \\
                                     & \aggsum(\rub) \renameto \rub                                                                            \\[2mm]
 \query_{proj}                       & \defas \projection_{\bgatt{G},\ubatt{G},\lbatt{G},lba,sga,uba,e_{c},e_{sg},e_{p}}(\query_{join}) \\
  e_{c}                                & \defas  \ifte{\theta_c}{1}{0} \renameto \rlb\\
  \theta_c &\defas \bigwedge_{A_i \in G}                                        \ubatt{A_i}=\ubatt{B_i} \wedge = \lbatt{A_i}=\lbatt{B_i} \wedge \lbatt{A_i} = \ubatt{A_i}\\
  e_{sg}                             & \defas \ifte{\theta_{sg}}{1}{0}
                                       \renameto \rbg                                                         \\
  \theta_{sg} &\defas \bigwedge_{A_i \in G} \bgMarker{A_i} = \bgMarker{B_i} \\
  e_{p} & \defas \ifte{\theta_{sg}}{\rub}{0}
                                       \renameto \rub \\[2mm]
  \query_{join}                      & \defas \query_{gbounds} \join_{\theta_{join}} \rename_{e_{rename}}(\rewrUAA{R})                       \\
  \theta_{join}                      &\defas \bigwedge_{A_i \in G}  \ubatt{A_i} \geq \lbatt{B_i} \wedge \ubatt{B_i} \geq \lbatt{A_i}                           \\
  e_{rename}                         & \defas \lbatt{A_1} \renameto \lbatt{B_1}, \ldots, \lbatt{A_n} \renameto \lbatt{B_n}, \ldots, \ubatt{A_n} \renameto \ubatt{B_n}                                                 \\[2mm]
  \query_{gbounds}                     & \defas                                                 \aggregation{\bgatt{G}}{e_{gbounds}}(\rewrUAA{R})  \\
  e_{gbounds}                      & \defas e_{bound}^{A_1}, \ldots, e_{bound}^{A_k} \tag{for $G = (A_1, \ldots, A_k)$}\\
  e_{bound}^{A}                      & \defas \aggmin(\lbatt{A}) \renameto \lbatt{A}, \aggmax(\ubatt{A}) \renameto \ubatt{A}
\end{align*}

The expressions that calculate bounds for aggregation function results ($lba$, $sga$, $uba$) are shown below. These expressions make use of expressions $lba_f$, $sga_f$, and $uba_f$ that are specific to the aggregation function $f$. We use expression $e_{gc}$ shown below to determine whether a tuple certainly belongs to a particular group, i.e., its group-by values are certain and its lower bound multiplicity is larger than zero. If a tuples group membership is uncertain, then we need account for the case where the tuple does not contribute to the aggregation function result. For that we calculate the minimum/maximum of $\amzero$ and $lba_f$/$uba_f$.
\begin{align*}
  lba   &\defas \ifte{e_{gc}}{lba_f}{\min(\amzero, lba_{f})} \renameto \lbatt{A}\\
  sga   &\defas sga_f \renameto \bgatt{A}\\
  uba   &\defas \ifte{e_{gc}}{uba_f}{\max(\amzero, lba_{f})} \renameto \ubatt{A}\\
  e_{gc} &\defas \theta_c \wedge \rlb > 0
\end{align*}

For $\aggsum$, we need to treat positive and negative numbers differently by multiplying them either with $\lbatt{row}$ or $\ubagg{row}$ to return the smallest/greatest possible aggregation function result.
\begin{align*}
 lba_{\aggsum}                                 & \defas \ifte{\lbatt{A} < 0}{\lbatt{A}\cdot\rub}{\lbatt{A} \cdot \rlb}                                          \\
 sga_{\aggsum}                                 & \defas \ifte{\bigwedge_{A_i \in G} \bgatt{A_i} = \bgatt{B_i}}{\bgatt{A}\cdot\rbg}{0}                                                       \\
 uba_{\aggsum}                                 & \defas \ifte{\ubatt{A}<0}{\ubatt{A}\cdot\rlb}{\ubatt{A}\cdot\rub}
\end{align*}

The neutral element of $\mmin$ is $\infty$ which is larger than any other value. Recall that $m \smbN{\mmin} k$ is the identify on $m$ except when $k=0$ where it returns $\infty$. Thus, the lowest possible value can be achieved if $k \neq 0$. Since $\rlb \leq \rub$, $\ifte{\rub > 0}{\lbatt{A}}{\infty}$ is a valid lower bound. Analog, $\ifte{\lub > 0}{\ubatt{A}}{\infty}$ is an upper bound. \ifnottechreport{The definition for $\aggmax$ is analog, please see \cite{techreport}.}
\begin{align*}
 lba_{\aggmin}                                 & \defas \ifte{\rub > 0}{\lbatt{A}}{\infty}                                           \\
 sga_{\aggmin}                                 & \defas \ifte{\bigwedge_{A_i \in G} \bgatt{A_i} = \bgatt{B_i}}{A\cdot\rbg}{\infty}                                                        \\
 uba_{\aggmin}                                 & \defas \ifte{\rlb > 0}{\ubatt{A}}{\infty}
\end{align*}

\iftechreport{
The neutral element of $\mmax$ is $-\infty$ which is smaller than any other value. $m \smbN{\mmin} k$ is the identify on $m$ except when $k=0$ where it returns $\infty$. Thus, the lowest possible value can be achieved if $k = 0$. Thus, $\ifte{\lub = 0}{-\infty}{\lbatt{A}}$ is a valid lower bound. Analog, $\ifte{\lub > 0}{\ubatt{A}}{\infty}$ is an upper bound.
\begin{align*}
 lba_{\aggmax}                                 & \defas \ifte{\rlb > 0 }{\lbatt{A}}{-\infty}                                           \\
 sga_{\aggmax}                                 & \defas \ifte{\bigwedge_{A_i \in G} \bgatt{A_i} = \bgatt{B_i}}{\bgatt{A}\cdot\rbg}{-\infty}                                                        \\
 uba_{\aggmax}                                 & \defas \ifte{\rub > 0}{\ubatt{A}}{-\infty}                                             \\
\end{align*}
}

\subsection{Correctness}
\label{sec:correctness}

To demonstrate that our encoding and rewrites correctly implement \abbrUAADB query semantics, we have to show that (i) the encoding is invertible, i.e., that there exists a mapping $\Dec$ such that $\Dec(\Enc(\rangeDB)) = \rangeDB$, and (ii) that the rewrite correctly simulates \abbrUAADB query semantics, i.e.,  $\rewrUAA{\query}(\Enc(\rangeDB)) = \Enc(\query(\rangeDB))$.

\begin{Theorem}[Rewrite Correctness]\label{theo:rewrite-correctness}
  Let $\rangeDB$ be a $\uaaN$-database, $\query$ be a $\raAgg$ query, then
  \begin{align*}
    \Dec(\Enc(\rangeDB)) &= \rangeDB \tag{$\Enc$ \textbf{is invertible}}\\
  \qmerge(\Enc(\rangeDB)) &= \Enc(\query(\rangeDB)) \tag{$\rewrUAA{\cdot}$ \textbf{is correct}}
  \end{align*}
\end{Theorem}
\begin{proof}

  \proofpara{$\Enc$ is invertible}
  Observe that by construction there exists a 1-to-1 mapping between the tuples in $\rangeRel$ and $\Enc(\rangeRel)$. A tuple $\rangeTup$ and its annotation $\rangeRel(\rangeTup)$ can be trivially reconstructed from the corresponding tuple $\tup$ in $\Enc(\rangeRel)$ by setting $\rangeTup.A = \uv{\tup.\lbatt{A}}{\tup.\bgatt{A}}{\tup.\ubatt{A}}$ for each attribute $A$ of $\rangeRel$ and then setting $\rangeRel(\rangeTup) = \ut{\tup.\rlb}{\tup.\rbg}{\tup.\rub}$.

  \proofpara{$\rewrUAA{\cdot}$ is correct}
  We prove the claim by induction over the structure of a query. Since we have proven that $\Dec(\Enc(\rangeDB)) = \rangeDB$, we  prove the correctness of $\rewrUAA{\cdot}$ by showing that $\Dec(\qmerge(\Enc(\rangeDB))) = \query(\rangeDB)$. We will we make use of this fact in the following.

  \proofpara{\textbf{Base case:} Table access $Q \defas R$}:
  $\rewrUAA{R}$ is the identify on $\Enc(\rangeRel)$. Thus, $\rewrUAA{R}$ does not contain value-equivalent tuples and $\qmerge$ is the identity on $\rewrUAA{R}$ and the claim holds.

  \proofpara{\textbf{Induction}}
    Assume that the claim holds for queries $\query_1$ and $\query_2$ modulo merging of value-equivalent tuples. We have to show that the claim holds modulo merging of value-equivalent tuples for each algebra
    operator applied to $\query_1$ (or $\query_1$ and $\query_2$ for binary operators). From this follows then that the claim holds for $\qmerge$ which merges such tuples. Consider an input database $\rangeDB$. In the following let $\rangeRel_1 = \query_1(\rangeDB)$ and $\rangeRel_2 = \query_2(\rangeDB)$. Similarly, let $\rel_1 = \rewrUAA{\query_1}(\Enc(\rangeDB))$
and $\rel_2 = \rewrUAA{\query_2}(\Enc(\rangeDB))$.

    \proofpara{Projection $\query \defas \projection_{U}(\query_1)$}
    Recall that $\projection_{U}(\query_1)$ is rewritten into

    $$\projection_{U, \lbatt{U}, \ubatt{U}, \rlb, \rbg, \rub}(\rewrUAA{\query_1}$$.

    Let $U' = (U, \lbatt{U}, \ubatt{U})$ and $U_{all} = (U, \lbatt{U}, \ubatt{U}, \rlb, \rbg, \rub)$.
The annotation of a tuple $\rangeTup$ in the result of $\query$ is the sum of annotations of all input tuples $\rangeOf{u}$ projected onto $\rangeTup$:
    \[
      \projection_{U}(\query_1)(\rangeTup) = \sum_{\rangeOf{u}: \rangeOf{u}.U = \rangeTup} \query_1(\rangeOf{u})
    \]
    Let $\{\rangeOf{u}_1, \ldots, \rangeOf{u}_m\}$ be the sets of tuples for
    which $\rangeOf{u_i}.U = \rangeTup$. Since,
    $\rangeOf{u_i}.U = \rangeOf{u_j}.A$, for any $i,j \in \{1, \ldots, m\}$ it
    follows that
    $\tenc{\rangeOf{u_i},\rangeRel_1(\rangeOf{u_i})}.U' = \tenc{\rangeOf{u_j},
      \rangeRel_1(\rangeOf{u_j})}.U'$ and in turn that
    $\tdec{\tenc{\rangeOf{u_i}}.U'} = \rangeTup$.  Note that,
    $\tenc{\rangeOf{u_i},\rangeRel_1(\rangeOf{u_i})}.U_{all} \neq
    \tenc{\rangeOf{u_j},\rangeRel_1(\rangeOf{u_j})}.U_{all}$ if
    $\rangeRel_1(\rangeOf{u_i}) \neq \rangeRel_1(\rangeOf{u_j})$. Based on the induction hypothesis, each tuple $\rangeOf{u_i}$ is encoded in $\rel_1$ as one or more value-equivalent tuples.
    Let
    $\tup_1, \ldots, \tup_l$ be the distinct tuples in the set
 $\{ \tup \mid i \in \exists i \in \{ 1, \ldots, m \}:  \rangeRel_1(\rangeOf{u_i}).U_{all}\}$.
  We use $\rangeOf{u_{i_1}}, \ldots, \rangeOf{u_{m_i}}$ to
  denote the tuples for which  $\rangeRel_1(\rangeOf{u_{i_j}}).U_{all} = \tup_i$.
  Furthermore, for any such tuple $\tup_i$, let $u_{i_{j_1}}$, \ldots, $u_{i_{o_j}}$ be the value-equivalent tuples that are projected onto $\tup_i$.
  Then,

  \begin{align*}
      & \Dec(\rewrUAA{\query}(\rangeDB))(\rangeTup)                                                                                                                              \\
    = &\sum_{i \in \{1, \ldots, l\}} \tdecr{\tup_i} \cdot \rewrUAA{\query}(\rangeDB)(\tup_i) \\
  \end{align*}
The definition of $\Enc$ ensures that every tuple $\tenc{\rangeOf{u_{i_j}}}$ is annotated  with $1$. Given the definition of projection, $\tup_i$ is annotated with the sum of annotations of all tuples $\tenc{\rangeOf{u_i}}$. That is, $\tup_i$ is annotated with $\sum_{j = 1}^{m_j} 1 = m_j$.
  \begin{align*}
    = &\sum_{i \in \{1, \ldots, l\}} \tdecr{\tup_i} \cdot u_{m_i}                       \\
    = & \sum_{i \in \{1, \ldots, l\}} \sum_{j \in \{1, \ldots, m_i\}} \tdecr{\tup_i}
  \end{align*}

$\rewrUAA{\projection_U(\query_1)}$ does retain the row annotation attributes of input tuples unmodified. Thus, we have:

  \begin{align*}
    = &\sum_{i \in \{1, \ldots, l\}} \sum_{j \in \{1, \ldots, m_i\}} \tdecr{\tenc{\rangeOf{u_i}}}
  \end{align*}

Since we assume that claim holds for $\query_1$, we know that

  \begin{align*}
    = & \sum_{i = \{ 1, \ldots, m \}} \tdecr{\tenc{\rangeOf{u_i}}}                                                                                                               \\
    = & \sum_{i = \{ 1, \ldots, m \}} \rangeRel_1(\tenc{\rangeOf{u_i}}) = \query(\rangeDB)(\rangeTup)
  \end{align*}

  We have proven that for any tuple $\rangeTup$, we have

  \[
    \Dec(\rewrUAA{\query}(\rangeDB))(\rangeTup) = \query(\rangeDB)(\rangeTup)
    \]

  which implies that $\Dec(\rewrUAA{\query}(\rangeDB)) = \query(\rangeDB)$.


   \proofpara{Selection $\query \defas \selection_{\theta}(\query_1)$}
   Consider a tuple $\rangeTup$ such that
   $\ubMarker{\query(\rangeDB)(\rangeTup)} > 0$, i.e., tuples that may exists in the
   result of the selection. Selection over $\uaaN$-relations calculates the
   annotation of a result tuple $\rangeTup$ by multiplying its annotation in the
   input with the bounds of the result of the selection condition evaluated over
   $\rangeTup$ mapped from a range annotated Boolean value to a $\uaaN$-value
   where $\btrue$ is mapped to $1$ and $\bfalse$ to $0$. $\rewrUAA{\query}$ uses
   a triple of deterministic expressions to compute the elements of
   $\theta(\rangeTup)$ individually over $\tup = \tenc{\rangeTup, \rangeRel_1(\rangeTup)}$.
   Thus, $\tdecr{\tup} = \query(\rangeDB)(\rangeTup)$ and the claim holds.

   \proofpara{Union $\query \defas \query_1 \union \query_2$}
   Consider a tuple $\rangeTup$ such that either $\rangeRel_1(\rangeTup) > 0$ or $\rangeRel_2(\rangeTup) > 0$. Let $T$ be the set of tuples that are value-equivalent to $\rangeTup$. Based on the induction hypothesis we know that $\rangeRel_1(\rangeTup) = \sum_{\tup \in T} \tdecr{\tup} \cdot  \ut{\rel_1(\tup)}{\rel_1(\tup)}{\rel_1(\tup)}$ and $\rangeRel_2(\rangeTup) = \sum_{\tup \in T} \tdecr{\tup} \cdot  \ut{\rel_2(\tup)}{\rel_2(\tup)}{\rel_2(\tup)}$.
Recall that $\rewrUAA{\query_1 \union \query_2} = \rewrUAA{\query_1} \union \rewrUAA{\query_2}$.
Based on the definition of union (semiring addition), we know that for any $\tup \in T$, it holds that $(\rel_1 \union \rel_2)(\tup) = \rel_1(\tup) + \rel_2(\tup)$. Thus,



   \begin{align*}
     \Dec(\rel)(\rangeTup) & = \sum_{\tup: \Dec(\tup) = \rangeTup} \left( \tdecr{\tup} \cdot  \ut{\rel_1(\tup)}{\rel_1(\tup)}{\rel_1(\tup)} \right)  \\
                           & \hspace{1.3cm} + \left(\tdecr{\tup} \cdot  \ut{\rel_2(\tup)}{\rel_2(\tup)}{\rel_2(\tup)})\right)                                       \\[3mm]
                           & =  \sum_{\tup: \Dec(\tup) = \rangeTup}  \tdecr{\tup} \cdot  \ut{\rel_1(\tup)}{\rel_1(\tup)}{\rel_1(\tup)}  \\
                           & \mathtab +  \sum_{\tup: \Dec(\tup) = \rangeTup}  \tdecr{\tup} \cdot  \ut{\rel_2(\tup)}{\rel_2(\tup)}{\rel_2(\tup)}  \\[3mm]
                           &= \rangeRel_1(\rangeTup) + \rangeRel_2(\rangeTup)                                                                        \\
                          &= \query(\rangeDB)(\rangeTup)
   \end{align*}



   \proofpara{Cross Product $\query \defas \query_1 \times \query_2$}
Consider a tuple $\rangeTup$ that is the result of joining tuples $\rangeTup_1$ from the result of $\query_1$ and $\rangeTup_2$ from the result of $\query_2$. The annotation of a result tuple $\rangeTup$ of $\query$ in $\uaaN$ relations are computed by multiplying the annotations of the input tuples $\rangeTup_1$ and $\rangeTup_2$ which are joined to form $\rangeTup$. Based on the induction hypothesis we know that $\rangeRel_1(\rangeTup) = \sum_{\tup \in T} \tdec{\tup} \cdot  \ut{\rel_1(\tup)}{\rel_1(\tup)}{\rel_1(\tup)}$ and $\rangeRel_2(\rangeTup) = \sum_{\tup \in T} \tdec{\tup} \cdot  \ut{\rel_2(\tup)}{\rel_2(\tup)}{\rel_2(\tup)}$. Recall that $\rewrUAA{\query_1 \crossprod \query_2} = \rewrUAA{\query_1} \crossprod \rewrUAA{\query_2}$. Using the fact that multiplication distributes over addition and that semiring operations are commutative and associativity we get:

   \begin{align*}
     \Dec(\rel)(\rangeTup) & = \hspace{-1cm}\sum_{\tup_1, \tup_2: \tdec{\tup_1} = \rangeTup_1 \wedge \tdec{\tup_2} = \rangeTup_2} ( (\tdecr{\tup_1} \cdot  \ut{\rel_1(\tup_1)}{\rel_1(\tup_1)}{\rel_1(\tup_1)}) \\
                           & \hspace{2cm}\cdot (\tdecr{\tup_2} \cdot  \ut{\rel_2(\tup_2)}{\rel_2(\tup_2)}{\rel_2(\tup_2)}))                                                                                                                                                                                                       \\[3mm]
   \end{align*}
In the following let $n_i$ denote $\tdecr{\tup_i} \cdot  \ut{\rel_i(\tup_i)}{\rel_i(\tup_i)}{\rel_i(\tup_i)}$ for $i \in \{1,2\}$.
   \begin{align*}
                           & = \sum_{\tup_1, \tup_2: \tdec{\tup_1} = \rangeTup_1 \wedge \tdec{\tup_2} = \rangeTup_2} \left(n_1 \cdot n_2\right)                                                                      \\
                           & = \left(\sum_{\tup_1: \tdec{\tup_1} = \rangeTup_1} n_1\right) \cdot \left(\sum_{\tup_2: \tdec{\tup_2} = \rangeTup_2} n_2\right)                                                         \\
                           & = \rangeRel_1(\rangeTup_1) + \rangeRel_2(\rangeTup_2)                                                                                                                                   \\
                           & = \query(\rangeDB)(\rangeTup)
   \end{align*}

   \proofpara{\abbrBG Combiner $\query \defas \combine(\query_1)$}
Recall that $\combine$ merges the attribute bounds of tuples that agree on their \abbrBG values and sums their annotations. $\rewrUAA{\combine(\query_1)}$ groups input tuples on their \abbrBG values. Each group contains all tuples that agree with each other on \abbrBG attribute values. Then the minimum (maximum) over attributes storing attribute bounds is computed to calculate the value of attributes storing lower (upper) bounds for attributes. The values of attributes storing tuple annotations are computed by summing up the values of these attributes for each group.

   \proofpara{Except $\query \defas \query_1 - \query_2$}
   Recall the definition of set difference over $\uaaK{\semK}$-relations:
   \begin{align*}
          \lbMarker{(\rangeRel_1 - \rangeRel_2)(\rangeTup)}  & \defas \lbMarker{\combine(\rangeRel_1)(\rangeTup)} \monK \sum_{\rangeTup \matches \rangeTup'} \ubMarker{\rangeRel_2(\rangeTup')} \\
     \bgMarker{(\rangeRel_1 - \rangeRel_2)(\rangeTup)} & \defas \bgMarker{\combine(\rangeRel_1)(\rangeTup)} \monK \sum_{\bgOf{\rangeTup} = \bgOf{\rangeTup'}} \bgMarker{\rangeRel_2(\rangeTup')} \\
     \ubMarker{(\rangeRel_1 - \rangeRel_2)(\rangeTup)}      & \defas \ubMarker{\combine(\rangeRel_1)(\rangeTup)} \monK \sum_{\rangeTup \equiv \rangeTup'} \lbMarker{\rangeRel_2(\rangeTup')}
   \end{align*}
   Note that $\rangeTup \matches \rangeTup'$ if the bounds of all attributes of  $\rangeTup$ and $\rangeTup'$ overlap, i.e., the tuples may represent the same tuple in some world. Furthermore, $\rangeTup \equiv \rangeTup'$ if these tuples are equal and are certain ($\lbMarker{\rangeTup.A} = \ubMarker{\rangeTup.A}$ for all attributes $A$. Obviously, $\rangeTup \equiv \rangeTup' \Rightarrow \rangeTup \matches \rangeTup'$. Since the rewrite for $\combine$ was proven to be correct above, we only need to prove that tuple annotations are calculated according to the definition repeated above. The first step of $\rewrUAA{\query}$ joins the results of $\rewrUAA{\combine(\query_1)}$ with $\rewrUAA{\query_2}$ based on overlap of their attribute bounds. Note that in the result of $\rewrUAA{\combine(\query_1)}$ each tuple $\rangeTup$ from $\combine(\query_1)$ is encoded as a single tuple $\tup$. However, in $\rel_2$ (the result of $\rewrUAA{\query_2}$), each tuple $\rangeTup_2$ from $\rangeRel_2$ may be encoded as multiple value equivalent tuples. The annotations of these tuples multiplied with the their row annotation attributes sum up to the annotation $\rangeRel_2(\rangeTup_2)$:
   \[
\rangeRel_2(\rangeTup) = \sum_{\tup_2: \tdec{\tup_2} = \rangeTup_2}    \tdecr{\tup_2} \cdot  \ut{\rel_2(\tup_2)}{\rel_2(\tup_2)}{\rel_2(\tup_2)}
\]

For each such tuple $\rangeTup_i$ let $\tup_{i,1}, \ldots, \tup_{i,n_i}$ be the set of these value-equivalent tuples. Thus, the join will pair $\tup$ with all such tuples for each tuple $\rangeTup_2$ for which $\rangeTup \matches \rangeTup'$.
After the join, the row annotation attributes of each RHS tuple $\tup'$ paired with a tuple $\tup$ is modified as follows: the lower bound is replaced with the upper bound (expression $e_{cv}$); the \abbrBG row annotation is retained unless $\tup$ and $\tup'$ do not agree on their \abbrBG attribute values (in this case the lower  bound is set to $0$); ant the upper bound is set to the lower bound if $\tup$ and $\tup'$ are equal on all attributes and are certain (otherwise the upper bound is set to $0$).
Afterwards, tuples are grouped based on their LHS values and the RHS row annotations are summed up. As shown above in the proof for cross product, the fact the one tuple $\rangeTup$ is encoded as multiple value-equivalent tuples is unproblematic for joins.
The net effect is that for each LHS tuple $\tup$ corresponding to a tuple $\rangeTup$ there exists one tuple $\tup_{ext}$ in the result of the aggregation that (i) agrees with $\tup$ on all attribute values and for which (ii) attribute $\lbatt{rrow}$ stores $\sum_{\rangeTup': \rangeTup \matches \rangeTup'} \ubMarker{\rangeRel_2(\rangeTup')}$, attribute $\bgatt{rrow}$ stores $\sum_{\rangeTup': \bgMarker{\rangeTup} = \bgMarker{\rangeTup'}} \bgMarker{\rangeRel_2(\rangeTup')}$, and finally attribute $\ubatt{rrow}$ stores $\sum_{\rangeTup \equiv \rangeTup'} \lbMarker{\rangeRel_2(\rangeTup')}$. In a last, step the projection is used to calculate the annotation of $\tup$ in the result corresponding to           and selection is applied to remove tuples which certainly do not exist (their upper bound row annotation is $0$).

   \proofpara{Aggregation $\query \defas \aggregation{G}{f(A)}(\query_1)$}
   To demonstrate the the rewrite rule for aggregation is correct, we need to show that (i) tuples are grouped according to the default grouping strategy (\Cref{def:default-grouping-strategy}; (ii) that group-by bounds for each output tuple are calculated as defined in \Cref{def:range-bounded-groups}; (iii) that aggregation function result bounds are computed following \Cref{def:agg-function-bounds}); and (iv) that the multiplicity bounds for each output are correct (~\Cref{def:aggr-op-semantics-wo-gb} and (~\Cref{def:aggr-op-semantics-gb}).

   Consider $\query_{gbounds}$ that is part of the rewrite for aggregation. This query groups the input on their \abbrBG group-by values and then calculates the group-by bounds for each group. Each group created in this way corresponds to one output group produced by the default grouping strategy (recall that this strategy defines one output group \abbrBG group-by value that exists in the input. Thus, (i) holds.

   Note that according to \Cref{def:default-grouping-strategy}, the group-by attribute bounds for an output are determined as the minimum (maximum) value of the lower (upper) bound on a group-by attribute across all tuples that have the same \abbrBG group-by values as the output. This set of tuples corresponds exactly to one group in $\query_{gbounds}$. Since, $\query_{bounds}$ computes group-by attribute bounds as the minimum (maximum) of the bounds of the input tuples belonging to a group, (ii) also holds.

   The next step in the rewritten query is $\query_{join}$ which pairs every tuple from $\query_{gbounds}$ with all input tuples that overlap in their group-by bounds with the output tuple's group-by bounds. Note that this is precisely the set $\tgrouping(\asgrp)$ for output group $\asgrp$ as iterated over in \Cref{def:agg-function-bounds} to calculate aggregation function result bounds. Recall that to calculate the lower bound of for an aggregation function result $\lbMarker{\rangeTup_i.f(A)}$ by summing up (in the aggregation monoid $\monoid$) for each tuple in $\tgrouping(\asgrp)$ either the minimum of $\amzero$ and $\lbMarker{(\rangeRel(\rangeTup) \smbpair_{\uaaN,\monoid} \rangeTup.A)}$ if $\uncertg{\gbAttrs}{\rangeRel}{\rangeTup}$ which is the case when either some group-by attributes of $\rangeTup$ are uncertain or if the tuple may not exit ($\lbMarker{\rangeRel(\rangeTup)} = 0$). Otherwise, $\lbMarker{(\rangeRel(\rangeTup) \smbpair_{\uaaN,\monoid} \rangeTup.A)}$ is used instead. Expressions $lba$, $sga$, and $uba$ used in the rewritten query implement this logic. Furthermore, it is trivial to see that expressions $lba_f$, $sga_f$, and $uba_f$ implement the semantics of $\smbpair_{\uaaN,\monoid}$. Thus, the computation of the lower bound in the rewritten query correctly reflects the computation in $\uaaN$. Note that each input tuple $\rangeTup$ of the aggregation may be encoded as multiple value-equivalent tuples in the encoding. However, this is not a problem, according to \ref{lem:K-addition-distributes-over-smb} if $\vec{k} = \vec{k_1} \uaanAdd \vec{k_2}$, then $\vec{k} \amysmbNAU \vec{m} = \vec{k_1} \amysmbNAU \vec{m} \madd{\rangeMof{\monoid}} \vec{k_2} \amysmbNAU \vec{m}$.
   Thus, (iii) holds for lower bounds. The prove for upper bounds is symmetric.
   Finally, for the multiplicity bounds of tuples all input tuples that agree with an output on their \abbrBG group-by values are considered by the default grouping strategy. For the upper bound the bound bound on the multiplicity of these inputs is summed up. For the lower bound, the result of summing up the lower bound  multiplicities for all tuples with certain group-by values (only these tuples are guaranteed to below to a group) is passed to $\duprem_{\semN}$ which returns $1$ if the sum is non-zero and $0$ otherwise.  The \abbrBG and upper bound multiplicities computed in the same way except that all tuples are considered (and $\duprem_{\semN}$ is not used for the upper bound). In the rewritten query this is achieved by conditionally replacing the multiplicity bounds of tuples that are not part of the sum to $0$ using a condition $\theta_c$ for the lower bound and $\theta_\abbrBG$ for the upper bound and \abbrBG multiplicities. $\duprem_{\semN}$ is implemented by replacing non-zero multiplicities with $1$ and using aggregation function $\aggmax$ instead of $\aggsum$ to combine the multiplicities of input tuples. Since the same multiplicity bounds are calculated by the rewritten query as for the query under $\uaaN$ semantics, (iv) holds. From (i) to (iv) follows that the aggregation rewrite is correct.
\end{proof}

\subsection{Optimizations for Joins}

One potential performance bottleneck of query evaluation over \abbrUAADBs is that joins may degenerate into cross products if the bounds of join attribute values are loose.
As shown in the example below,
in the worst case, each tuple from the LHS input may join with every tuple from the RHS leading to a join result whose size is quadratic in the input size. Even if most join attribute values are certain, the DBMS is likely to chose a nested loop join since we join on inequalities to test for overlap of join attribute bounds leading to $\bigO(n^2)$ runtime for the join.

\begin{figure*}[t]
	\centering
    \begin{minipage}{0.48\linewidth}
	\begin{subtable}{0.4\linewidth}
	\centering
	\begin{tabular}{ cc}
		\textbf{A}  & \underline{\semqN} \\
		\cline{1-1}
		$\uv{1}{1}{2}$ & \ut{2}{2}{3}\\
		$\uv{1}{2}{2}$ & \ut{1}{1}{2}\\
	\end{tabular}
	\caption{\abbrUAADB relation $\rangeOf{R}$}
	\label{table:join-input-r}
	\end{subtable}
	\begin{subtable}{0.4\linewidth}
	\centering
	\begin{tabular}{ cc}
		\textbf{C}  & \underline{\semqN} \\
		\cline{1-1}
		$\uv{1}{3}{3}$ & \ut{1}{1}{1}\\
		$\uv{1}{2}{2}$ & \ut{1}{2}{2}\\
	\end{tabular}
	\caption{\abbrUAADB relation $\rangeOf{S}$}
	\label{table:join-input-s}
	\end{subtable}

	\begin{subtable}{\linewidth}
	\centering
	\begin{tabular}{ c|cc}
		\textbf{A}  & \textbf{C}  & \underline{\semqN} \\
		\cline{1-2}
		$\uv{1}{2}{2}$ & $\uv{1}{2}{2}$ & \ut{1}{2}{4}\\
	\end{tabular}
	\caption{\abbrBGW result of $R \join_{A=C} S$}
	\label{table:join-example-bg-output}
	\end{subtable}

	\begin{subtable}{\linewidth}
	\centering
	\begin{tabular}{ c|cc}
		\textbf{A}  & \textbf{C}  & \underline{\semqN} \\
		\cline{1-2}
		$\uv{1}{1}{2}$ & $\uv{1}{3}{3}$ & \ut{0}{0}{3}\\
		$\uv{1}{1}{2}$ & $\uv{1}{2}{2}$ & \ut{0}{0}{6}\\
		$\uv{1}{2}{2}$ & $\uv{1}{3}{3}$ & \ut{0}{0}{2}\\
		$\uv{1}{2}{2}$ & $\uv{1}{2}{2}$ & \ut{1}{2}{4}\\
	\end{tabular}
	\caption{\abbrUAADB result of $R \join_{A=C} S$}
	\label{table:join-example-UAA-output}
	\end{subtable}

	\caption{Join rewriting without optimization}
	\label{table:join_no_op}
  \end{minipage}
  \begin{minipage}{0.48\linewidth}
	\centering

	\begin{subtable}{0.4\linewidth}
	\centering
	\begin{tabular}{ cc}
		\textbf{A}  & \underline{\semqN} \\
		\cline{1-1}
		$\uv{1}{1}{1}$ & \ut{0}{2}{2}\\
		$\uv{2}{2}{2}$ & \ut{0}{1}{1}\\
	\end{tabular}
	\caption{$\splitbg(R)$}
	\label{table:join-input-r-bg}
	\end{subtable}
	\begin{subtable}{0.4\linewidth}
	\centering
	\begin{tabular}{ cc}
		\textbf{C}  & \underline{\semqN} \\
		\cline{1-1}
		$\uv{3}{3}{3}$ & \ut{0}{1}{1}\\
		$\uv{2}{2}{2}$ & \ut{0}{2}{2}\\
	\end{tabular}
	\caption{$\splitbg(S)$}
	\label{table:join-input-s-bg}
	\end{subtable}

	\begin{subtable}{0.4\linewidth}
	\centering
	\begin{tabular}{ cc}
		\textbf{A}  & \underline{\semqN} \\
		\cline{1-1}
		$\uv{1}{1}{2}$ & \ut{0}{0}{3}\\
		$\uv{1}{2}{2}$ & \ut{0}{0}{2}\\
	\end{tabular}
	\caption{$\splitpos(R)$}
	\label{table:join-input-r-pos}
	\end{subtable}
	\begin{subtable}{0.4\linewidth}
	\centering
	\begin{tabular}{ cc}
		\textbf{C}  & \underline{\semqN} \\
		\cline{1-1}
		$\uv{1}{3}{3}$ & \ut{0}{0}{1}\\
		$\uv{1}{2}{2}$ & \ut{0}{0}{2}\\
	\end{tabular}
	\caption{$\splitpos(S)$}
	\label{table:join-input-s-pos}
	\end{subtable}

	\begin{subtable}{0.4\linewidth}
	\centering
	\begin{tabular}{ cc}
		\textbf{A}  & \underline{\semqN} \\
		\cline{1-1}
		$\uv{1}{1}{2}$ & \ut{0}{0}{5}\\
	\end{tabular}
	\caption{$\compressPos_{A,1}(\splitpos(R))$}
	\label{table:join-input-r-pos-cmp}
	\end{subtable}
	\begin{subtable}{0.4\linewidth}
	\centering
	\begin{tabular}{ cc}
		\textbf{C}  & \underline{\semqN} \\
		\cline{1-1}
		$\uv{1}{2}{3}$ & \ut{0}{0}{3}\\
	\end{tabular}
	\caption{$\compressPos_{C,1}(\splitpos(S))$}
	\label{table:join-input-s-pos-cmp}
	\end{subtable}

	\begin{subtable}{\linewidth}
	\centering
	\begin{tabular}{ c|cc}
		\textbf{A}  & \textbf{C}  & \underline{\semqN} \\
		\cline{1-2}
		$\uv{2}{2}{2}$ & $\uv{2}{2}{2}$ & \ut{0}{2}{2}\\
		$\uv{1}{1}{2}$ & $\uv{1}{2}{3}$ & \ut{0}{0}{15}\\
	\end{tabular}
	\caption{$\rewrOpt{R \join_{A=C} S}$}
	\label{table:join-example-compressed}
	\end{subtable}

	\caption{Join rewriting with optimization}
	\label{table:join_with_op}
  \end{minipage}
\end{figure*}

\begin{Example}
  Figure~\ref{table:join_no_op} shows an example of joining two \abbrUAADB relations $R \join_{A=C} S$. The \abbrBGW result  consists of a single tuple. However, the \abbrUAADB result is a cross product of the two input tables since join attribute bounds of all tuples from $\rangeRel$ overlap with the join attribute bounds of all tuples from $\rangeOf{S}$.
\end{Example}

In order to reduce the running time of joins, we introduce an optimized version of the rewrite rule for join. This optimization trades accuracy for performance by compressing the overestimation of possible answers encoded by the input relations of the join.
We introduce an operator called \textit{split} that splits the input relation $\rangeRel$ into two parts: $\splitbg(R)$ encodes the \abbrBGW removing all attribute-level uncertainty and $\splitpos(\rangeRel)$ encodes the over estimation of possible worlds encoded by $\rangeRel$. We will show that for any incomplete $\semN$-relation $\prel$ that is bounded by a $\uaaN$-relation $\rangeRel$, $prel$ is also bound by $\splitbg(\rangeRel) \cup \splitpos(\rangeRel)$. We define these two operations below. For each tuple, $\rangeTup$, $\splitbg(\rangeRel)$ contains a corresponding tuple $\rangeTup'$ from which all attribute uncertainty has been removed by setting $\lbMarker{\rangeTup'.A} = \bgMarker{\rangeTup'.A} = \ubMarker{\rangeTup'.A} = \bgMarker{\rangeTup.A}$. The annotation of such a tuple $\rangeTup'$ is determined as follows.: $\bgMarker{\splitbg(\rangeRel)(\rangeTup')} = \bgMarker{\rangeRel(\rangeTup)}$, $\ubMarker{\splitbg(R)(\rangeTup)} = \bgMarker{\rangeRel}(\rangeTup)$, i.e., the overestimation of possible annotations is removed and $\lbMarker{\splitbg(R)(\rangeTup)} = \lbMarker{\rangeRel(\rangeTup)}$ if the tuple's attribute values are all uncertain and to $0$ otherwise. $\splitpos(\rangeRel)$ retains the tuples from $\rangeRel$, keep the $\ubMarker{\rangeRel(\rangeTup)}$ as the upper bound of a tuple's annotation, and sets $\lbMarker{\splitpos(\rangeRel)(\rangeTup)} = \bgMarker{\splitpos(\rangeRel)(\rangeTup)} = 0$. Consider a \abbrUAADB relation $\rangeRel$ with schema $\bar{A} = (A_1, \ldots, A_n)$. Let $\certainName(\rangeTup)$ denote the tuple derived from $\rangeTup$ by making all attribute values certain, i.e., replacing each attribute value $\uv{c_1}{c_2}{c_3}$ with $\uv{c_2}{c_2}{c_2}$.

\begin{align*}
  \lbMarker{\splitbg(\rangeRel)(\rangeTup)} &\defas \sum_{\rangeTup': \certainName(\rangeTup') = \rangeTup}
                                              \begin{cases}
                                                \lbMarker{\rangeRel(\rangeTup')} &\mathtext{if} \bigwedge_{i \in \{ 1, \ldots, n\}} \lbMarker{\rangeTup'.A_i} = \ubMarker{\rangeTup'.A_i}\\
                                                0 &\mathtext{otherwise}
                                              \end{cases}\\
  \bgMarker{\splitbg(\rangeRel')(\rangeTup)} &=   \ubMarker{\splitbg(\rangeRel)(\rangeTup)} \defas \sum_{\rangeTup': \certainName(\rangeTup') = \rangeTup} \bgMarker{\rangeRel(\rangeTup)}
\end{align*}
\begin{align*}
  \lbMarker{\splitpos(\rangeRel)(\rangeTup)} &\defas 0 &
  \bgMarker{\splitpos(\rangeRel)(\rangeTup)} &\defas 0 &
  \ubMarker{\splitpos(\rangeRel)(\rangeTup)} &\defas \ubMarker{\rangeRel(\rangeTup)}\\
\end{align*}

We rewrite the two split operators as shown below.

\begin{align*}
  \rewrUAA{\splitbg(R)} &=\projection_{e_{A}, e_{c} \renameto \rlb, \rbg, \rbg \renameto \rub)}(\selection_{\rbg>0}(\rangeRel))\\
  e_{A} &= \bgatt{A_1} \renameto \lbatt{A_1},
          \bgatt{A_1} \renameto \bgatt{A_1},
          \bgatt{A_1} \renameto \ubatt{A_1},
          \ldots, \\ 
    e_{c} &= \ifte{\bigwedge_{i \in \{1,\ldots,n\}} \lbMarker{A_i} = \ubMarker{A_i}}{\rlb}{0}\\
 \rewrUAA{\splitpos(R)} &= \projection_{\lbMarker{A},A,\ubMarker{A},0\renameto \rlb, 0\renameto \rbg, \rub}(\rangeRel)
\end{align*}


Based on the following lemma, we can split any \abbrUAADB relation $\rangeRel$ without loosing its bounding properties while preserving  the \abbrBGW encoded by $\rangeRel$.

\begin{Lemma}[Split preserves bounds]\label{lem:split-preserves-bounds}
Let $\rangeRel$ be a $\uaaN$-relation that bounds an incomplete $\semN$-relation $\prel$, then $\splitbg(\rangeRel) \union \splitpos(\rangeRel)$ also bounds $\prel$. Furthermore, $\bgOf{\rangeRel} = \bgOf{\splitbg(\rangeRel) \union \splitpos(\rangeRel)}$.
\end{Lemma}
\begin{proof}
  We first show that $\bgOf{\rangeRel} = \bgOf{\splitbg(\rangeRel) \union \splitpos(\rangeRel)}$. Since, the \abbrBG annotations of $\splitpos(\rangeRel)$ are zero and $\bgOf{\splitbg(\rangeRel)} = \bgOf{\rangeRel}$ by definition, the claim holds.
Let $\rangeRel_{split} = \splitbg(\rangeRel) \union \splitpos(\rangeRel)$.
  To demonstrate that the split operator preserves bounds, we have to show that for any possible world $\rel \in \prel$ we can extend a tuple matching $\TM$ based on which $\rangeRel$ bounds $\rel$ to a tuple matching $\TM_{split}$ based on which $\rangeRel_{split}$ bounds $\rel$. Consider a tuple $\rangeTup$ such that $\rangeRel(\rangeTup) \neq 0$. We distinguish two cases. If $\bigwedge_{i \in \{ 1, \ldots, n\}} \lbMarker{\rangeTup.A_i} = \ubMarker{\rangeTup.A_i}$, i.e., the tuple's attribute values are certain then $\rangeRel_{split}(\rangeTup) = \splitbg(\rangeRel)(\rangeTup) + \splitpos(\rangeRel)(\rangeTup) = (\lbMarker{\rangeRel(\rangeTup)},\sum_{\rangeTup': \certainName(\rangeTup') = \rangeTup}\bgMarker{\rangeRel(\rangeTup')}, \sum_{\rangeTup': \certainName(\rangeTup') = \rangeTup} \bgMarker{\rangeRel(\rangeTup')} + \ubMarker{\rangeRel(\rangeTup)})$. Thus, for every tuple $\tup$ we can set $\TM_{split}(\rangeTup, \tup) = \TM(\rangeTup, \tup)$ and wrt. $\rangeTup$ the tuple matching $\TM_{split}$ is fulfilling the requirements for being a tuple matching based on which $\rangeRel_{spli}$ bounds $\rel$ since its annotation include the bounds of $\rangeRel(\rangeTup)$. Otherwise, at least one value of $\rangeTup$ is uncertain and $\rangeRel_{split}(\rangeTup) = \splitpos(\rangeRel)(\rangeRel) = (0,0,\ubMarker{\rangeRel(\rangeTup)})$. Thus, we can safely set $\TM_{split}(\rangeTup, \tup) = \TM(\rangeTup, \tup)$ and $\TM_{split}$ is a tuple matching based on which $\rangeRel_{split}$ bounds $\rel$.
\end{proof}

To improve the performance of joins, we split both input relations of the join. We then employ another new operator $\compressPos_{A,n}$ that compresses the output of $\splitpos$ into a relation with $n$ tuples by grouping the input tuples into $n$ buckets based on their $A$ values. For that we split the range of $A$ values appearing in the input into $n$ buckets containing roughly the same number of values each. All tuples from a bucket are aggregated into a single result tuple by merging their attribute bounds and summing up their annotations. Let $\rangeTup_b$ denote the tuple constructed for bucket $b$ in this fashion. Let $B = \{ b_1, \ldots, b_n \}$ be the set of buckets for $\compressPos_{A,n}(\rangeRel)$.
\begin{align*}
  \compressPos_{A,n}(\rangeRel)(\rangeTup) =
  \begin{cases}
    \ut{0}{0}{\sum_{\rangeTup' \in b} \ubMarker{\rangeRel(\rangeTup')}} & \mathtext{if} \exists b \in B: \rangeTup = \rangeTup_b \\
    0                                                                   & \mathtext{otherwise}                                   \\
  \end{cases}
\end{align*}
\BG{Discuss how to deal with intermediate results}

We rewrite the $\compressPos$ operator as shown below. Assume that the bucket $b_i$ covers the interval  $[l_i,u_i]$ from the domain of $A$.

\begin{align*}
  \rewrUAA{\compressPos_{A,n}(\rangeRel)} &= \projection_{\bar{\lbatt{A}},\bar{\bgatt{A}},\bar{\ubatt{A}},0 \renameto \rlb, 0 \renameto \rbg, \rub}(\query_{merge})\\
  \query_{merge} &\defas \aggregation{B}{e_{merge}, \aggsum(\rub) \renameto \rub}(\rangeRel)\\
  e_{merge} &= min(\lbatt{A_1}) \renameto \lbatt{A_1}, min(\bgatt{A_1}) \renameto \bgatt{A_1}, max(\ubatt{A_1}) \renameto \ubatt{A_1}, \ldots\\
  e_{bucketize} &= \ifte{\ubatt{A} \geq l_1 \wedge \lbatt{A} \leq u_1}{1}{(\mathbf{if} \ldots} 
\end{align*}

Note that $\compressPos$ does not preserve the \abbrBGW encoded by its input. This is not problematic, because we only apply $\compressPos$ to the output of $\splitpos$ for which the \abbrBG annotation of each tuple is zero anyways.

\begin{Lemma}[$\compressPos$ preserves bounds]\label{lem:compression-preserves-bounds}
Let $\rangeRel$ be a $\uaaN$-relation that bounds an incomplete $\semN$-relation $\prel$, $n \in \semN$, and $A \in \schemaOf(\rangeRel)$, then $\compressPos_{A,n}(\rangeRel)$ also bounds $\prel$. 
\end{Lemma}
\begin{proof}

Let $\rel$ be one possible world of $\prel$.  Let $\TM$ be a tuple matching based on which $\rangeRel$ bounds $\rel$. Furthermore, let $\rangeTup_r$ denote the result tuple produced for a bucket $b$. We construct a tuple matching $\TM_{\compressPos}$ which bounds $\prel$ by setting for all buckets $b$ and tuple $\tup$ from $\rel$:
  \begin{align*}
    \TM_{\compressPos}(\rangeTup_b, \tup) = \sum_{\rangeTup \in b}\TM(\rangeTup, \tup)
  \end{align*}
Since by definition $\compressPos_{A,n}(\rangeRel)(\rangeTup_b) = \sum_{\rangeTup \in b} \rangeRel(\rangeTup)$ and since $\tup \tmatch \rangeTup \Rightarrow \tup \tmatch \rangeTup_b$ for all $\rangeTup \in b$ (because the attribute bounds of $\rangeTup_b$ are defined to cover the attribute bounds of all tuples from bucket $b$), we have that $\TM_{\compressPos}$ is a tuple matching based on which $\compressPos_{A,n}(\rangeRel)$ bounds $\rel$.
\end{proof}

Our optimized rewrite for join first splits both inputs of the join.  We join
$\splitpos$ and $\splitbg$ separately and then union the result. For a join with
condition $\theta$, let $\bgMarker{\theta}$ denote the result of replacing
references to an attribute $A$ with $\bgatt{A}$. Note that for $\splitbg$ since
all attribute values are certain, we can apply $\bgMarker{\theta}$ instead of
having to apply range-anntotated expression evaluation. Thus, the join over
$\splitbg$ will result only in minimal overhead compared to a regular join.
Since $\splitpos(\rangeRel) \join \splitpos(\rangeOf{S})$ may potentially
produce a large number of results, we apply $\compressPos$ to the inputs to
bound the size of the join result. Thus, we control the worst case join result
size by setting the parameter $n$ of $\compressPos$. In principle we can
compress on any attribute of the input relations. However, it is typically
better to choose attributes that are referenced in $\theta$, e.g., for a
condition $A=B$ if we compress on $A$ respective $B$ using $n$ buckets with the
same bucket boundaries for $A$ and $B$, then the join result will contain at
most $n$ results since each tuple from $\compressPos_{A,n}(\rangeRel)$ will join
at most with one tuple from $\compressPos_{B,n}(\rangeOf{S})$. Let $A$ ($B$)
denote an attribute from $\query_1$ ($\query_2$) that appears in $\theta$,
preferably in an equality comparison. The optimized rewrite for join $\rewrOpt{\cdot}$
using $A$ and $B$ is defined below.
\begin{align*}
\hspace{-3cm}  \rewrOpt{\query_1 \join_{\theta} \query_2} &\defas Q_{sg} \union Q_{pos}
\end{align*}
\begin{align*}
  Q_{sg} &\defas \projection_{\bar{A},\bar{B},\lbatt{\bar{A}},\lbatt{\bar{B}}\ubatt{\bar{A}},\ubatt{\bar{B}},e_{c},e_{sg},e_{p}}(\\ &\hspace{1cm}\rewrUAA{\splitbg(\query_1)} \join_{\bgMarker{\theta}} \rewrUAA{\splitbg(\query_2)})\\[3mm]
  Q_{pos} &\defas \rewrUAA{\compressPos_{A,n}(\splitpos(\query_1)) \join_{\theta} \compressPos_{B,n}(\splitpos(\query_2))}\\
   e_{c}                     & \defas \query_1.\rlb \cdot \query_2.\rlb \renameto \rlb                                                                                                           \\
 e_{sg}                    & \defas \query_1.\rbg \cdot \query_2.\rbg \renameto \rbg                                                                                                           \\
 e_{p}                     & \defas \query_1.\rub \cdot \query_2.\rub \renameto \rub
\end{align*}

\begin{lemma}[The optimized rewrite preserves bounds]\label{lem:opt-join-rewr-preserves-bounds}
  The optimized join rewrite is correct, i.e., for any $\uaaN$-database and query $\query \defas \query_1 \join_{\theta} \query_2$, let $\query_{optMerge}$ be query $\query_{merge}$, but using $\rewrOpt{\cdot}$ instead of $\rewrUAA{\cdot}$ for joins. Then,
  \[
    \Dec(\query_{merge}(\Enc(\rangeDB))) \dbleq \Dec(\query_{optMerge}(\Enc(\rangeDB)))
    \]
\end{lemma}
\begin{proof}
  Based on \Cref{lem:split-preserves-bounds} and
  \Cref{lem:compression-preserves-bounds} we know that the split and compression
  operators preserve bounds. As mentioned above, all tuples are attribute-level
  certain in the result of $\splitbg$. Thus,
  \[
    \rewrUAA{\splitbg(R)} \join_{\bgMarker{\theta}} \rewrUAA{\splitbg(S)})
  \]
  and
  \[
    \rewrUAA{\splitbg(R) \join_\theta \splitbg(S)}
\]
    are equivalent.
  Since the rewrite for union is returning of the rewritten inputs we know that $\rewrOpt{\query_1 \join_{\theta} \query_2}$ is equivalent to $\rewrUAA{Q_{sg'} \union \query_{pos'}}$ where $\query_{sg'}$ and $\query_{pos'}$ are defined as shown below.

  \begin{align*}
  \rewrUAA{Q_{sg'} \union Q_{pos'}}\\
  Q_{sg'} &\defas \splitbg(\query_1) \join_{\bgMarker{\theta}} \splitbg(\query_2)\\
  Q_{pos'} &\defas \compressPos_{A,n}(\splitpos(\query_1)) \join_{\theta} \compressPos_{B,n}(\splitpos(\query_2))\\
  \end{align*}

  Using the fact that join distributes over union, we can rewrite $\query_{sg'} \union \query_{pos'}$ into:
  \begin{align*}
   (\splitbg(\query_1) \union \splitpos(\query_1)) \join_{\theta} (\splitbg(\query_2) \union \splitpos(\query_2))
  \end{align*}
Since replacing $\query$ with  $\splitbg(\query) \union \splitpos(\query)$ preserves bounds, it follows that $\rewrOpt{\cdot}$ is correct.
\end{proof}

\begin{Example}
	\Cref{table:join_with_op} shows the result of joining the two tables from \Cref{table:join_no_op} using the optimized rewrite. By sacrificing accuracy, the number of result tuples can be reduced by limiting the number of output tuples.
  \end{Example}

\subsection{Optimization for Aggregation}
Similar to the optimization for joins, the self-join used in the rewrite for
aggregation to determine which tuples could possibly belong to a group may also
degenerate into a cross product if the attribute-level bounds are loose,
resulting in a potential performance bottleneck. We now introduce an optimized
version of the rewrite rule for aggregation. This optimization improves over the
naive aggregation rewrite in two aspects: (i) we piggy-back the computation of
\abbrBG aggregation function results on the computation of output groups
($\query_{gbounds}$) and (ii) we use the compression operator used in the
optimize join rewrite to compress the RHS of subquery $\query_{join}$ from the
naive aggregation rewrite.

\begin{align*}
  \rewrOpt{\aggregation{G}{f(A)}(R)} & \defas \aggregation{'f(\bgMarker{A})',\bgatt{G},\ubatt{G},\lbatt{G}}{e_{aggbounds}}(Q_{proj})                                                                                             \\
	e_{aggbounds}                    & \defas
                                      f(\ubatt{A}), f(\lbatt{A}),                                                                                                                                     \\
                                     &\hspace{1cm} \aggmax(\rlb) \renameto \rlb,                                                                                                                                       \\
                                     &\hspace{1cm} \aggmax(\rbg) \renameto \rbg,                                                                                                                                       \\
                                     &\hspace{1cm} \aggsum(\ubMarker{\rattj}) \renameto \rub
\end{align*}
\begin{align*}
 \query_{proj}                       & \defas \projection_{'f(\bgMarker{A})',\bgatt{G},\ubatt{G},\lbatt{G},lba,uba,\rlb,\rbg,\ubMarker{\rattj}}(\query_{join})                                                                        \\[3mm]
 \query_{join}                       & \defas \query_{gbounds} \join_{\theta_{join}} \compressPos_{A,n}(\rename_{e_{rename}}(\rewrUAA{R}))                                                                 \\
  \theta_{join}                      & \defas \bigwedge_{A_i \in G}  \ubatt{A_i} \geq \lbatt{B_i} \wedge \ubatt{B_i} \geq \lbatt{A_i}                                                                      \\
  e_{rename}                         & \defas A_1 \renameto B_1, \ldots, A_n \renameto B_n, \ratt \renameto \rattj
\end{align*}
\begin{align*}
  \query_{gbounds}                   & \defas                                                 \aggregation{\bgatt{G}}{e_{gbounds},\aggmax(e_{c}) \renameto \rlb,\aggmax(\rbg) \renameto \rbg}(\rewrUAA{R}) \\
  e_{c}                              & \defas  \ifte{\theta_c \wedge \rlb>0}{1}{0}                                                                                                                         \\
  \theta_c                           & \defas \bigwedge_{A_i \in G}                                        \ubatt{A_i}=\ubatt{B_i} \wedge = \lbatt{A_i}=\lbatt{B_i} \wedge \lbatt{A_i} = \ubatt{A_i}       \\
  e_{gbounds}                        & \defas e_{bound}^{A_1}, \ldots, e_{bound}^{A_k}, sga \tag{for $G = (A_1, \ldots, A_k)$}                                                                                  \\
  e_{bound}^{A}                      & \defas \aggmax(\ubatt{A}) \renameto \ubatt{A},\aggmin(\lbatt{A}) \renameto \lbatt{A}
\end{align*}

As mentioned above, the main goal of the optimization is to decrease the number
of input tuples for the join operation.  Note how in the optimized rewrite shown
above, a compression operator is applied on the RHS input of the join which
limits the number of tuples that need to be matched to the output groups
produced by the LHS. Because we still need to produce the correct \abbrBG
aggregation results for each output group, we now need to perform this when
calculating output groups in the LHS since compression may group multiple
\abbrBG groups into one tuple in the compressed RHS which makes it impossible to
compute the precise \abbrBG aggregation function results from the RHS.  There
are also some minor additional modifications like pre-computing the
certain/selected guess tuple multiplicities.

\begin{lemma}[The optimized rewrite preserves bounds]\label{lem:opt-join-rewr-preserves-bounds}
  The optimized aggregation rewrite is correct, i.e., for any $\uaaN$-database and query $\query \defas \aggregation{G}{f(A)}(\query_1)$, let $\query_{optMerge}$ be query $\query_{merge}$, but using $\rewrOpt{\cdot}$ instead of $\rewrUAA{\cdot}$ for aggregation. Then
  \[
    \Dec(\query_{merge}(\Enc(\rangeDB))) \dbleq \Dec(\query_{optMerge}(\Enc(\rangeDB)))
  \]
\end{lemma}
\begin{proof}
We already proved that the naive aggregation rewrite preserves bounds (\Cref{theo:aggregation-preserves-bounds}).
The optimized rewrite we computes \bestG{} results in the LHS of the join. However, this does not affect the computed values.   We also apply the compression operator to compress the RHS. As shown in \Cref{lem:compression-preserves-bounds}, this operator preserves bounds. Consider two tuples $\tup_{left}$ and $\tup_{right}$ from the LHS and RHS of naive rewrite that are joined by $\query_{join}$. In the optimized rewrite, $\tup_{right}$ may have been merged with multiple other tuples into a compressed tuple $\tup_{compressed}$. By definition of $\compressPos$, the attribute bounds of $\tup_{compress}$ include the attribute bounds of $\tup_{right}$. Since the join condition of $\query_{join}$ is the same in the naive and optimized versions of the aggregation rewrite, we know that $\tup_{left}$ joins with $\tup_{compress}$. Using the fact the $\compressPos$ is bound preserving it follows that the aggregation function result bounds computed by the optimized aggregation rewrite have to include the bounds produced by the naive rewrite. The group-by bounds are computed in the same way in both cases and, thus, both rewrite compute the same group-by bounds for each result tuple. Thus, it follows that the optimized aggregation rewriting is correct.
\end{proof}

\section{Creating \abbrUAADBs}
\label{sec:creating-abbruaadbs}

  In this section we discuss how to translate data represented in incomplete and probabilistic data models into \abbrUAADBs such that the generated \abbrUAADB bounds the input uncertain database. Furthermore, we enable the user to specify bounds for an expression. This enables integration of \abbrUAADBs with Lenses~\cite{Yang:2015:LOA:2824032.2824055,BB19} to generate \abbrUAADBs that encode the uncertainty introduced by data cleaning and curation heuristics.

\subsection{\captialTIs}
\label{sec:tip}
A tuple-independent database (\abbrTI) $\pdb$ is a database where each tuple $\tup$ is marked as optional or not. The incomplete database represented by a \abbrTI $\pdb$ is the set of instances that include all non-optional tuples and some subset of the optional tuples. That is, the existence of a tuple $t$ is independent of the existence of any other tuple $t'$. In the probabilistic version of \abbrTIs each tuple is associated with its marginal probability. The probability of a possible world is then the product of the probability of all tuples included in the world multiplied by the product of $1 - \prob(t)$ for all tuples from $\pdb$ that are not part of the possible world.
We define a translation function $\doTTRTI$ for \abbrTIs that returns a \abbrAUDB relation with certain attribute values. We extend this function to databases in the obvious way. A tuple's lower multiplicity bound is $1$ if the tuples is not certain (marked as optional or its probability is less than $1$). That is, like for \abbrUADBs, $\doTTRTI$ encodes exactly the certain answers of the \abbrTI and the \abbrBG world is selected as the world that contains all tuples whose marginal probability is larger than or equal to $0.5$. Note that this is indeed one of possible worlds that have the highest probability among all worlds encoded by the \abbrTI. Let $\TTRTI(\tup)$ denote the range-annotated tuple which encodes the certain tuple $\tup$ and $\pdb$ be a  \abbrPTIrel. We define:

\[
  \forall{A \in \schemaOf(\prel)}:\lbMarker{\TTRTI(\tup).A} = \ubMarker{\TTRTI(\tup).A}=\bgMarker{\TTRTI(\tup).A}=\tup.A
  \]

Using this definition, we define $\doTTRTI$ for probabilistic \abbrTIs as shown below.

\begin{align*}\label{eq:tip-labeling}
  \lbMarker{\doTTRTI(\prel)(\TTRTI(\tup))}            & \defas
                                                      \begin{cases}
                                                        1 & \mathtext{if} \prob(\tup) = 1 \\
                                                        0 & \mathtext{otherwise}
                                                      \end{cases}                         \\
\bgMarker{\doTTRTI(\prel)(\TTRTI(\tup))}              & \defas
                                                      \begin{cases}
                                                        1 & \mathtext{if} \prob(\tup) \geq 0.5 \\
                                                        0 & \mathtext{otherwise}
                                                      \end{cases}                         \\
\ubMarker{\doTTRTI(\prel)(\TTRTI(\tup))}              & \defas
                                                      \begin{cases}
                                                        1 & \mathtext{if} \prob(\tup) > 0 \\
                                                        0 & \mathtext{otherwise}
                                                      \end{cases}                         \\
\end{align*}

\begin{Theorem}[$\doTTRTI$ is bound preserving]\label{theo:h-homo-UA}
Given a probabilistic \abbrTI  $\pdb$, $\doTTRTI(\pdb)$ is a bound for $\pdb$.
\end{Theorem}
\begin{proof}
  By definition of the translation function $\doTTRTI$, all attribute-values are
  certain. Thus, with exception of the upper bound on a tuple's multiplicity
  ($1$ is a trivial upper bound on the multiplicity of any tuple in an \abbrTI),
  the claim follows from~\cite{FH19}[Theorem 2] which proved the lower bounding
  property of this translation.  Thus, the result of $\doTTRTI$ trivially bounds
  $\prel$.  We omit the proof for incomplete \abbrTIs since it is analog.
\end{proof}

\subsection{\abbrXDBs}
\label{sec:x-dbs}
An \abbrXDB~\cite{DBLP:conf/vldb/AgrawalBSHNSW06}, records for each tuple a
number of alternatives. Such alternatives are encoded as so-called x-tuples.  An
x-tuple $\xTup$ is simply a set of tuples $\{t_1, \ldots, t_n\}$ with a label
indicating whether the x-tuple is optional. We use $\card{\xTup}$ to denote the
number of alternatives of x-tuple $\xTup$.  x-relations are sets of such
x-tuples and x-databases are sets of x-relations. A possible world of an
x-relation $\rel$ is a deterministic relations that is generated by selecting at
most one alternative $t \in \xTup$ for every x-tuple $\xTup$ from $\rel$ if
$\xTup$ is optional, or exactly one if it is not optional.  That is, each
x-tuple is assumed to be independent of the others, and its alternatives are
assumed to be disjoint (hence the name block-independent incomplete database).
The probabilistic version of \abbrXDBs (also called a \textit{block-independent
  database}~\cite{DBLP:conf/vldb/AgrawalBSHNSW06}) assigns each alternative a
probability such that $\prob(\xTup) = \sum_{t \in \xTup} \prob(t) \leq 1$. Thus,
a probabilistic x-tuple is optional if $\prob(\xTup) < 1$.
For an x-tuple $\xTup = \{ t_1, \ldots, t_n \}$, we use $\xselect{\xTup}$ to denote the alternative of $\xTup$ with the highest probability among all alternatives, picking the first alternative if there are multiple such alternatives. For instance, for $\xTup = \{ t_1, t_2 \}$ with $\prob(t_1) = 0.5$ and $\prob(t_2) = 0.5$ we get $\xselect{\xTup} = t_1$.

We now define a
translation function $\doTTRX$ that maps x-dbs into \abbrUAADBs. In the result of this translation each x-tuple from the input is encoded as a single range-annotated tuple whose attribute bounds include the attribute values of all alternatives of the x-tuple.  For the \abbrBG values of a tuple we choose the alternative with the highest probability.
We use $\TTRX(\xTup)$ to denote a range-annotated tuple that we construct as shown below such that it bounds all alternatives for x-tuple $\xTup$. For any $A \in \schemaOf(\xTup)$ we define:

\begin{align*}
  \lbMarker{\TTRX(\xTup).A}   & = \min_{\tup \in \xTup}\tup.A    \\
  \bgMarker{\TTRX(\xTup).A} & = \xselect{\xTup}.A \\
  \ubMarker{\TTRX(\xTup).A} & = \max_{\tup \in \xTup}\tup.A
\end{align*}

X-tuples in an x-relation are certain if $\prob(\xTup)$, i.e., if some
alternative of $\xTup$ exists in every possible world.  For the incomplete
version of x-dbs, a tuple is certain if it is not optional. We set the lower
multiplicity bound for such tuples to $1$. The lower multiplicity bound for all
other tuples is set to $0$.  Any alternative of an x-tuple is possible.  This
alternative $\xselect{\xTup}$ is part of the \abbrBGW if
$(1 - \prob(\xTup)) \leq \prob(\xselect{\xTup})$, i.e., it is more likely that
$\xselect{\xTup}$ exists than that no alternative of the x-tuple exists.  The
multiplicity upper bound for any $\TTRX(\tup)$ is $1$ and its lower bound is $1$
iff the x-tuple is not optional (one of the alternatives of the x-tuple exists
in every world). As mentionIn probabilistic \abbrXDBs we check that
$\prob(\xTup) = 1$. We the show version of $\doTTRX$ for probabilistic \abbrXDBs
below.

\begin{align*}
  \lbMarker{\doTTRX(\prel)(\xTup)} = &
                              \begin{cases}
                                1    & \mathtext{if} \prob(\xTup) = 1                               \\
                                0    & \mathtext{otherwise}                                         \\
                              \end{cases}                                                           \\
  \bgMarker{\doTTRX(\prel)(\xTup)} = & \begin{cases}
                                1    & \mathtext{if} (1 - \prob(\xTup)) \leq \prob(\xselect{\xTup}) \\
                                0    & \mathtext{otherwise}                                         \\
                              \end{cases}                                                           \\
  \ubMarker{\doTTRX(\prel)(\xTup)} = & \begin{cases}
                                1    & \mathtext{if} \prob(\xTup) > 0                               \\
                                0    & \mathtext{otherwise}                                         \\
                              \end{cases}                                                           \\
\end{align*}

\begin{Theorem}[$\doTTRX$ preserves bounds]\label{theo:h-homo-UA-TUL-correct}
Given an \abbrXtable $\pdb$, $\doTTRX(\pdb)$ bounds $\pdb$.
\end{Theorem}
\begin{proof}
  Trivially, $\TTRX(\xTup)$ bounds all alternatives of $\xTup$ by
  construction. Each possible world contains at most one alternative per x-tuple
  (exactly one if the x-tuple is certain). Thus, $\ubMarker{\TTRX(\xTup)} = 1$
  is an upper bound on the multiplicity of any x-tuple's alternative in every
  world. The lower bound multiplicity $\lbMarker{\TTRX(\xTup)}$ is $1$ if
  $\prob(\xTup) = 1$ and $0$ otherwise. Thus, $\lbMarker{\TTRX(\xTup)}$ a lower
  bound the multiplicity of an alternative of the x-tuple in every world. Since
  x-tuples are assumed to be independent of each other and each possible world contains at
  most one alternative of an x-tuple, the probability of a possible world $\db$
  of an \abbrXtable $\prel$ is calculated as

  \[
    \left( \prod_{\tup \in \db} \prob(\tup) \right) \cdot \left( \prod_{\xTup: (\not \exists \tup \in \xTup: \tup \in \db)} 1 - \prob(\xTup) \right)
  \]

  Thus, the \abbrBG world is indeed the world with the highest
  probability of the \abbrXDB, because it contains the highest probability
  alternative for each x-tuple (or no alternative if this is the highest
  probability option).
\end{proof}

\subsection{\abbrCtables}
\label{sec:c-tables}
In contrast to the proof of \Cref{theo:intractability-of-tight-bounds}, we now
consider \textit{\abbrCtables}~\cite{DBLP:journals/jacm/ImielinskiL84} where
variables from the set $\Sigma$ of variable symbols can also be used as
attributes values, i.e., tuples over $\dataDomain \cup \Sigma$. Recall from the
proof of \Cref{theo:intractability-of-tight-bounds} that an
\abbrCtable~\cite{DBLP:journals/jacm/ImielinskiL84}
$\prel = (\rel, \lcond, \gcond)$ is a relation $\rel$ paired with (i) a global
condition $\gcond$ which is also a logical condition over $\Sigma$ and (ii) a
function $\lcond$ that assigns to each tuple $\tup \in \rel$ a logical condition
over $\Sigma$. Given a valuation $\mu$ that assigns to each variable from
$\Sigma$ a value, the global condition and all local conditions evaluate to
either $\btrue$ or $\bfalse$. The incomplete database represented by a
generalized \abbrCtable $\prel$ is the set of all relations $\rel$ such that
there exists a valuation $\mu$ for which $\mu(\gcond)$ is true and
$\rel = \{ \mu(\tup) \mid \mu(\lcond(\tup)) \}$, i.e., $\rel$ contains all
tuples for which the local condition evaluates to true where each variable in
the tuple is replaced based on $\mu$.

Since determining whether a tuple in a c-table is certain is
\conpcomplete\footnote{
Determine the certain answers to a query over a \abbrCoddTable is \conpcomplete~\cite{V86,AK91}.
  Since, the result of any
first order query over a \abbrCoddTable can be encoded as a \abbrCtable and
evaluating a query in this fashion is in \ptime, it follows that determining
whether a tuple is certain in a \abbrCtable cannot be in \ptime, because otherwise we could use this to compute the certain answers to a query in \ptime by evaluating it in over a \abbrCtable encoding the input \abbrCoddTable and then calculating the certain tuples of the resulting \abbrCtable.}, we settle for a transformation that creates lower and upper bounds for both attribute-values and multiplicities that are not tight. Using a constraint solver, we can determine (i) whether the local condition $\lcond(\tup)$ of a tuple is a tautology (the tuple exists in every world) and (ii) lower and upper bounds on the value of a tuple's attribute. For instance, for the lower bound we have to solve the following optimization problem:

\begin{flushleft}
  \textbf{Minimize:}
  \begin{align*}
    \mu(\tup.A)
  \end{align*}
  \textbf{Subject to:}
  \begin{align*}
    \lcond(\mu(\tup))\\
    \forall x \in \Sigma: \mu(x) \in \dataDomain
  \end{align*}
\end{flushleft}
We use $\min(\tup.A)$ and $\max(\tup.A)$ to denote the results of these optimization problems.
While solving constraints is not in \ptime, this can be acceptable if the size of formulas used in local conditions is relatively small. If that is not the case, we can trade tightness of bounds for performance and fall back to a simpler method, e.g., in worst case the minimum and maximum values of $\dataDomain$ are safe bounds and the bounds for a constant value is the constant itself.
To select a \abbrBGW, we have to find a valuation $\mu$ such that the global condition holds. In general, this may be computationally hard. For \abbrCtables without global conditions, any valuation would do. Let $\mu_{\abbrBG}$ denote the valuation we select.
To define our transformation, we again first  define a function $\TTRC$ that maps tuples and their local conditions to range-annotated tuples. For all $A \in \schemaOf(\tup)p$ this function is defined as:
\begin{align*}
  \lbMarker{\TTR(\tup).A} & \defas \min(\tup.A)         \\
  \bgMarker{\TTR(\tup).A} & \defas \mu_{\abbrBG}(\tup.A) \\
  \ubMarker{\TTR(\tup).A} & \defas \max(\tup.A)
\end{align*}

\newcommand{\istautology}{\textsc{isTautology}}
\newcommand{\issat}{\textsc{isSatisfiable}}

Now to determine the multiplicity bounds of tuples we have to reason about whether a local condition is a tautology and whether it is satisfiable. This can again be checked using constraint solvers.
We can fall back to a $\ptime$ method that can only detect certain types of tautologies if necessary, e.g., if the local condition is a conjunction of inequalities.
Note that even the first version does not guarantee tight bounds, because it (i) ignores global conditions and (ii) does not take into account that the set of tuples encoded by two tuples from the input \abbrCtable may overlap (e.g., a tuple may be certain even though no local condition is a tautology).
We define predicate $\istautology(\psi)$ that evaluates to true if $\psi$ is a tautology (potentially using the optimization as described above). Similarity, $\issat(\psi)$ is true if $\psi$ is satisfiable. Using these predicates we define the multiplicity bounds generated by $\doTTRC$ as shown below.

\begin{align*}
  \lbMarker{\doTTRC(\prel)(\TTRC(\tup, \lcond(\tup)))} & \defas
                                         \begin{cases}
                                           1           & \mathtext{if} \istautology(\psi)          \\
                                           0           & \mathtext{otherwise}                      \\
                                         \end{cases}                                               \\
  \bgMarker{\doTTRC(\prel)(\TTRC(\tup, \lcond(\tup)))} & \defas
                                         \begin{cases}
                                           1           & \mathtext{if} \mu_{\abbrBG}(\lcond(\tup)) \\
                                           0           & \mathtext{otherwise}                      \\
                                         \end{cases}                                               \\
  \ubMarker{\doTTRC(\prel)(\TTRC(\tup, \lcond(\tup)))}                 & \defas
                                         \begin{cases}
                                           1           & \mathtext{if} \issat(\psi)                \\
                                           0           & \mathtext{otherwise}                      \\
                                         \end{cases}                                               \\
\end{align*}

\begin{Theorem}[$\doTTRC$ is bound preserving]\label{theo:h-homo-UA-c-table}
Given an incomplete database $\pdb$ encoded as \abbrCtables, $\doTTRC(\pdb)$ bounds $\pdb$.
\end{Theorem}
\begin{proof}
Based on the definition of \abbrCtables, any tuple whose local condition is a tautology exists in every possible world. Furthermore, if the local condition of a tuple is satisfiable then it may exist in some world (only if the global condition for the valuation that satisfies the local condition evaluates to true). Finally, by construction the tuples of the  \abbrUAADB created by $\doTTRC$ for a \abbrCtable $\prel$ bound the possible values of the tuples of $\prel$ across all possible worlds. By construction $\mu_{\abbrBG}$ is a possible world of $\prel$. Thus, $\doTTRC(\prel)$ bounds $\prel$.
\end{proof}

The probabilistic version of \abbrCtables~\cite{DBLP:journals/debu/GreenT06}  associates each variable with a probability distribution over its possible values. Variables are considered independent of each other. The probability of a variable assignment $\mu$ is then the product of each probabilities $\prob(x = \mu(x))$ for each $x \in \Sigma$. The probability of a possible world is the sum of the probabilities of all valuations that produce this world.
Our translation scheme can be adapted to the probabilistic version of \abbrCtables by picking the \abbrBG by selecting the highest probability assignment for each variable $x$ to approximate the world with the highest probability. For probabilistic \abbrCtables with global conditions, finding a possible may again be computationally hard.



\subsection{Lenses}
\label{sec:lenses}
Lenses as presented in~\cite{Yang:2015:LOA:2824032.2824055,BB19,BS20} are a principled method for exposing the uncertainty in the result of data cleaning, curation and integration methods as incomplete data. For instance, when repairing primary key violations by picking one tuple $\tup$ for each set of tuples with the same key, the choice of $\tup$ is typically made based on some heuristic. However, all other possible picks cannot be ruled out in general. A lens applies such a cleaning heuristic to select on possible repair as a \abbrBGW and then encodes the space of possible repairs for a method as an incomplete database. Technically, this is achieved using the so-called variable-generating algebra which allows queries to introduce uncertainty in the form of random variables. The result are \abbrVCtables which generalize \abbrCtables by allowing symbolic expressions as attribute values. In our previous primary key repair example, for a key $k$ and the set of tuples $T = \{ \tup \mid \tup.K = k\}$ which have this key value ($K$ are the key attributes), random variables are introduced for each attribute $A$ such that $\card{\{ \tup.A \mid \tup \in T \}} > 1$, i.e., the attribute's value depends on the choice of repair for $k$. The possible value for such a variable $x_{k,A}$ are then all values that appear in $T$:
$x_{k,A} \in \{ \tup.A \mid \tup \in T \}$. This operation can be implemented in the variable-generating algebra as a query that uses aggregation grouping on the key attributes $K$, to check for each non-key attribute and key value $k$ whether there is more than one tuple with this key value. If this is the case, then the value of each attribute is replaced with a variable if there is more than one value for this attribute in a group.

\begin{Example}[Cleaning with Lenses]\label{ex:cleaning-with-lenses}
  Consider repairing the key of a relation $R(A,B)$ with key $A$. In systems
  like Mimir~\cite{Yang:2015:LOA:2824032.2824055} and Vizier~\cite{BB19,BS20}
  that support lenses, variables are introduced by queries while their
  distribution is specified separately. Here we assume that the construct
  $Var(name)$ creates a new random variables with name $name$. Using this construct, we can generate a \abbrVCtable the encodes all possible repairs using the query shown below. We first count the number of distinct values of attribute $B$ for each key (subquery \lstinline!keys!). Then in the outer query we replace each $B$ value with a variable if there is more than one possible value. Note that a full solution also requires us to separately specify the possible values for these variables. However, we omit this here.
\begin{lstlisting}
SELECT A,
       CASE WHEN numB > 1
            THEN Var(tid || 'B')
       ELSE theB
       END AS B
  FROM (SELECT A,
               count(DISTINCT B) AS numB,
               min(B) AS theB
          FROM R
         GROUP BY A) keys
\end{lstlisting}
\end{Example}

To support \abbrUAADBs as an approximation of \abbrVCtables created by Lenses
(and thus to support tracking of a wide range of cleaning and curation
operations), we can either develop a transformation for \abbrVCtables or define
a construct similar to the one for \abbrVCtables used in the example above to
allow uncertainty to be introduced as part of a query. We opted for the second
option, since it would allow new Lenses to be implemented directly in
\abbrUAADBs without the need for an excursion to \abbrVCtables for which
evaluation of certain operators (e.g., joins) is expensive. Towards that goal we
introduce a construct $\makeuncert{\lbMarker{e}}{\bgMarker{e}}{\ubMarker{e}}$
which takes three expressions $\lbMarker{e}$, $\bgMarker{e}$, $\ubMarker{e}$
that calculate a \abbrBG value, and an lower an upper bound for a value. We
require that this construct can only be applied to \abbrUAADBs, assumed to have
been created using one of our transformations explained above. For deterministic
inputs we produce a dummy transformation $\doTTRD$ that assumes that all tuples and
attribute values are certain.

\begin{Example}\label{ex:make-uncertain-application}
  Reconsider the key repair task from \Cref{ex:cleaning-with-lenses}. We can create a \abbrUAADB bounding the space of repairs using the query shown below. Here we select the minimum $B$ value for each group as the \abbrBG value for attribute $B$.
\begin{lstlisting}
SELECT A,
       CASE WHEN numB > 1
            THEN
       ELSE $\makeuncertName$(minB,minB,maxB)
       END AS B
  FROM (SELECT A,
               count(DISTINCT B) AS numB,
               min(B) AS minB,
               max(B) AS maxB
          FROM $\doTTRD(R)$
         GROUP BY A) keys
\end{lstlisting}
\end{Example}

\newcommand{\methAUDB}{\texttt{AU-DB}\xspace}
\newcommand{\methDet}{\texttt{Det}\xspace}
\newcommand{\methLibkin}{\texttt{Libkin}\xspace}
\newcommand{\methUADB}{\texttt{UA-DB}\xspace}
\newcommand{\methMayBMS}{\texttt{MayBMS}\xspace}
\newcommand{\methMCDB}{\texttt{MCDB}\xspace}
\newcommand{\methTrio}{\texttt{Trio}\xspace}
\newcommand{\methSymb}{\texttt{Symb}\xspace}
\newcommand{\realQ}[2]{\ensuremath{Q_{#1,#2}}}
\newcommand{\pdbench}{\texttt{PDBench}\xspace}

\newcommand{\numaggops}{\#agg-ops\xspace}

\begin{figure*}[t]
  \begin{minipage}{0.62\linewidth}
  \centering
  \begin{minipage}{1.0\linewidth}
\begin{subfigure}[b]{.48\textwidth}
	\centering
  \includegraphics[width=1\linewidth, trim=0cm 4cm 3cm 4.5cm, clip]{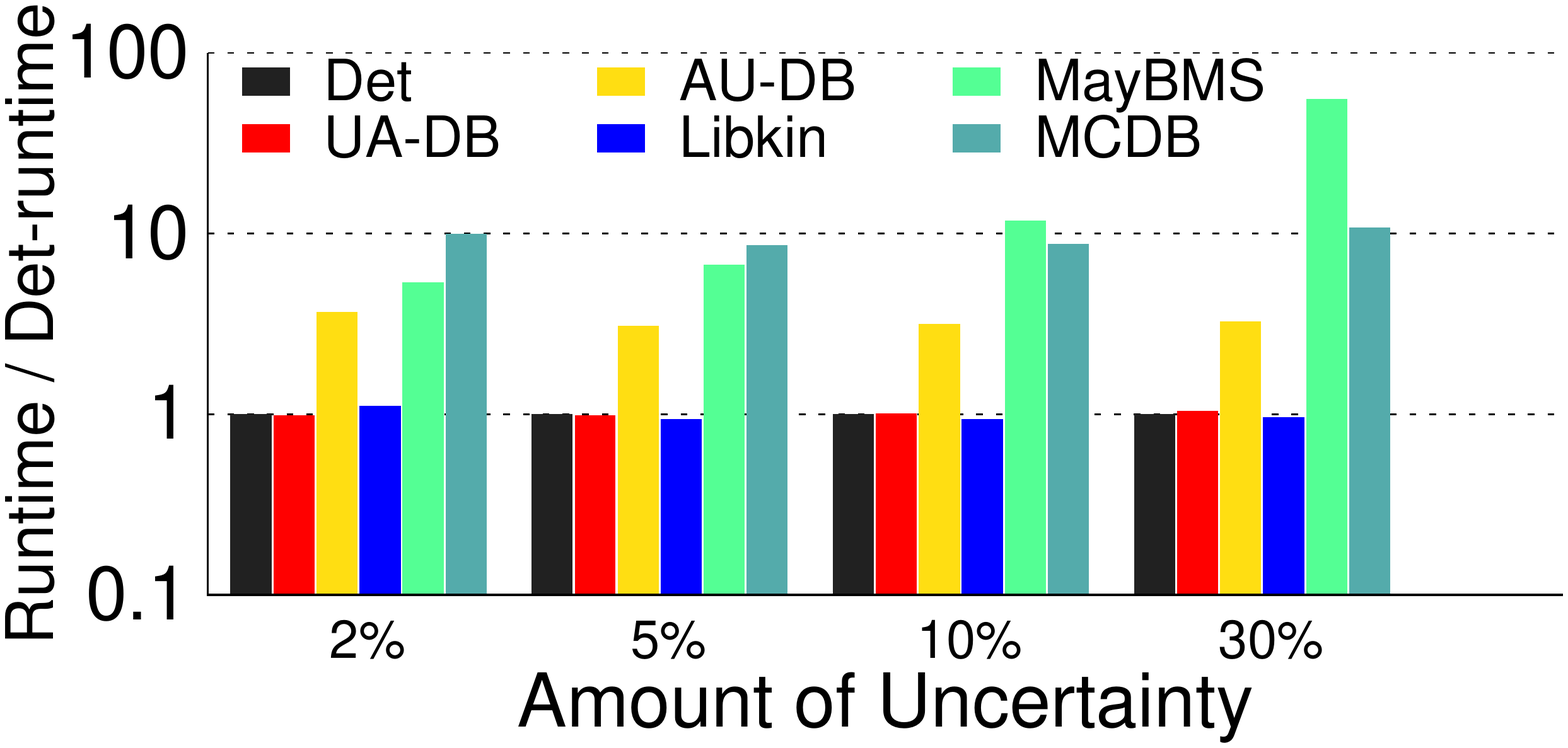}
  \vspace*{-8mm}
  \bfcaption{Varying uncertainty (1GB)}
  \label{fig:runtime_uncert}
\end{subfigure}
\begin{subfigure}[b]{.48\textwidth}
	\centering
  \includegraphics[width=1\linewidth, trim=0cm 4cm 3cm 4.5cm, clip]{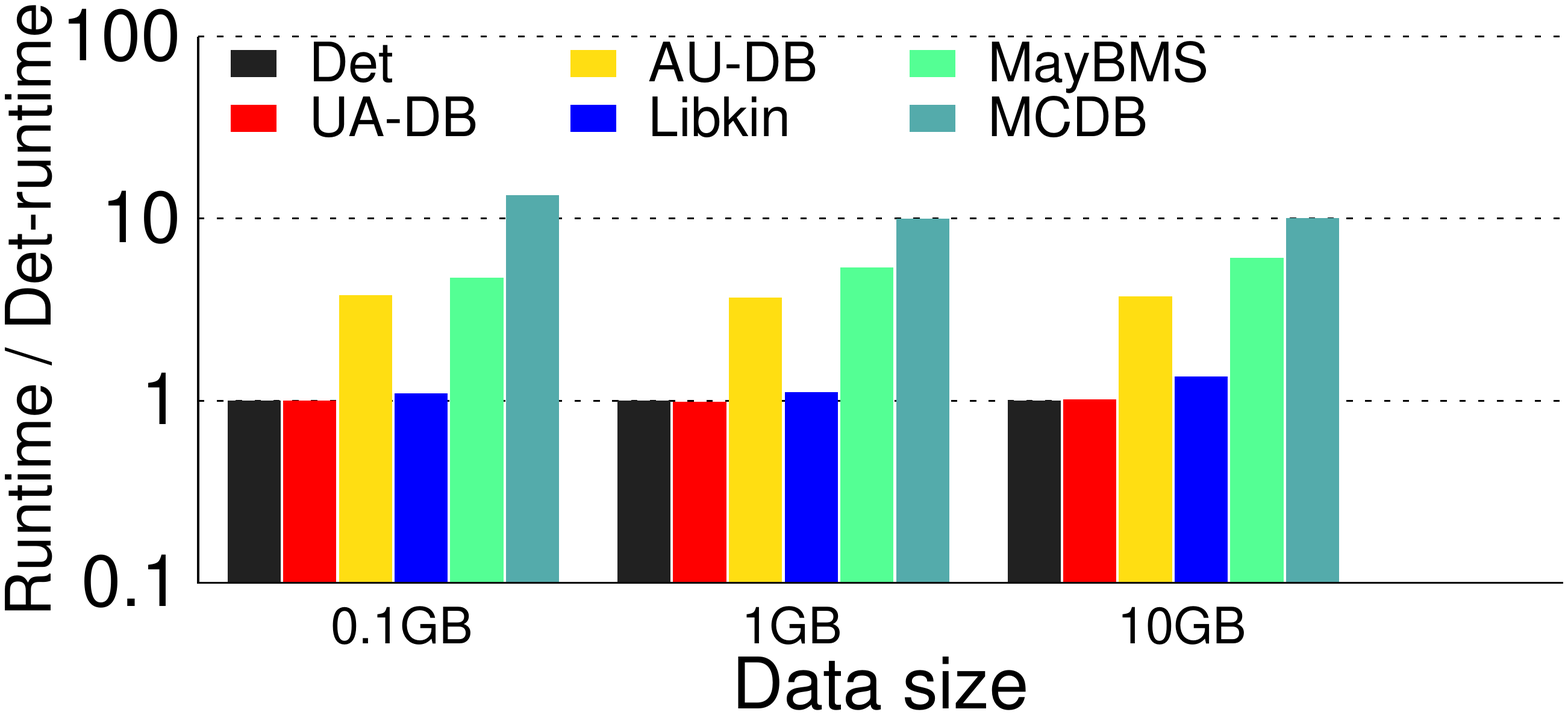}
  \vspace*{-8mm}
  \bfcaption{Varying DB size (2\% uncertainty)}
  \label{fig:runtime_data}
\end{subfigure}
\end{minipage}\\[-5mm]
  \caption{PDbench Queries}
  \label{fig:exp-pdbench}
\end{minipage}
\begin{minipage}{0.36\linewidth}
	\centering
  	\includegraphics[width=0.8\textwidth, trim=0cm 1cm 2cm 4.5cm, clip]{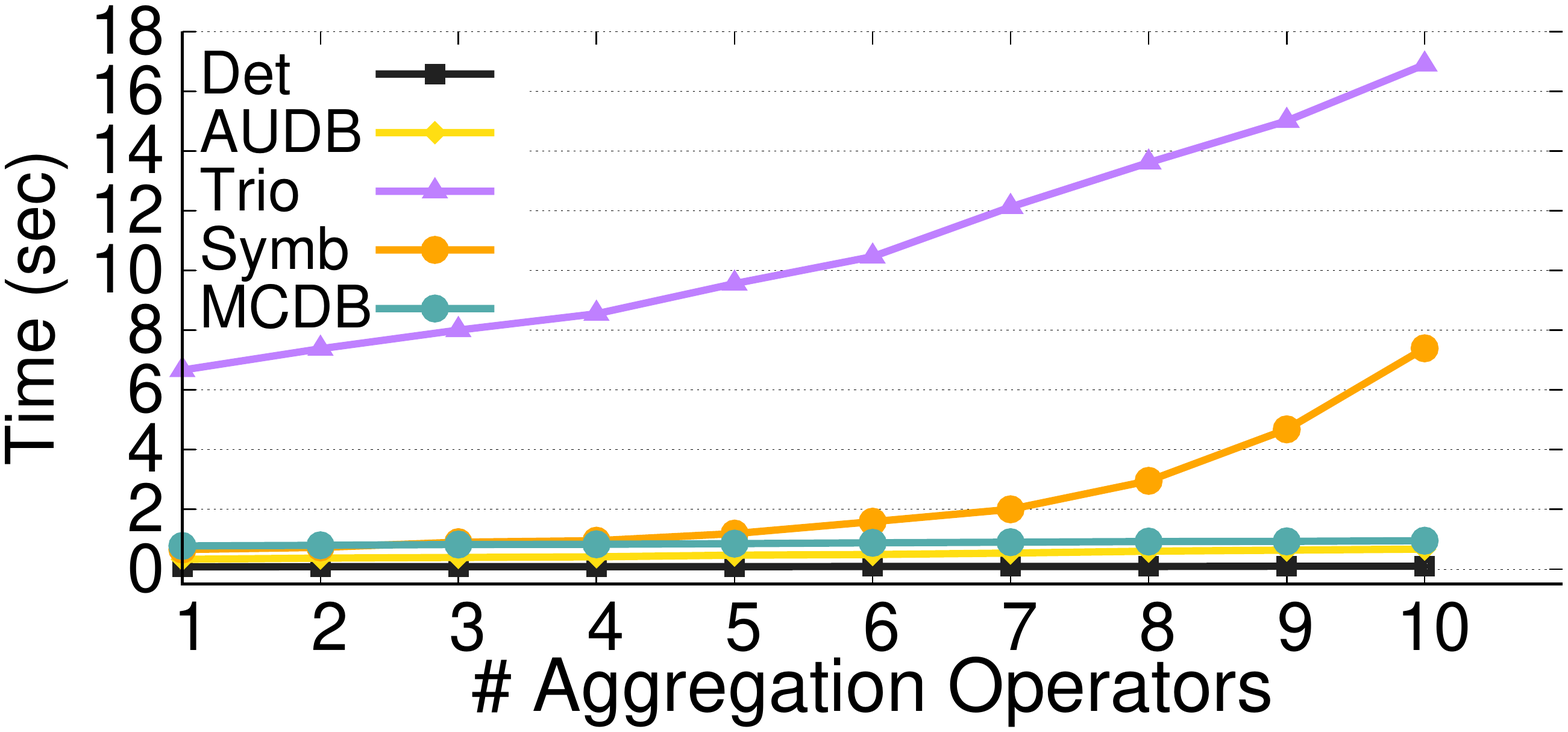}
  	\vspace*{-8mm}
  	\bfcaption{Simple aggregation over TPC-H data}
  	\label{fig:trio}
\end{minipage}
\end{figure*}


\section{Experiments}
\label{sec:experiments}

We compare 
\abbrAUDBs (\textbf{\methAUDB}) implemented on Postgres against
(1)~\textbf{\methDet}:~Deterministic \abbrBGQP;
(2)~\textbf{\methLibkin}: An under\--ap\-prox\-i\-ma\-tion of certain answers~\cite{L16a,GL16};
(3)~\textbf{\methUADB}: An under\--approximation of certain answers combined with \abbrBGQP~\cite{FH19};
(4)~\textbf{\methMayBMS}: \methMayBMS used to compute all possible answers\footnote{Times listed for \methMayBMS \revm{and \methMCDB} include only computing possible answers and not computing probabilities.}; 
(5)~\textbf{\methMCDB}: Database sampling (10 samples) in the spirit of MCDB~\cite{jampani2008mcdb} to over-approximate certain answers;
(6)~\textbf{\methTrio}: A probabilistic DB with bounds for aggregation~\cite{DBLP:conf/vldb/AgrawalBSHNSW06}; and
(7)~\textbf{\methSymb}: An SMT solver (Z3) calculating aggregation result bounds based on the symbolic representation from~\cite{AD11d}.
All experiments are run on a 2$\times$6 core AMD Opteron 4238 CPUs, 128GB RAM, 4$\times$1TB 7.2K HDs (RAID 5). We report the average 
of 10 runs.
\BG{Our experimental results support the following claims:
(i) In contrast to \methLibkin and \methUADB, which support a less expressive range of queries and report less detail (no possible answers or bounds), \abbrAUDB performance is competititve (within a factor of 2-5);
(ii) \methAUDB outperforms probabilistic databases including \methTrio, and \methMayBMS that can enumerate all possible answers and can distinguish certain tuples; and
(iii) \methAUDB also outperforms \methMCDB, the only other approach that supports all the queries we support (noting that it can not distinguish between certain and possible tuples, nor bound possible answers).}

\subsection{Uncertain TPC-H (\pdbench)}
We use \pdbench~\cite{antova2008fast}, a modified TPC-H data generator~\cite{tpch} that creates an \abbrXDB (\termBI) with attribute-level uncertainty by replacing random attributes with multiple randomly selected possible alternatives.
We directly run \methMayBMS queries (without probability computations) on its native columnar data representation.
For \methMCDB, we approximate  
tuple bundles with 10 samples.
We apply \methLibkin on a database with labeled nulls for uncertain attributes using the optimized rewriting from~\cite{GL16}.
We run 
\methDet on one randomly selected world --- this world is also used as the \abbrBGW for \methUADB and \methAUDB.
We construct an \abbrAUDB instance by annotating each cell in this world with the minimum and maximum possible values for this cell across all 
worlds.
For \methUADB we mark all tuples with at least one uncertain value as uncertain.

\mypar{\pdbench Queries}
To evaluate the overhead of our approach compared to \abbrUADBs for queries supported by this model, we reproduce the experimental setup from~\cite{FH19}, which uses the queries of \pdbench (simple SPJ queries).
%
With a scale factor 1 (\revm{SF1}) database ($\sim$1GB  per world), we evaluate scalability relative to amount of uncertainty.
Using \pdbench, we vary the percentage of uncertain cells: 2\%, 5\%, 10\% and 30\%.
Each uncertain cell has up to 8 possible values
\revm{picked uniformly at random over the whole domain,
resulting in large ranges, a worst-case scenario for \abbrAUDBs} and a best-case scenario for \methMayBMS. 
As \Cref{fig:runtime_uncert} shows, our approach has constant overhead (a factor of $\sim$ 5),  
resulting from the many possible tuples created by joins on attributes with ranges across the entire domain.
%
To evaluate scalability, we use 100MB, 1GB, and 10GB datasets (SF 0.1, 1, and 10) and fix the uncertainty percentage (2\%).
As evident from \Cref{fig:runtime_data}, \abbrAUDBs scale linearly in the \revc{SF} for such queries.

\begin{figure}[t]
  \centering
  {\footnotesize
  \begin{tabular}{cc|r|r|r|r|r}
    \multicolumn{2}{c|}{\thead{Queries}} & \thead{2\%/SF0.1} & \thead{2\%/SF1} & \thead{5\%/SF1} & \thead{10\%/SF1} & \thead{30\%/SF1} \\ \hline
    \multirow{2}{*}{Q1}                  & \methAUDB         & 1.607           & 15.636          & 15.746           & 15.811 & 16.021 \\
                                         & \methDet          & 0.560           & 1.833           & 1.884            & 1.882  & 1.883  \\
                                         & \methMCDB       & 5.152           & 19.107          & 18.938           & 19.063 & 19.279 \\
    \hline
    \multirow{2}{*}{Q3}                  & \methAUDB         & 0.713           & 7.830           & 8.170            & 8.530  & 7.972  \\
                                         & \methDet          & 0.394           & 1.017           & 1.058            & 1.092  & 1.175  \\
                                         & \methMCDB       & 4.112           & 11.138          & 11.222           & 10.936 & 11.454 \\
                                         \hline
    \multirow{2}{*}{Q5}                  & \methAUDB         & 0.846           & 8.877           & 8.803            & 8.839  & 8.925  \\
                                         & \methDet          & 0.247           & 0.999           & 1.012            & 1.123  & 1.117  \\
                                         & \methMCDB       & 2.599           & 10.152          & 10.981           & 11.527 & 11.909 \\
                                         \hline
    \multirow{2}{*}{Q7}                  & \methAUDB         & 0.791           & 7.484           & 7.537            & 7.303  & 7.259  \\
                                         & \methDet          & 0.145           & 0.977           & 0.985            & 0.989  & 1.044  \\
                                         & \methMCDB       & 1.472           & 10.123          & 10.277           & 10.749 & 10.900 \\
                                         \hline
    \multirow{2}{*}{Q10}                 & \methAUDB         & 0.745           & 7.377           & 7.283            & 7.715  & 8.012  \\
                                         & \methDet          & 0.263           & 1.024           & 0.993            & 1.004  & 1.015  \\
                                         & \methMCDB       & 2.691           & 10.743          & 10.937           & 11.826 & 11.697 \\


  \end{tabular}
  }\\[-3mm]
  \bfcaption{TPC-H query performance (runtime in sec)}
  \label{fig:tpch}
  \trimfigurespacing
\end{figure}

\mypar{TPC-H queries}
We now evaluate 
 actual
TPC-H queries on \pdbench data.
These queries contain aggregation with uncertain group-by attributes (only supported by \methAUDB and \methMCDB).
Results are shown in \Cref{fig:tpch}.
For most queries, \methAUDB has an overhead factor of between 3-7 over \methDet.
This overhead is mainly due to additional columns and scalar expressions.
Compared to \methMCDB, \methAUDB is up to 570\% faster, while producing hard bounds instead of an estimation.


\mypar{Simple Aggregation}
We use a simple aggregation query with certain group-by attributes on an SF0.1 instance to compare against a wider range of approaches, varying the number of aggregation operators (\numaggops).
For systems that do not support subqueries like \methTrio, operator outputs are materialized as tables. 
In this experiment, \methTrio produces incorrect answers, as its representation of aggregation results (bounds) is not closed under queries; We are only interested in its performance.
\Cref{fig:trio} shows the runtime of our technique compared to
\methTrio which is significantly slower, and
\methSymb (only competitive for low \numaggops values). 

\begin{figure*}[t]
  \centering
  \begin{minipage}{1.0\linewidth}
\begin{subfigure}[b]{.245\textwidth}
	\centering
  \includegraphics[width=0.99\textwidth, trim=0cm 1.5cm 0.5cm 0cm, clip]{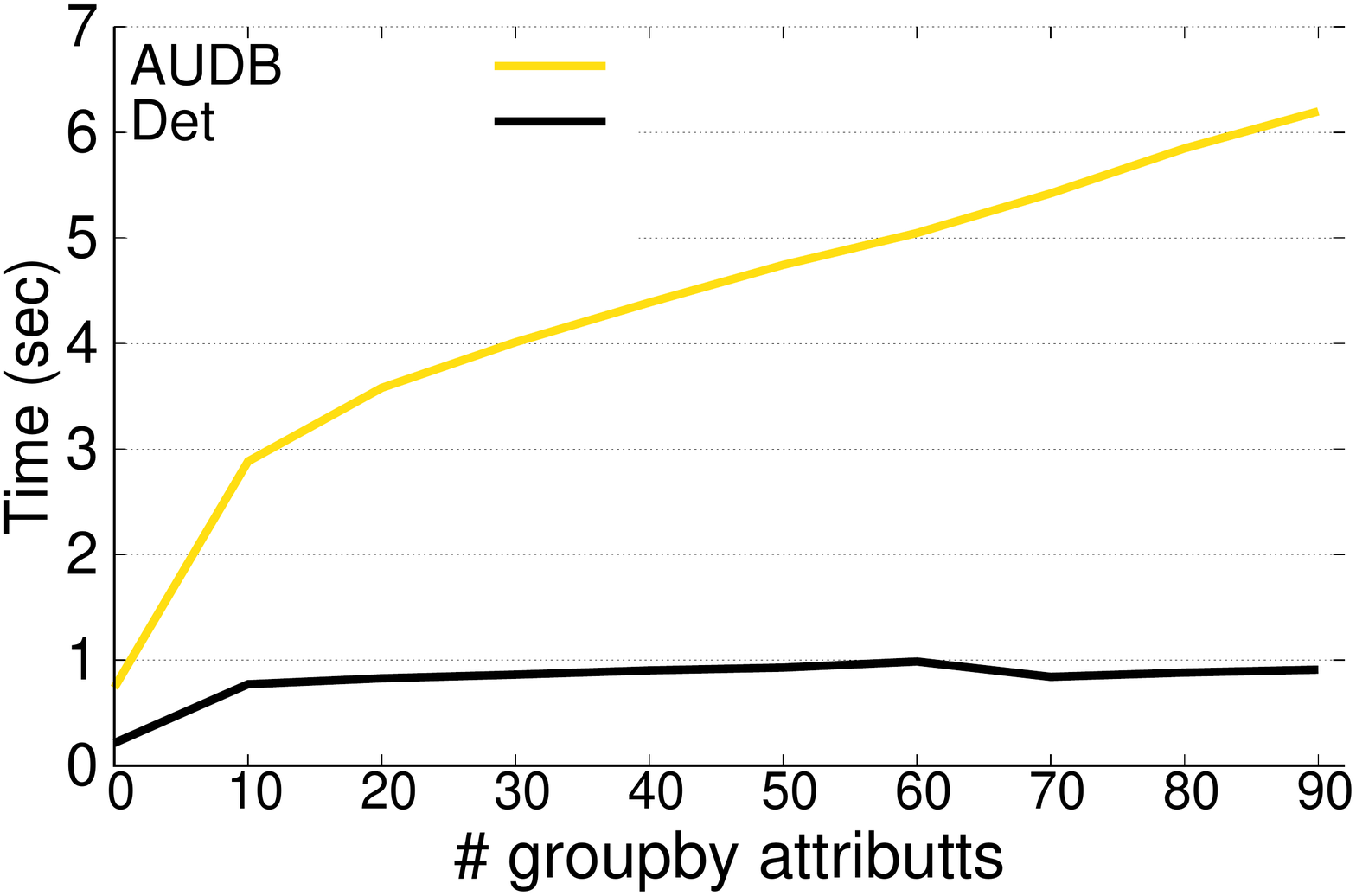}
  \vspace*{-6mm}
  \bfcaption{Varying \#group by}
  \label{fig:micro1}
\end{subfigure}
\begin{subfigure}[b]{.245\textwidth}
	\centering
  \includegraphics[width=0.99\textwidth, trim=0cm 1.5cm 0.5cm 0cm, clip]{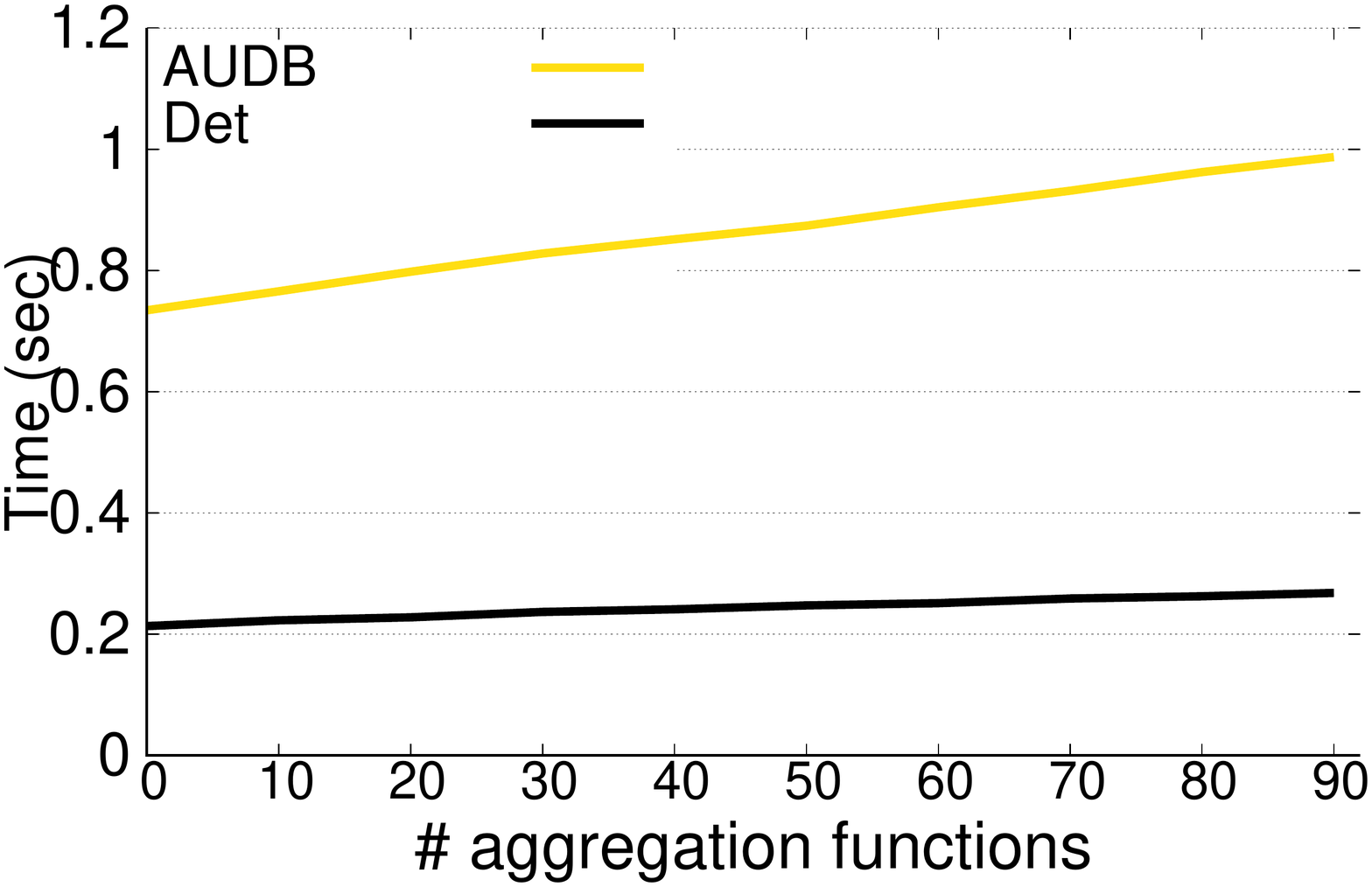}
  \vspace*{-6mm}
  \bfcaption{Varying \#aggregation}
  \label{fig:micro2}
\end{subfigure}
\begin{subfigure}[b]{.245\textwidth}
	\centering
  \includegraphics[width=0.99\textwidth, trim=0cm 1.5cm 0cm 0cm, clip]{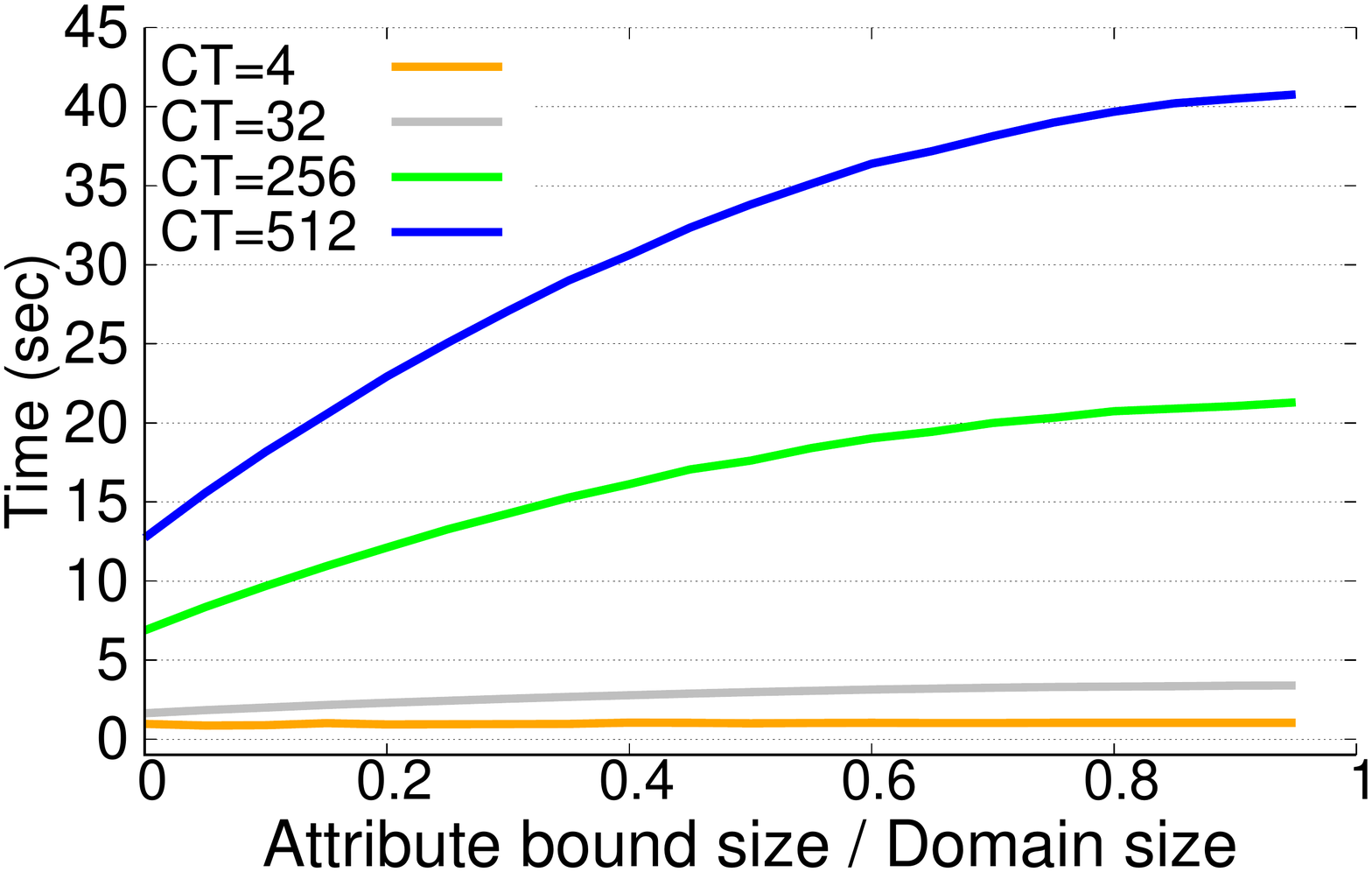}
  \vspace*{-6mm}
  \bfcaption{Varying attribute range}
  \label{fig:micro3}
\end{subfigure}
\begin{subfigure}[b]{.245\textwidth}
	\centering
  \includegraphics[width=0.99\textwidth, trim=0.5cm 5cm 2cm 1cm, clip]{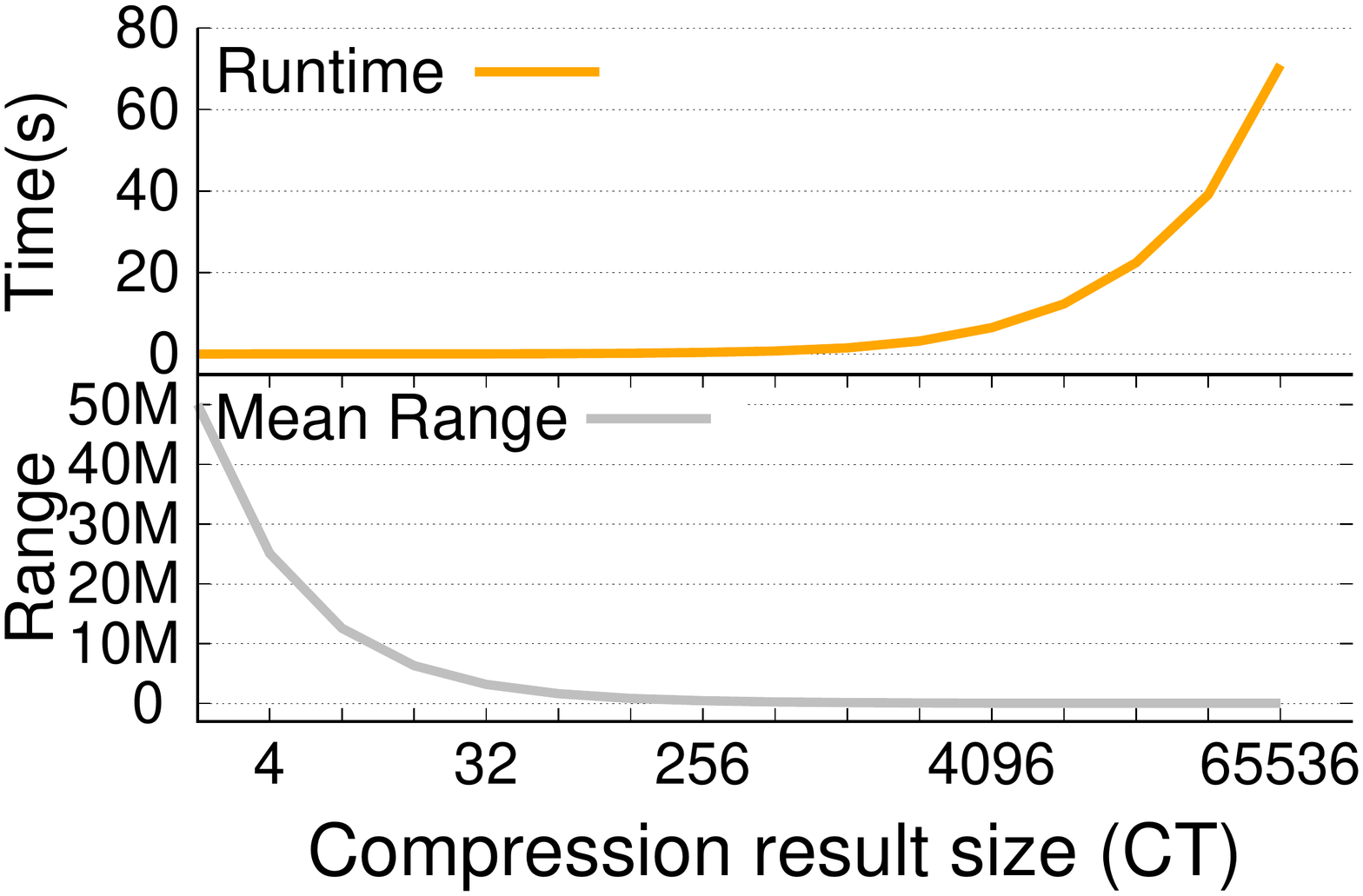}
  \vspace*{-6mm}
  \bfcaption{Varying compression rate}
  \label{fig:micro4}
\end{subfigure}\\[-8mm]
\bfcaption{Aggregation Microbenchmarks - Performance and Accuracy}
\end{minipage}
\end{figure*}

\begin{figure}
  \begin{minipage}{1.0\linewidth}
\begin{subfigure}[b]{.49\textwidth}
	\centering
  \includegraphics[width=0.9\textwidth, trim=0cm 0cm 0cm 3.5cm, clip]{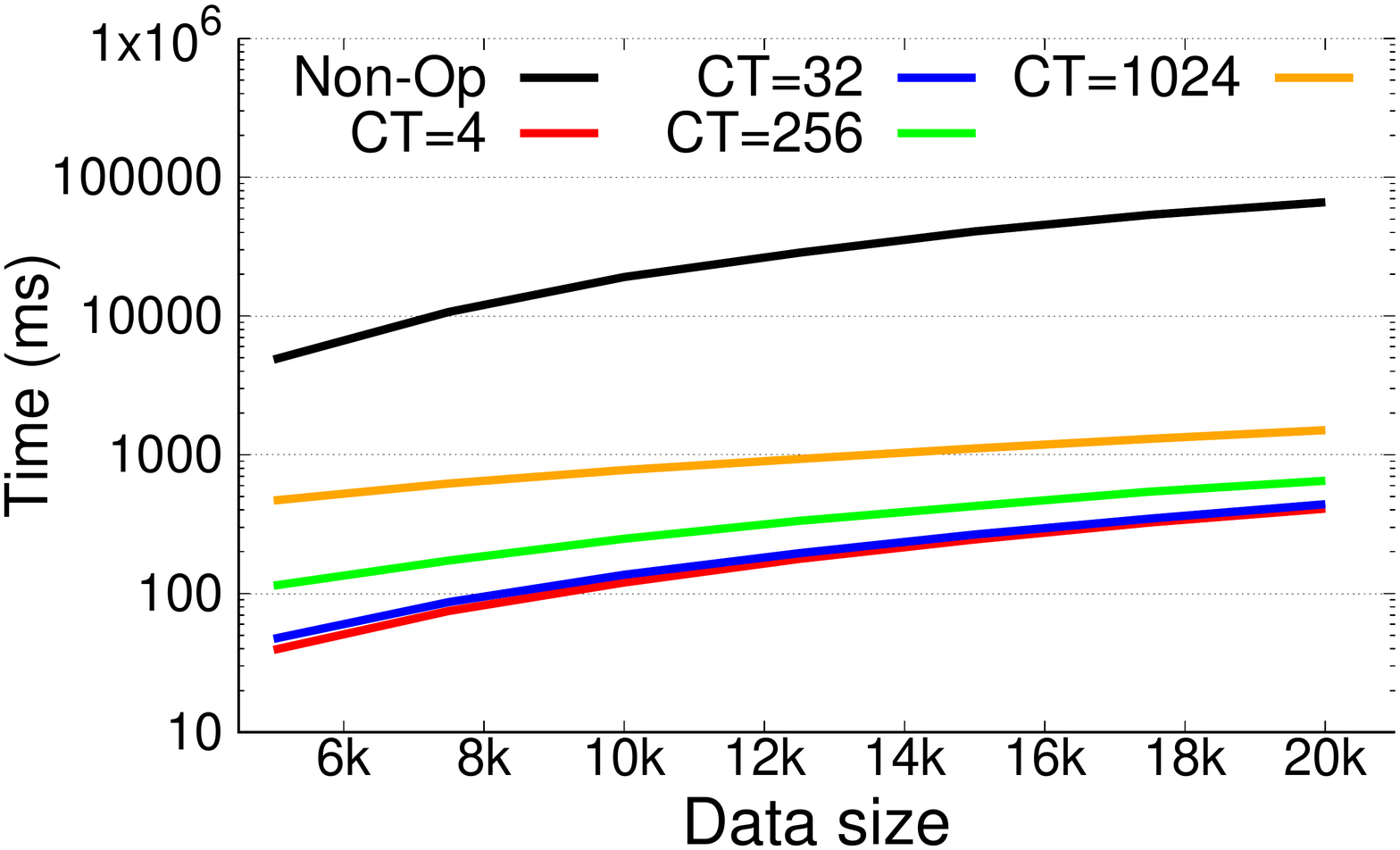}
  \vspace*{-5mm}
  \bfcaption{Runtime}
  \label{fig:micro5}
\end{subfigure}
\begin{subfigure}[b]{.49\textwidth}
	\centering
  \includegraphics[width=0.9\textwidth, trim=0cm 0cm 0cm 3.5cm, clip]{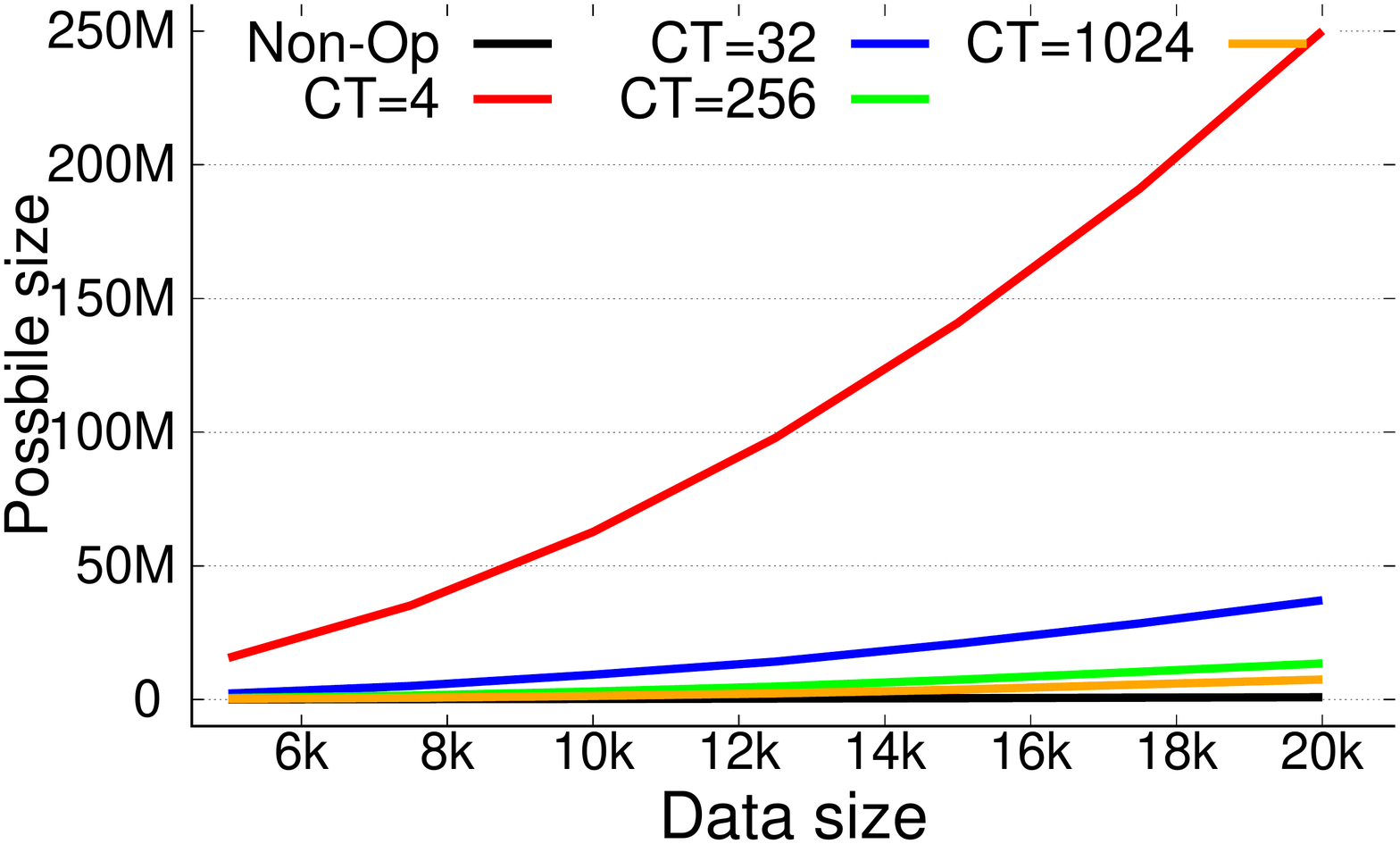}
  \vspace*{-5mm}
  \bfcaption{Attribute-level Accuracy}
  \label{fig:micro6}
\end{subfigure}
\vspace*{-4mm}
  \bfcaption{Join Optimizations - Performance and Accuracy}
  \label{fig:microbench}
\end{minipage}
\trimfigurespacing
\end{figure}


\begin{figure}
  \begin{minipage}{1.0\linewidth}
\begin{subfigure}[b]{.49\textwidth}
	\centering
  \includegraphics[width=0.9\textwidth, trim=0cm 0cm 0cm 3.5cm, clip]{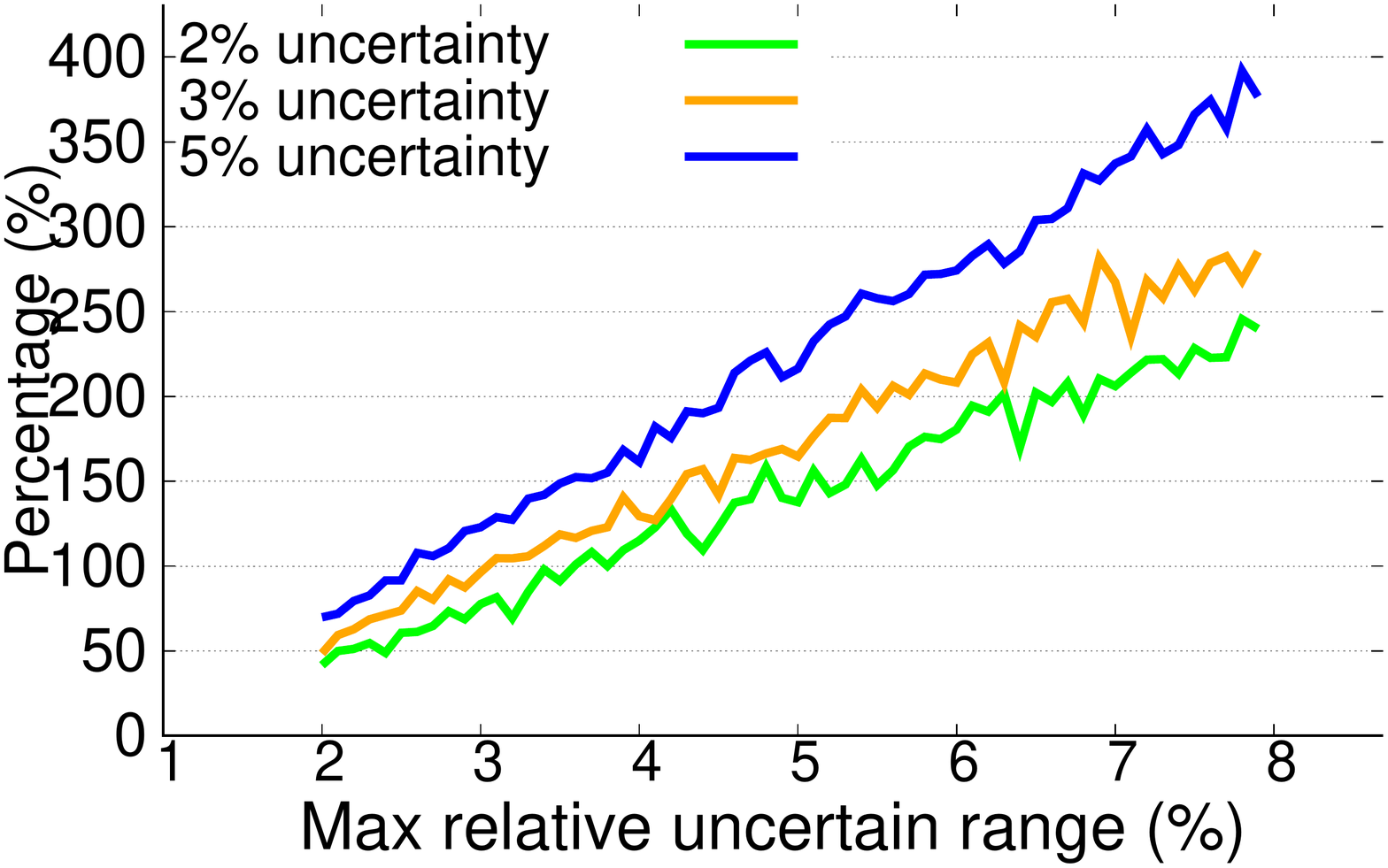}
  \vspace*{-5mm}
  \bfcaption{\revc{Over-grouping}}
  \label{fig:microovergroup}
\end{subfigure}
\begin{subfigure}[b]{.49\textwidth}
	\centering
  \includegraphics[width=0.9\textwidth, trim=0cm 0cm 0cm 3.5cm, clip]{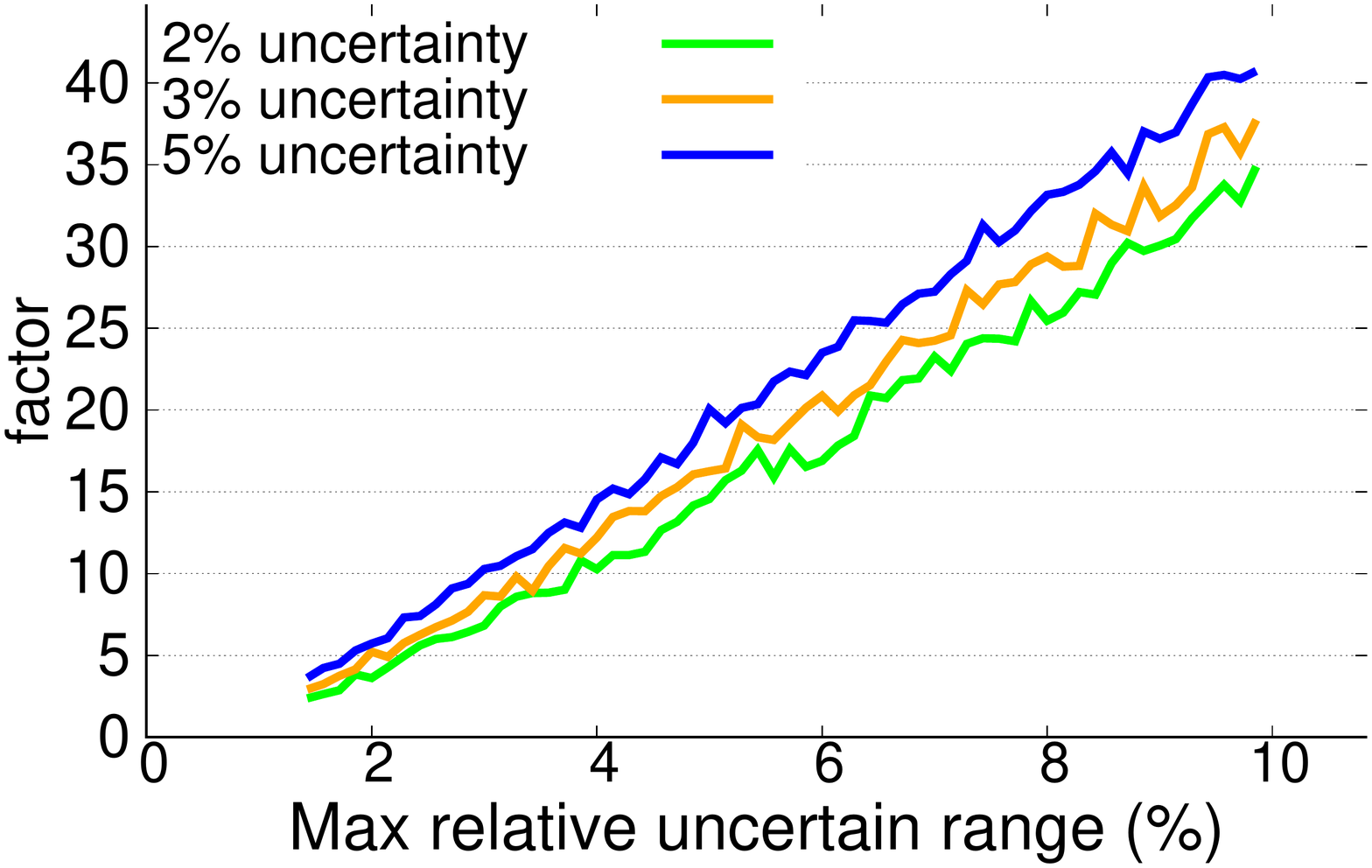}
  \vspace*{-5mm}
  \bfcaption{\revc{Range over-estimation}}
  \label{fig:microoverrange}
\end{subfigure}
\vspace*{-4mm}
  \bfcaption{\revc{Aggregation - varying attribute range}}
  \label{fig:rangeagg}
\end{minipage}
\trimfigurespacing
\end{figure}

\begin{figure}
  \centering
  {\footnotesize
  \begin{tabular}{cc|r|r|r|r}
    \multicolumn{2}{c|}{Comp. Size}	&	\thead{1 join}	&	\thead{2 joins} & \thead{3 joins} & \thead{4 joins} \\ \hline

    \multirow{2}{*}{\thead{4}}	& 3\%	&	0.004         & 0.006           & 0.009          & 0.015	\\
    					& 10\%	& 0.004         & 0.007          & 0.010          & 0.015	\\
    \hline
    \multirow{2}{*}{\thead{16}}	& 3\%	& 0.005         & 0.008           & 0.012          & 0.017	\\
    					& 10\%	& 0.005         & 0.009          & 0.012          & 0.017	\\
    \hline
    \multirow{2}{*}{\thead{64}}	& 3\%	& 0.009         & 0.027           & 0.47          & 0.069		\\
   						& 10\%	& 0.009         & 0.029          & 0.049          & 0.070	\\
    \hline
    \multirow{2}{*}{\thead{256}}	& 3\%	& 0.036         & 0.308           & 0.627          & 0.969	\\
    						& 10\%	& 0.043         & 0.337          & 0.660          & 1.019 \\
    \hline
    \thead{No}	& 3\%          & 0.216         & 1.351           & 6.269          & 29.639	\\
    \thead{Comp.}					& 10\%	& 0.213         & 2.565          & 29.379          & 333.695	\\
    \hline
     \end{tabular}
  }\\[-4mm]
  \bfcaption{\reva{Join query performance (runtime in sec)}}
  \label{fig:micronumberjoin}
  \trimfigurespacing
\end{figure}

\subsection{Micro-benchmarks}
We use a synthetic table with 100 attributes with  uniform random values to evaluate the performance and accuracy 
of our approach.  
\BG{We vary attribute-level uncertainty: \% of uncertain cells and average size of bounds (percentage of an attribute's domain).}


\mypar{Varying number of group-by attributes}
We use an \aggsum \revm{aggregation} with 1 to 99 group-by attributes on a table with 35k rows and 5\% uncertainty.
Our implementation applies an aggregate analog of the join optimization described in \Cref{sec:joinOpt}: possible groups are compressed before being joined with the output groups (see~\cite{techreport}).
 This improves performance when there are fewer result groups.
As~\Cref{fig:micro1} shows, overhead over \methDet is up to a factor of 6 to 7.

\mypar{Varying number of aggregates}
Using a similar query and dataset, and 1 group-by attribute, we vary the number of aggregation functions from 1 to 99.
As \Cref{fig:micro2} shows, the overhead of our approach compared to \methDet varies between a factor of 5 to 6.

\mypar{Compression Trade-off for Aggregation}
We evaluate the tradeoffs between tightness and compression for aggregation using \aggsum aggregation with group-by. \Cref{fig:micro4} shows the runtime overhead of our approach over \methDet when increasing the number of tuples in the compressed pre-aggregation result. The input table has 10\% uncertainty and 10k rows.
For tightness we calculate tight bounds for the aggregation function results for each possible group (a group that exists in at least one world). We then measure for each such group the relative size of our approximate bounds compared to the maximally tight bounds and report the average of this number.
\BG{While tightness is affected by the compression, \revm{even for aggressive compression we get decent bounds.}}

\mypar{Attribute Bound Size}
Next, we vary the average size of attribute-level bounds (same query as above).
We generate tables with 35k rows each and 5\% uncertainty, varying the range of uncertain attribute values from 0\% to 100\% of the attribute's domain.
We measure runtime, varying the number of tuples in the compressed result (\CT) for the pre-aggregation step.  
\Cref{fig:micro3} shows that for more aggressive compression, the runtime of our approach is only slightly affected by the size of attribute-level bounds.
\revc{We also measure the effect of the attribute range on precision. We generate \abbrXDBs with 2\%, 3\%, and 5\% of uncertain tuples (10 alternatives per uncertain tuple) varying  attribute ranges from 1\% to 10\% of the entire value domain. We create an \abbrAUDB from the \abbrXDB (\cite{techreport} details how this is achieved). \Cref{fig:microovergroup,fig:microoverrange} show the percentage of over-grouping for \methAUDB (increase in group size, because of over-estimation of possible group-by attribute values) and relative factor of aggregation result range over-estimation. 
  The range over-estimation grows faster than over-grouping, as it is affected by uncertainty in aggregation function inputs as well as the over-grouping.
}


\mypar{Join Optimizations}
\revm{Next, we evaluate the  
impact of our join  optimization.} 
\Cref{fig:micro5} shows the runtime for a single equality join (log-scale) varying the size of both input relations \revm{from 5k to 20k rows containing 3\% uncertain values ranging over 2\% of the value domain}. The optimized version is between $\sim$ 1 and $\sim$ 2 orders of magnitude  faster depending on the compression rate (i.e., \CT). 
As a simple accuracy measure, \Cref{fig:micro5} shows the number of possible tuples in the  join result. 
\BG{We should state the data size and highlight the large difference!}
\reva{ 
Next, we join tables of 4k rows with 3\% or 10\% uncertainty and vary the number of joins (1 to 4 chained equality joins, i.e., no overlap of join attributes between joins).  
As shown in \Cref{fig:micronumberjoin}, joins without optimization are up to 4 orders of magnitude more expensive, because
of the nested loop joins that are needed for interval-overlap joins and resulting large result relations. 
}

\newcommand{\cellNA}{\textcolor{red}{N.A.}}
\newcommand{\rag}{\textbf{GB}}
\newcommand{\raa}{\textbf{AGG}}
\newcommand{\ras}{\textbf{SPJ}}
\begin{figure}[t]
  \centering
\resizebox{0.9\linewidth}{!}{
{\footnotesize
  \begin{tabular}{ccc|r|c|c|c|l|l}
    \multicolumn{3}{c|}{\thead{Datasets}}                                                                    & \thead{Time}  & \thead{cert.} & \multicolumn{2}{c|}{\thead{attr. bounds}} & \thead{pos.tup.} & \thead{pos.tup.}                                      \\
 \multicolumn{3}{c|}{\thead{\& Queries}}                                                                     & \thead{(sec)} & \thead{tup.}  & \thead{min} & \thead{max} & \thead{by id} & \thead{by val}        \\ \hline
    \multirow{8}{*}{\rotatebox[origin=c]{90}{\parbox[c]{1cm}{\centering \textbf{Netflix}~\cite{netflixdata}                                                                                                                                                               \\(1.9\%, 2.1)}}}
    &               & \methAUDB			& \textbf{0.011}		& 100\%			& 1	 		& 1 		& 100\%  & 100\% \\
	& \realQ{n}{1}  & \methTrio			& 0.900				& 100\%			& 1	 		& 1 		& 100\%  & 100\% \\
	& \ras          & \methMCDB			& 0.049				& \cellNA		& 1 		& 1 		& 99.6\%   & 98.5\% \\
	&          		& \revm{\methUADB}	& \revm{0.006} 		& \revm{100\%}	& \cellNA	& \cellNA	& \revm{99.1\%} & 97.3\%  \\
    \cline{2-9}
	&               & \methAUDB          & \textbf{0.082} 	& 100\%         & 1			& 4       & 	100\%  & 100\% \\
	& \realQ{n}{2}  & \methTrio          & 1.700 			& 100\%         & 1        	& 1    		& 	98.8\%  & 98.0\% \\
	& \rag          & \methMCDB          & 0.118 			& \cellNA       & 1        	& 1    		&	99.9\%  & 97.9\% \\
	&          		& \revm{\methUADB}   & \revm{0.009} 			& \revm{0\%}       & \cellNA       & \cellNA     & \revm{99.3\%} & 95.7\% \\
	\hline
    \multirow{8}{*}{\rotatebox[origin=c]{90}{\parbox[c]{1cm}{\centering \textbf{Crimes}~\cite{crimesdata}                                                                                                                                                              \\(0.1\%, 3.2)}}}
    &               & \methAUDB          	& \textbf{1.58}   	& 100\%         & 1   		& 1		& 100\%  & 100\% \\
	& \realQ{c}{1}  & \methTrio          	& 59.0 				& 100\%          & 1         & 1    & 100\% & 100\% \\
	& \ras          & \methMCDB          	& 6.91				& \cellNA        & 0.6       & 1    & 99.9\% & 92.1\% \\
	&          		& \revm{\methUADB}		& \revm{0.63} 			& \revm{100\%}       & \cellNA         & \cellNA     & \revm{99.9\%} & 87.5\% \\
    \cline{2-9}
	&               & \methAUDB          & \textbf{2.09} 	& 100\%        & 1  		& 1.01 	& 100\% & 100\% \\
	& \realQ{c}{2}  & \methTrio          & 103.1 			& 100\%        & 1       	& 1    	& 100\%  & 100\% \\
	& \rag          & \methMCDB          & 5.24			& \cellNA      & 0.99    		& 0    	& 100\% & $\sim0\%$ \\
	&          & \revm{\methUADB}        & 0.47			& 0\%       & \cellNA      & \cellNA     & \revm{100\%} & $\sim0\%$ \\
	\hline
    \multirow{8}{*}{\rotatebox[origin=c]{90}{\parbox[c]{1cm}{\centering \textbf{Healthcare}~\cite{healthcaredata}                                                                                                                                                                \\(1.0\%, 2.7)}}}
    &               & \methAUDB 		& \textbf{0.179} 	& 99.5\%		& 1	 	& 1		& 100\% 	& 100\% \\
	& \realQ{h}{1}  & \methTrio         & 20.6 			& 100\%     & 1		& 1      	& 100\% & 100\% \\
	& \ras          & \methMCDB         & 0.501				& \cellNA   & 0.4	    	& 1    	& 99.9\% & 87.6\% \\
	&          		& \revm{\methUADB}  & 0.042				& \revm{98.2\%}     & \cellNA     & \cellNA     & \revm{99.3\%} & 65.4\% \\
	\cline{2-9}
	&               & \methAUDB          & \textbf{0.859} 	& 100\%        & 1     		& 45 	& 100\% & 100\% \\
	& \realQ{h}{2}  & \methTrio          & 29.2 			& 100\%        & 1         	& 1    	& 100\%  & 100\% \\
	& \rag          & \methMCDB          & 2.31			& \cellNA      & 0.78      		& 1    	& 100\% & $\sim0\%$ \\
	&          & \revm{\methUADB}        &	0.235		& \revm{0\%}       & \cellNA     & \cellNA     & \revm{100\%} & $\sim0\%$ \\
  \end{tabular}
  }}\\[-3mm]
  \bfcaption{\revm{Real world data - performance and accuracy}}
  \label{fig:realq}
  \trimfigurespacing
\end{figure}

\iftechreport{
\begin{figure}[t]
  \centering
{\footnotesize
  \begin{tabular}{ccc|r|c|c|c|c|l}
    \multicolumn{3}{c|}{\thead{Datasets}}                                                                    & \thead{Time}  & \thead{cert.} & \multicolumn{3}{c|}{\thead{cert. tup. attr. bounds}} & \thead{pos.}                                       \\
 \multicolumn{3}{c|}{\thead{\& Queries}}                                                                     & \thead{(sec)} & \thead{tup.}  & \thead{min}                                          & \thead{median} & \thead{max} & \thead{tup.}        \\ \hline
    \multirow{6}{*}{\rotatebox[origin=c]{90}{\parbox[c]{1cm}{\centering \textbf{Customers}                                                                                                                                                               \\(2.9\%, 3.9)}}}  &               & \methAUDB          & \textbf{0.09} 
                                                                                                             & 100\%         & 1             & 1                                                    & 1              & 100\%                             \\
                                                                                                             & \realQ{c}{1}  & \methTrio          & 0.54                                                 & 100\%          & 1           & 1    & 1    & 100\% \\
                                                                                                             & \ras          & \methMCDB          & 0.18                                                 & \cellNA        & 0.99        & 1    & 1    & 91\%  \\
                                                                                                             &          & \revm{\methUADB}          &                                                   &          &          &      &      &    \\
    \cline{2-9}
                                                                                                             &               & \methAUDB          & \textbf{1.25} 
                                                                                                             & 100\%         & 1             & 1                                                    & 523            & 100\%                             \\
                                                                                                             & \realQ{c}{2}  & \methTrio          & 58.80                                                & 100\%          & 1           & 1    & 1    & 89\%  \\
                                                                                                             & \rag          & \methMCDB          & 1.92                                                 & \cellNA        & 1           & 1    & 1    & 93\%  \\
                                                                                                             &          & \revm{\methUADB}          &                                                   &          &          &      &      &    \\
                                                                                                              \hline
    \multirow{6}{*}{\rotatebox[origin=c]{90}{\parbox[c]{1cm}{\centering \textbf{Treatments}                                                                                                                                                              \\(4.3\%, 4.6)}}} &               & \methAUDB          & \textbf{0.26}    
                                                                                                             & 100\%         & 1             & 1                                                    & 1              & 100\%                             \\
                                                                                                             & \realQ{t}{1}  & \methTrio          & 144.00                                               & 100\%          & 1           & 1    & 1    & 100\% \\
                                                                                                             & \ras          & \methMCDB          & 0.91                                                 & \cellNA        & 0.65        & 1    & 1    & 95\%
                                                                                                             \\
                                                                                                             &          & \revm{\methUADB}          &                                                   &          &          &      &      &    \\
    \cline{2-9}
                                                                                                             &               & \methAUDB          & \textbf{0.27} 
                                                                                                             & 100\%         & 1             & 51                                                   & 901            & 100\%                             \\
                                                                                                             & \realQ{t}{2}  & \methTrio          & 17.70                                                & 100\%          & 1           & 1    & 1    & 78\%  \\
                                                                                                             & \rag          & \methMCDB          & 0.75                                                 & \cellNA        & 0.98        & 0.98 & 1    & 91\%
                                                                                                             \\
                                                                                                             &          & \revm{\methUADB}          &                                                   &          &          &      &      &     \\
                                       \hline
    \multirow{6}{*}{\rotatebox[origin=c]{90}{\parbox[c]{1cm}{\centering \textbf{Employee}                                                                                                                                                                \\(2.8\%, 5.2)}}}   &               & \methAUDB          & \textbf{2.40} 
                                                                                                             & 100\%         & 1             & 1                                                    & 1              & 100\%                             \\
                                                                                                             & \realQ{e}{1}  & \methTrio          & 1524.00                                              & 100\%          & 1           & 1    & 1    & 100\% \\
                                                                                                             & \ras          & \methMCDB          & 5.25                                                 & \cellNA        & 0.95        & 1    & 1    & 82\%
                                                                                                             \\
                                                                                                             &          & \revm{\methUADB}          &                                                   &          &          &      &      &     \\
    \cline{2-9}
                                                                                                             &               & \methAUDB          & \textbf{0.28} 
                                                                                                             & 100\%         & 1             & 1                                                    & 1              & 100\%                             \\
                                                                                                             & \realQ{e}{2}  & \methTrio          & 36.70                                                & 100\%          & 1           & 1    & 1    & 100\% \\
                                                                                                             & \ras          & \methMCDB          & 0.49                                                 & \cellNA        & 0.95        & 1    & 1    & 87\%
                                                                                                             \\
                                                                                                             &          & \revm{\methUADB}          &                                                   &          &          &      &      &     \\
                                                                                                             \hline
    \multirow{6}{*}{\rotatebox[origin=c]{90}{\parbox[c]{1cm}{\centering \textbf{Tax}                                                                                                                                                                     \\(1.7\%, 4.4)}}}        &               & \methAUDB          & \textbf{2.21} 
                                                                                                             & 100\%         & 1             & 1                                                    & 75             & 100\%                             \\
                                                                                                             & \realQ{ta}{1} & \methTrio          & 16.40                                                & 100\%          & 1           & 1    & 1    & 98\%  \\
                                                                                                             & \rag          & \methMCDB          & 3.35                                                 & \cellNA        & 0.99        & 0.99 & 1    & 89\%
                                                                                                             \\
                                                                                                             &          & \revm{\methUADB}          &                                                   &          &          &      &      &     \\
    \cline{2-9}
                                                                                                             &               & \methAUDB          & \textbf{0.08} 
                                                                                                             & 100\%         & 1             & 1                                                    & 1              & 100\%                             \\
                                                                                                             & \realQ{ta}{2} & \methTrio          & 2.77                                                 & 100\%          & 1           & 1    & 1    & 100\% \\
                                                                                                             & \raa          & \methMCDB          & 0.34                                                 & \cellNA        & 0.99        & 0.99 & 0.99 & 100\%
                                                                                                             \\
                                                                                                             &          & \revm{\methUADB}          &                                                   &          &          &      &      &     \\

  \end{tabular}
  }\\[-3mm]
  \bfcaption{Synthethic data - performance and accuracy (OLD)}
  \label{fig:realqold}
  \trimfigurespacing
\end{figure}
}

\subsection{Real World Data}
\label{sec:real-world-data}

For this experiment, we repaired key violations for real world datasets (references shown in \Cref{fig:realq}). 
To repair key violations, we group tuples by their key attributes \revm{so that each group represents all possibilities of a single tuple with the corresponding key value}. For each group, \revm{we randomly pick one tuple for the \abbrBGW and use all tuples in the group to determine its attribute bounds as the  minimum (maximum) value within the group}.
\Cref{fig:realq} shows for each dataset the percentage of tuples with uncertain values and for all such tuples the average number of possibilities. We generated SPJ (\ras) and simple aggregation queries with group-by (\rag) 
for each of these datasets (query types are shown in \Cref{fig:realq}, see~\cite{techreport} for additional details).
%
%
\Cref{fig:realq} shows the runtime for these queries comparing \methAUDB with \methMCDB, \methTrio and \revm{\methUADB}. \methAUDB is significantly faster than \methTrio and consistently outperforms \methMCDB.
\revm{As a comparison point and 
  to calculate our quality metrics}, for each query we calculated the precise set of certain and possible tuples and exact bounds for attribute-level uncertainty in the query result.
\revm{We execute those queries in each system} and report the recall of certain and possible tuples it returns \revm{versus the exact result. Note that for possible tuple recall, we report two metrics. The first ignores attribute-level uncertainty. Possible tuples are grouped by their key (or group-by values for aggregation queries) and we measure the percentage of returned groups (a group is ``covered'' if at least one possible tuple from the group is returned). The second metric just measures the percentage of all possible tuples (without grouping) that are returned.
}
We also measure the tightness of attribute-level bounds for certain rows
by measuring for each tuple the average size of its attribute-level bounds relative to exact bounds. \Cref{fig:realq} shows the minimum and maximum of this metric across all certain result tuples. 
Since \methMCDB relies on samples, it (i) may not return all possible tuples and (ii) calculating bounds for attributes values from the sample, we get bounds that may not cover all possible values. Furthermore, \methMCDB cannot distinguish between certain and possible tuples. For \methTrio the bounds on aggregation results are tight, but \methTrio does not support uncertainty in group-by attributes (no result is returned for a group with uncertain group-by values). As shown in~\Cref{fig:realq}, our attribute-level bounds are close to the tight bounds produced by \methTrio for most of the certain result tuples. \methMCDB does not return all possible aggregation result values (the ones not covered by the samples). Furthermore, we never miss possible tuples like both \methTrio and \methMCDB, and seldomly report a certain tuple as uncertain, while \methMCDB cannot distinguish certain from possible. \revm{\methUADB has performance close to conventional (\abbrBGQP) query processing and outperforms all other methods. However, \abbrUADBs provide no attribute level uncertainty and only contain tuples from the \abbrBGW and, thus,  miss most possible tuples. Furthermore, aggregates over \abbrUADBs will not return any certain answers, as doing so requires having a bound on all possible input tuples for the aggregate and often additionally requires attribute-level uncertainty (the group exists certainly in the result, but the aggregation function result for this group is uncertain).
For aggregates over \abbrUADBs, the range of the attribute bounds is significantly affected by the attribute domain and the aggregation functions used.
$\realQ{n}{2}$ and $\realQ{c}{2}$ use \aggmax and \aggcount, which return a relatively small over-estimation of the actual bounds.
$\realQ{h}{2}$ uses \aggsum, where the larger domain for the attribute over which we are aggregating over, and the combined effect of over-grouping and over-estimation of possible attribute values results in a larger over-estimation.}\iftechreport{\\} \ifnottechreport{\\[-6mm]}


\iftechreport{
\BG{Add subsections for datasets, explain what these queries are doing}
The real world queries are listed below with brief explanation of what the queries are doing:
\begin{lstlisting}
Qn1:
SELECT title, release_year, director
FROM netflix
WHERE release_year < '2017';
\end{lstlisting}
Select all shows with year earlier than 2017.\\
\begin{lstlisting}
Qn2:
SELECT director, MAX(release_year)
FROM netflix
GROUP BY director;
\end{lstlisting}
What is the year of the most current show for each director.\\
\begin{lstlisting}
Qc1:
SELECT date, block, District 
FROM crimes 
WHERE Primary_Type='HOMICIDE' AND Arrest='False';
\end{lstlisting}
What is the date, block and district of all HOMICIDE crimes that are not arrested.\\
\begin{lstlisting}
Qc2:
SELECT year, count(*) 
FROM crimes 
GROUP BY year; 
\end{lstlisting}
Count the number of crimes for each year.\\
\begin{lstlisting}
Qh1:
SELECT Facility_Name, Measure_Name, score 
FROM healthcare_c 
WHERE state != 'TX' 
AND state !='CA' 
AND measure_id='HAI_1_SIR';
\end{lstlisting}
What is the facility name, measure name and score for all records that measuring HAI\_1\_SIR except state TX and CA.\\
\begin{lstlisting}
Qh2:
SELECT sum(score) 
FROM healthcare_c 
GROUP BY Facility_Name;
\end{lstlisting}
What is the total score for each facility.\\


}

\iftechreport{
	Detailed configurations for each microbench test is listed in figure~\ref{fig:config}.
	\begin{figure*}[t]
  		\centering
  		\begin{tabular}{c|c|c|c|c|c|c}
   		 \thead{Test} & \thead{\#Rows}  & \thead{Domain} & \thead{Range} & \thead{Uncert.\%} & \thead{CT} & \thead{Query}     \\ \hline
   		Groupby  & 35k &	[1,100]   & 5 & 5\% & $2^5$ & SELECT SUM(a0) FROM t GROUP BY [...]\\
   		Aggregation  & 35k &	[1,100]   & 5 & 5\% & $2^5$ & SELECT [...] FROM t GROUP BY a0\\
   		Range  & 35k &	[1,100k]  & 5k-100k & 10\% & $2^{2,5,8,9}$ & SELECT a0,SUM(a1) FROM t GROUP BY a0\\
   		Compression  & 10k &	[1,10k]  & 20 & 2\% & $2^{1 \sim 16}$ & SELECT a0,SUM(a1) FROM t GROUP BY a0\\
   		Join  & 5k-20k &	[1,1k]  & 15 & 3\% & non,$2^{2,5,8,10}$ & SELECT * FROM t1 JOIN t2 ON t1.a0 = t2.a0\\
   		Multi-join  & 4k &	[1,4k]  & 300 & 3\%,10\% & non,$2^{2,4,6,8}$ & ... (t1 JOIN t2 ON t1.a1 = t2.a0) JOIN t3 on t2.a1=t3.a0 ... \\
  		\end{tabular}
 		 \bfcaption{Microbenchmark Configurations}
 		 \label{fig:config}
 	\end{figure*}
	A brief data description is listed in figure~\ref{fig:datadesc}.
	\begin{figure*}[t]
  		\centering
  		\begin{tabular}{c|c|c|c}
   		 \thead{Dataset} & \#columns & \#rows & \thead{source} \\ \hline
   		Netflix  		& 12	& >6K	& \url{https://www.kaggle.com/shivamb/netflix-shows}  \\
   		Crimes  		& 22	& >1.4M	& \url{https://www.kaggle.com/currie32/crimes-in-chicago}  \\
   		Healthcare  	& 15	& >171K	& \url{https://data.medicare.gov/data/hospital-compare}  \\
  		\end{tabular}
 		 \bfcaption{Realworld Data}
 		 \label{fig:datadesc}
 	\end{figure*}
}


\section{Conclusions}
\label{sec:conclusions}

We present \termUAADBs (\abbrUAADBs), an efficient scheme for approximating certain and possible answers for 
full relational algebra and aggregation. 
\BG{Towards this goal we introduce \termUAADBs (\abbrUAADBs) which encode both attribute- and tuple-level uncertainty as bounds on the attribute values and multiplicity of tuples.}
Our approach stands out in that it is (i) more general in terms of supported queries than most past work, (ii) has guaranteed \ptime data complexity, and (iii) compactly encodes over-approximations of incomplete databases.
\BG{We demonstrate the efficiency of our approach over a wide range of datasets and queries, including TPC-H queries.}
In future work, we will investigate extensions of this model for queries with ordering (top-k queries and window functions). \BG{DELETED: and investigate how to compress \abbrUAADB bounds to improve performance.}
We will also explore how to manage non-ordinal categorical attributes.
\BG{Furthermore, we will investigate how to reduce the storage and, thus, also runtime overhead of \abbrUAADBs by compressing bounds and by only recording bounds for uncertain attribute values.
Since tuples in \abbrAUDBs can be interpreted as boxes spanned by the bounds on attribute values, we will explore how techniques from spatial databases can be adapted for our model.}


\bibliographystyle{abbrv}
\bibliography{tech_report.bib}

%

\end{document}